\documentclass[10pt,english]{amsart}

\usepackage{babel}
\usepackage[latin1]{inputenc}
\usepackage[T1]{fontenc}
\usepackage[babel=true]{csquotes}
\usepackage{amssymb}
\usepackage{a4wide}
\usepackage{amsthm}
\usepackage[all]{xy}
\usepackage{graphicx}
\usepackage{amsmath}
\usepackage{array}
\usepackage{vaucanson-g}
\usepackage{color}
\usepackage{subfigure}
\usepackage{hhline,arydshln}
\usepackage{enumerate}
\usepackage{soul}
\usepackage{mathtools}
\usepackage{caption}
\usepackage{hyperref}

\theoremstyle{plain} \newtheorem{proposition}{Proposition}[section]
\theoremstyle{plain} \newtheorem{theorem}[proposition]{Theorem}
\theoremstyle{plain} \newtheorem{corollary}[proposition]{Corollary}
\theoremstyle{plain} \newtheorem{lemma}[proposition]{Lemma}
\theoremstyle{remark} \newtheorem{remark}[proposition]{Remark}
\theoremstyle{definition} \newtheorem{example}[proposition]{Example}
\theoremstyle{definition} \newtheorem{definition}[proposition]{Definition}
\theoremstyle{plain}	
\theoremstyle{plain}	\newtheorem{fact}[proposition]{Fact}

\makeatletter
\renewcommand{\thefigure}{\ifnum \c@section>\z@ \thesection.\fi
 \@arabic\c@figure}
\@addtoreset{figure}{section}
\makeatother

\makeatletter
\renewcommand{\thetable}{\ifnum \c@section>\z@ \thesection.\fi
 \@arabic\c@table}
\@addtoreset{table}{section}
\makeatother

\def\A{\mathcal{A}}

\def\G{\mathcal{G}}
\def\S{\mathcal{S}}
\def\k{\mathfrak{K}}

\def\N{\mathbb{N}}
\def\Z{\mathbb{Z}}

\def\bw{\mathbf{w}}
\def\bx{\mathbf{x}}

\newcommand{\fac}[2]{{\rm Fac}_{#1}(#2)}
\newcommand{\card}[1]{{\rm Card}(#1)}
\newcommand{\alp}{{\rm Alph}}
\newcommand{\CPref}{{\rm CP}}
\newcommand{\CSuff}{{\rm CS}}

\definecolor{light-gray}{gray}{0.7}

\def\restriction#1#2{\mathchoice
              {\setbox1\hbox{${\displaystyle #1}_{\scriptstyle #2}$}
              \restrictionaux{#1}{#2}}
              {\setbox1\hbox{${\textstyle #1}_{\scriptstyle #2}$}
              \restrictionaux{#1}{#2}}
              {\setbox1\hbox{${\scriptstyle #1}_{\scriptscriptstyle #2}$}
              \restrictionaux{#1}{#2}}
              {\setbox1\hbox{${\scriptscriptstyle #1}_{\scriptscriptstyle #2}$}
              \restrictionaux{#1}{#2}}}
\def\restrictionaux#1#2{{#1\,\smash{\vrule height .8\ht1 depth .85\dp1}}_{\,#2}} 

\renewcommand{\phi}{\varphi}
\newcommand{\emptyword}{\varepsilon}

\newcommand{\definir}[1]{\index{#1}\textit{#1}}
\makeindex

\title[An $S$-adic characterization of minimal subshifts with $1 \leq p(n+1) - p(n) \leq 2$]{An $S$-adic characterization of minimal subshifts with first difference of complexity $1 \leq p(n+1) - p(n) \leq 2$}

\begin{document}


\author{Julien Leroy}	


\address{Mathematics Research Unit, FSTC, University of Luxembourg, 6, rue Coudenhove-
Kalergi, L-1359 Luxembourg, Luxembourg}	

\email{julien.leroy[at]uni.lu}


\begin{abstract}
In [Ergodic Theory Dynam. System, 16 (1996)  663--682], S. Ferenczi proved that any minimal subshift with first difference of complexity bounded by 2 is $S$-adic with $\card S \leq 3^{27}$. In this paper, we improve this result by giving an $\S$-adic charaterization of these subshifts with a set $\S$ of 5 morphisms, solving by this way the $S$-adic conjecture for this particular case.
\end{abstract}





\maketitle

\section{Introduction}

A classical tool in the study of sequences (or infinite words) with values in an alphabet $A$ is the \textit{complexity function} $p$ that counts the number $p(n)$ of words of length $n$ that appear in the sequence. Thus this function allows to measure the regularity in the sequence. For example, it allows to describe all ultimately periodic sequences as exactly being  those for which $p(n) \leq n$ for some length $n$~\cite{Morse-Hedlund}. By extension, this function can obviously be defined for any language or any symbolic dynamical system (or \textit{subshift}). For surveys over the complexity function, see~\cite{Allouche_survey,Ferenczi_survey} or~\cite[Chapter~4]{CANT}.

The complexity function can also be used to define the class of \textit{Sturmian sequences}: it is the family of aperiodic sequences with minimal complexity $p(n) = n +1$ for all lengths $n$. Those sequences are therefore defined over a binary alphabet (because $p(1) =2$) and a large literature is devoted to them (see~\cite[Chapter~1]{Lothaire} and~\cite[Chapter~6]{Pytheas-Fogg} for surveys). In particular, these sequences admit several equivalent definitions such as natural codings of rotations with irrational angle or aperiodic balanced sequences. Moreover, it is well known~\cite{Morse-Hedlund} that the subshifts they generate can be obtained by successive iterations of two morphisms (or substitutions) $R_0$ and $R_1$ defined (when the alphabet $A$ is $\{0,1\}$) by $R_0(0) = 0$, $R_0(1) = 10$, $R_1(0) = 01$ and $R_1(1) = 1$. To generate not all Sturmian subshifts but all sturmian sequences it is necessary~\cite{Mignosi-Seebold,Berthe-Holton-Zamboni} to consider two additional morphisms $L_0$ and $L_1$ defined by $L_0(0) = 0$, $L_0(1) = 01$, $L_1(0) = 10$ and $L_1(1) = 1$. In general, a sequence (or subshift) obtained by such a method, that is, obtained by successive iterations of morphisms belonging to a set $S$, is called an \textit{$S$-adic} sequence (or subshift), accordingly to the terminology of adic systems introduced by A. M.  Vershik~\cite{Vershik-Livshits}.


Beside Sturmian sequences, many other families of sequences are usually studied in the literature. Among them one can find generalizations of Sturmian sequences, such as codings of rotations~\cite{Didier_codage_combinatoire,Rote} or of intervals exchanges~\cite{Rauzy_iet,Ferenczi-Zamboni_language}, Arnoux-Rauzy sequences~\cite{Arnoux-Rauzy} 
and episturmian sequences~\cite{Glen-Justin}. One can also think about \textit{automatic sequences}~\cite{Allouche-Shallit} linked to automata theory and morphisms.

An interesting point is that all these mentioned sequences have a linear complexity, i.e., there exist a constant $D$ such that for all positive integers $n$, $p(n) \leq Dn$. In addition, we can usually associate  a (generally finite) set $S$ of morphisms to these sequences in such a way that they are $S$-adic. It is then natural to ask whether there is a connection between the fact of being $S$-adic and the fact of having a linear complexity. Both notions cannot be equivalents since, thanks to Pansiot's work~\cite{Pansiot}, there exist purely morphic sequences with a quadratic complexity. However, we can imagine a stronger notion of $S$-adicity that would be equivalent to having a linear complexity. In other words, we would like to find a condition $C$ such that \textit{a sequence has a sub-linear complexity if and only if it is $S$-adic satisfying the condition $C$}. This problem is called the \textit{$S$-adic conjecture} and is due to B . Host. Up to now, we have no idea about the nature of the condition $C$. It may be a condition on the set $S$ of morphisms, or a condition on the way in which they must occur in the sequence of morphisms. There exist examples~\cite{Durand-Leroy-Richomme} supporting the idea that the answer should be a combination of both, supporting the difficulty of the conjecture.

A difficulty of the conjecture is that all known $S$-adic representations of families of sequences strongly depend on the nature of these sequences which makes general properties difficult to extract. In addition, the characterization of everywhere growing purely morphic sequences with linear complexity (obtained by Pansiot) can only be generalized into a sufficient condition for $S$-adic sequences~\cite{Durand_LR,Durand_corrigentum} and many (\textit{a priori} natural) conditions over $S$-adic sequences are even not sufficient to guarantee a linear complexity~\cite{Durand-Leroy-Richomme}. Nevertheless, S. Ferenczi~\cite{Ferenczi} provided a general method that, given any uniformly recurrent sequence with linear complexity, produces an $S$-adic representation with a finite set $S$ of morphisms and such that all images of letters under the product of morphisms have length growing to infinity. By a refinement of Ferenczi's proof, the author~\cite{Leroy} managed to highlight a few more necessary conditions of these $S$-adic representations, but which unfortunately were not sufficient to ensure linear complexity. A different (although closely linked) proof of that result can also be obtained using a generalization of return words~\cite{Leroy-Richomme}, a tool that has been helpful to find an $S$-adic characterization of the family of linearly recurrent sequences~\cite{Durand_LR,Durand_corrigentum} (that includes the primitive substitutive sequences~\cite{Durand_UR,Durand-Host-Skau}).

In Ferenczi-s proof, the algorithm that produces the morphisms is based on an extensive use of \textit{Rauzy graphs}. These graphs are powerful tools to study combinatorial properties of sequences or subshifts. For example, they are the basis of a strong Cassaigne's result proving that a sequence has a sub-linear complexity if and only if the first difference of its complexity $p(n+1)-p(n)$ is bounded (see~\cite{Cassaigne_big_thm}). They also allowed T. Monteil~\cite[Chapter~7]{CANT}, \cite[Chapter~5]{Monteil_phd} to improve a result due to M. Boshernitzan~\cite{Boshernitzan} by giving a better bound on the number of ergodic invariant measures of a subshift. However, these graphs are usually difficult to compute as soon as the complexity exceeds a very low level. For this reason, the extraction of properties of the $S$-adic representation from these graphs is usually hard. Anyway, applying these methods to subshifts for which the difference of complexity $p(n+1)-p(n)$ is no more than to 2 for every $n$, Ferenczi succeeded to prove that the number of morphisms built in such a way is less than $3^{27}$.

In this paper, we strongly improve this bound and show the existence of a set $\S$ of 5 morphisms such that any minimal subshift with first difference of complexity bounded by 2 is $\S$-adic. Furthermore, we give necessary and sufficient conditions on sequences in $\S^{\N}$ to be an $\S$-adic representation of such a subshift. In other words, we solve the $S$-adic conjecture for this particular case. This characterization contains the subshifts with complexity $2n$, some of which were studied by G. Rote~\cite{Rote}.

As a corollary, the obtained $\S$-adic representations provide Bratteli-Vershik representations of the concerned subshifts. Historically, O. Bratteli~\cite{Bratteli} introduced infinite graphs (subsequently called \textit{Bratteli diagrams}) partitioned in levels in order to approximate $C^*$-algebras. With other motivations, Vershik~\cite{Vershik} associated dynamics (\textit{adic transformations}) to these diagrams by introducing a lexicographic ordering on the infinite paths of the diagrams. This ordering is induced by a partial order on the arcs between two consecutive levels, it can then be defined by an adjacent matrix between the two considered levels and thus by a morphism. For more details, see~\cite[Chapter~6]{CANT} and see~\cite{Wargan} for the link between Bratteli diagrams and $S$-adic systems.

By a refinement of Vershik's constructions, the authors of~\cite{Herman-Putnam-Skau} have proved that any minimal Cantor system is topologically isomorphic to a Bratteli-Vershik system (Vershik already obtained this result in~\cite{Vershik} in a measure theoretical context). These Bratteli-Vershik representations are helpful in dynamics, mainly with problems about recurrence. But, being given a minimal Cantor system, it is generally difficult to find a \enquote{canonical} Bratteli-Vershik representation (see~\cite{Durand-Host-Skau} for examples). However, Ferenczi proved that for minimal subshift with sub-linear complexity, the number of morphisms read on the associated Bratteli diagram (in a measure theoretical context) is finite~\cite{Ferenczi}. In particular, he obtained an upper bound on the rank of these systems and proved that they cannot be strongly mixing. In addition, Durand showed that, in the case of linearly recurrent subshifts, the morphisms appearing in the $S$-adic representation are exactly those read on the Bratteli diagram. Furthermore, unlike in Ferenczi's result, the subshift is topologically conjugated to the Bratteli-Vershik system. Similarly to that last case, the $\S$-adic representations obtained in this paper are exactly those that can be read on a Bratteli-Vershik system which is topologically conjugated to the $\S$-adic subshift~\cite{Durand-Leroy}.

The paper is organized as follows. Section~\ref{section: Backgrounds} contains all needed definitions and backgrounds. Section~\ref{section: adicity of minimal subshift} concerns $S$-adic representations of minimal subshift. We define the tools that are needed for the announced $\S$-adic characterization in a more general case. In Section~\ref{section: $S$-adicity of subshifts with complexity $2n$}, we start a detailed description of Rauzy graphs corresponding to minimal subshifts with first difference of complexity bounded by 2. This allows us to explicitly compute all needed morphisms and we show that they all can be decomposed into compositions of only five morphisms. In Section~\ref{section: caraterisation}, we improve the results obtained in Section~\ref{section: $S$-adicity of subshifts with complexity $2n$} by studying even more the sequences of possible evolutions of Rauzy graphs. This allows us to obtain an $S$-adic characterization, hence the condition $C$ of the conjecture for this particular case.

\section{Backgrounds}
\label{section: Backgrounds}

\subsection{Words, sequences and languages}
\label{subsection: words and sequences}

We assume that readers are familiar with combinatorics on words; for basic (possibly omitted)
definitions we follow \cite{Lothaire1983,Lothaire,CANT}.

Given an \definir{alphabet} $A$, that is a finite set of symbols called \definir{letters}, we denote by $A^*$ the set of all finite words over $A$ (that is the set of all finite sequences of elements of $A$).
As usual, the \definir{concatenation} of two words $u$ and $v$ is simply denoted $uv$. 
It is well known that the set $A^*$ embedded with the concatenation operation is a free monoid with neutral element $\emptyword$, the \definir{empty word}.

For a word $u = u_1 \cdots u_{\ell}$ of length $|u|=\ell$, we write $u[i,j] = u_i \cdots u_j$ for $1 \leq i \leq j \leq \ell$. A word $v$ is a \textit{factor} of a word $u$ (or \textit{occurs at position $i$} in $u$) if $u[i,j] = v$ for some integers $i$ and $j$. It is a \textit{prefix} (resp. \textit{suffix}) if $i = 1$ (resp. $j = |u|$). The \definir{language} of $u$ is the set $\fac{}{u}$ of all factors of $u$;

A \definir{two-sided sequence} (resp. \definir{one-sided sequence}) is an element of $A^{\N}$ (resp. $A^{\Z}$); sequences will be denoted by bold letters. 
When no information are given, \textit{sequence} means \textit{two-sided sequence}.
With the product topology of the discrete topology over $A$, $A^{\Z}$ and $A^{\N}$ are compact metric spaces.

We extend the notions of factor, prefix and suffix to two-sided sequences (resp. one-sided sequences) putting $i,j \in \Z$ (resp. $i,j \in \N$), $i \leq j$, $i = -\infty$ (resp. $i=0$) for prefixes and $j=+\infty$ for suffixes.

Let $u$ be a non-empty finite word over $A$. We let $u^{\omega}$ (resp. $u^{\infty}$) denote the one-sided sequence  $uuu\cdots$ (resp. two-sided sequence $\cdots uuu.uuu \cdots$) composed of consecutive copies of $u$. A (one-sided) sequence $\bw$ is \definir{periodic} if there is a word $u$ such that $\bw \in \{ u^{\omega}, u^{\infty} \}$.

A sequence $\bw$ is \definir{recurrent} if every factor occurs infinitely often. It is \definir{uniformly recurrent} if it is recurrent and every factor occurs with bounded gaps, \textit{i.e.}, if $u$ is a factor of $\bw$, there is a constant $K$ such that for any integers $i,j$ such that $\bw[i,i+|u|-1]$ and $\bw[j,j+|u|-1]$ are two consecutive occurrences of $u$ in $\bw$, then $|i-j| \leq K$.

\subsection{Subshifts and minimality}
\label{subsection: Subshifts and minimality}

A \definir{subshift} over $A$ is a couple $(X,\restriction{T}{X})$ (or simply $(X,T)$) where $X$ is a closed $T$-invariant ($T(X) = X$) subset of $A^{\Z}$ and $T$ is the \definir{shift transformation} $T : A^{\Z} \rightarrow A^{\Z}, \ (\bw_i)_{i \in \Z} \mapsto (\bw_{i+1})_{i \in \Z}$.

The \definir{language} of a subshift $X$ is the union of the languages of its elements and we denote it by $\fac{}{X}$.

Let $\mathbf{w}$ be a sequence (or a one-sided sequence) over $A$. We denote by $X_{\bw}$ the set $\{ \bx \in A^{\Z} \mid \bx[i,j] \in \fac{}{\bw} \text{ for all } i,j \in \Z, i \leq j \}$. Then, $(X_{\mathbf{w}},T)$ is a subshift called the \definir{subshift generated by $\mathbf{w}$}. For $\bw \in A^\Z$, we have $X_{\bw} = \overline{\{ T^n(\bw) \mid n \in \Z \}}$

A subshift $(X,T)$ is \textit{periodic} whenever $X$ is finite. Observe that in this case, $X$ contains only periodic sequences. It is \textit{minimal} if the only closed $T$-invariant subsets of $X$ are $X$ and $\emptyset$, or, equivalently, if for all $\bw \in X$, we have $X = X_{\bw}$. We also have that $(X_{\bw},T)$ is minimal if and only if $\mathbf{w}$ is uniformly recurrent.

In the sequel, we will mostly consider minimal subshifts.

\subsection{Factor complexity and special factors}

The \definir{factor complexity} of a subshift $X$ is the function $p_X : \N \to \N$ that counts the number of words of each length that occur in elements of $X$, \textit{i.e.}, $p_X(n) = \card{\fac{n}{X}}$, where $\fac{n}{X} = \fac{}{X} \cap A^n$. 

The first difference of complexity $s(n) = p(n+1)-p(n)$ is closely related to special factors~\cite{Cassaigne_resume}. 
A word $u$ in $\fac{}{X}$ is a \definir{right special factor} (resp. a \definir{left special factor}) if there are two letters $a$ and $b$ in $A$ such that $ua$ and $ub$ (resp. $au$ and $bu$) belong to $\fac{}{X}$. 
For $u$ in $\fac{}{X}$, if $\delta^+u$ (resp. $\delta^-u$) denotes the number of letters $a$ in $A$ such that $ua$ (resp. $au$) is in $\fac{}{X}$ we have
\begin{eqnarray}
\label{eq 1}
	p_X(n+1)-p_X(n) & = & 	\sum_{\substack{u \in \fac{n}{X} \\ u \text{ right special}}} \underbrace{(\delta^+u-1)}_{\geq 1} \\
\label{eq 2}		& = &	\sum_{\substack{u \in \fac{n}{X} \\ u \text{ left special}}} \underbrace{(\delta^-u-1)}_{\geq 1}
\end{eqnarray}

It is well known~\cite{Morse-Hedlund} that a subshift is aperiodic if and only if $p_X(n) \geq n+1$ for all $n$, or, equivalently, if there is at least one right (resp. left) special factor of each length.

The second difference of complexity $s(n+1)-s(n)$ is related to \definir{bispecial factors}, \textit{i.e.}, to facotrs that are both left and right special. Indeed, if $u$ is a bispecial factor in $\fac{}{X}$, its \definir{bilateral order} is $m(u) = \# (\fac{}{X} \cap A u A) - \delta^+u - \delta^-u +1$ and we have\footnote{Observe that for non-bispecial factors $u$, we have $m(u)=0$.}
\[
	s_X(n+1) - s_X(n) = \sum_{u \in \fac{n}{X}} m(u).
\]
A bispecial factor $u$ is said to be \definir{weak} (resp. \definir{neutral}, \definir{strong}) whenever $m(u)<0$ (resp. $m(u)=0$, $m(u)>0$).

\subsection{Morphisms and $S$-adicity}

Given two alphabets $A$ and $B$, a free monoid morphism, or simply \definir{morphism} $\sigma$, is a map from $A^*$ to $B^*$ such that $\sigma(uv) = \sigma(u) \sigma(v)$ for all words $u$ and $v$ over $A$ (note this implies $\sigma(\emptyword) = \emptyword$). 
It is well known that a morphism is completely determined by the images of letters.

When a morphism is not \definir{erasing}, that is the images of letters are never the empty word, the notion of morphism extends naturally to (one-sided) sequences.
If $A = B$ and if there is a letter $a \in A$ such that $\sigma(a) \in a A^*$, then $\sigma^{\omega}(a) = \lim_{n \to +\infty} \sigma^n(a^{\omega})$ is a one-sided sequence which is a fixed point of $\sigma$. If there is also a letter $b$ such that $\sigma(b) \in A^* b$, then the two-sided sequence $\sigma^{\omega}(\prescript{\omega}{}b.a^\omega) = \lim_{n \to +\infty} \sigma^n(\cdots bbb.aaa \cdots)$ is also a fixed point of $\sigma$.

Let $\bw$ be a sequence over $A$. An \definir{adic representation} of $\mathbf{w}$ is given by a sequence $(\sigma_n: A_{n+1}^* \rightarrow A_n^*)_{n \in \N}$ of morphisms and a sequence $(a_n)_{n \in \N}$ of letters, $a_i \in A_i$ for all $i$ such that $A_0 = A$, $\lim_{n \rightarrow +\infty} |\sigma_0 \sigma_1 \cdots \sigma_n (a_{n+1})| = +\infty$ and
\[
	\bw = \lim_{n \rightarrow +\infty} \sigma_0 \sigma_1 \cdots \sigma_n (a_{n+1}^{\infty}).
\]
The sequence $(\sigma_n)_{n \in \N}$ is the \definir{directive word} of the representation. Let $S$ be a set of morphisms. We say that $\bw$ is \definir{$S$-adic} (or that $\bw$ is \definir{directed by $(\sigma_n)_{n \in \N}$})  if $(\sigma_n)_{n \in \N} \in S^{\N}$. In the sequel, we will say that a sequence $\bw$ is $S$-adic whenever there is a set $S$ of morphisms such that $\bw$ admits an $S$-adic representation. We say that a subshift $(X,T)$ is \definir{$S$-adic} if it is the subshift generated by an $S$-adic sequence.

Let $(\sigma_n)_{n \in \N}$ be a sequence of morphisms. The sequence of morphisms $(\tau_n: B_{n+1}^* \rightarrow B_n^*)_{n \in \N}$ is a \definir{contraction} of $(\sigma_n: A_{n+1}^* \rightarrow A_n^*)_{n \in \N}$ if there is a sequence of integers $(i_n)_{n \in \N}$ such that for all $n$ in $\N$, $B_n = A_{i_n}$ and 
\[
	\tau_n = \sigma_{i_n} \sigma_{i_n+1} \cdots \sigma_{i_{n+1}-1}. 
\]

A sequence $(\sigma_n)_{n \in \N}$ of morphisms is said to be \definir{weakly primitive} if for all $r \in \N$, for all $s > r$ and for all letters $a \in A_r$ and $b \in A_{s+1}$, the letter $a$ occurs in $\sigma_r \cdots \sigma_s(b)$.  A sequence $(\sigma_n)_{n \in \N}$ of morphisms is said to be \definir{primitive} if it is weakly primitive and there is a constant $k$ such that $s$ can be replaced by $r+k$.

\begin{remark}
\label{remark: weak prim and prim}
A sequence of morphisms is weakly primitive if and only if it admits a contraction which is primitive.
\end{remark}

\subsection{Rauzy graphs}
\label{subsection: Rauzy graphs}

Let $(X,T)$ be a subshift over an alphabet $A$. 

\begin{definition}
The \definir{Rauzy graph of order $n$} of $(X,T)$ (also called \textit{graph of words of length $n$}), denoted by $G_n(X)$ (or simply $G_n$), is the labelled directed graph $(V(n),E(n))$, where the set $V(n)$ of vertices is $\fac{n}{X}$ and	there is an edge from $u$ to $v$ if there exist some letters $a$ and $b$ in $A$ such that $ub = av \in \fac{n+1}{X}$; 
this edge is labelled\footnote{In the literature, there are different ways of labelling the edges. Indeed, the edges are sometimes labelled by the letter $a$, by the letter $b$, by the ordered pair $(a,b)$ or by the word $av$.} by $ub$ and is denoted by $(u,(a,b),v)$.
\end{definition}

Let us introduce some notation: for an edge $e=(u,(a,b),v)$, let us call $o(e) = u$ its \definir{outgoing vertex}, $i(e) = v$ its \definir{incoming vertex}, $\lambda_L(e)=a$ its \definir{left label}, $\lambda_R(e)=b$ its \definir{right label} and $\lambda(e) = ub = av$ its \definir{full label}. Same definitions hold for labels of paths (left and right labels being words of same length as the considered path) where we naturally extend the map $\lambda$ to the set of paths by $\lambda\left( (u_0,(a_1,b_1),u_1)(u_1,(a_2,b_2),u_2)\cdots (u_{\ell-1},(a_{\ell},b_{\ell}),u_{\ell}) \right) = u_0 b_1 b_2 \cdots b_{\ell}$. In this paper we will mostly consider right labels.

\begin{example} 
Let $(X_{\phi},T)$ be the subshift generated by the \textit{Fibonacci sequence} $\phi^{\omega}(0)$ where $\phi$ is the morphism defined by $\phi(0)=01$, $\phi(1)=0$. Figure~\ref{Rauzy graphs of the Fibonacci sequence} represents the three first Rauzy graphs of $(X_{\phi},T)$ (with full labels on the edges).
\begin{figure}[h!tbp]
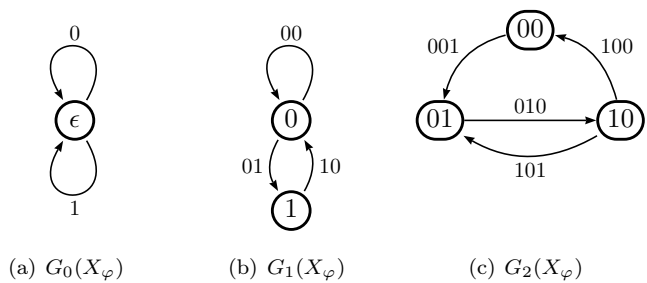

\centering
\subfigure[$G_0(X_{\phi})$]{
\label{figure: rauzy graph G_0 of the fibonacci sequence}
\scalebox{0.6}{
\begin{VCPicture}{(0.5,0)(3.5,5)}
\ChgEdgeLabelScale{0.8}
\StateVar[\epsilon]{(2,3)}{epsilon}
\LoopR[0.5]{90}{epsilon}{0}
\LoopL[0.5]{-90}{epsilon}{1}
\end{VCPicture}
}}
\qquad
\subfigure[$G_1(X_{\phi})$]{
\label{figure: rauzy graph G_1 of the fibonacci sequence}
\scalebox{0.6}{
\begin{VCPicture}{(0.5,0)(3.5,5)}
\ChgEdgeLabelScale{0.8}
\StateVar[0]{(2,3)}{0}
\StateVar[1]{(2,1)}{1}
\LoopR[0.5]{90}{0}{00}
\LArcR[0.5]{0}{1}{01}
\LArcR[0.5]{1}{0}{10}
\end{VCPicture}
}}
\qquad
\subfigure[$G_2(X_{\phi})$]{
\label{figure: rauzy graph G_2 of the fibonacci sequence}
\scalebox{0.6}{
\begin{VCPicture}{(0,0)(4,5)}
\ChgEdgeLabelScale{0.8}
\StateVar[01]{(0,3)}{01}
\StateVar[10]{(4,3)}{10}
\StateVar[00]{(2,5)}{00}
\EdgeL[0.5]{01}{10}{010}
\LArcR[0.5]{10}{00}{100}
\LArcR[0.5]{00}{01}{001}
\LArcL[0.5]{10}{01}{101}
\end{VCPicture}
}}

\caption{First Rauzy graphs of the Fibonacci sequence}
\label{Rauzy graphs of the Fibonacci sequence}
\end{figure}
\end{example}

\begin{remark}
\label{remark: uniformly recurrent strongly connected}
Any minimal subshift has only strongly connected Rauzy graphs (that is, for all vertices $u$ and $v$ of $G_n$ there is a path from $u$ to $v$).
\end{remark}
We say that a vertex $v$ is \textit{right special} (resp. \textit{left special}, \textit{bispecial}) if it corresponds to a right special (resp. left special, bispecial) factor.

By definition of Rauzy graphs, any word $u \in \fac{}{X}$ is the full label of a path in $G_n(X)$ for $n < |u|$. Figure~\ref{figure: rauzy graph G_1 of the fibonacci sequence} shows that the converse is not true: the word $000$ is the full label of a path of length $2$ but does not belong to $\fac{}{X_{\phi}}$. Hence a path $p$ is said to be \definir{allowed} if $\lambda(p) \in \fac{}{X}$. The next proposition follows immediately from definitions.

\begin{proposition}
\label{prop: path in Rauzy graphs}
Let $G_n$ be a Rauzy graph of order $n$. For all paths $p$ of length $\ell \leq n$ in $G_n$, the left (resp. right) label of $p$ is a prefix (resp. a suffix) of $o(p)$ (resp. of $i(p)$).
\end{proposition}

\begin{definition}
The \textit{reduced Rauzy graph} of order $n$ of $(X,T)$ is the directed graph $g_n(X)$ such that 
\begin{list}{-}{}
\item 	the vertices are the vertices of $G_n(X)$ that are either special or \enquote{boundary}, \textit{i.e.}, at least one value in $\{\delta^+v,\delta^-v\}$ is null and 
\item 	there is an edge from $u$ to $v$ is there is a path $p$ in $G_n(X)$ from $u$ to $v$ such that all interior vertices of $p$ are not special.
\end{list}
\end{definition}
The (left, right and full) labels of an edge in $g_n(X)$ are the (left, right and full) labels of the corresponding path in $g_n(X)$. To avoid any confusion, edges of reduced Rauzy graphs are represented by double lines. Figure~\ref{figure: figure fibonacci reduced} represents the reduced Rauzy graph $g_{\phi}(2)$ with full labels on the edges.

\begin{figure}[h!tbp]
\centering
\scalebox{0.6}{
\begin{VCPicture}{(-1,0)(5,5)}
\ChgEdgeLabelScale{0.8}
\StateVar[01]{(0.5,2.5)}{ab}
\StateVar[10]{(3.5,2.5)}{ba}
\EdgeLineDouble
\VCurveL[]{angleA=45,angleB=135,ncurv=2}{ba}{ab}{1001}
\VCurveR[]{angleA=-45,angleB=-135,ncurv=2}{ba}{ab}{101}
\EdgeL{ab}{ba}{010}
\end{VCPicture}
}
\caption{$g_{2}(X_{\phi})$}
\label{figure: figure fibonacci reduced}
\end{figure}

\section{Adicity of minimal subshifts using Rauzy graphs}
\label{section: adicity of minimal subshift}

Let $(X,T)$ be a minimal subshift over an alphabet $A$. In this section we prove the following theorem.

\begin{theorem}
\label{thm adic minimal}
An aperiodic subshift $(X,T)$ is minimal if and only if it is primitive and proper $S$-adic. Moreover, if $X$ does not have linear complexity, then $S$ is infinite.
\end{theorem}

The construction of the $S$-adic representation is based on the evolution of Rauzy graphs. Similar construction can be found in~\cite{Leroy} (see also~\cite{Ferenczi}) where we give a method to build a $S$-adic representation of any uniformly recurrent sequence with a sub-linear complexity. In that paper, the construction is based on the $n$-segments although here we work with the $n$-circuits (see Section~\ref{subsection: Definition of the morphisms} below for the definition). However the techniques are the same.

\subsection{$n$-circuits}

For $n \in \N$, an \definir{$n$-circuit} is a non-empty path $p$ in $G_n(X)$ such that $o(p) = i(p)$ is a right special vertex and no interior vertex of $p$ is $o(p)$.

\begin{remark}
An $n$-circuit is not necessarily an allowed path of $G_n(X)$. Indeed, consider the subshift $X_{\mu}$ generated by the \definir{Thue-Morse sequence} $\mu^{\omega}(0)$ where $\mu$ is the \definir{Thue-Morse morphism} defined by $\mu(0)=01$ and $\mu(1)=10$. The path
\[
	010 \to \left( 101 \to 011 \to 110 \to 101 \right)^3 \to 010
\]
in Figure~\ref{fig: Rauzy graph of order 3 of bt} is a $3$-circuit and its full label contains the word $(101)^3$ which is not a factor of $\mu^{\omega}(0)$ since the Thue-Morse sequence is cube-free~\cite{Thue_1,Thue_2}.
\begin{figure}[h!tbp]
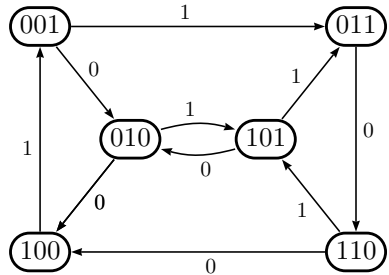

\centering
\scalebox{0.6}{
\begin{VCPicture}{(0,-0.5)(7,5)}
\ChgEdgeLabelScale{0.8}
\StateVar[001]{(0,5)}{001}
\StateVar[011]{(7,5)}{011}
\StateVar[010]{(2,2.5)}{010}
\StateVar[101]{(5,2.5)}{101}
\StateVar[100]{(0,0)}{100}
\StateVar[110]{(7,0)}{110}
\EdgeL{001}{011}{1}
\EdgeL{001}{010}{0}
\EdgeL{011}{110}{0}
\EdgeL{010}{100}{0}
\ArcL{010}{101}{1}
\ArcL{101}{010}{0}
\EdgeL{101}{011}{1}
\EdgeL{010}{100}{0}
\EdgeL{100}{001}{1}
\EdgeL{110}{101}{1}
\EdgeL{110}{100}{0}
\end{VCPicture}
}
\caption[Rauzy graph of order 3 of the Thue-Morse sequence.]{Rauzy graph of order 3 (with right labels on the edges) of $X_{\mu}$.}
\label{fig: Rauzy graph of order 3 of bt}
\end{figure}
\end{remark}

\begin{remark}
\label{remark: lien entre n-circuit et mot de retour}
The notion of $n$-circuit is closely related to the notion of return word. Let us recall that if $u \in \fac{}{X}$, a \definir{return word} to $u$ in $X$ is a non-empty word $v$ such that $uv \in \fac{}{X}$ and that contains exactly two occurrences of $u$, one as a prefix and one as a suffix. If $u$ is a right special vertex in $G_n(X)$, then $\{ \lambda_R(v) \mid v \text{ allowed } n\text{-circuit starting from } u\}$ is exactly the set of return words to $u$. 
\end{remark}

\begin{fact}
\label{fact: alphabet fini}
A subshift is minimal if and only if for all $n$, the number of its allowed $n$-circuits is finite.
\end{fact}

The next lemma is also well known.

\begin{lemma}
\label{lemma: pas de court}
Let $(X,T)$ be a minimal and aperiodic subshift. Then 
\[
	\lim_{n \to +\infty} \min\{|\lambda_R(p)| \mid p \text{ allowed } n\text{-circuit}\} = +\infty.
\]
\end{lemma}


\subsection{Definition of the morphisms of the adic representation}
\label{subsection: Definition of the morphisms}

The adic representation that we will compute is based on the behaviour of $n$-circuits when $n$ increases. 
To this aim we define a map $\psi_n$ on the set of paths of $G_{n+1}(X)$ in the following way. 
For each path $p$ in $G_{n+1}(X)$ with right label $\lambda_R(p)=u$, $\psi_n(p)$ is the unique path $q$ in $G_n(X)$ whose right label is $\lambda_R(q)=u$ and such that $o(q)$ and $i(q)$ are suffixes of $o(p)$ and $i(p)$ respectively.
The next lemma is obvious.

\begin{lemma}
\label{lemma: decomposition circuit}
Let $(X,T)$ be a subshift. If $u \in \fac{n+1}{X}$ is a right special factor, then for all allowed $(n+1)$-circuit $p$ starting from $u$, there exist some allowed $n$-circuits $q_1, q_2, \dots, q_k$ starting from the right special factor $u[2,n+1] \in \fac{n}{X}$ such that $\psi_n(p) = q_1 q_2 \cdots q_k$. Moreover, if $G_n(X)$ does not contain any bispecial vertex, then $\psi_n$ is a bijective map such that for every allowed $(n+1)$-circuit, $\psi_n(p)$ is an allowed $n$-circuit.
\end{lemma}

Lemma~\ref{lemma: decomposition circuit} allows to define some morphisms coding how the $n$-circuits can be concatenated to create the $(n+1)$-circuits. However we can see in this lemma that we can only put in relations the $n$-circuits and $(n+1)$-circuits that are starting in vertices with the same suffix of length $n$. Lemma~\ref{lemma: suite de speciaux droits} below allows to choose some particular vertices; it comes from aperiodicity and from the observation that any suffix of a right special factor is also right special.

\begin{lemma}
\label{lemma: suite de speciaux droits}
Let $(X,T)$ be an aperiodic subshift on an alphabet $A$. There exists an infinite sequence $(U_n \in \fac{n}{X})_{n \in \N}$ such that for all $n$, $U_n$ is a right special factor and is a suffix of $U_{n+1}$.
\end{lemma}

\begin{definition}
\label{definition: definition des morphismes}
Let $(X,T)$ be a minimal and aperiodic subshift and let $(U_n)_{n \in \N}$ be a sequence as in Lemma~\ref{lemma: suite de speciaux droits}. 
For each non-negative integer $n$, we let $\A_n$ denote the set of allowed $n$-circuits starting from $U_n$ ($\A_n$ is finite due to Fact~\ref{fact: alphabet fini}).
Now define the alphabet $A_n = \{0,1,\dots,\card{\A_n}-1 \}$ and consider a bijection $\theta_n : A_n \rightarrow \A_n$. We can extend $\theta_n$ to an isomorphism by putting $\theta_n(ab) = \theta_n(a) \theta_n(b)$ for all letters $a,b$ in $A_n$ (observe that $\theta_n(a) \theta_n(b)$ might not be a path in $G_n(X)$). Then, for all $n$ we define the morphism $\gamma_n : A_{n+1}^* \rightarrow A_n^*$ as the unique morphism satisfying
\[
	\theta_n \gamma_n = \psi_n \theta_{n+1}.
\]
\end{definition}

\begin{remark}
\label{remark: sigme = identity}
Let $(i_n)_{n \in \N}$ be the increasing sequence of non-negative integers such that there is a bispecial factor in $\fac{k}{X}$ if and only if $k = i_n$ for some $n$. It is a direct consequence of Lemma~\ref{lemma: decomposition circuit} that if $k \notin \{i_n \mid n \in \N\}$, then the morphism $\gamma_k$ is simply a bijective and letter-to-letter morphism. This morphism only depends on the differences that could exist between $\theta_k$ and $\theta_{k+1}$. In that case, we can suppose without loss of generality that $\theta_k$ and $\theta_{k+1}$ satisfy $\psi_k \theta_{k+1} (i) = \theta_k(i)$ for all letters $i$ in $A_{k+1}$ so that $\gamma_k$ is the identity morphism. As a consequence, to build an adic representation of a subshift, it would suffice to consider the subsequence $(\gamma_{i_n})_{n \in \N}$ of $(\gamma_n)_{n \in \N}$. Depending on the context, we will sometimes consider the sequence $(\gamma_n)_{n \in \N}$ or the subsequence $(\gamma_{i_n})_{n \in \N}$.
\end{remark}

\begin{remark}
\label{remark = A_0 = A}
If the alphabet of $(X,T)$ is $A=\{0,\dots,k\}$, the Rauzy graph $G_0(X)$ is as in Figure~\ref{Rauzy graph of order 0} so we have $\lambda(\A_0) = A$. We can suppose that $\theta_0$ is such that $\lambda_R \theta_0(a) = a$ for all $a \in A_0$.
\begin{figure}[h!tbp]
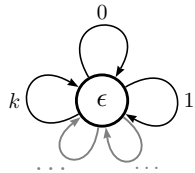

\centering
\scalebox{0.6}{
\begin{VCPicture}{(0,1)(5,4)}
\ChgEdgeLabelScale{0.8}
\LargeState
\State[\epsilon]{(2.5,2.5)}{epsilon}
\LoopN[0.5]{epsilon}{0}
\LoopE[0.5]{epsilon}{1}
\LoopL[0.5]{180}{epsilon}{k}
\DimEdge	\CLoopL[0.5]{-60}{epsilon}{\dots}	\RstEdge
\DimEdge	\CLoopL[0.5]{-120}{epsilon}{\dots}	\RstEdge
\end{VCPicture}
}
\caption{Rauzy graph $G_0$ of any subshift over $\{0,\dots, k\}$}
\label{Rauzy graph of order 0}
\end{figure}
\end{remark}

\subsubsection{An example}
\label{subsubsection: An example}

Consider a graph as represented in Figure~\ref{figure: Sturmian graph with bispecial factor and labels} and let us give all possible evolutions from it. The letters $a$ and $b$ (resp. $\alpha$ and $\beta$) represent the right (resp. left) extending letters of $U_n$.

\begin{figure}[h!tbp]
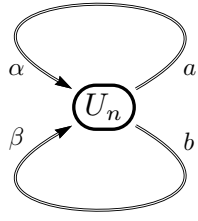

\centering
\scalebox{0.7}{
\begin{VCPicture}{(0,0.5)(5,4.5)}
\ChgEdgeLabelScale{0.8}
\StateVar[U_n]{(2.5,2.5)}{B}
\EdgeLineDouble
\VCurveR[.1]{angleA=30,angleB=150,ncurv=7}{B}{B}{a}		\LabelR[.9]{\alpha}
\VCurveL[.1]{angleA=-30,angleB=-150,ncurv=7}{B}{B}{b}	\LabelL[.9]{\beta}
\end{VCPicture}
}
\caption{Reduced Rauzy graph $g_n$ with some additional labels}
\label{figure: Sturmian graph with bispecial factor and labels}
\end{figure}

By definition of the Rauzy graph, the words $\alpha U_n$, $\beta U_n$, $U_na$ and $U_nb$ are vertices of $G_{n+1}$. Since the subshifts we are considering satisfy $p(n+1)-p(n) \geq 1$ for all $n$, at least one of the vertices $\alpha U_n$ and $\beta U_n$ is right special and at least one of the vertices $U_na$ and $U_nb$ is left special. Moreover, by definition of the reduced Rauzy graphs, the two loops of $g_n$ become edges respectively from $U_na$ to $\alpha U_n$ and from $U_nb$ to $\beta U_n$. Thus, the only missing information are which edges are starting from $\alpha U_n$ and $\beta U_n$ and which edges are arriving to $U_na$ and $U_nb$. By minimality, $G_{n+1}$ has to be strongly connected so we have only three possibilities (2 of them being symmetric). The possible evolutions are represented at Figure~\ref{figure: evolution of a sturmian graph}.

\begin{figure}[h!tbp]
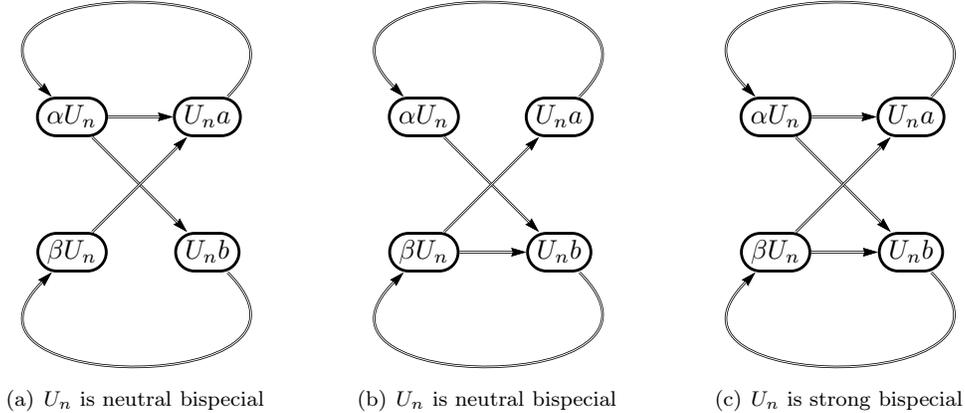

\begin{center}
\subfigure[$U_n$ is neutral bispecial]{
\label{figure: Evolution of a sturmian graph with ordinary bispecial factor}
\scalebox{0.6}{
\begin{VCPicture}{(-0.5,-2)(5.5,6.5)}
\ChgEdgeLabelScale{0.8}
\StateVar[\alpha U_n]{(1,4)}{alpha B}
\StateVar[\beta U_n]{(1,1)}{beta B}
\StateVar[U_na]{(4,4)}{Ba}
\StateVar[U_nb]{(4,1)}{Bb}
\EdgeLineDouble
\Edge{alpha B}{Ba}
\Edge{alpha B}{Bb}
\Edge{beta B}{Ba}
\VCurveR[]{angleA=45,angleB=135,ncurv=2}{Ba}{alpha B}{}
\VCurveR[]{angleA=-45,angleB=-135,ncurv=2}{Bb}{beta B}{}
\end{VCPicture}
}}
\qquad
\subfigure[$U_n$ is neutral bispecial]{
\label{figure: Evolution of a sturmian graph with ordinary bispecial factor'}
\scalebox{0.6}{
\begin{VCPicture}{(-0.5,-2)(5.5,6.5)}
\ChgEdgeLabelScale{0.8}
\StateVar[\alpha U_n]{(1,4)}{alpha B}
\StateVar[\beta U_n]{(1,1)}{beta B}
\StateVar[U_na]{(4,4)}{Ba}
\StateVar[U_nb]{(4,1)}{Bb}
\EdgeLineDouble
\Edge{alpha B}{Bb}
\Edge{beta B}{Bb}
\Edge{beta B}{Ba}
\VCurveR[]{angleA=45,angleB=135,ncurv=2}{Ba}{alpha B}{}
\VCurveR[]{angleA=-45,angleB=-135,ncurv=2}{Bb}{beta B}{}
\end{VCPicture}
}}
\qquad
\subfigure[$U_n$ is strong bispecial]{
\label{figure: Evolution of a sturmian graph with strong bispecial factor}
\scalebox{0.6}{
\begin{VCPicture}{(-0.5,-2)(5.5,6.5)}
\ChgEdgeLabelScale{0.8}
\StateVar[\alpha U_n]{(1,4)}{alpha B}
\StateVar[\beta U_n]{(1,1)}{beta B}
\StateVar[U_na]{(4,4)}{Ba}
\StateVar[U_nb]{(4,1)}{Bb}
\EdgeLineDouble
\Edge{alpha B}{Ba}
\Edge{alpha B}{Bb}
\Edge{beta B}{Ba}
\Edge{beta B}{Bb}
\VCurveR[]{angleA=45,angleB=135,ncurv=2}{Ba}{alpha B}{}
\VCurveR[]{angleA=-45,angleB=-135,ncurv=2}{Bb}{beta B}{}
\end{VCPicture}
}}
\end{center}
\caption{Possible evolutions of the graph represented in Figure~\ref{figure: Sturmian graph with bispecial factor and labels}}
\label{figure: evolution of a sturmian graph}
\end{figure}

Suppose that the bijection $\theta_n$ maps $0$ to the $n$-circuit starting with an $a$ and $1$ to the $n$-circuit starting with a $b$. Consider the same definition of $\theta_{n+1}$ for the two first evolutions (since $\# \A_{n+1} = 2$). For the third one, suppose that $\# \A_{n+1} = r+1$ ($1 \leq r < + \infty$) and that if $U_{n+1} = \alpha U_n$ (resp. $\beta U_n$), $\theta_{n+1}(0)$ is the loop starting with the edge from $\alpha U_n$ to $U_na$ (resp. $\beta U_n$ to $U_nb$) and let $k_1, \dots, k_r$ be integers such that $\theta_{n+1}(i)$ is the path going to $U_nb$ (resp. to $U_na$) and going $k_i$ times through the loop $U_nb \rightarrow \beta U_n \rightarrow U_nb$ (resp. $U_na \rightarrow \alpha U_n \rightarrow U_na$) before coming back to $\alpha U_n$ (resp. $\beta U_n$).

Then for the two first possible evolutions, the morphisms coding them are respectively
\begin{eqnarray}
	\begin{cases}
		0 \mapsto 0	\\
		1 \mapsto 10
	\end{cases}
	&
	\text{and }
	&
	\begin{cases}
		0 \mapsto 1	\\
		1 \mapsto 01
	\end{cases}
\end{eqnarray}
and the morphism coding the third evolution is one of the following, depending on the choice of $U_{n+1}$:
\begin{eqnarray}
\label{equation: morphisms 1 to 8}
	\begin{cases}
		0 \mapsto 0		\\
		1 \mapsto 1^{k_1} 0 	\\
		2 \mapsto 1^{k_2} 0		\\
		\vdots	\\
		r \mapsto 1^{k_r} 0
	\end{cases}
	&
	\begin{cases}
		0 \mapsto 1		\\
		1 \mapsto 0^{k_1} 1 	\\
		2 \mapsto 0^{k_2} 1		\\
		\vdots	\\
		r \mapsto 0^{k_r} 1	
	\end{cases}
\end{eqnarray}

\subsection{Adic representation of $\fac{}{X}$}
\label{subsection: Adicity of L(X)}

The next two result shows \textit{a posteriori} that this makes sense to build an $S$-adic representation using $n$-circuits: it states that when considering a sequence $(U_n)_{n \in \N}$ as Lemma~\ref{lemma: suite de speciaux droits}, the labels of $n$-circuits starting from $U_n$ provide the entire language of $X$ when $n$ goes to infinity.

\begin{lemma}
\label{lemma: langage de X et circuits}
Let $(X,T)$ be a minimal and aperiodic subshift. If $(U_n \in \fac{n}{X})_{n \in \N}$ is a sequence of right special vertices such that $U_n$ is suffix of $U_{n+1}$, then for all $n$
\begin{multline}
\label{eq 1 lemma langage}
	\fac{}{\{ \lambda_R(p) \mid p \text{ allowed } (n+1)\text{-circuit starting from } U_{n+1}\}^*} 	\\
	\subseteq \fac{}{\{ \lambda_R(p) \mid p \text{ allowed } n\text{-circuit starting from } U_n\}^*} 
\end{multline}
and
\begin{equation}
\label{eq 2 lemma langage}
	\bigcap_{n \in \N} \fac{}{\{ \lambda_R(p) \mid p \text{ allowed } n\text{-circuit starting from } U_n\}^* } = \fac{}{X}.
\end{equation}
Furthermore, for all non-negative integers $\ell$, there is a non-negative integer $N_\ell$ such that 
\begin{equation}
\label{eq 3 lemma langage}
	\fac{\leq \ell}{\{ \lambda_R(p) \mid p \text{ allowed } N_\ell\text{-circuit starting from } U_n\}^*} = \fac{\leq \ell}{X},
\end{equation}
where $\fac{\leq \ell}{X}$ stands for $\bigcup_{0 \leq n \leq \ell} \fac{n}{X}$.
\end{lemma}

\begin{proof}
Indeed, \eqref{eq 1 lemma langage} directly follows from Lemma~\ref{lemma: decomposition circuit} and~\eqref{eq 2 lemma langage} and~\eqref{eq 3 lemma langage} are consequences of the minimality. 
\end{proof}

The next result is just a reformulation of Lemma~\ref{lemma: pas de court} and of Lemma~\ref{lemma: langage de X et circuits}.

\begin{corollary}
\label{cor: langage}
Let $(X,T)$ be a minimal and aperiodic subshift and let $(\gamma_n)_{n \in \N}$ be the sequence of morphisms as in Definition~\ref{definition: definition des morphismes}. We have 
\[
	\min_{n \to +\infty} \min_{a \in A_{n+1}} |\gamma_0 \cdots \gamma_n(a)| = +\infty
\]
and for all sequences of letters $(a_n \in A_n)_{n \in \N}$,
\[
	\bigcap_{n \in \N} \fac{}{\gamma_0 \cdots \gamma_n(a_{n+1})} = \fac{}{X}. 
\]
\end{corollary}

\subsection{Adic representation of $X$}
\label{subsection: Adicity of X}

In this section we prove that, up to a little change, the directive word $(\gamma_n)_{n \in \N}$ introduced in Definition~\ref{definition: definition des morphismes} is an adic representation of a sequence in $X$, \textit{i.e.}, we provide a slightly different directive word $(\tau_n:B_{n+1}^* \to B_n^*)_{n \in \N}$ such that $(\tau_0 \cdots \tau_n(a_{n+1}^{\infty}))_{n \in \N}$ converges in $A^{\Z}$ to $\bw \in X$. The proof of Theorem~\ref{thm adic minimal} follows immediately from Lemma~\ref{lemma: convergence taU_n} and Proposition~\ref{prop: durand primitif}.

The next result shows that there is a contraction $(\gamma'_n:A_{n+1}'^* \to A_n'^*)_{n \in \N}$ of $(\gamma_n)_{n \in \N}$ such that every morphism $\gamma_n'$ is \definir{right proper}, \textit{i.e.}, there is a letter $a \in A_n'$ such that $\gamma_n'(A_{n+1}') \subset A_n'^* a$.

\begin{lemma}
\label{lemma: toutes les images finissent pareil}
Let $(X,T)$ be a minimal subshift and let $(\gamma_n)_{n \in \N}$ be the sequence of morphisms defined in Definition~\ref{definition: definition des morphismes}. For all non-negative integers $r$ there is an integer $s >r$ and a letter $a$ in $A_{r}$ such that $\gamma_r \cdots \gamma_s(A_{s+1}) \subset A_r^* a$.
\end{lemma}

\begin{proof}
Let $(U_n)_{n \in \N}$ be a sequence as defined in Lemma~\ref{lemma: suite de speciaux droits} and let $(\gamma_n)_{n \in \N}$ be the sequence of morphisms as in Definition~\ref{definition: definition des morphismes}. Let $r$ be a non-negative integer. By definition, for all integers $j >r$, $U_j$ is a word in $\fac{}{X}$ that admits $U_r$ as a suffix. Consequently, we can associate to $U_j$ a path $p_j$ of length $j-r$ in $G_r$ such that $\lambda(p_j) = U_j$ and $i(p_j)=U_r$. 

By Lemma~\ref{lemma: pas de court} there is an integer $s > k$ such that all $s$-circuits starting from $U_s$ have length at least $k+r$. Let $c$ be such an $s$-circuit. Since $U_k$ is a suffix of length $k$ of $U_s$, we deduce from Proposition~\ref{prop: path in Rauzy graphs} that $D_k$ is a suffix of $\lambda_R(c)$. Let $q_{k,c}$ the suffix of $c$ with right label $U_k$ and let $t_{k,c}$ the suffix of $q_{k,c}$ of length $k$. By construction, we have $o(t_{k,c}) \in A^* U_r$, $i(t_{k,c}) \in A^* U_r$ and $\lambda_R(t_{k,c}) = \lambda_R(p_k)$, and so, $\psi_r \psi_{r+1} \cdots \psi_{s-1}(t_{k,c}) = p_k$. Denoting $a = \theta_r^{-1}(p_k) \in A_r$ and observing that $A_s = \{ \theta_s^{_1} \mid c = s\text{-circuit starting from } U_s \}$, we have $\gamma_r \cdots \gamma_{s-1}(A_s) \in A_r^* a$. 
\end{proof}

The following trick allows us to define the directive word $(\tau_n)_{n \in \N}$ mentioned above. If $\sigma: A^* \to B^*$ is a right proper morphism with ending letter $r \in B$, then its \definir{left conjugate} is the morphism $\sigma^{(L)} : A^* \to B^*$ defined by $\sigma^{(L)}(a) = r u$ whenever $\sigma(a) = ur$. Thus, it is a \definir{left proper} morphism, \textit{i.e.}, there is a letter $a \in B$ such that $\sigma(A) \subset a B^*$ (in our case, $a = r$).

\begin{lemma}
\label{lemma: conjugue de morphism}
If $\sigma: A^* \rightarrow B^*$ is a right proper morphism and if $\bw$ is a sequence in $A^{\Z}$, then $T(\sigma^{(L)}(\bw)) = \sigma(\bw)$. In particular, $\fac{}{\sigma^{(L)}(\bw)} = \fac{}{\sigma(\bw)}$.
\end{lemma}

\begin{fact}
\label{fact: conjugue}
Let $(\gamma_n')_{n \in \N}$ be a contraction of $(\gamma_n)_{n \in \N}$ such that all morphisms $\gamma_n'$ are right proper. Every morphism $\tau_n = \gamma_{2n}' \gamma_{2n+1}'^{(L)}$ is both left and right proper.  
\end{fact}

\begin{lemma}
\label{lemma: convergence taU_n}
The directive word $(\tau_n: B_{n+1}^* \to B_n^*)_{n \in \N}$ is proper and weakly primitive and such that $(\tau_0 \cdots \tau_n(b_{n+1}))_{n \in \N}$ converges in $A^\Z$ to $\bw \in X$ for sequences $(b_n \in B_n)_{n \in \N}$.
\end{lemma}
\begin{proof}
The convergence is ensured by the fact that all morphisms are left and right proper. The fact that the limit $\bw$ belongs to $X$ follows from Corollary~\ref{cor: langage} and Lemma~\ref{lemma: conjugue de morphism}. The weak primitivity comes from the minimality of $(X,T)$.
\end{proof}

\begin{proposition}[Durand~\cite{Durand_LR,Durand_corrigentum}]
\label{prop: durand primitif}
If $(X,T)$ is a primitive $S$-adic subshift with $S$ finite, then $X$ has linear complexity.
\end{proposition}


\section{$\S$-adicity of minimal subshifts satisfying $1 \leq p(n+1)-p(n) \leq 2$}
\label{section: $S$-adicity of subshifts with complexity $2n$}

In this chapter we present Theorem~\ref{thm: 2n} which is an improvement of Theorem~\ref{thm adic minimal} for the particular case of minimal subshifts with first difference of complexity bounded by 2.
For this class of complexity, Ferenczi~\cite{Ferenczi} proved that the amount of morphisms needed for the $S$-adic representations is less than $3^{27}$. Here, we significantly improve this bound by giving the set $\S$ of 5 morphisms that are actually needed.
To avoid unnecessary repetitions, we only sketch the proof of Theorem~\ref{thm: 2n} on an example. We will later prove Theorem~\ref{thm: 2n final} which is an improvement of the former.
In all this chapter, the set $\S$ is the set of morphisms $\{G,D,M,E_{01},E_{12}\}$ where
\begin{eqnarray*}
\label{eqnarray: morphisms GDME}
	G:	\begin{cases}
			0 \mapsto 10	\\
			1 \mapsto 1		\\
			2 \mapsto 2
		\end{cases}
	&
	D:	\begin{cases}
			0 \mapsto 01	\\
			1 \mapsto 1		\\
			2 \mapsto 2
		\end{cases}
	&
	M:	\begin{cases}
			0 \mapsto 0	\\
			1 \mapsto 1		\\
			2 \mapsto 1
		\end{cases}
\end{eqnarray*}
\begin{eqnarray*}
	E_{01}:	\begin{cases}
				0 \mapsto 1		\\
				1 \mapsto 0		\\
				2 \mapsto 2
			\end{cases}
	&
	E_{12}:	\begin{cases}
				0 \mapsto 0		\\
				1 \mapsto 2		\\
				2 \mapsto 1
			\end{cases}
\end{eqnarray*}

\begin{theorem}
\label{thm: 2n}
Let $\G$ be the graph represented in Figure~\ref{figure: graph of graphs}. There is a non-trivial way to label the edges of $\G$ with morphisms in $\S^*$ such that for any minimal subshift $(X,T)$ satisfying $1 \leq p_X(n+1)-p_X(n) \leq 2$ for all $n$, there is an infinite path $p$ in $\G$ whose label $(\sigma_n)_{n \in \N} \in \S^{\N}$ is a directive word of $(X,T)$. Furthermore, $(\sigma_n)_{n \in \N}$ is weakly primitive and admits a contraction that contains only proper morphisms.
\end{theorem}

This result is based on a detailed description of all possible Rauzy graphs of minimal subshifts with the considered complexity. The Rauzy graphs of such subshifts can have only 10 different shapes. These shapes correspond to vertices of $\G$. The edges of $\G$ are given by the possible evolutions of these graphs and are labelled by morphisms coding these evolutions (see Section~\ref{subsubsection: An example}). The theorem is obtained by showing that these labels belong to $\S^*$. In the next section, we will study even more the evolutions of Rauzy graphs in order to obtain an $\S$-adic characterization of the considered subshifts.

From now on, $(X,T)$ satisfies the conditions of Theorem~\ref{thm: 2n}, \textit{i.e.}, it is minimal and is such that $1 \leq p_X(n+1) - p_X(n) \leq 2$ for all $n$. Consequently,  we have $p_X(n) \leq 2n$ for all $n \geq 1$ when $\card{A}=2$ and $p_X(n) \leq 2n+1$ for all $n$ when $\card{A}=3$.

We also consider notation introduced in Definition~\ref{definition: definition des morphismes} and Remark~\ref{remark: sigme = identity}, \textit{i.e.}, for every $n$, the morphism $\gamma_n$ describes the evolution from $G_n$ to $G_{n+1}$ and $(i_n)_{n \in \N}$  is the sequence of integers such that there is a bispecial factor in $\fac{k}{X}$ if and only if $k \in \{i_n \mid n \in \N\}$.

\subsection{10 shapes of Rauzy graphs}	
\label{subsection: 10 chapes of rauzy graphs}

In this section we describe the possible shapes of Rauzy graphs for the considered class of complexity.

From Equation~\eqref{eq 1} (page~\pageref{eq 1}) the hypothesis on the complexity implies  that for all integers $n$, there are either one right special factor $u$ of length $n$ with $\delta^+(u) \in \{2,3\}$ or two right special factors $v_1$ and $v_2$ with $\delta^+(v_1) = \delta^+(v_2) = 2$. From Equation~\eqref{eq 2} we can make a similar observation for the left special factors. Hence for all integers $n$, we have the following possibilities: 
\begin{enumerate}
	\item	there is one right special factor $r$ and one left special factor $l$ of length $n$ with $\delta^+(r) = \delta^- (l) \in \{2,3\}$ (Figure~\ref{figure: Rauzy graphs with one right special factor and one left special factor});
	\item	there is one right special factor $r$ and two left special factors $l_1$ and $l_2$ of length $n$ with $\delta^+(r) = 3$ and $\delta^- (l_1) = \delta^- (l_2) =2$ (Figure~\ref{figure: graph with 1 right special factor});
	\item	there are two right special factors $r_1$ and $r_2$ and one left special factor $l$ of length $n$ with $\delta^+(r_1) = \delta^+ (r_2) =2$ and $\delta^-(l) = 3$ (Figure~\ref{figure: graph with 1 left special factor}); 
	\item	there are two right special factors $r_1$ and $r_2$ and two left special factors $l_1$ and $l_2$ of length $n$ with $\delta^+(r_1) = \delta^+ (r_2) = \delta^- (l_1) = \delta^-(l_2) =2$ (Figure~\ref{figure: Rauzy graphs with two left and right special factors}).
\end{enumerate}

From these possibilities we can deduce that for all $n$, $g_n(X)$ only has eight possible shapes: those represented from Figure~\ref{figure: Rauzy graphs with one right special factor and one left special factor} to Figure~\ref{figure: Rauzy graphs with two left and right special factors}. Reduced Rauzy graphs in Figure~\ref{figure: Rauzy graphs with one right special factor and one left special factor} are well-known: they correspond to reduced Rauzy graphs of Sturmian sequences (Figure~\ref{figure: Sturmian graph}) or of Arnoux-Rauzy sequences (Figure~\ref{figure: Arnoux-Rauzy graph}). Reduced Rauzy graphs in Figure~\ref{figure: Rauzy graphs with two left and right special factors} have also been studied by Rote~\cite{Rote}.
Observe that in these figures, the edges represented by dots may have length 0. In this case, the two vertices they link are merged to one vertex.

{
\begin{figure}[h!tbp]
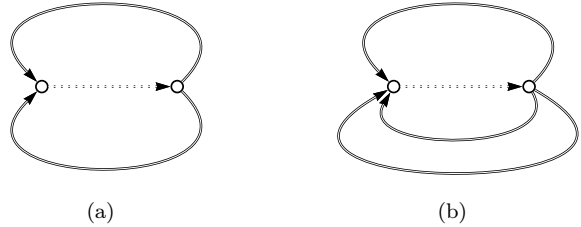

\centering
\subfigure[]{
\label{figure: Sturmian graph}
\scalebox{0.6}{
\begin{VCPicture}{(-1,0)(5,4)}
\ChgEdgeLabelScale{0.8}
\VSState{(0.5,2.5)}{L}
\VSState{(3.5,2.5)}{R}
\EdgeLineDouble
\VCurveL[]{angleA=45,angleB=135,ncurv=2}{R}{L}{}
\VCurveR[]{angleA=-45,angleB=-135,ncurv=2}{R}{L}{}
\ChgEdgeLineStyle{dotted}	\EdgeL{L}{R}{}	\RstEdgeLineStyle
\end{VCPicture}
}}
\qquad
\subfigure[]{
\label{figure: Arnoux-Rauzy graph}
\scalebox{0.6}{
\begin{VCPicture}{(-1,0)(5,4)}
\ChgEdgeLabelScale{0.8}
\VSState{(0.5,2.5)}{L}
\VSState{(3.5,2.5)}{R}
\EdgeLineDouble
\VCurveL[]{angleA=45,angleB=135,ncurv=2}{R}{L}{}
\VCurveR[]{angleA=-60,angleB=-120,ncurv=1}{R}{L}{}
\VCurveR[]{angleA=-30,angleB=-150,ncurv=3}{R}{L}{}
\ChgEdgeLineStyle{dotted}	\EdgeL{L}{R}{}	\RstEdgeLineStyle
\end{VCPicture}
}}
\caption[Reduced Rauzy graphs with 1 left and 1 right special factor.]{Reduced Rauzy graphs with one left special factor and one right special factor.}
\label{figure: Rauzy graphs with one right special factor and one left special factor}
\end{figure}

\begin{figure}[h!tbp]
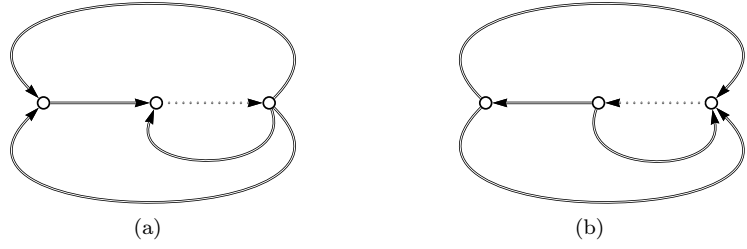

\centering
\subfigure[]{
\label{figure: graph with 1 right special factor}
\scalebox{0.6}{
\begin{VCPicture}{(-1,0)(7,4)}
\ChgEdgeLabelScale{0.8}
\VSState{(0.5,2.5)}{L_1}
\VSState{(3,2.5)}{L_2}
\VSState{(5.5,2.5)}{R}
\EdgeLineDouble
\VCurveL[]{angleA=45,angleB=135,ncurv=1.5}{R}{L_1}{}
\VCurveR[]{angleA=-45,angleB=-135,ncurv=1.5}{R}{L_1}{}
\VCurveR[]{angleA=-70,angleB=-110,ncurv=1.2}{R}{L_2}{}
\EdgeL{L_1}{L_2}{}
\ChgEdgeLineStyle{dotted}	\EdgeL{L_2}{R}{}	\RstEdgeLineStyle
\end{VCPicture}
}}
\qquad
\subfigure[]{
\label{figure: graph with 1 left special factor}
\scalebox{0.6}{
\begin{VCPicture}{(-1,0)(7,4)}
\ChgEdgeLabelScale{0.8}
\VSState{(0.5,2.5)}{L_1}
\VSState{(3,2.5)}{L_2}
\VSState{(5.5,2.5)}{R}
\EdgeLineDouble
\ReverseArrow
\VCurveL[]{angleA=45,angleB=135,ncurv=1.5}{R}{L_1}{}
\VCurveR[]{angleA=-45,angleB=-135,ncurv=1.5}{R}{L_1}{}
\VCurveR[]{angleA=-80,angleB=-110,ncurv=1.2}{R}{L_2}{}
\EdgeL{L_1}{L_2}{}
\ChgEdgeLineStyle{dotted}	\EdgeL{L_2}{R}{}	\RstEdgeLineStyle
\end{VCPicture}
}}
\caption[Reduced Rauzy graphs with different numbers of left and right.]{Reduced Rauzy graphs with different numbers of left and right special factors.}
\label{figure: Rauzy graphs with different number of left and right special factors.}
\end{figure}

\begin{figure}[h!tbp]
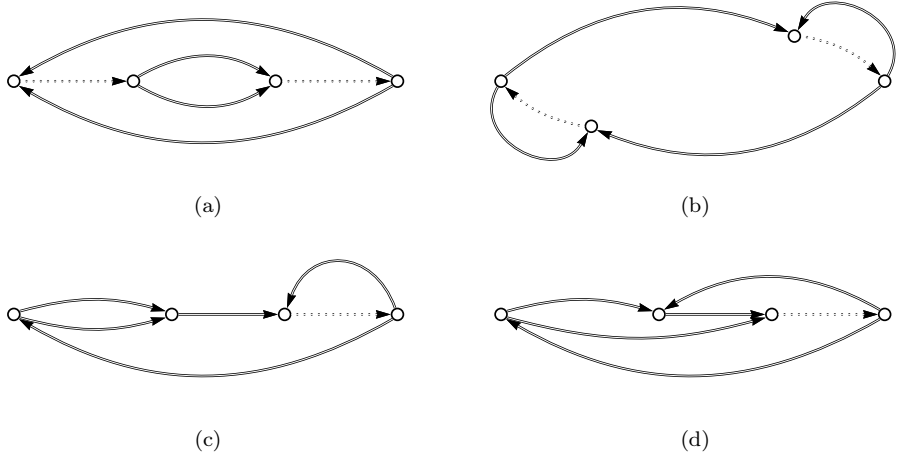

\centering
\subfigure[]{
\scalebox{0.6}{
\begin{VCPicture}{(0,0)(9,4)}
\label{figure: graph with no loop}
\ChgEdgeLabelScale{0.8}
\VSState{(0,2.5)}{L_1}
\VSState{(2.65,2.5)}{R_1}
\VSState{(5.8,2.5)}{L_2}
\VSState{(8.5,2.5)}{R_2}
\EdgeLineDouble
\LArcL{R_1}{L_2}{}
\LArcR{R_1}{L_2}{}
\LArcL{R_2}{L_1}{}
\LArcR{R_2}{L_1}{}
\ChgEdgeLineStyle{dotted}	\EdgeL{L_1}{R_1}{}	\EdgeL{L_2}{R_2}{}	\RstEdgeLineStyle
\end{VCPicture}
}}
\qquad
\subfigure[]{
\scalebox{0.6}{
\begin{VCPicture}{(0,0)(9,4)}
\label{figure: graph with 2 loops}
\ChgEdgeLabelScale{0.8}
\VSState{(0,2.5)}{R_1}
\VSState{(6.5,3.5)}{L_1}
\VSState{(8.5,2.5)}{R_2}
\VSState{(2,1.5)}{L_2}
\EdgeLineDouble
\LArcL{R_1}{L_1}{}
\VCurveL[]{angleA=-120,angleB=-120,ncurv=1.2}{R_1}{L_2}{}
\VCurveL[]{angleA=60,angleB=60,ncurv=1.2}{R_2}{L_1}{}
\LArcL{R_2}{L_2}{}
\ChgEdgeLineStyle{dotted}	\ArcL{L_1}{R_2}{}	\ArcL{L_2}{R_1}{}	\RstEdgeLineStyle
\end{VCPicture}
}}
\\
\subfigure[]{
\label{figure: graph with 1 classic loop}
\scalebox{0.6}{
\begin{VCPicture}{(0,0)(9,4)}
\ChgEdgeLabelScale{0.8}
\VSState{(0,2.5)}{R_1}
\VSState{(3.5,2.5)}{L_1}
\VSState{(6,2.5)}{L_2}
\VSState{(8.5,2.5)}{R_2}
\EdgeLineDouble
\ArcL{R_1}{L_1}{}
\ArcR{R_1}{L_1}{}
\EdgeL{L_1}{L_2}{}
\LArcL{R_2}{R_1}{}
\VCurveL[]{angleA=110,angleB=70,ncurv=1.2}{R_2}{L_2}{}
\ChgEdgeLineStyle{dotted}	\EdgeL{L_2}{R_2}{}	\RstEdgeLineStyle
\end{VCPicture}
}}
\qquad
\subfigure[]{
\label{figure: graph with 1 long loop}
\scalebox{0.6}{
\begin{VCPicture}{(0,0)(9,4)}
\ChgEdgeLabelScale{0.8}
\VSState{(0,2.5)}{R_1}
\VSState{(3.5,2.5)}{L_1}
\VSState{(6,2.5)}{L_2}
\VSState{(8.5,2.5)}{R_2}
\EdgeLineDouble
\ArcL{R_1}{L_1}{}
\ArcR{R_1}{L_2}{}
\EdgeL{L_1}{L_2}{}
\LArcL{R_2}{R_1}{}
\LArcR{R_2}{L_1}{}
\ChgEdgeLineStyle{dotted}	\EdgeL{L_2}{R_2}{}	\RstEdgeLineStyle
\end{VCPicture}
}}
\caption[Reduced Rauzy graphs with 2 left and 2 right special factors.]{Reduced Rauzy graphs with two left and two right special factors.}
\label{figure: Rauzy graphs with two left and right special factors}
\end{figure}
}

From Remark~\ref{remark: sigme = identity}, it is enough to consider Rauzy graphs of order $i_n$, $n \in \N$. To this aim, we have to merge the vertices that are linked by dots in Figure~\ref{figure: Rauzy graphs with one right special factor and one left special factor} to Figure~\ref{figure: Rauzy graphs with two left and right special factors}. Observe that both Figures~\ref{figure: graph with no loop} and~\ref{figure: graph with 2 loops} give rise to two different graphs: one with one bispecial vertex and one right special vertex and one with two bispecial vertices. This gives rise to 10 different types of graphs. They are represented in Figure~\ref{figure: Rauzy graphs with at least 1 bispecial vertex}.


\begin{figure}[h!tbp]
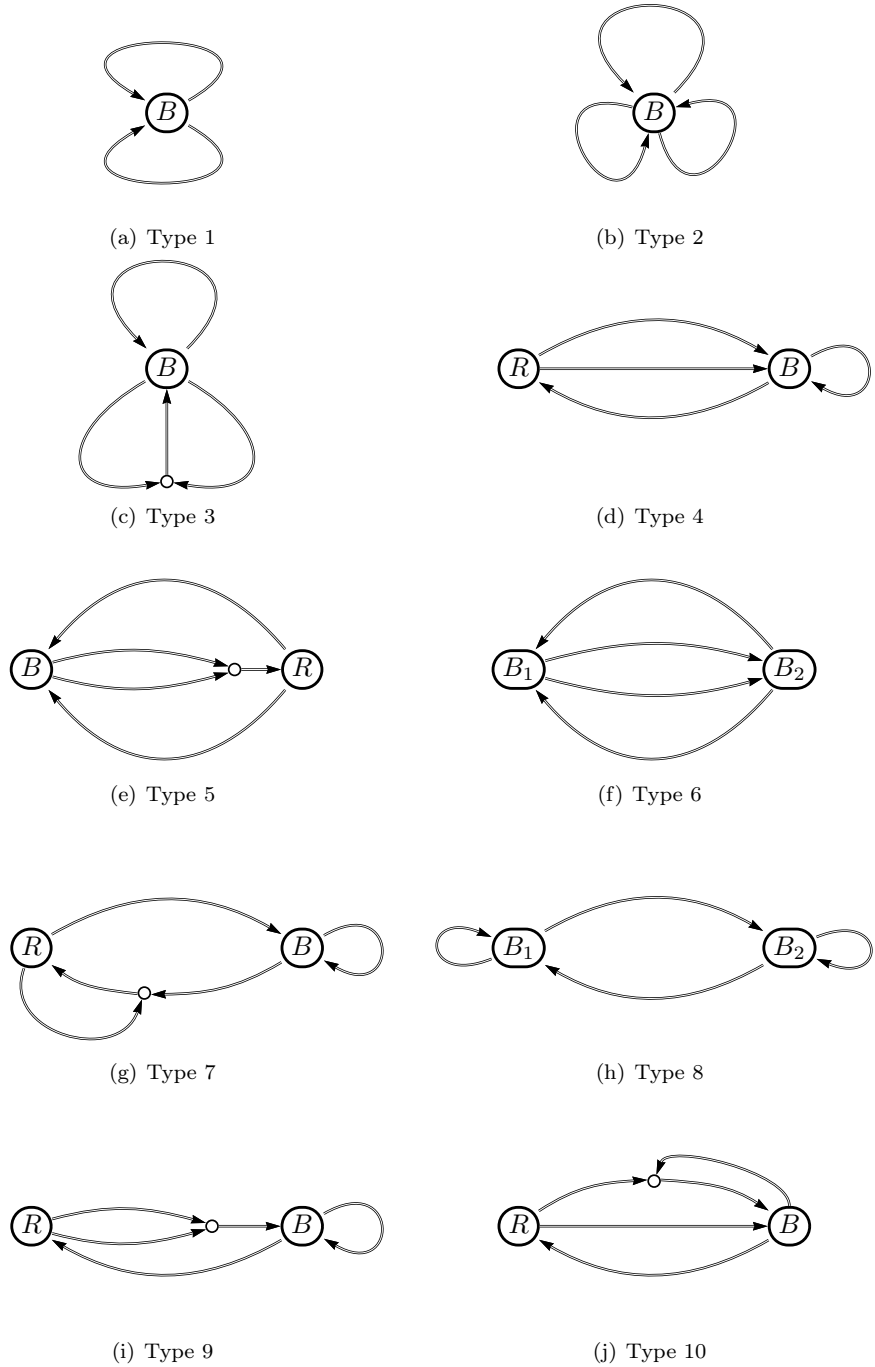

\centering
\subfigure[Type 1]{
\label{figure: type 1}
\scalebox{0.6}{
\begin{VCPicture}{(-0.5,0)(8.5,5)}
\ChgEdgeLabelScale{0.5}
\StateVar[B]{(4,2.5)}{B}
\EdgeLineDouble
\VCurveR[]{angleA=30,angleB=150,ncurv=7}{B}{B}{}
\VCurveL[]{angleA=-30,angleB=-150,ncurv=7}{B}{B}{}
\end{VCPicture}
}}
\qquad
\subfigure[Type 2]{
\label{figure: type 2}
\scalebox{0.6}{
\begin{VCPicture}{(-0.5,0)(8.5,5)}
\ChgEdgeLabelScale{0.5}
\StateVar[B]{(4,2.5)}{B}
\EdgeLineDouble
\VCurveR[]{angleA=45,angleB=135,ncurv=8}{B}{B}{}
\VCurveR[]{angleA=165,angleB=255,ncurv=8}{B}{B}{}
\VCurveR[]{angleA=285,angleB=15,ncurv=8}{B}{B}{}
\end{VCPicture}
}}
\\
\subfigure[Type 3]{
\label{figure: type 3}
\scalebox{0.6}{
\begin{VCPicture}{(-0.5,0)(8.5,5)}
\ChgEdgeLabelScale{0.5}
\StateVar[B]{(4,3)}{B}
\VSState{(4,0.5)}{L}
\EdgeLineDouble
\VCurveL[]{angleA=45,angleB=135,ncurv=8}{B}{B}{}
\VCurveR[]{angleA=-30,angleB=-10,ncurv=2}{B}{L}{}
\VCurveR[]{angleA=-150,angleB=190,ncurv=2}{B}{L}{}
\EdgeL{L}{B}{}
\end{VCPicture}
}}
\qquad
\subfigure[Type 4]{
\label{figure: type 4}
\scalebox{0.6}{
\begin{VCPicture}{(-0.5,0)(8.5,5)}
\ChgEdgeLabelScale{0.5}
\StateVar[R]{(1,3)}{R}
\StateVar[B]{(7,3)}{B}
\EdgeLineDouble
\LArcL{R}{B}{}
\EdgeL{R}{B}{}
\LoopE{B}{}
\LArcL{B}{R}{}
\end{VCPicture}
}}
\\
\subfigure[Type 5]{
\scalebox{0.6}{
\begin{VCPicture}{(-0.5,0)(8.5,5)}
\label{figure: type 5}
\ChgEdgeLabelScale{0.5}
\StateVar[B]{(1,2.5)}{R}
\StateVar[R]{(7,2.5)}{B}
\VSState{(5.5,2.5)}{L}
\EdgeLineDouble
\ArcL{R}{L}{}
\ArcR{R}{L}{}
\Edge{L}{B}
\VCurveL[.05]{angleA=-130,angleB=-50,ncurv=1}{B}{R}{}
\VCurveR[.05]{angleA=130,angleB=50,ncurv=1}{B}{R}{}		
\end{VCPicture}
}}
\qquad
\subfigure[Type 6]{
\scalebox{0.6}{
\begin{VCPicture}{(-0.5,0)(8.5,5)}
\label{figure: type 6}
\ChgEdgeLabelScale{0.5}
\StateVar[B_1]{(1,2.5)}{B_1}
\StateVar[B_2]{(7,2.5)}{B_2}
\EdgeLineDouble
\ArcL{B_1}{B_2}{}
\ArcR{B_1}{B_2}{}
\VCurveL[.05]{angleA=-130,angleB=-50,ncurv=1}{B_2}{B_1}{}
\VCurveR[.05]{angleA=130,angleB=50,ncurv=1}{B_2}{B_1}{}
\end{VCPicture}
}}
\\
\subfigure[Type 7]{
\scalebox{0.6}{
\begin{VCPicture}{(-0.5,0)(8.5,5)}
\label{figure: type 7}
\ChgEdgeLabelScale{0.5}
\StateVar[R]{(1,2.5)}{R}
\StateVar[B]{(7,2.5)}{B}
\VSState{(3.5,1.5)}{L}
\EdgeLineDouble
\LArcL{R}{B}{}
\ArcL{B}{L}{}
\LoopE{B}{}
\ArcL{L}{R}{}
\VCurveR{angleA=-110,angleB=-110,ncurv=1.15}{R}{L}{}
\end{VCPicture}
}}
\qquad
\subfigure[Type 8]{
\scalebox{0.6}{
\begin{VCPicture}{(-0.5,0)(8.5,5)}
\label{figure: type 8}
\ChgEdgeLabelScale{0.5}
\StateVar[B_1]{(1,2.5)}{B_1}
\StateVar[B_2]{(7,2.5)}{B_2}
\EdgeLineDouble
\LArcL{B_1}{B_2}{}
\LArcL{B_2}{B_1}{}
\CLoopL{0}{B_2}{}
\CLoopL{180}{B_1}{}
\end{VCPicture}
}}
\\
\subfigure[Type 9]{
\label{figure: type 9}
\scalebox{0.6}{
\begin{VCPicture}{(-0.5,0)(8.5,5)}
\ChgEdgeLabelScale{0.5}
\StateVar[R]{(1,2.5)}{R}
\VSState{(5,2.5)}{L}
\StateVar[B]{(7,2.5)}{B}
\EdgeLineDouble
\ArcL{R}{L}{}
\ArcR{R}{L}{}
\EdgeL{L}{B}{}
\LArcL{B}{R}{}
\LoopL{0}{B}{}
\end{VCPicture}
}}
\qquad
\subfigure[Type 10]{
\label{figure: type 10}
\scalebox{0.6}{
\begin{VCPicture}{(-0.5,0)(8.5,5)}
\ChgEdgeLabelScale{0.5}
\StateVar[R]{(1,2.5)}{R}
\VSState{(4,3.5)}{L}
\StateVar[B]{(7,2.5)}{B}
\EdgeLineDouble
\ArcL{R}{L}{}
\ArcL{L}{B}{}
\Edge{R}{B}{}
\LArcL{B}{R}{}
\VCurveL[]{angleA=90,angleB=70,ncurv=0.5}{B}{L}{}
\end{VCPicture}
}}
\caption{Reduced Rauzy graphs with at least one bispecial vertex.}
\label{figure: Rauzy graphs with at least 1 bispecial vertex}
\end{figure}

\begin{remark}
\label{remark: type sans bispecial}
In the sequel, we sometimes talk about the type of a reduced Rauzy graph $g_k$ with $k \notin \{i_n \mid n \in \N\}$. In that case, the type of that graph is simply the type of $g_{\min \{i_n \mid i_n \geq k\}}$. This makes no confusions since if $R$ is a right special vertex in a Rauzy graph, the circuits starting from it have the same right labels (and full labels) of those starting from the smallest bispecial vertex (in a Rauzy graph of larger order) containing $R$ as a suffix. We also sometimes talk about the type of a Rauzy graph (and not reduced Rauzy graph). This simply corresponds to the type of the corresponding reduced Rauzy graph.
\end{remark}

\subsubsection{Graph of graphs}
\label{subsubsection: graph of graphs}

Now that we have defined all types of graphs, we can check which evolutions are available, \textit{i.e.}, which type of graphs can evolve to which type of graphs. It is clear that a given Rauzy graph cannot evolve to any type of Rauzy graphs. For example, if $G_n$ is a graph of type 4, both right special vertices can be extended by only two letters. Since for any word $u$ and for any suffix $v$ of $u$, we have $\delta^+(v) \geq  \delta^+(u)$, the graph $G_n$ will never evolve to a graph of type 2 or 3. Section~\ref{subsubsection: An example} shows that a graph of type 1 can evolve to graphs of type 1, 7 or 8.

By computing all available evolutions, we can define the \definir{graph of graphs} as the directed graph with 10 vertices (one for each type of Rauzy graph) such that there is an edge from $i$ to $j$ if a Rauzy graph of type $i$ can evolve to a Rauzy graph of type $j$. This graph is represented in Figure~\ref{figure: graph of graphs}. 
A detailed computation of evolutions is available in Section~\ref{appendix: evolutions}. 

\begin{figure}[h!tbp]
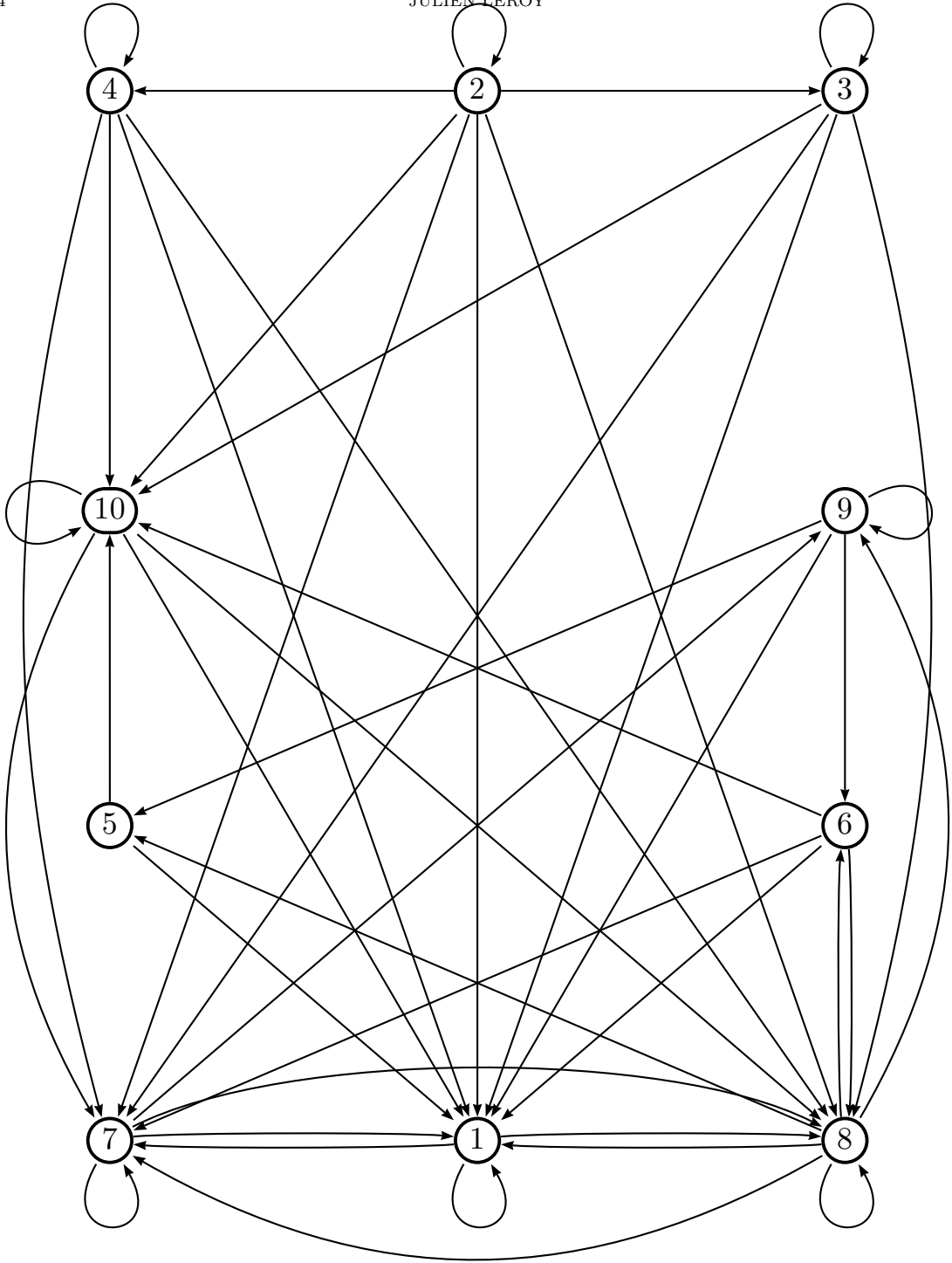

\centering
\scalebox{0.8}{
\begin{VCPicture}{(0,-3)(14,21)}
\StateVar[4]{(0,20)}{10}				\StateVar[2]{(7,20)}{1} 			\StateVar[3]{(14,20)}{2}	
\StateVar[10]{(0,12)}{8}												\StateVar[9]{(14,12)}{9}
\StateVar[5]{(0,6)}{4}													\StateVar[6]{(14,6)}{5}
\StateVar[7]{(0,0)}{6}					\StateVar[1]{(7,0)}{3}				\StateVar[8]{(14,0)}{7}
\LoopN{10}{}	\LoopN{1}{}		\LoopN{2}{}		
\LoopS{6}{}		\LoopS{7}{}		\LoopW{8}{}		\LoopS{3}{}		\LoopE{9}{}
\EdgeL{1}{2}{}		\EdgeL{1}{10}{}		\EdgeL{1}{6}{}		\EdgeL{1}{7}{}		\EdgeL{1}{3}{}		\EdgeL{1}{8}{}
\EdgeL{2}{3}{}		\EdgeL{2}{6}{}		\ArcL{2}{7}{}			\Edge{2}{8}{}
\VArcL{arcangle=5,ncurv=.3}{3}{6}{}		\VArcL{arcangle=5,ncurv=.3}{3}{7}{}
\EdgeL{4}{3}{}		\EdgeL{4}{8}{}
\EdgeL{5}{3}{}		\EdgeL{5}{6}{}		\VArcL{arcangle=5,ncurv=.3}{5}{7}{}	\EdgeL{5}{8}{}
\VArcL{arcangle=5,ncurv=.3}{6}{3}{}		\VArcL{arcangle=27,ncurv=.5}{6}{7}{}		\EdgeL{6}{9}{}
\VArcL{arcangle=5,ncurv=.3}{7}{3}{}		\EdgeL{7}{4}{}		\VArcL{arcangle=5,ncurv=.3}{7}{5}{}	\LArcL{7}{6}{}		\LArcR{7}{9}{}
\EdgeL{8}{3}{}		\LArcR{8}{6}{}		\EdgeL{8}{7}{}
\EdgeL{9}{3}{}		\EdgeL{9}{4}{}		\EdgeL{9}{5}{}			
\EdgeL{10}{3}{}		\ArcR{10}{6}{}		\EdgeL{10}{7}{}		\EdgeL{10}{8}{}
\end{VCPicture}
}
\caption{Graph of graphs.}
\label{figure: graph of graphs}
\end{figure}

\subsection{A critical result}	
\label{subsection: a critical result}

Now that we know all possible Rauzy graphs we have to deal with, we can define the bijections $\theta_n$ of Definition~\ref{definition: definition des morphismes}.
A first necessary condition to need only a finite set of morphisms is that the alphabets $A_n$ are bounded.
In this section we prove that when the first difference of complexity is bounded by 2, they always contain 2 or 3 letters. 
This result seems to be inherent to that class of complexity~\cite{Durand-Leroy-Richomme}. 

We need two technical lemmas to simplify the proof that $\card{A_n} \in \{2,3\}$ for all $n$. 

\begin{lemma}
\label{lemma: bispecial fort donne 2 speciaux droits}
Let $A$ be an alphabet. If $(X,T)$ is a minimal subshift over $A$ satisfying $p(n+1)-p(n) \leq 2$ for all $n$ and if $B$ is a strong bispecial factor of $X$, then any right special factor of length $\ell > |B|$ admits $B$ as a suffix.
\end{lemma}
\begin{proof}
Indeed, $B$ being supposed to be strong bispecial, its bilateral order $m(B)$ is positive. Observe that, by definition, $m(B)>0$ is equivalent to the inequality
\[
	\sum_{aB \in \fac{}{X}} (\delta^+(aB) - 1) > \delta^+(B) - 1,
\]
which is true only if there are at least two letters $a$ and $b$ in $A$ such that $aB$ and $bB$ are right special (since $\delta^+(aB) \leq \delta^+(B)$). As there can exist at most 2 right special factors of each length (because $p(n+1)-p(n) \leq 2$) and as any suffix of a right special factor is still a right special factor, the result holds. 
\end{proof}

The following result is a direct consequence of Lemma~\ref{lemma: bispecial fort donne 2 speciaux droits}.

\begin{corollary}
\label{corollary: v_n = B}
Let $(X,T)$ be a minimal subshift satisfying $1 \leq p(n+1)-p(n) \leq 2$ for all $n$ and let $(U_n)_{n \in \N}$ be a sequence of right special factors of $X$ fulfilling the conditions of Lemma~\ref{lemma: suite de speciaux droits}. For any strong bispecial factor $B$ of length $n$ of $X$, we have $B = U_n$. In particular, if there are infinitely many strong bispecial factors in $\fac{}{X}$, the sequence $(U_n)_{n \in \N}$ of Lemma~\ref{lemma: suite de speciaux droits} is unique.
\end{corollary}

\begin{lemma}
\label{lemma: trop de boucles donne un bispecial fort}
Let $G_n$ be a Rauzy graph. If there is a right special vertex $R$ in $G_n$ with $\delta^+ (R)= 2$, an $n$-circuit $q$ starting from $R$, two paths $p$ and $s$ in $G_n$ and two integers $k_1$ and $k_2$, $k_1 < k_2 -1$, such that \begin{enumerate}
	\item $i(p) = o(s) = R$;	
	\item $p$ is not a suffix of $q$;
	\item $q$ is not a suffix of $p$;
	\item the first edge of $s$ is not the first edge of $q$;
	\item both paths $pq^{k_1}s$ and $pq^{k_2}s$ are allowed;
\end{enumerate}
then there is a strong bispecial factor $B$ that admits $R$ as a suffix.
\end{lemma}
\begin{proof}
Since $i(p) = o(q) = R$ but $p$ and $q$ are not suffix of each other, there is a left special vertex $L$ in $G_n$ and two edges $e_1$ in $p$ and $e_2$ in $q$ such that $p$ and $q$ agree on a path $q'$ from $L$ to $R$ and $i(e_1) = i(e_2) = L$. Let $\alpha$ and $\beta$ be the respective left labels of $e_1$ and $e_2$. Let also $a$ and $b$ respectively denote the right labels of the first edge of $q$ and of $s$. By hypothesis we have $a \neq b$.

Now let us prove that the word $\lambda(q'q^{k_1})$ is strong bispecial. As the paths $pq^{k_1}s$ and $pq^{k_2}s$ are allowed, the four words $\alpha \lambda(q'q^{k_1}) a$, $\alpha \lambda(q'q^{k_1}) b$, $\beta \lambda(q'q^{k_1}) a$ and $\beta \lambda(q'q^{k_1}) b$ belong to $\fac{}{X}$. Consequently we have
\[
	\delta^+(\alpha \lambda(q'q^{k_1})) + \delta^+(\beta \lambda(q'q^{k_1}))  = 4.
\]
Moreover, as the word $\lambda(q'q^{k_1})$ admits $R$ as a suffix, we have $\delta^+(\lambda(q'q^{k_1})) \leq \delta^+(R) = 2$ which implies that $m(\lambda(q'q^{k_1})) >0$.
\end{proof}

\begin{proposition}
\label{prop: 3 circuits par special droit}
Let $(X,T)$ be a minimal subshift satisfying $1 \leq p(n+1)-p(n) \leq 2$ for all $n$ and let $(U_n)_{n \geq N}$ be a sequence of right special factors fulfilling the conditions of Lemma~\ref{lemma: suite de speciaux droits}.
Then for all right special factors $U_n$, there are at most 3 allowed $n$-circuits starting from $U_n$.
\end{proposition}

\begin{proof}
Suppose that there exist 4 allowed $n$-circuits starting from the vertex $U_n$ in the graph $G_n(X)$ and let us have a look at all possible reduced Rauzy graphs. We see that this is possible only if there exist two right special factors of length $n$. More precisely, this is only possible if $U_n$ corresponds to the leftmost right special vertex in Figures~\ref{figure: graph with 1 left special factor},~\ref{figure: graph with 1 classic loop} and~\ref{figure: graph with 1 long loop} or to any right special vertex in Figures~\ref{figure: graph with no loop} and~\ref{figure: graph with 2 loops} (as these two graphs present a kind of \enquote{symmetry}). We will show that for each of these graphs, the existence of 4 $n$-circuits starting from the described vertices implies that the other right special factor $R$ of length $n$ is a suffix of a strong bispecial factor $B$ of length $m \geq n$ in $\fac{}{X}$. Then, due to Corollary~\ref{corollary: v_n = B}, $U_m = B$ so $U_n$ is not a suffix of $U_m$ which contradicts the hypothesis.

The result clearly holds for graphs as represented in Figure~\ref{figure: graph with no loop} and it is a direct consequence of Lemma~\ref{lemma: trop de boucles donne un bispecial fort} for graphs as represented at Figure~\ref{figure: graph with 2 loops} (since the existence of 4 $n$-circuits implies that 3 of them goes through the loop respectively $k_1$, $k_2$ and $k_3$ times, $k_1 < k_2 < k_3$).

For graphs as represented in Figure~\ref{figure: graph with 1 classic loop}, we have to consider several cases. To be clearer, Figure~\ref{figure: graph with 1 classic loop with some labels} represents the same graph with some labels. The letters $\alpha$ and $\beta$ are the left extending letters of $L_1$ in $\fac{}{X}$ and the letters $a$ and $b$ are the right extending letters of $R_2$ in $\fac{}{X}$. If there are three $n$-circuits starting from $R_1$, going through a same simple path from $R_1$ to $L_1$ and passing through the loop $p = L_2 \rightarrow R_2 \rightarrow L_2$ respectively $k_1$, $k_2$, and $k_3$ times, $k_1 < k_2 < k_3$, then we can conclude using Lemma~\ref{lemma: trop de boucles donne un bispecial fort}. Otherwise, for both simple paths from $R_1$ to $L_1$, there are two $n$-circuits passing through it. Let $k_{\alpha,1}$ and $k_{\alpha,2}$, $k_{\alpha,1} < k_{\alpha,2}$ (resp. $k_{\beta,1}$ and $k_{\beta,2}$, $k_{\beta,1}< k_{\beta,2}$) be the number of times that the two circuits passing through the edge with left label $\alpha$ (resp. $\beta$) can pass through the loop $p$. If $k_{\alpha,1} < k_{\alpha,2} - 1$ or if $k_{\beta,1} < k_{\beta,2} - 1$ or if $k_{\alpha,1} \neq k_{\beta,1}$, we conclude using Lemma~\ref{lemma: trop de boucles donne un bispecial fort}. Otherwise, we have $k_{\alpha,1} = k_{\beta,1}$ and $k_{\alpha,2} = k_{\beta,2} = k_{\alpha,1} + 1$ and we can easily check that the full label of the path $q = L_1 \left( \rightarrow L_2 \rightarrow R_2 \right)^{k_{\alpha,1}}$ is a strong bispecial factor.

\begin{figure}[h!tbp]
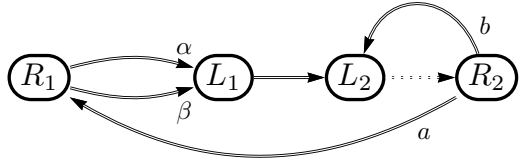

\centering
\scalebox{0.7}{
\begin{VCPicture}{(0,0)(9,4)}
\ChgEdgeLabelScale{0.8}
\StateVar[R_1]{(0,2.5)}{R_1}
\StateVar[L_1]{(3.5,2.5)}{L_1}
\StateVar[L_2]{(6,2.5)}{L_2}
\StateVar[R_2]{(8.5,2.5)}{R_2}
\EdgeLineDouble
\ArcL[.9]{R_1}{L_1}{\alpha}
\ArcR[.9]{R_1}{L_1}{\beta}
\EdgeL{L_1}{L_2}{}
\LArcL[.1]{R_2}{R_1}{a}
\VCurveR[.1]{angleA=110,angleB=70,ncurv=1.2}{R_2}{L_2}{b}
\ChgEdgeLineStyle{dotted}	\EdgeL{L_2}{R_2}{}	\RstEdgeLineStyle
\end{VCPicture}
}
\caption{Graph as in Figure~\ref{figure: graph with 1 classic loop} with some labels.}
\label{figure: graph with 1 classic loop with some labels}
\end{figure}

The cases of graphs as represented at Figures~\ref{figure: graph with 1 left special factor} and~\ref{figure: graph with 1 long loop} can be treated in a similar way.
\end{proof}

Proposition~\ref{prop: 3 circuits par special droit} cannot be extended to the general case. Indeed, there exist~\cite{Durand-Leroy-Richomme} minimal subshift with linear complexity and such that the number of $n$-circuits to any factor of length $n$ increases with $n$.

\subsection{A procedure to assign letters to circuits}	
\label{subsection: a procedure to assign letters to circuits}

Now let us explicitly determine the bijections $\theta_{i_n}$. 
We would like to define them for each graph represented at Figure~\ref{figure: Rauzy graphs with at least 1 bispecial vertex} in such a way that two Rauzy graphs of same type provide the same bijection $\theta_n$. 
In that case, a given evolution (from $G_{i_n}$ to $G_{i_{n+1}}$) would always provide the same morphism $\gamma_{i_n}$ (which is equal to $\gamma_{i_n} \cdots \gamma_{i_{n+1}-1}$) of Definition~\ref{definition: definition des morphismes}. However, we will see that it is sometimes impossible to give enough details about $\theta_{i_n}$ so that the morphisms are sometimes defined up to permutations of the letters.

From Lemma~\ref{prop: 3 circuits par special droit} we know that $\card{A_{i_n}} \in \{2,3\}$ for all $n$ (1 is not enough since the number of $i_n$-circuits is at least $\delta^+(U_{i_n}) \geq 2$). 
From Definition~\ref{definition: definition des morphismes} we then have $A_{i_n} \in \left\{ \{ 0,1 \}, \{0,1,2\} \right\}$ depending on $n$.

Observe that, in the description of the bijections $\theta_{i_n}$ below, we sometimes express some restrictions on the number of times that some circuits can pass through a loop in the consider type of Rauzy graph. The reason for this is that if the circuits do not satisfy those restrictions, the right special factor that is not $U_{i_n}$ is a suffix of a strong bispecial factor (by Lemma~\ref{lemma: trop de boucles donne un bispecial fort}) which contradicts Corollary~\ref{corollary: v_n = B}. 

If $G_n$ is a Rauzy graph, then an \definir{$n$-segment} is a path that starts in a right special vertex and ends in a right special vertex and that does not go through any other right special vertex.

\begin{enumerate}
\label{enumerate: definition des circuits}

	\item \textbf{Type 1}:
	\label{item 1} there exists only one right special vertex and the two possible circuits are the two loops. One is $\theta_{i_n}(0)$ and the other is $\theta_{i_n}(1)$ and we cannot be more precise (like we are for graphs of type 2 or 3 below).
	
	\item \textbf{Type 2 and 3}:
	\label{item 2} also here there exists only one right special vertex and the three possible circuits are the three possible loops $\theta_{i_n}(0)$, $\theta_{i_n}(1)$ and $\theta_{i_n}(2)$. However, as shown by Figure~\ref{figure: graph of graphs}, the only graphs that can evolve to a graph of type 2 (resp. of type 3) are the graphs of type 2 (resp. of type 2 and 3). Moreover after such an evolution, the right labels of the three loops start with the same letter as before the evolution. Consequently we suppose that for all $i \in \{0,1,2\}$, $i$ is prefix of $\lambda_R \circ \theta_{i_n}(i)$.
	
	\item \textbf{Type 4}:
	\label{item 4} first consider $U_{i_n} = R$. There exist two segments from $R$ to $B$. Consequently, there exist at least two circuits $\theta_{i_n}(0)$ and $\theta_{i_n}(1)$, each of them passing through one of the two segments and looping respectively $k$ and $\ell$ times, $k + \ell \geq 1$, in the loop $B \rightarrow B$ before coming back to $R$. If there exists a third circuit, then we suppose it starts with the same segment as the circuit $\theta_{i_n}(0)$ does, and then goes through the loop exactly $k-1$ times. In this case, we must have $\ell \leq k$. If the third circuit does not exist, then we suppose that $k \geq \ell$ so we have $k \geq \ell \geq 0$ and $k+ \ell \geq 1$. 
	
	Now consider $U_{i_n} = B$. There exist exactly three circuits: the circuit that does not pass through the vertex $R$ is denoted by $\theta_{i_n}(0)$ and the two others, $\theta_{i_n}(1)$ and $\theta_{i_n}(2)$, are going to the vertex $R$ and then are coming back to $B$ with one of the two segments from $R$ to $B$.
	
	\item \textbf{Type 5 and 6}:
	\label{item 56} as a consequence of Remark~\ref{remark: type sans bispecial}, the circuits are the same whatever the type of graphs is. Moreover, from the symmetry of theses graphs, it is useless to make a distinction between the two right special vertices. Suppose $U_{i_n} = R$ for a graph of type 5. There exist four possible circuits (but Proposition~\ref{prop: 3 circuits par special droit} implies that only three among them are allowed) and we only impose some restrictions to their labels: the circuits $\theta_{i_n}(0)$ and $\theta_{i_n}(1)$ must pass through two different segments from $R$ to $B$ and through two different segments from $B$ to $R$. If the third circuit $\theta_{i_n}(2)$ exists, then it pass through the same segment from $R$ to $B$ as $\theta_{i_n}(0)$ does and through the same segment from $B$ to $R$ as $\theta_{i_n}(1)$ does.
	
	\item \textbf{Type 7 and 8}:
	\label{item 78} like for graphs of type 5 or 6, the starting vertex and the type of the graph does not change anything to the definition of the circuits. Suppose $U_{i_n} = R$ for a graph of type 7. We consider that $\theta_{i_n}(0)$ is the circuit that does not pass through the vertex $B$. The circuit $\theta_{i_n}(1)$ goes to $B$, passes through the loop $B \rightarrow B$ $k$ times, $k \geq 1$, and then comes back to $R$. The circuit $\theta_{i_n}(2)$, if it exists, is the same as $\theta_{i_n}(1)$ but passes through the loop $B \rightarrow B$ $k-1$ times instead of $k$ times.
	
	\item \textbf{Type 9}:
	\label{item 9} suppose $U_{i_n} = R$. Like for graphs of type 4, we consider the two circuits $\theta_{i_n}(0)$ and $\theta_{i_n}(1)$, each of them going through different segments from $R$ to $B$ and looping respectively $k$ and $\ell$ times in the loop $B \rightarrow B$, $k + \ell \geq 1$, before coming back to $R$. However for these graphs, $k$ and $\ell$ must satisfy $k-\ell \leq 1$ otherwise the vertex $B$ would become strong bispecial (see Lemma~\ref{lemma: trop de boucles donne un bispecial fort}). Moreover, if the third circuit $\theta_{i_n}(2)$ exists, we suppose it starts like $\theta_{i_n}(0)$ does and passes through the loop exactly $k-1$ times. In this case, the circuit $\theta_{i_n}(1)$ cannot go through the loop $k+1$ times otherwise $B$ would again become strong bispecial. Hence we always suppose $k \geq \ell$. Consequently, $\ell$ can only take the values $k-1$ and $k$ even if the circuit $\theta_{i_n}(2)$ does not exist.

	Now suppose $U_{i_n} = B$. There exist exactly three circuits: the circuit that does not pass through the vertex $R$ is $\theta_{i_n}(0)$ and the two other circuits, $\theta_{i_n}(1)$ and $\theta_{i_n}(2)$, are going to the vertex $R$ and then are coming back to $B$ with one of the two segments from $R$ to $B$.	
	
	\item \textbf{Type 10}:
	\label{item 10} suppose $U_{i_n} = R$. Let $x$ denote the segment from $R$ to $B$ that passes only through non-left-special vertices; $y$ is the other segment from $R$ to $B$. We consider that $\theta_{i_n}(0)$ (resp. by $\theta_{i_n}(1)$) is the circuit that starts with $y$ (resp. with $x$), passes $k$ times (resp. $\ell$ times) through the loop $B \rightarrow B$, $k + \ell \geq 1$, and then comes back to $R$. If the third circuit $\theta_{i_n}(2)$ exists, then it starts with $x$ or $y$ and loops respectively $k-1$ or $\ell -1$ times before coming back to $R$. Moreover, if $\theta_{i_n}(2)$ starts with $x$, then we must have $k \leq \ell -1$ and if $\theta_{i_n}(2)$ starts with $y$, then we must have $\ell \leq k$ (because of Lemma~\ref{lemma: trop de boucles donne un bispecial fort}). 

Now suppose $U_{i_n} = B$. There are exactly three circuits. The loop $B \rightarrow B$ is $\theta_{i_n}(0)$, the circuit passing through the segment $y$ is $\theta_{i_n}(1)$ and the circuit passing through $x$ is $\theta_{i_n}(2)$.

\end{enumerate}

\subsection{Computation of the morphisms $\gamma_n$}	
\label{subsection: computation of the morphisms}

Now that we know the bijections $\theta_{i_n}$, we can compute the morphisms $\gamma_{i_n}$ of Definition~\ref{definition: definition des morphismes} (knowing $\gamma_{i_n}$ is enough since we have supposed that for all $k \notin \{i_n \mid n \in \N\}$, $\gamma_k = id$). 
As announced at the beginning of the section, we only present the method on the example of Section~\ref{subsubsection: An example}. 
A detailed computation of all evolutions and all corresponding morphisms is available in Appendix~\ref{appendix: evolutions}.
However, not all morphisms in that list will be needed to get the $\S$-adic characterization of Section~\ref{section: caraterisation}. At each step, we will provide the concerned morphisms.

Suppose $G_{i_n}$ is a graph of type 1 as in Figure~\ref{figure: Sturmian graph with bispecial factor and labels} (on page~\pageref{figure: Sturmian graph with bispecial factor and labels}). By definition of $\theta_{i_n}$ for this type of graphs, $\theta_{i_n}(0)$ and $\theta_{i_n}(1)$ are the two loops of the graph. Suppose that $\theta_{i_n}$ maps $0$ to the $i_n$-circuit starting with the letter $a$ and $1$ to the $i_n$-circuit starting with the letter $b$. For the two first evolutions (Figure~\ref{figure: Evolution of a sturmian graph with ordinary bispecial factor} and~\ref{figure: Evolution of a sturmian graph with ordinary bispecial factor'}), $G_{i_{n+1}}$ is again of type 1. By definition of $\theta_{i_{n+1}}$ for this type of graphs, we therefore have two possibilities for each evolution. Indeed, in Figure~\ref{figure: Evolution of a sturmian graph with ordinary bispecial factor} we have either 
\[
	(\psi_{i_n} \circ \theta_{i_n+1}(0), \psi_{i_n} \circ \theta_{i_n+1}(1)) = (\theta_{i_n}(0), \theta_{i_n}(10))
\] 
or 
\[
	(\psi_{i_n} \circ \theta_{i_n+1}(0), \psi_{i_n} \circ \theta_{i_n+1}(1)) = (\theta_{i_n}(10), \theta_{i_n}(0))
\] 
and in Figure~\ref{figure: Evolution of a sturmian graph with ordinary bispecial factor'} we have either 
\[
	(\psi_{i_n} \circ \theta_{i_n+1}(0), \psi_{i_n} \circ \theta_{i_n+1}(1)) = (\theta_{i_n}(01), \theta_{i_n}(1))
\]
or 
\[
	(\psi_{i_n} \circ \theta_{i_n+1}(0), \psi_{i_n} \circ \theta_{i_n+1}(1)) = (\theta_{i_n}(1), \theta_{i_n}(01)).
\]
The four morphisms labelling the edge from 1 to 1 in the graph of graphs are therefore
\begin{eqnarray}
\label{eq: 1 to 1}
	\begin{cases}
		0 \mapsto 0	\\
		1 \mapsto 10
	\end{cases}
	&	&
	\begin{cases}
		0 \mapsto 10	\\
		1 \mapsto 0
	\end{cases}	\nonumber
	\\
	\\
	\begin{cases}
		0 \mapsto 01	\\
		1 \mapsto 1
	\end{cases}
	&	&
	\begin{cases}
		0 \mapsto 1	\\
		1 \mapsto 01
	\end{cases}	\nonumber
\end{eqnarray}

For the third evolution (Figure~\ref{figure: Evolution of a sturmian graph with strong bispecial factor}), the bijection $\theta_{i_n+1}$ (hence $\theta_{i_{n+1}}$) depends on $U_{i_n+1}$. If $U_{i_n+1} = \alpha B$ we have 
\[
	(\psi_{i_n}  \theta_{i_n+1}(0),\psi_{i_n}  \theta_{i_n+1}(1),\psi_{i_n}  \theta_{i_n+1}(2)) 
	= (\theta_{i_n}(0), \theta_{i_n}(1^k0), \theta_{i_n}(1^{k-1}0)) 
\]
for an integer $k \geq 2$ (remember that the circuit $\theta_{i_n+1}(2)$ might not exist). Similarly, if $U_{i_n+1} = \beta B$ we have
\[
	(\psi_{i_n}  \theta_{i_n+1}(0),\psi_{i_n}  \theta_{i_n+1}(1),\psi_{i_n}  \theta_{i_n+1}(2)) 
	= (\theta_{i_n}(1), \theta_{i_n}(0^k1), \theta_{i_n}(0^{k-1}1)) 
\]
for an integer $k \geq 2$. Consequently, there are infinitely many morphisms labelling the edges from 1 to 7 and from 1 to 8 (one for each $k \geq 2$) but they all have one of the following two shapes:

\begin{eqnarray}
\label{equation: morphisms 1 to 8 sur 3 lettres}
	\begin{cases}
		0 \mapsto 0		\\
		1 \mapsto 1^{k} 0 	\\
		2 \mapsto 1^{k-1} 0	
	\end{cases}
	&
	\text{and }
	&
	\begin{cases}
		0 \mapsto 1		\\
		1 \mapsto 0^{k} 1 	\\
		2 \mapsto 0^{k-1} 1
	\end{cases}.
\end{eqnarray}

\subsection{Sketch of proof of Theorem~\ref{thm: 2n}}	
\label{subsection: sketch of proof}

Let us briefly recall the way the proof can be obtained. Section~\ref{subsection: 10 chapes of rauzy graphs} describes how to build the graph of graphs $\G$ (Figure~\ref{figure: graph of graphs}). Then, Section~\ref{subsection: a critical result} states that the morphisms $\gamma_{i_n}$ of Definition~\ref{definition: definition des morphismes} are defined over alphabets of 2 of 3 letters. Section~\ref{subsection: a procedure to assign letters to circuits} and Section~\ref{subsection: computation of the morphisms} explicitly compute the morphisms. 

Due to Lemma~\ref{lemma: toutes les images finissent pareil}, Lemma~\ref{lemma: conjugue de morphism}, Fact~\ref{fact: conjugue} and Lemma~\ref{lemma: convergence taU_n}, the sequence $(\gamma_{i_n})_{n \in \N}$ can be slightly modified into a weakly primitive and proper directive word of $(X,T)$ (by contracting it and considering some left conjugates of the obtained morphisms). Therefore, what remains to show is that the morphisms $\gamma_{i_n}$ are compositions of morphisms in $\S$ as well as the left conjugates of the contracted morphisms.

Let us keep on considering the example of Section~\ref{subsubsection: An example}. The morphisms in Equation~\eqref{eq: 1 to 1} clearly belong to $\S^*$ as well as their respective left conjugates. Those in Equation~\eqref{equation: morphisms 1 to 8 sur 3 lettres} and their respective left conjugates also admit a decomposition: we define the morphisms of $\S^*$
\begin{eqnarray*}
	D_{12}:  	\begin{cases}
					0 \mapsto 0		\\
					1 \mapsto 12	\\
					2 \mapsto 2
				\end{cases}
	&
	D_{20}:  	\begin{cases}
					0 \mapsto 0		\\
					1 \mapsto 1		\\
					2 \mapsto 20
				\end{cases}
	&
	G_{21}:  	\begin{cases}
					0 \mapsto 0		\\
					1 \mapsto 1		\\
					2 \mapsto 12
				\end{cases}
\end{eqnarray*}
and obtain
\begin{eqnarray*}
	M G_{21}^{k-2} D_{20} D_{12} = 	
	\begin{cases}
		0 \mapsto 0		\\
		1 \mapsto 1^{k} 0 	\\
		2 \mapsto 1^{k-1} 0	
	\end{cases}
	&
	E_{01} M G_{21}^{k-2} D_{20} D_{12} = 	
	\begin{cases}
		0 \mapsto 1		\\
		1 \mapsto 0^{k} 1 	\\
		2 \mapsto 0^{k-1} 1
	\end{cases}.
\end{eqnarray*}
For the left conjugates, we simply have to replace $D_{12}$ and $D_{20}$ respectively by
\begin{eqnarray*}
	G_{12}:  	\begin{cases}
					0 \mapsto 0		\\
					1 \mapsto 21	\\
					2 \mapsto 2
				\end{cases}
	&
	\text{and}
	&
	G_{20}:  	\begin{cases}
					0 \mapsto 0		\\
					1 \mapsto 1		\\
					2 \mapsto 02
				\end{cases}.
\end{eqnarray*}
 
On that example, we see that the result holds, \textit{i.e.}, both $\gamma_{i_n}$ and $\gamma_{i_n}^{(L)}$ belong to $\S^*$. It is actually always true that $\gamma_{i_n}$ belongs to $\S^*$. But, not all morphisms $\gamma_{i_n}$ are right proper, making $\gamma_{i_n}^{(L)}$ undefined. However, one can always find a composition $\Gamma = \gamma_{i_n} \cdots \gamma_{i_{n+m}}$ such that $\Gamma$ is right proper and $\Gamma^{(L)}$ belongs to $\S^*$. This will be explained with more details in Theorem~\ref{thm: 2n final}.


\begin{remark}
\label{rem: different morphismes peuvent etiquetter}
Observe that a given edge in $\G$ may be labelled by several morphisms. This is due not only to a lack of precision in the definition of the bijections $\theta_{i_n}$ but also to the number of possibilities that exist for a given Rauzy graph to evolve to a given type of Rauzy graph. For example, consider a graph of type 8 as in Figure~\ref{figure: graph of type 8 with labels}. 

\begin{figure}[h!tbp]
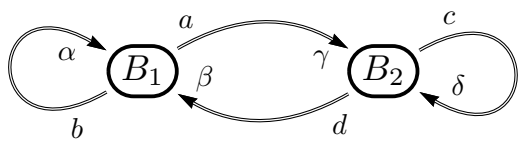

\centering
\scalebox{0.8}{
\begin{VCPicture}{(1,1.5)(9,5)}
\ChgEdgeLabelScale{0.8}
\StateVar[B_1]{(3,3)}{B_1}
\StateVar[B_2]{(7,3)}{B_2}
\EdgeLineDouble
\LArcL[.1]{B_1}{B_2}{a}		\LabelR[.9]{\gamma}
\LArcL[.1]{B_2}{B_1}{d}		\LabelR[.9]{\beta}
\LoopL[.1]{0}{B_2}{c}			\LabelR[.9]{\delta}
\LoopL[.1]{180}{B_1}{b}		\LabelR[.9]{\alpha}
\end{VCPicture}
}
\caption{Rauzy graph of type 8 with some labels.}
\label{figure: graph of type 8 with labels}
\end{figure}

This graph can evolve to a graph of type 7 or 8 (depending on the length of some paths) in two different ways:
\begin{list}{-}{}
	\item	either one of the bispecial factors $B_1$ and $B_2$ is a strong bispecial factor and the other one is a weak bispecial factor;
	\item or both of them are neutral bispecial factors and the two new right special factors are $\alpha B_1$ and $\delta B_2$.
\end{list}
Indeed, the two other cases do not satisfy the hypothesis on the subshift: two weak bispecial factors delete all right special factors so the subshift is either not minimal (when the graph is not strongly connected anymore) or periodic (when the graph keeps being strongly connected) and two strong bispecial factors provide 4 right special factors so we do not have $p(n+1)-p(n) \leq 2$ anymore.

The Rauzy graphs obtained in both available cases are represented at Figure~\ref{figure: evolutions of a type 8}. They are of type 7 or 8 depending on the respective length of the paths $B_1b \rightarrow \alpha B_1$ and $B_1a \rightarrow \beta B_1$ for Figure~\ref{figure: first evolution of a type 8} and on the respective length of the paths $B_1b \rightarrow \alpha B_1$ and $B_2c \rightarrow \delta B_2$ for Figure~\ref{figure: second evolution of a type 8}. These two possibilities of evolution to a same type of graphs imply that the edges $8 \rightarrow 7$ and $8 \rightarrow 8$ in $\G$ are labelled by several morphisms.

\begin{figure}[h!tbp]
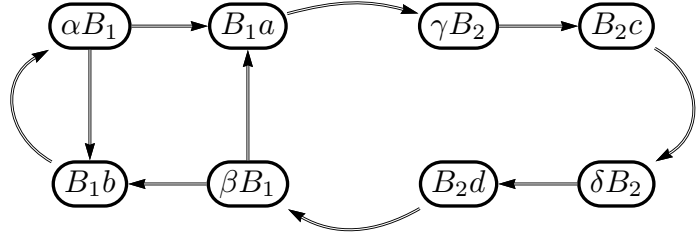
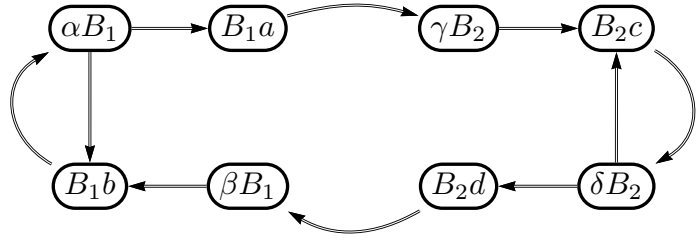

\begin{center}
\subfigure[$B_1$ is strong and $B_2$ is weak]{
\label{Subfigure: $B_1$ strong and $B_2$ weak}
\scalebox{0.7}{
\begin{VCPicture}{(1,-0.5)(12,5)}
\label{figure: first evolution of a type 8}
\ChgEdgeLabelScale{0.8}
\StateVar[\alpha B_1]{(1,4)}{alphaB}
\StateVar[B_1b]{(1,1)}{Bb}
\StateVar[B_1a]{(4,4)}{Ba}
\StateVar[\beta B_1]{(4,1)}{betaB}
\StateVar[\gamma B_2]{(8,4)}{gammaB}
\StateVar[B_2d]{(8,1)}{Bd}
\StateVar[B_2c]{(11,4)}{Bc}
\StateVar[\delta B_2]{(11,1)}{deltaB}
\EdgeLineDouble
\Edge{alphaB}{Ba}
\Edge{alphaB}{Bb}
\Edge{betaB}{Bb}
\Edge{gammaB}{Bc}
\Edge{deltaB}{Bd}
\Edge{betaB}{Ba}
\VCurveR[]{angleA=150,angleB=-150,ncurv=1}{Bb}{alphaB}{}
\ArcL{Ba}{gammaB}{}
\VCurveR[]{angleA=-30,angleB=30,ncurv=1}{Bc}{deltaB}{}
\LArcL{Bd}{betaB}{}
\end{VCPicture}
}}
\\
\subfigure[Both $B_1$ and $B_2$ are ordinary]{
\label{Subfigure: $B_1$ and $B_2$ ordinary}
\scalebox{0.7}{
\begin{VCPicture}{(1,-0.5)(12,5)}
\label{figure: second evolution of a type 8}
\ChgEdgeLabelScale{0.8}
\StateVar[\alpha B_1]{(1,4)}{alphaB}
\StateVar[B_1b]{(1,1)}{Bb}
\StateVar[B_1a]{(4,4)}{Ba}
\StateVar[\beta B_1]{(4,1)}{betaB}
\StateVar[\gamma B_2]{(8,4)}{gammaB}
\StateVar[B_2d]{(8,1)}{Bd}
\StateVar[B_2c]{(11,4)}{Bc}
\StateVar[\delta B_2]{(11,1)}{deltaB}
\EdgeLineDouble
\Edge{alphaB}{Ba}
\Edge{alphaB}{Bb}
\Edge{betaB}{Bb}
\Edge{gammaB}{Bc}
\Edge{deltaB}{Bd}
\Edge{deltaB}{Bc}
\VCurveR[]{angleA=150,angleB=-150,ncurv=1}{Bb}{alphaB}{}
\ArcL{Ba}{gammaB}{}
\VCurveR[]{angleA=-30,angleB=30,ncurv=1}{Bc}{deltaB}{}
\LArcL{Bd}{betaB}{}
\end{VCPicture}
}}
\end{center}
\caption{Evolutions from 8 to 7 or 8.}
\label{figure: evolutions of a type 8}
\end{figure}
\end{remark}


\section{$\S$-adic characterization}	
\label{section: caraterisation}

Theorem~\ref{thm: 2n} states that any minimal and aperiodic subshift $(X,T)$ with first difference of complexity bounded by two admits a directive word $(\sigma_n)_{n \in \N} \in \S^\N$ which is linked to a path in $\G$. However, the converse is false (see Section~\ref{subsection: valid paths}). A possible way to get an $\S$-adic characterization of the considered subshifts would be to describe exactly all infinite paths in $\G$ that really correspond to the sequences of evolutions of Rauzy graphs of such subshifts. By achieving this, we would determine the condition $C$ of the $S$-adic conjecture for this particular case. This is the aim of this section and this will lead to Theorem~\ref{thm: 2n final}.

In the sequel, to alleviate notations we let $[u,v,w]$ denote the morphism
\[
	\begin{cases}
		0 \mapsto u \\
		1 \mapsto v \\
		2 \mapsto w
	\end{cases}
\]
and when some letters are not completely determined (that is if some circuits can play the same role), we use the letters $x,y$ and $z$. 

For example, the morphisms in Equation~\eqref{equation: morphisms 1 to 8 sur 3 lettres} will be denoted by one morphism: $[x,y^{k_1}x,y^{k_1-1}x]$ and it is understood that $\{x,y\} = \{0,1\}$. Observe that $x$ and $y$ depend on the type of graphs we come from. Indeed, when coding the evolution of a graph of type 1, we cannot have $\{x,y\} = \{0,2\}$ by definition of $\theta_{i_n}$ for such graphs. Moreover, if for example the letters $0$, $x$ and $y$ occur in an image, it is understood that $0$, $x$, and $y$ are pairwise distinct. 

We also need to introduce the following notation. For $x,y \in \{0,1,2\}$, $x \neq y$, the morphisms $D_{x,y}$ and $E_{x,y}$ are respectively defined by
\[
	D_{x,y}:	\begin{cases}
					x \mapsto xy		\\
					y \mapsto y		\\
					z \mapsto z
				\end{cases}
	\quad 	\text{and}	\quad
	E_{x,y}:	\begin{cases}
					x \mapsto y		\\
					y \mapsto x		\\
					z \mapsto z
				\end{cases}	
\]

\subsection{Valid paths}	
\label{subsection: valid paths}

The first step to get the $\S$-adic characterization is to understand how we can describe the \enquote{good labelled paths} in $\G$, hence the good sequences of evolutions. To this aim, we introduce the notions of \textit{valid directive word} and of \textit{valid path}.


\begin{definition}
\label{def: valid path}
An infinite and labelled path $p$ in $\G$ is \definir{valid} if there is a minimal subshift with first difference of complexity bounded by 2 for which the sequence $(\gamma_{i_n})_{n \in \N}$ of Definition~\ref{definition: definition des morphismes} (and Remark~\ref{remark: sigme = identity}) labels $p$. 

We extend the notion of \textit{validity} to prefix and suffixes of $p$, \textit{i.e.}, a path is a \definir{valid prefix} (resp. \definir{valid suffix}) if it is a prefix (resp. suffix) of a valid path. We also extend it to sequences of morphisms in $\S^*$, \textit{i.e.}, a sequence of morphisms is \definir{valid} if it is the label of a valid path (or valid prefix or valid suffix).

\end{definition}

There exist several reasons for which a given labelled path in $\G$ is not valid: two conditions (due to Theorem~\ref{thm: 2n}) are that its label has to be weakly primitive and must admit a contraction that contains only 
right\footnote{In the definition of valid directive word, we did not consider left conjugates of morphisms so the property of being proper becomes being right proper.}
proper morphisms. Example~\ref{ex: not valid path 2} and Example~\ref{ex: not valid path 1} below show two sequences of evolutions which are forbidden because their respective directive words do not satisfy the weak primitivity.

\begin{example}
\label{ex: not valid path 2}
Sturmian subshifts have Rauzy graphs of type 1 for all $n$. Thus, for all $n$, $\gamma_{i_n}$ is one of the morphisms given in Equation~\eqref{eq: 1 to 1}.
However if, for instance, we consider that for all $n$, the morphism $\gamma_{i_n}$  is $[0,10]$, the directive word is not weakly primitive and the sequence of Rauzy graphs $(G_{i_n})_{n \in \N}$ is such that for all $n$, $i_n = n$ and $\lambda_R (\theta_n(0)) = 0$ and $\lambda_R (\theta_n(1)) = 10^n$ (the reduced Rauzy graph $g_n$ is represented in Figure~\ref{figure: 0001000}).
It actually corresponds to the subshift generated by the sequence $\bw = \cdots 000.1000\cdots$ which has complexity $p(n)=n+1$ for all $n$ but which is not minimal.

\begin{figure}[h!tbp]
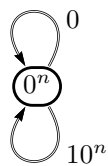

\centering
\scalebox{0.6}{
\begin{VCPicture}{(0,0)(5,5)}
\StateVar[0^n]{(2.5,2.5)}{B}
\EdgeLineDouble
\LoopR{90}{B}{0}
\LoopL{-90}{B}{10^n}
\end{VCPicture}
}
\caption{Reduced Rauzy graph $g_n$ of $\cdots 000.1000 \cdots$.}
\label{figure: 0001000}
\end{figure}
\end{example}

\begin{example}
\label{ex: not valid path 1}
Let us consider a path in $\G$ that ultimately stays in the vertex 9. Figure~\ref{figure: evolution of a type 9} represents the only way for a Rauzy graph $G_{i_n}$ of type 9 to evolve to a Rauzy graph of type 9. We can see that in this evolution, the $i_n$-circuit $\theta_{i_n}(0)$ starting from the vertex $B$ (i.e., the loop that does not pass through the vertex $R$) \enquote{stays unchanged} in $G_{i_n+1}$, i.e., $\psi_{i_n}(\theta_{i_n+1}(0)) = \theta_{i_n}(0)$. Consequently, we have $\lim_{n \to +\infty} |\theta_{i_n}(0)| < +\infty$: a contradiction with Lemma~\ref{lemma: pas de court} (the circuit is trivially allowed). One can also check  that for all morphisms $\gamma_{i_n}$ coding such an evolution, we have $\gamma_{i_n}(0) = 0$. As there is no other evolution from a Rauzy graph of type 9 to a Rauzy graph of type 9, the directive word cannot be weakly primitive. 

\begin{figure}[h!tbp]
\centering
\subfigure[Before evolution]{
\label{figure: graph of type 9 before evolution}
\scalebox{0.6}{
\begin{VCPicture}{(-0.5,0)(8.5,5)}
\ChgEdgeLabelScale{0.5}
\StateVar[R]{(1,2.5)}{R}
\VSState{(5,2.5)}{L}
\StateVar[B]{(7,2.5)}{B}
\EdgeLineDouble
\ArcL{R}{L}{}
\ArcR{R}{L}{}
\EdgeL{L}{B}{}
\LArcL{B}{R}{}
\LoopL{0}{B}{}
\end{VCPicture}
}}
\subfigure[After evolution]{
\label{figure: graph of type 9 after evolution}
\scalebox{0.6}{
\begin{VCPicture}{(-0.5,0)(8.5,5)}
\ChgEdgeLabelScale{0.5}
\StateVar[R]{(1,2.5)}{R}
\VSState{(5,2.5)}{L}
\VSState{(6,3)}{B1}
\VSState{(6,2)}{B2}
\VSState{(7,3)}{B3}
\VSState{(7,2)}{B4}
\EdgeLineDouble
\ArcL{R}{L}{}
\ArcR{R}{L}{}
\EdgeL{L}{B1}{}
\EdgeL{B1}{B3}{}
\EdgeL{B4}{B2}{}
\EdgeL{B4}{B3}{}
\LArcL{B2}{R}{}
\VCurveR[]{angleA=10,angleB=-10,ncurv=3}{B3}{B4}{}
\end{VCPicture}
}}
\caption{Evolution of a graph of type 9 to a graph of type 9.}
\label{figure: evolution of a type 9}
\end{figure}
\end{example}

The two previously given conditions (being weak primitivity and proper) are not sufficient to be a valid directive word: there is also a \enquote{local condition} that has to be satisfied. Indeed, Example~\ref{ex: not valid path 3} below shows that for some prefixes $\gamma_{i_0} \cdots \gamma_{i_k}$ labelling a finite path $p$ in $\G$, not every edge starting from $i(p)$ is allowed.

\begin{example}
\label{ex: not valid path 3}
Consider a graph $G_{i_n}$ of type 1 that evolves to a graph as in Figure~\ref{figure: Evolution of a sturmian graph with strong bispecial factor} (Page~\pageref{figure: Evolution of a sturmian graph with strong bispecial factor}), hence to a graph of type 7 or 8. We write $R_1 = \alpha B$ and $R_2 = \beta B$ and suppose that $v_{i_n+1} = R_1$. The morphism coding this evolution is $[x,y^kx,y^{k-1}x]$ for some integer $k \geq 2$. If we suppose $k \geq 3$, this means that the circuits $\theta_{i_n+1}(1)$ and $\theta_{i_n+1}(2)$ respectively go through $k-1$ and $k-2$ times in the loop $R_2 \rightarrow R_2$. By construction of the Rauzy graphs, this means that the shortest bispecial factor $B'$ admitting $R_2$ as a suffix is a neutral bispecial factor. Let $m>n$ be an integer such that $B'$ is a bispecial vertex in $G_{i_m}$. Since $B'$ is neutral bispecial, there is a right special factor $R'$ of length $i_m+1$ that admits $B'$ as a suffix. Moreover, since $v_{i_m}$ is not $B'$ (as $R_1$ has to be a suffix of $v_{i_m}$), the right special factor $v_{i_m+1}$ is not $R'$. Consequently there are two right special factors in $G_{i_m+1}$ so $G_{i_m+1}$ is not of type 1.
\end{example}

To be a valid labelled path in $\G$ the three previous examples show that a given path $p$ must necessary satisfy at least two conditions: a local one about its prefixes (Example~\ref{ex: not valid path 3}) and a global one about weak primitivity (Example~\ref{ex: not valid path 2} and Example~\ref{ex: not valid path 3}). The next result states that the converse is true.

\begin{proposition}
\label{prop: valid path}
An infinite and labelled path $p$ in $\G$ 
is valid if and only if both following conditions are satisfied.
\begin{enumerate}
	\item All prefixes of $p$ are valid\footnote{a local condition};
	\item its label is weakly primitive and a contraction of it contains only right proper morphisms\footnote{a global condition}.
\end{enumerate}
\end{proposition}

\begin{proof}
The first condition is obviously necessary and the second condition comes from Theorem~\ref{thm: 2n}. 
For the sufficient part, if all prefixes of $p$ are valid, it implies that we can build a sequence of Rauzy graphs $(G_n)_{n \in \N}$ such that for all $n$, $G_n$ is as represented in Figure~\ref{figure: Rauzy graphs with one right special factor and one left special factor} to Figure~\ref{figure: Rauzy graphs with two left and right special factors} and evolves to $G_{n+1}$. To these Rauzy graphs we can associate a sequence of languages $(L(G_n))_{n \in \N}$ defined as the set of finite words labelling paths in $G_n$. By construction we obviously have $L(G_{n+1}) \subset L(G_n)$ and the language 
\[
	L = \bigcap_{n \in \N} L(G_n)
\]
is factorial\footnote{For every word $u$ in $L$, $\fac{}{u} \subset L$.}, prolongable\footnote{For every word $u$ in $L$, there are some letters $a$ and $b$ such that $au$ and $ub$ are in $L$.} and such that $1 \leq p_L(n+1)-p_L(n) \leq 2$ for all $n$ (where $p_L$ is the complexity function of the language). Thus, it defines a subshift $(X,T)$ whose language is $L$ and which, by construction, is such that the sequence $(\gamma_{i_n})_{n \in \N}$ of Definition~\ref{definition: definition des morphismes} labels $p$.
%
\end{proof}

\subsection{Decomposition of the problem}	
\label{subsection: decomposition}

Our aim is now to describe exactly the set of all valid paths in $\G$. The idea is to modify the graph of graphs $\G$ in such a way that the \enquote{local condition} to be a valid path (the first point of Proposition~\ref{prop: valid path}) is treated by the graph\footnote{In other words, we would like to modify $\G$ in such a way that all finite paths are valid.}. We also would like that for any minimal subshift with $p(n+1)-p(n) \leq 2$, a contraction of $(\gamma_{i_n})_{n \in \N}$ that contains infinitely many right proper morphisms\footnote{to be able to consider left conjugates} labels a path in $\G$. In that case, we will only have to take care at the weak primitivity, which is rather easy to check. But, we actually will see that modifying the graph $\G$ as wanted will not be possible. There will still remain some vertices $v$ such that for some finite paths arriving in $v$, some edges $e$ starting from $v$ make the path $pe$ not valid. However, we will manage to describe the local condition for these vertices so this will still provides an $\S$-adic characterization. 
The computations in the next section are sometimes a little bit heavy to check. The reader can find some help (figures with evolutions of graphs, list of morphisms coding these evolutions, decomposition of them into $\S^*$, etc.) in the appendices.


The graph of graphs $\G$ contains 4 strongly connected components: 
\[
	C_1 = \{2\}, \ C_2 = \{3\}, \ C_3 = \{4\}, \ C_4 = \{1,5,6,7,8,9,10\}. 
\]
Any infinite path in $\G$ ends in one component $C_i$ and is valid if and only if the prefix leading to $C_i$ is valid and if the infinite suffix staying in $C_i$ is valid and fit with the prefix. Thus, to describe all valid paths in $\G$, we can separately describe the valid suffixes in each component and then study how the components are linked together.

\begin{remark}
\label{rem: validity of prefix}
By hypothesis on $p(1)-p(0)$, a valid path $p$ in $\G$ always starts from the vertex $1$ or from the vertex $2$ (depending on the size of the alphabet: 2 or 3). Therefore, when studying the validity of a path in the component $C_2$, $C_3$ or $C_4$, we only study the validity of its suffix that always stays in that component (even for $C_4$ since a path ultimately staying in the component $C_4$ might start in the vertex $2$). By contrary, studying the validity of the suffix of a path ultimately staying in $C_1$ is the same as studying the validity of the entire path. 
\end{remark}

\subsection{Valid paths in $C_1$}
\label{subsection: valid paths in C_1}

This component corresponds to the class of Arnoux-Rauzy subshifts which is well known~\cite{Arnoux-Rauzy}. The morphisms $\gamma_{i_n}$ that code an evolution in that component are right proper and are easily seen to belong to $\S^*$, as well as their respective left conjugates.
\begin{eqnarray*}
	\forall n, &	\gamma_{i_n} 		\in \{[0,10,20],[01,1,21],[02,12,2]\}	\\
	\forall n, &	\gamma_{i_n}^{(L)} 	\in \{[0,01,02],[10,1,12],[20,21,2]\}		
\end{eqnarray*}

Arnoux and Rauzy~\cite{Arnoux-Rauzy} gave an $S$-adic description of the so-called Arnoux-Rauzy subshifts by considering the morphisms $[0,10,20]$, $[01,1,21]$ and $[02,12,2]$. They proved the following result.

\begin{proposition}[Arnoux and Rauzy~\cite{Arnoux-Rauzy}]
\label{lemma: allowed path for type 2}
A labelled path $p$ in $\G$ is valid and corresponds to an Arnoux-Rauzy subshift if and only if it goes only through  vertex $2$ and the three morphisms $[0,10,20]$, $[01,1,21]$ and $[02,12,2]$ occur infinitely often in the label of $p$.
%
\end{proposition}


\subsection{Valid paths in $C_2$}
\label{subsection: component C_2}

This component contains only the vertex $3$ of $\G$ and the morphisms $\gamma_{i_n}$ that code an evolution in this component are one of the following
\[
	[0,10,20],[01,1,21],[02,12,2], [0,10,2], [01,1,2], [02,1,2], [0,1,20], [0,1,21], [0,12,2] ;
\]
they belong to $\S^*$.

Observe that not all these morphisms are right proper and we could even find an infinite sequence of them that would not admit a contraction with only right proper morphisms (for instance, $[0,10,2]^{\omega}$). The reason is that not all finite composition of these morphisms correspond to a valid finite sequence of evolution of Rauzy graphs. The next lemma describes this fact.  


\begin{lemma}
\label{lemma: component C_2}
Let $(X,T)$ be a minimal and aperiodic subshift with first difference of complexity bounded by 2. Let $(\gamma_{i_n})_{n \in \N}$ be the directive word of Definition~\ref{definition: definition des morphismes}. Suppose that both $\gamma_{i_n}$ and $\gamma_{i_{n+1}}$ are coding an evolution from a graph of type 3 to a graph of type 3. Then if $\gamma_{i_n}$ is equal to
\[
	D_{y,x}D_{z,x} \quad (\text{resp. }D_{x,y})
\]
for $\{x,y,z\} = \{0,1,2\}$, then $\gamma_{i_{n+1}}$ can only be one of the three following morphisms
\[
	D_{y,x}D_{z,x}, D_{x,y}, D_{x,z} \quad (\text{resp. } D_{y,z}D_{x,z}, D_{z,y}, D_{z,x})
\]
\end{lemma}

\begin{proof}
We only have to look at the behaviour of the Rauzy graph when it evolves. Figure~\ref{figure: evolution d'une graphe de type 3} shows the two possibilities for a graph of type 3 to evolve to a graph of type 3. When computing the morphisms coding these evolutions, we see that what is important to know is which letter corresponds to the top loop in Figure~\ref{fig: type 3 avant evolution}. Indeed, if $\theta_{i_n}(x)$ corresponds to the top loop in Figure~\ref{fig: type 3 avant evolution}, the three available morphisms are (the second must be counted twice since $y$ can be replaced by $z$)
\begin{eqnarray*}
	\begin{cases}
		x \mapsto x		\\
		y \mapsto yx	\\
		z \mapsto zx
	\end{cases}	 
	&
	\text{and}
	&
	\begin{cases}
		x \mapsto xy	\\
		y \mapsto y		\\
		z \mapsto z
	\end{cases}.
\end{eqnarray*}
The evolution represented in Figure~\ref{figure: evolution 1 of type 3} is coded by the first morphism and the evolution represented in Figure~\ref{figure: evolution 2 of type 3} is coded by the second one (where $\theta_{i_n}(y)$ is the leftmost loop in Figure~\ref{fig: type 3 avant evolution}). 

After the first evolution, the graph becomes again a graph as in Figure~\ref{fig: type 3 avant evolution} where the circuit $\theta_{i_{n+1}}(x)$ still corresponds to the top loop. The available morphisms are therefore the same as before the evolution.

After the second evolution, the graph becomes again a graph as in Figure~\ref{fig: type 3 avant evolution} but the top loop is the circuit $\theta_{i_{n+1}}(z)$. The available morphisms are therefore the same as before the evolution but with $x$ and $z$ exchanged.

\begin{figure}[h!tbp]
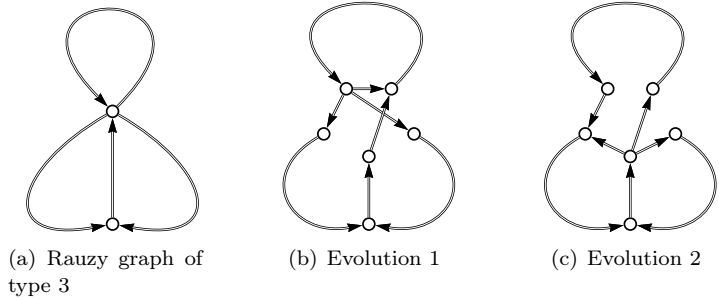

\centering
\subfigure[Rauzy graph of type 3]{
\label{fig: type 3 avant evolution}
\scalebox{0.6}{
\begin{VCPicture}{(0,0)(4,5)}
\ChgEdgeLabelScale{0.5}
\VSState{(2,3)}{B}
\VSState{(2,0.5)}{L}
\EdgeLineDouble
\VCurveL[]{angleA=45,angleB=135,ncurv=38}{B}{B}{}
\VCurveR[]{angleA=-30,angleB=-10,ncurv=2}{B}{L}{}
\VCurveR[]{angleA=-150,angleB=190,ncurv=2}{B}{L}{}
\EdgeL{L}{B}{}
\end{VCPicture}}}
\qquad
\subfigure[Evolution 1]{
\label{figure: evolution 1 of type 3}
\scalebox{0.6}{
\begin{VCPicture}{(0,0)(4,5)}
\ChgEdgeLabelScale{0.8}
\VSState{(1.5,3.5)}{HG}
\VSState{(2.5,3.5)}{HD}
\VSState{(1,2.5)}{MG}
\VSState{(3,2.5)}{MD}
\VSState{(2,2)}{MM}
\VSState{(2,0.5)}{BM}
\EdgeLineDouble
\Edge{BM}{MM}
\VCurveL[]{angleA=45,angleB=135,ncurv=5.5}{HD}{HG}{}
\VCurveR[]{angleA=-30,angleB=-10,ncurv=1.5}{MD}{BM}{}
\VCurveR[]{angleA=-150,angleB=190,ncurv=1.5}{MG}{BM}{}
\Edge{HG}{HD}
\Edge{HG}{MG}
\Edge{HG}{MD}
\Edge{MM}{HD}
\end{VCPicture}
}}
\qquad
\subfigure[Evolution 2]{
\label{figure: evolution 2 of type 3}
\scalebox{0.6}{
\begin{VCPicture}{(0,0)(4,5)}
\ChgEdgeLabelScale{0.8}
\VSState{(1.5,3.5)}{HG}
\VSState{(2.5,3.5)}{HD}
\VSState{(1,2.5)}{MG}
\VSState{(3,2.5)}{MD}
\VSState{(2,2)}{MM}
\VSState{(2,0.5)}{BM}
\EdgeLineDouble
\Edge{BM}{MM}
\VCurveL[]{angleA=45,angleB=135,ncurv=5.5}{HD}{HG}{}
\VCurveR[]{angleA=-30,angleB=-10,ncurv=1.5}{MD}{BM}{}
\VCurveR[]{angleA=-150,angleB=190,ncurv=1.5}{MG}{BM}{}
\Edge{MM}{HD}
\Edge{MM}{MG}
\Edge{MM}{MD}
\Edge{HG}{MG}
\end{VCPicture}
}}
\caption{Evolutions of a graph of type 3 to a graph of type 3.}
\label{figure: evolution d'une graphe de type 3}
\end{figure}
\end{proof}

Thanks to the previous lemma, if $\gamma_{i_k} \gamma_{i_{k+1}} \gamma_{i_{k+2}} \cdots$ labels a valid suffix that stays in component $C_2$, then for all $n \geq k$, $\gamma_{i_n} \gamma_{i_{n+1}}$ is a right proper morphism. Consequently, we obtain the following result. We let the reader that all involved morphisms (as well as their respective left conjugates when they exist) belongs to $\S^*$.

\begin{proposition}
\label{proposition: allowed path for type 3}
An infinite path $p$ in $\G$ labelled by $(\gamma_{i_n})_{n \geq N}$ is a valid suffix that always stays in vertex $3$ if and only if there is a contraction $(\alpha_n)_{n \geq N}$ of $(\gamma_{i_n})_{n \geq N}$ such that 
\begin{enumerate}

	\item 
	$(\alpha_n)_{n \geq N}$labels an infinite path in the graph represented in Figure~\ref{figure: zoom du vertex 3} with 
	\begin{enumerate}

		\item 	for all $x \in \{0,1,2\}$, the loop on $V_x$ is labelled by morphisms in 
		\[
			F_x = \left\{ D_{y,x} D_{z,x},\ D_{x,y} D_{z,y} \mid \{x,y,z\} = \{0,1,2\} \right\};
		\]
	
		\item 	for all $x,y \in \{0,1,2\}$, $x \neq y$, the edge from $V_x$ to $V_y$ is labelled by morphisms in 
		\[
			F_{x \to y} = \left\{ D_{x,z},\ D_{x,y} D_{z,x} \mid z \notin \{x,y\} \right\};
		\]
	\end{enumerate}
	
	\item
	$(\alpha_n)_{n \geq N}$	contains infinitely many right proper morphisms;
	
	\item
	for all $x \in \{0,1,2\}$, there are infinitely many integers $n \geq N$ such that $D_{y,x}$ is a factor of $\alpha_n$ for some $y \in \{0,1,2\}$.

\end{enumerate}
%
%
%
%
%
%
%
%
\end{proposition}

\begin{figure}[h!tbp]
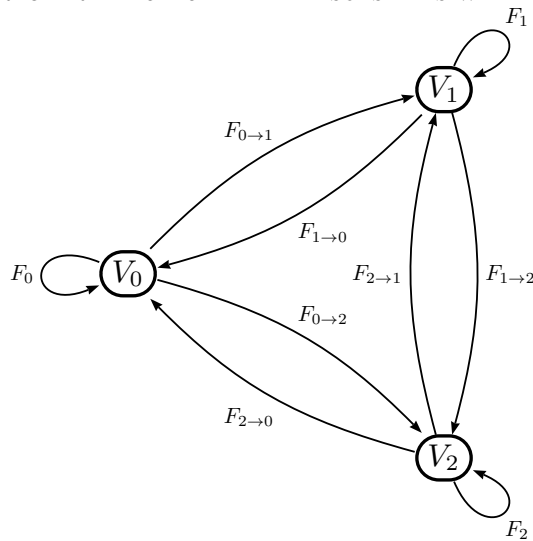

\centering
\scalebox{0.7}{
\begin{VCPicture}{(0,-2.3)(5,7.5)}
\ChgEdgeLabelScale{0.7}
\StateVar[V_0]{(-0.5,3)}{V0}
\StateVar[V_1]{(5.5,6.5)}{V1}
\StateVar[V_2]{(5.5,-0.5)}{V2}
\CLoopW[.5]{V0}{F_0}
\CLoopNE[.5]{V1}{F_1}
\CLoopSE[.5]{V2}{F_2}
\ArcL[.5]{V0}{V1}{F_{0 \to 1}}
\ArcL[.5]{V1}{V2}{F_{1 \to 2}}
\ArcL[.5]{V2}{V0}{F_{2 \to 0}}
\ArcL[.5]{V0}{V2}{F_{0 \to 2}}
\ArcL[.5]{V2}{V1}{F_{2 \to 1}}
\ArcL[.5]{V1}{V0}{F_{1 \to 0}}
\end{VCPicture}
}
\caption{Graph corresponding to component $C_2$ in $G$.}
\label{figure: zoom du vertex 3}
\end{figure}

\begin{proof}


Our aim is to describe valid suffix in $\G$ that stay in vertex 3, accordingly to Proposition~\ref{prop: valid path}. 

Let us start with Condition 1 (\textit{i.e.}, the local one). The morphisms that code an evolution from a graph of type 3 to a graph of type 3 (and their decomposition into $\S^*$) are listed at the beginning of Section~\ref{subsection: component C_2}. However, Lemma~\ref{lemma: component C_2} shows that they cannot be composed in every way. When computing the morphisms coding the different evolutions (see Figure~\ref{figure: evolution d'une graphe de type 3}), we see that what is important is which letter corresponds to the top loop in Figure~\ref{fig: type 3 avant evolution}. Consequently, we can  \enquote{split} the vertex 3 in $\G$ into 3 vertices $V_0$, $V_1$ and $V_2$, each $V_x$ corresponding to the fact that the circuit $\theta_{i_n}(x)$ only goes through non-left special vertices (\textit{i.e.}, corresponds to the top loop in Figure~\ref{fig: type 3 avant evolution}) and we put some edges between these vertices if the corresponding evolution is available. Then we label the graph as follows: for all $x,y \in \{0,1,2\}$ such that $x \neq y$, we let $F_x$ denote the set of morphisms labelling the loop on $V_x$ and we let $F_{x \to y}$ denote the set of morphisms labelling the edge from $V_x$ to $V_y$. Of course, $F_x$ and $F_{x \to y}$ contain the morphism corresponding to the evolution, \textit{i.e.}, $F_x$ contains the morphism $D_{y,x} D_{z,x}$
and $F_{x \to y}$ contains the morphism $D_{x,z}$. 
Defining $F_x$ and $F_{x \to y}$ this way ensures that the local condition is satisfied.

Before considering the second condition of Proposition~\ref{prop: valid path}, let us modify the sets $F_x$ and $F_{x\to y}$ accordingly to what we explained in Section~\ref{subsection: decomposition}, \textit{i.e.}, in such a way that a contraction $(\alpha_n)_{n \geq N}$ of $(\gamma_{i_n})_{n \geq N}$ contains infinitely many right proper morphisms and labels a path in Figure~\ref{figure: zoom du vertex 3}.
As all non-right proper morphisms belong to some set $F_{x \to y}$, this can easily be done as follows: for all $x,y,z \in \{0,1,2\}$, $x \neq y$, $y \neq z$, one can check that the morphism $D_{x,z} D_{y,x} \in F_{x \to y} F_{y \to z}$ is right proper and labels a finite path from $V_x$ to $V_z$. Consequently, for all $x$ and all $y,z$ such that $\{x,y,z\} = \{0,1,2\}$ we can add in $F_x$ the morphism $D_{x,z} D_{y,z}$ and we add in $F_{x \to z}$ the morphism $D_{x,z} D_{y,x}$. By doing this, the existence of $(\alpha_n)_{n \geq N}$ is ensured.

Now let us describe all labelled paths in Figure~\ref{figure: zoom du vertex 3} with weakly primitive label (Condition 2 of Proposition~\ref{prop: valid path}). The morphisms in $F_x$ and in $F_{x \to y}$, $x,y \in \{0,1,2\}$, are composed of morphisms $D_{u,v}$ for some $u,v \in \{0,1,2\}$. Let us prove that the label $(\alpha_n)_{n \geq N}$ of a path in Figure~\ref{figure: zoom du vertex 3} is weakly primitive if and only if for all $x \in \{0,1,2\}$, there are infinitely many integers such that $D_{y,x}$ is a factor of $\alpha_n$ for some $y \in \{0,1,2\}$, $y \neq x$. The condition is trivially necessary since if for all $y$, $D_{y,x}$ is not a factor of $\alpha_n$ for $n$ not smaller than some integer $m \geq N$, then $x$ does not belong to $\alpha_m \cdots \alpha_{m+k}(z)$ for all $z \neq x$ and all integers $k \geq 0$. 
It is also sufficient. Indeed, it is clear that if, for $\{x,y,z\} = \{0,1,2\}$, the three morphisms $D_{x,y}$, $D_{y,z}$ and $D_{z,x}$ occur infinitely often as factors of $(\alpha_n)_{n \geq m}$, then the directive word is weakly primitive. Thus, to satisfy the condition without inducing the weak primitivity, the set of morphisms that occur infinitely often as factors of $(\alpha_n)_{n \geq m}$ has to be included in $\{D_{x,y}, D_{y,z}, D_{y,x}, D_{x,z}\}$. This is in contradiction with the way the morphisms have to be composed (governed by Figure~\ref{figure: zoom du vertex 3}).
\end{proof}

\subsection{Preliminary lemmas for $C_3$ and $C_4$}
\label{subsection: preliminary lemmas for components}

In both types of graphs of component $C_1$ and $C_2$, there is only one right special vertex. This makes the computation of valid paths easier to compute than when there are two right special factors. Indeed, if $R_1$ and $R_2$ are two bispecial factors in a Rauzy graph $G_{i_n}$, the circuits starting from $R_1$ impose some restrictions on the behaviour of $R_2$, \textit{i.e.}, on the way it will make the graph evolve when it will become bispecial (see Example~\ref{ex: not valid path 3} where the explosion of the bispecial vertex $B'$ is governed by $\theta_{i_n}(1)$ and $\theta_{i_n}(2)$). Such a thing cannot happen for graphs of type 2 and 3, \textit{i.e.}, the local condition of Proposition~\ref{prop: valid path} can be easily expressed. In this section, we introduce some notations and we give some lemmas that will be helpful to study valid paths in components $C_3$ and $C_4$.

First, let us briefly explain what we will mean when talking about the \textit{explosion} of a bispecial factor. Roughly speaking, \enquote{explosion} describes the behaviour of a bispecial vertex when the Rauzy graph evolves. These vertices are of a particular interest since those are the only ones that can change the shape of a graph (hence they are the only ones that determine the morphisms $\gamma_{i_n}$ since they depend on the shape of the graphs). 

Indeed, let us consider a non-special vertex $V$ in a Rauzy graph $G_{i_n}$. Since $V$ is not special, there are exactly two vertices $V_p$ and $V_s$ in $G_{i_n+1}$ such that $V$ is prefix of $V_p$ and suffix of $V_s$ and there is always an edge from $V_p$ to $V_s$. Consequently, the behaviour of $V$ when $G_{i_n}$ evolves does not change the shape of $G_{i_n}$. One can make similar observation for left (but not right) special vertices and for right (but not left) special vertices. The difference is that, for left special vertices (resp. for right special vertices), there are several vertices $V_p^{(1)}, \dots, V_p^{(k)}$ with $k = \delta^-V>1$ (resp. $V_s^{(1)}, \dots, V_s^{(k)}$ with $k = \delta^+V>1$) that admit $V$ as a prefix (resp. as a suffix) and for all $i$, $1 \leq i \leq k$, there is an edge from $V_p^{(i)}$ to $V_s$ (resp. from $V_p$ to $V_s^{(i)}$). Consequently, the behaviour of $V$ when $G_{i_n}$ evolves does not change the shape of $G_{i_n}$ either.

For bispecial vertices $V$, this is not true anymore. Indeed, in $G_{i_n+1}$ there are several vertices $V_p^{(1)}, \dots, V_p^{(k)}$ with $k = \delta^-V>1$ and several vertices $V_s^{(1)}, \dots, V_s^{(\ell)}$ with $\ell = \delta^+V>1$ that respectively admit $V$ as a prefix and as a suffix.  Moreover, the number of edges between $\left\{ V_p^{(1)}, \dots, V_p^{(k)} \right\}$ and $\left\{ V_s^{(1)}, \dots, V_s^{(\ell)} \right\}$ depends on the bilateral order of $V$. Therefore the behaviour of $V$ when $G_{i_n}$ evolves can strongly change the shape of $G_{i_n}$ (by increasing or decreasing the number of special vertices for example).

The next lemma gives a method to build a sequence $(\eta_{j_n})_{n \in \N}$ of morphisms which is a little bit different from $(\gamma_{i_n})_{n \in \N}$ and that will help us to describe the valid paths in $C_3$ and $C_4$.

\begin{lemma}
\label{lemma: decomposition du mot directeur}
Let $(X,T)$ be a minimal subshift with first difference of complexity satisfying $1 \leq p(n+1)-p(n) \leq 2$ for all $n$ and let $(i_n)_{n \in \N}$ be the increasing sequence of integers such that $\fac{k}{X}$ contains a bispecial factor of $X$ if and only if $k \in \{i_n \mid n \in \N\}$. There is a non-decreasing sequence $(j_n)_{n \in \N}$ of integers such that $j_n \leq i_n$ for all $n$ and a sequence $(\eta_n)_{n \in \N}$ of morphisms in $\S^*$ such that or all $n$, $\eta_n$ codes the explosion of a unique bispecial factor of length $j_n$ in $G_{j_n}(X)$.
\end{lemma}

\begin{proof}
First it is obvious that if a Rauzy graph $G_{i_n}$ contains two bispecial vertices, making them explode at the same time or separately produces the same graph $G_{i_n+1}$ (hence $G_{i_{n+1}}$). Consequently, since $\gamma_{i_n}$ describes how a graph evolves to the next one, we can decompose it into two morphisms $\gamma_{i_n}^{(1)}$ and $\gamma_{i_n}^{(2)}$ such that $\gamma_{i_n} = \gamma_{i_n}^{(1)} \gamma_{i_n}^{(2)}$, each one describing the explosion of one of the two bispecial vertices. Then it suffices to show that we can decompose $\gamma_{i_n}^{(1)}$ and $\gamma_{i_n}^{(2)}$ into morphisms of $\S$. This is actually obvious. Indeed, if there are two bispecial vertices, the graph can only be of type 6 or of type 8. Then, making only one bispecial vertex explode corresponds to considering that it is actually respectively of type 5 or 7 and we know that these morphisms belong to $\S^*$. However, we have to make it carefully: if $B_1$ and $B_2$ are the two bispecial vertices in $G_{i_n}$ and if, for instance, $B_1$ is strong, we have to make $B_2$ explode before $B_1$ otherwise the explosion of $B_1$ would yield a graph with 3 right special vertices and this does not correspond to any type of graphs as considered in Figure~\ref{figure: Rauzy graphs with at least 1 bispecial vertex}. In other words, $\gamma_{i_n}^{(1)}$ has to correspond to the explosion of $B_2$ and $\gamma_{i_n}^{(2)}$ has to correspond to the explosion of $B_1$.

To conclude the proof, it suffices to build the sequences $(j_n)_{n \in \N}$ and $(\eta_n)_{n \in \N}$. From what precedes, the first one is simply the sequence $(i_n)_{n \in \N}$ but such that when $G_{i_n}$ contains two bispecial factors, then $i_n$ occurs twice in a row in $(j_n)_{n \in \N}$. The second one is the sequence $(\gamma_{i_n})_{n \in \N}$ but such that when $G_{i_n}$ contains two bispecial vertices, we split $\gamma_{i_n}$ into $\gamma_{i_n}^{(1)}$ and $\gamma_{i_n}^{(2)}$.
\end{proof}

\begin{example}
Let us consider a path $p$ in $\G$ that ultimately stays in the set of vertices $\{7,8\}$. When the Rauzy graph $G_{i_n}$ is of type 7, there is a unique bispecial factor so the morphism $\gamma_{i_n}$ satisfies the conditions of the lemma, \textit{i.e.}, it corresponds to a morphism in $(\eta_m)_{m \in \N}$. On the other hand, when $G_{i_n}$ is of type 8, its two possible evolutions are represented at Figures~\ref{Subfigure: $B_1$ strong and $B_2$ weak} and~\ref{Subfigure: $B_1$ and $B_2$ ordinary} on page~\pageref{Subfigure: $B_1$ strong and $B_2$ weak}. Suppose that the starting vertex $U_{i_n}$ corresponds to the vertex $B_1$ in Figure~\ref{figure: graph of type 8 with labels} (page~\pageref{figure: graph of type 8 with labels}) and suppose that $G_{i_n}$ evolves as in Figure~\ref{Subfigure: $B_1$ strong and $B_2$ weak} with $U_{i_n+1}$ equals to $\alpha B_1$; the others cases are analogous. We have $\gamma_{i_n} = [0, 1^k0, (1^{k-1}0)]$. To decompose it as announced in Lemma~\ref{lemma: decomposition du mot directeur}, it suffices to consider that $G_{i_n}$ is of type 7 with $B_2$ as bispecial vertex. We make this bispecial vertex explode like it is supposed to do (\textit{i.e.} like a weak bispecial factor). This makes the graph evolve to a graph $G_{i_n}'$ of type 1 (whose bispecial vertex is $B_1$) and we consider that the morphism coding this evolution is $\eta_m = [0,1]$. Now it suffices to make this new graph $G_{i_n}'$ evolve to a graph of type 7 or 8 with the morphism $\eta_{m+1} = [0,1^k0,(1^{k-1}0)]$. We then have $\gamma_{i_n} = \eta_m \eta_{m+1}$ and these new morphisms satisfy the condition 2 in Lemma~\ref{lemma: decomposition du mot directeur}. They can easily be decomposed by morphisms in $\S$ since $\eta_m = id$ and $\eta_{m+1} = \gamma_{i_n}$.
\end{example}

\begin{definition}
\label{remark: definition of j_n}
Let $(j_n)_{n \in \N}$ and $(\eta_n)_{n \in \N}$ be as in Lemma~\ref{lemma: decomposition du mot directeur}. For all $n$ we let $B_{j_n}$ denote the bispecial factor of length $j_n$ whose explosion is coded by $\eta_n$. 
\end{definition}

The following result directly follows from the definition of the morphisms $\eta_n$.

\begin{lemma}
\label{lemme: eta = identite ssi}
Let $(j_n)_{n \in \N}$ and $(\eta_n)_{n \in \N}$ be as in Lemma~\ref{lemma: decomposition du mot directeur}. The morphism $\eta_n$ is a letter-to-letter morphism if and only if $B_{j_n} \neq U_{j_n}$ (where $(U_n)_{n \in \N}$ is the sequence of starting vertices of the circuits).
\end{lemma}

\begin{remark}
\label{remark: choix des circuits et consequences}
Observe that, as illustrated by Example~\ref{ex: not valid path 3}, when $B_{j_n} \neq U_{j_n}$, the evolution of $G_{j_n}$ is influenced by the last morphism $\eta_k$, $k<n$, such that $B_{j_k} = U_{j_k}$. Indeed, as we have seen in Section~\ref{subsection: a procedure to assign letters to circuits}, the circuits starting from $U_{j_k}$ may depend on some parameters (the number of loops they contain for instance) and there exist some restrictions to these parameters\footnote{For instance, when there are two parameters $k$ and $\ell$, one of them can sometimes not be greater than the other one.}. Actually, considering a particular morphism $\eta_k$ corresponds to determining these parameters. Since some of these circuits go through the other right special vertex in $G_{j_k}$ (if it exists), these parameters influence the behaviour of this right special vertex. 

On the other hand, when $B_{j_n} = U_{j_n}$, there are no restrictions on the possibilities for $\eta_n$ since we do not have any information on the circuits starting from the right special vertex that is not $U_{j_n}$. Also, for graphs in components $C_3$ and $C_4$ there are no restrictions on the labels of the circuits like there are for Rauzy graphs of type\footnote{For those graphs, the right label of $\theta_{i_n}(x)$ always starts with $x$.} 2 or 3. Consequently, all possible morphisms are allowed. However, some of these morphisms are only \textit{locally} allowed, \textit{i.e.}, even if a morphism is allowed, some \enquote{infinite choices} containing it may be forbidden. Indeed, Example~\ref{ex: not valid path 1} shows that a graph of type $9$ can evolve to a graph of type 9 (so there is an allowed evolution) but it cannot ultimately keep being a graph of type 9 otherwise $(\gamma_{i_n})_{n \in \N}$ would not be everywhere growing. To be clearer, the circuits starting in the right special vertex that is not $U_{j_n}$ also depend on some parameters and, as for the circuits starting from $U_{j_n}$, there are some restrictions on them. Those parameters are \textit{partially} determined by the morphism $\eta_n$. For instance let us consider the evolution of a graph of type 9 as in Figure~\ref{figure: evolution of a type 9} (Page~\pageref{figure: evolution of a type 9}) such that $U_{j_n}$ corresponds to the vertex $B$ in Figure~\ref{figure: graph of type 9 before evolution}. This evolution implies that all circuits starting from the vertex $R$ in Figure~\ref{figure: graph of type 9 before evolution} go through the loop $B \rightarrow B$ at least once.
\end{remark}

\subsection{Valid paths in $C_3$}
\label{subsection: component C_3}

This component only contains the vertex $4$ in $\G$ and this type of graphs contains two right special vertices. Moreover, these two right special vertices cannot be bispecial at the same time since there is only one left special factor of each length. Consequently, we have $j_n = i_n$ and $\eta_n = \gamma_{i_n}$ for all $n$ and, as explained in Remark~\ref{remark: choix des circuits et consequences}, we can \textit{locally} choose any morphism we want when $U_{i_n} = B_{i_n}$ and we have to be careful when $U_{i_n} \neq B_{i_n}$. In other words, when $U_{i_n}$ is the vertex $R$ in Figure~\ref{figure: graph of type 4}, the choice of the morphism $\gamma_{i_n}$ is restrained by the latest morphism $\gamma_{i_m}$, $m <n$, such that $U_{i_m}$ is the vertex $B$. We let the reader check that this morphism $\gamma_{i_m}$ is either
\begin{eqnarray*}
	[0 x^ky,x^{\ell}y,(0 x^{k-1}y)] 	& 	\text{or}	& 	[x^ky,0 x^{\ell}y,(x^{k-1}y)]	
\end{eqnarray*}
with $\{x,y\} = \{1,2\}$, $k \geq 1$ and  $k \geq \ell \geq 0$.
\begin{figure}[h!tbp]
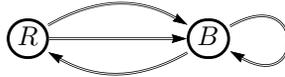

\centering
\scalebox{0.6}{
\begin{VCPicture}{(0,0)(6,2)}
\ChgEdgeLabelScale{0.5}
\StateVar[R]{(1,1)}{R}
\StateVar[B]{(5,1)}{B}
\EdgeLineDouble
\LArcL{R}{B}{}
\EdgeL{R}{B}{}
\LoopE{B}{}
\LArcL{B}{R}{}
\end{VCPicture}
}
\caption{Rauzy graph of type 4.}
\label{figure: graph of type 4}
\end{figure}

Lemma~\ref{lemma: allowed path for type 4} below expresses the consequences of this morphism $\gamma_{i_m}$. 

\begin{lemma}
\label{lemma: allowed path for type 4}
Let $m \in \N$ and $G_{i_m}$ be a Rauzy graph of type 4.
\newline 
Suppose that $U_{i_m} = R$ and that the two $i_m$-circuits $\theta_{i_m}(0)$ and $\theta_{i_m}(1)$ go through the loop $k$ and $\ell$ times respectively, with $k \geq 1$ and $k \geq \ell \geq 0$.
\newline
If the circuit $\theta_{i_m}(2)$ exists:
\begin{enumerate}[i.]
	\item	if $\ell=k$, the Rauzy graph will evolve to a graph $G_{i_n}$, $n >m$ of type 10 such that $U_{i_n}$ corresponds to the vertex $B$ in Figure~\ref{figure: type 10} (page~\pageref{figure: type 10}) and the evolution from $G_{i_m}$ to $G_{i_n}$ is coded by the morphism $[1,0,2]$;
	
	\item	if $\ell=k-1$, the Rauzy graph will evolve to a graph $G_{i_n}$, $n >m$ of type 4 such that $U_{i_n}$ corresponds to the vertex $B$ in Figure~\ref{figure: graph of type 4} just above and the evolution from $G_{i_m}$ to $G_{i_n}$ is coded by a morphism in $\{[1,0,2],[1,2,0]\}$;
	
	\item	if $\ell<k-1$, the Rauzy graph will evolve to a graph $G_{i_n}$, $n >m$ of type 7 or 8 such that $U_{i_n}$ corresponds to one of the vertices $R$ and $B$ in Figure~\ref{figure: type 7} and to one of the vertices $B_1$ and $B_2$ in Figure~\ref{figure: type 8}. The evolution from $G_{i_n}$ to $G_{i_m}$ is coded by the morphism $[1,0,2]$ and we refer to Lemma~\ref{Lemma: graph of type 7 or 8} with $k := k - \ell - 1$ to know what will next happen.
\end{enumerate}

If the circuit $\theta_{i_m}(2)$ does not exist:
\begin{enumerate}[i.]
	\item	if $\ell = k$ or $\ell=k-1$, the graph will evolve to a graph $G_{i_n}$, $n >m$ of type 1 such that $U_{i_n}$ corresponds to the vertex $B$ in Figure~\ref{figure: type 1} and the evolution from $G_{i_m}$ to $G_{i_n}$ is coded by in morphism in $\{ [0,1],[1,0] \}$;
	\item	if $\ell<k-1$, the graph will evolve to a graph $G_{i_n}$, $n >m$ of type 7 or 8 such that $U_{i_n}$ corresponds to one of the vertices $R$ and $B$ in Figure~\ref{figure: type 7} and to one of the vertices $B_1$ and $B_2$ in Figure~\ref{figure: type 8}. The evolution from $G_{i_m}$ to $G_{i_n}$ is coded by the morphism $[1,0]$ and we refer to Lemma~\ref{Lemma: graph of type 7 or 8} with $k := k - \ell - 1$ to know what happens next.
\end{enumerate}
\end{lemma}

\begin{proof}
It suffices to see how the graph evolves. Indeed, when the vertex $B$ explodes, we have eight possibilities represented at Figure~\ref{figure: evolutions of a type 4 with C} and Figure~\ref{figure: evolutions of a type 4 without C}. The main thing to notice is that if both circuits\footnote{The reader is invited to check the definition of $\theta_{i_m}$ for such graphs on page~\pageref{item 4}.} $\theta_{i_m}(0)$ and $\theta_{i_m}(1)$ can go through the loop $B \rightarrow B$ respectively $k$ and $\ell$ times with $k$ and $\ell$ greater than 1 (observe that in this case, the circuit $\theta_{i_m}(2)$ goes through that loop $k-1$ times), the graph will evolve as in Figure~\ref{figure: evolution 1 of type 4 with C} and the new circuits $\theta_{i_m+1}(0)$ and $\theta_{i_m+1}(1)$ will go through the loop respectively $k-1$ and $\ell -1$ times (so $k-2$ times for $\theta_{i_m+1}(2)$). The computation of the morphisms is left to the reader.
\end{proof}

\begin{figure}[h!tbp]
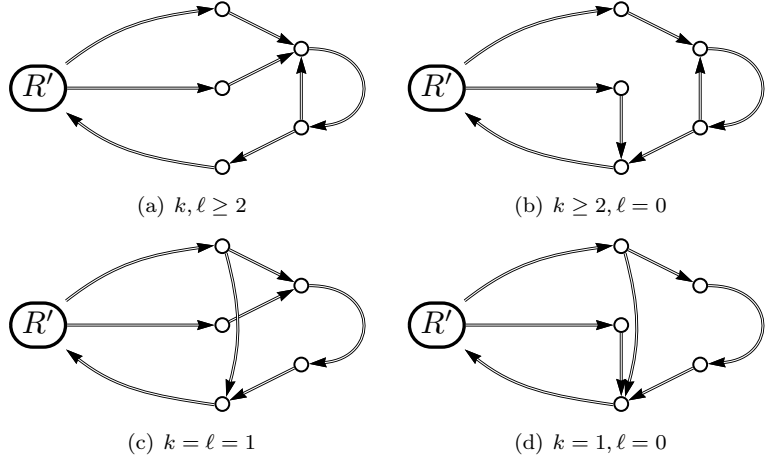

\begin{center}
\subfigure[$k, \ell \geq 2$]{
\label{figure: evolution 1 of type 4 with C}
\scalebox{0.7}{
\begin{VCPicture}{(0,0.5)(6,4)}
\ChgEdgeLabelScale{0.8}
\StateVar[R']{(0,2.5)}{R'}
\VSState{(3.5,4)}{HG}
\VSState{(3.5,2.5)}{MG}
\VSState{(3.5,1)}{BG}
\VSState{(5,3.25)}{HD}
\VSState{(5,1.75)}{BD}
\EdgeLineDouble
\Edge{HG}{HD}
\Edge{MG}{HD}
\Edge{BD}{BG}
\Edge{BD}{HD}
\EdgeLineDouble
\ArcL{R'}{HG}{}
\Edge{R'}{MG}
\ArcL{BG}{R'}{}
\VCurveR[]{angleA=0,angleB=5,ncurv=1.8}{HD}{BD}{}
\end{VCPicture}
}}
\qquad
\subfigure[$k \geq 2, \ell = 0$]{
\label{figure: evolution 2 of type 4 with C}
\scalebox{0.7}{
\begin{VCPicture}{(0,0.5)(6,4)}
\ChgEdgeLabelScale{0.8}
\StateVar[R']{(0,2.5)}{R'}
\VSState{(3.5,4)}{HG}
\VSState{(3.5,2.5)}{MG}
\VSState{(3.5,1)}{BG}
\VSState{(5,3.25)}{HD}
\VSState{(5,1.75)}{BD}
\EdgeLineDouble
\Edge{HG}{HD}
\Edge{MG}{BG}
\Edge{BD}{BG}
\Edge{BD}{HD}
\EdgeLineDouble
\ArcL{R'}{HG}{}
\Edge{R'}{MG}
\ArcL{BG}{R'}{}
\VCurveR[]{angleA=0,angleB=5,ncurv=1.8}{HD}{BD}{}
\end{VCPicture}
}}
\\
\subfigure[$k = \ell = 1$]{
\label{figure: evolution 3 of type 4 with C}
\scalebox{0.7}{
\begin{VCPicture}{(0,0.5)(6,4)}
\ChgEdgeLabelScale{0.8}
\StateVar[R']{(0,2.5)}{R'}
\VSState{(3.5,4)}{HG}
\VSState{(3.5,2.5)}{MG}
\VSState{(3.5,1)}{BG}
\VSState{(5,3.25)}{HD}
\VSState{(5,1.75)}{BD}
\EdgeLineDouble
\Edge{HG}{HD}
\Edge{MG}{HD}
\Edge{BD}{BG}
\ArcL{HG}{BG}{}
\EdgeLineDouble
\ArcL{R'}{HG}{}
\Edge{R'}{MG}
\ArcL{BG}{R'}{}
\VCurveR[]{angleA=0,angleB=5,ncurv=1.8}{HD}{BD}{}
\end{VCPicture}
}}
\qquad
\subfigure[$k = 1, \ell = 0$]{
\label{figure: evolution 4 of type 4 with C}
\scalebox{0.7}{
\begin{VCPicture}{(0,0.5)(6,4)}
\ChgEdgeLabelScale{0.8}
\StateVar[R']{(0,2.5)}{R'}
\VSState{(3.5,4)}{HG}
\VSState{(3.5,2.5)}{MG}
\VSState{(3.5,1)}{BG}
\VSState{(5,3.25)}{HD}
\VSState{(5,1.75)}{BD}
\EdgeLineDouble
\Edge{HG}{HD}
\Edge{MG}{BG}
\Edge{BD}{BG}
\ArcL{HG}{BG}{}
\EdgeLineDouble
\ArcL{R'}{HG}{}
\Edge{R'}{MG}
\ArcL{BG}{R'}{}
\VCurveR[]{angleA=0,angleB=5,ncurv=1.8}{HD}{BD}{}
\end{VCPicture}
}}
\end{center}
\caption{Evolutions of a graph of type 4 with 3 circuits starting from $R$.}
\label{figure: evolutions of a type 4 with C}
\end{figure}

\begin{figure}[h!tbp]
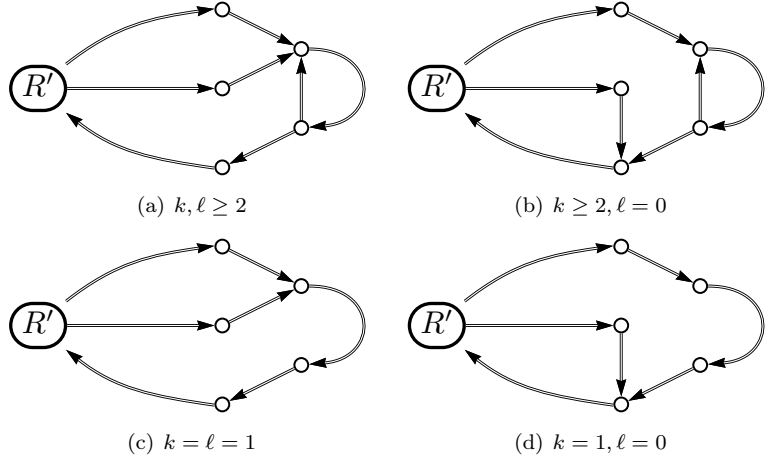

\begin{center}
\subfigure[$k, \ell \geq 2$]{
\label{figure: evolution 1 of type 4 without C}
\scalebox{0.7}{
\begin{VCPicture}{(0,0.5)(6,4)}
\ChgEdgeLabelScale{0.8}
\StateVar[R']{(0,2.5)}{R'}
\VSState{(3.5,4)}{HG}
\VSState{(3.5,2.5)}{MG}
\VSState{(3.5,1)}{BG}
\VSState{(5,3.25)}{HD}
\VSState{(5,1.75)}{BD}
\EdgeLineDouble
\Edge{HG}{HD}
\Edge{MG}{HD}
\Edge{BD}{BG}
\Edge{BD}{HD}
\EdgeLineDouble
\ArcL{R'}{HG}{}
\Edge{R'}{MG}
\ArcL{BG}{R'}{}
\VCurveR[]{angleA=0,angleB=5,ncurv=1.8}{HD}{BD}{}
\end{VCPicture}
}}
\qquad
\subfigure[$k \geq 2, \ell = 0$]{
\label{figure: evolution 2 of type 4 without C}
\scalebox{0.7}{
\begin{VCPicture}{(0,0.5)(6,4)}
\ChgEdgeLabelScale{0.8}
\StateVar[R']{(0,2.5)}{R'}
\VSState{(3.5,4)}{HG}
\VSState{(3.5,2.5)}{MG}
\VSState{(3.5,1)}{BG}
\VSState{(5,3.25)}{HD}
\VSState{(5,1.75)}{BD}
\EdgeLineDouble
\Edge{HG}{HD}
\Edge{MG}{BG}
\Edge{BD}{BG}
\Edge{BD}{HD}
\EdgeLineDouble
\ArcL{R'}{HG}{}
\Edge{R'}{MG}
\ArcL{BG}{R'}{}
\VCurveR[]{angleA=0,angleB=5,ncurv=1.8}{HD}{BD}{}
\end{VCPicture}
}}
\\
\subfigure[$k = \ell = 1$]{
\label{figure: evolution 3 of type 4 without C}
\scalebox{0.7}{
\begin{VCPicture}{(0,0.5)(6,4)}
\ChgEdgeLabelScale{0.8}
\StateVar[R']{(0,2.5)}{R'}
\VSState{(3.5,4)}{HG}
\VSState{(3.5,2.5)}{MG}
\VSState{(3.5,1)}{BG}
\VSState{(5,3.25)}{HD}
\VSState{(5,1.75)}{BD}
\EdgeLineDouble
\Edge{HG}{HD}
\Edge{MG}{HD}
\Edge{BD}{BG}
\EdgeLineDouble
\ArcL{R'}{HG}{}
\Edge{R'}{MG}
\ArcL{BG}{R'}{}
\VCurveR[]{angleA=0,angleB=5,ncurv=1.8}{HD}{BD}{}
\end{VCPicture}
}}
\qquad
\subfigure[$k = 1, \ell = 0$]{
\label{figure: evolution 4 of type 4 without C}
\scalebox{0.7}{
\begin{VCPicture}{(0,0.5)(6,4)}
\ChgEdgeLabelScale{0.8}
\StateVar[R']{(0,2.5)}{R'}
\VSState{(3.5,4)}{HG}
\VSState{(3.5,2.5)}{MG}
\VSState{(3.5,1)}{BG}
\VSState{(5,3.25)}{HD}
\VSState{(5,1.75)}{BD}
\EdgeLineDouble
\Edge{HG}{HD}
\Edge{MG}{BG}
\Edge{BD}{BG}
\EdgeLineDouble
\ArcL{R'}{HG}{}
\Edge{R'}{MG}
\ArcL{BG}{R'}{}
\VCurveR[]{angleA=0,angleB=5,ncurv=1.8}{HD}{BD}{}
\end{VCPicture}
}}
\end{center}
\caption{Evolutions of a graph of type 4 with 2 circuits starting from $R$.}
\label{figure: evolutions of a type 4 without C}
\end{figure}

Now we can determine the valid suffixes in component $C_3$. Moreover, in $\G$ we can rename the vertex $4$ by $4B$, meaning that we always have $U_{i_n} = B$.

\begin{proposition}
\label{proposition: valid path in C3}
An infinite path $p$ in $\G$ labelled by $(\gamma_{i_n})_{n \geq N}$ is a valid suffix that always stays in vertex $4$ and that is such that $U_{i_N}$ is bispecial if and only if there is a contraction $(\alpha_n)_{n \geq N}$ of $(\gamma_{i_n})_{n \geq N}$ such that 


\begin{enumerate}

	\item 	for all $n \geq N$, 
\begin{multline*}
\alpha_n \in \left\{ [0,10,20], [0,20,10], 
		[x^{k-1}y,0x^ky,0x^{k-1}y] , [x^{k-1}y,0x^{k-1}y,0x^ky],	\right.\\
 			\left.	 [0x^{k-1}y,x^ky,x^{k-1}y], [0x^{k-1}y,x^{k-1}y,x^ky] \mid k \geq 1 \right\}
\end{multline*}
with $\{x,y\} = \{1,2\}$; 

	\item 	for all $r \geq N$,  
\[
	(\alpha_n)_{n \geq r} \notin  {\{ [0,10,20], [0,20,10] \}}^{\omega}
\]
and
\[		
	(\alpha_n)_{n \geq r} \notin  \left\{ [0x^{k-1}y,x^ky,x^{k-1}y], [0x^{k-1}y,x^{k-1}y,x^ky] \mid k \geq 1 \right\}^{\omega} 
\]				


\end{enumerate}
\end{proposition}

\begin{proof}
Our aim is to describe valid suffix in $\G$ that stay in vertex 4, accordingly to Proposition~\ref{prop: valid path}. 

Let us start with Condition 1. Given a graph $G_{i_n}$ of type 4 with $U_{i_n} = B$, the morphism $\gamma_{i_n}$ coding the evolution to a graph of type 4 and such that
\begin{itemize}
	\item[(a)] $U_{i_{n+1}} = B$ are $[0,10,20]$ and $[0,20,10]$;
	
	\item[(b)] $U_{i_{n+1}} = R$ are $[0x^ky,x^{\ell}y,0x^{k-1}y]$ and $[x^ky,0x^{\ell}y,x^{k-1}y]$. 
\end{itemize}
%
Let $(k_n)_{n \geq N}$ be the subsequence of $(i_n)_{n \geq N}$ such that $U_{i_n}$ is bispecial if and only if $i_n \in \{k_n \mid n \geq N\}$ and let $(\alpha_n)_{n \geq N}$ be the contraction of $(\gamma_{i_n})_{n \geq N}$ defined by $\alpha_n = \gamma_{k_n} \gamma_{k_n+1} \cdots \gamma_{k_{n+1}-1}$. Using Lemma~\ref{lemma: allowed path for type 4}, we obtain that $(\gamma_{i_n})_{n \geq N}$ has valid prefixes if and only if all morphisms $\alpha_n$ belong to 
\begin{multline*}
\left\{ [0,10,20], [0,20,10], 
		[x^{k-1}y,0x^ky,0x^{k-1}y] , [x^{k-1}y,0x^{k-1}y,0x^ky],	\right.\\
 			\left.	 [0x^{k-1}y,x^ky,x^{k-1}y], [0x^{k-1}y,x^{k-1}y,x^ky] \mid k \geq 1 \right\}.
\end{multline*}
Indeed, the exponent $k$ (resp. $\ell$) in the morphisms given above (in (b)) corresponds to the number of times the circuit $\theta_{i_{n+1}}(0)$ (resp. $\theta_{i_{n+1}}(0)$) goes through the loop $B \to B$. Consequently, we must have $\ell = k-1$.


Now let us consider Condition 2. All morphisms $\alpha_n$ are right proper so we only have to take care of the weak primitivity and it is easily seen that $(\alpha_n)_{n \geq N}$ is weakly primitive if and only if for all $r \geq N$,  
\[
	(\alpha_n)_{n \geq r} \notin  {\{ [0,10,20], [0,20,10] \}}^{\omega}
\]
and
\[		
	(\alpha_n)_{n \geq r} \notin  \left\{ [0x^{k-1}y,x^ky,x^{k-1}y], [0x^{k-1}y,x^{k-1}y,x^ky] \mid k \geq 1 \right\}^{\omega} 
\]				
with $\{x,y\} = \{1,2\}$.
\end{proof}

\subsection{Valid paths in $C_4$}
\label{subsection: component C_4}

This component of $\G$ contains the vertices 1, 5, 6, 7, 8, 9 and 10. As for component $C_3$, we need some lemmas to determine the consequences of some morphisms $\gamma_{i_n}$ on the sequence $(\gamma_{i_k})_{k \geq n+1}$. The difficulty in determining the valid paths in this component lies in the fact that we have to take care of the length of some paths in the Rauzy graphs to know which morphism is allowed. Indeed, the morphisms that code the evolutions to Rauzy graphs of type 5 or 6 (and 7 or 8) are the same and the precise type depends on the lengths of the path $p_1$ and $p_2$ in Figure~\ref{figure: graph with no loop'} (and of the lengths of the paths $u_1$, $u_2$, $v_1$ and $v_2$ in Figure~\ref{figure: graph with 2 loops'}. When the Rauzy graph $G_{i_n}$ is of type 6 or 8 (\textit{i.e.}, when $|p_1|=|p_2|$ or when $|u_1| = |u_2|$), we know from Lemma~\ref{lemma: decomposition du mot directeur} that we can decompose the morphism $\gamma_{i_n}$ into two morphisms, each one corresponding to the explosion of one bispecial vertex. On the other hand, if for example $|u_1| >> |u_2|+|v_2|$ in Figure~\ref{figure: graph with 2 loops'} and if we denote by $B_1(1), B_1(2), \dots$ (resp. $B_2(1), B_2(2), \dots$) the bispecial vertices (ordered by increasing length) in the Rauzy graphs of larger order that admit $R_1$ (resp. $R_2$) as a suffix, we will see that many vertices $B_1(i)$ will explode before that $B_2(1)$ explodes. Consequently not all morphisms are allowed.

\begin{figure}[h!tbp]
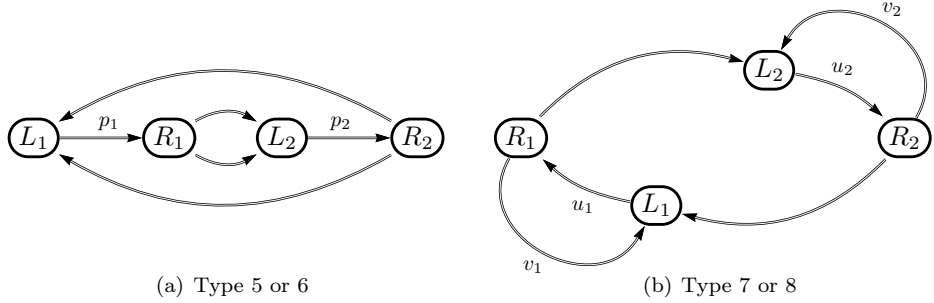

\centering
\subfigure[Type 5 or 6]{
\scalebox{0.6}{
\begin{VCPicture}{(0,-0.5)(9,5)}
\label{figure: graph with no loop'}
\ChgEdgeLabelScale{0.8}
\StateVar[L_1]{(0,2.5)}{L_1}
\StateVar[R_1]{(3,2.5)}{R_1}
\StateVar[L_2]{(5.5,2.5)}{L_2}
\StateVar[R_2]{(8.5,2.5)}{R_2}
\EdgeLineDouble
\LArcL{R_1}{L_2}{}
\LArcR{R_1}{L_2}{}
\LArcL{R_2}{L_1}{}
\LArcR{R_2}{L_1}{}
\EdgeL[.6]{L_1}{R_1}{p_1}
\EdgeL[.4]{L_2}{R_2}{p_2}
\end{VCPicture}
}}
\qquad
\subfigure[Type 7 or 8]{
\scalebox{0.6}{
\begin{VCPicture}{(0,-0.5)(9,5)}
\label{figure: graph with 2 loops'}
\ChgEdgeLabelScale{0.8}
\StateVar[R_1]{(0,2.5)}{R_1}
\StateVar[L_2]{(5.5,4)}{L_2}
\StateVar[R_2]{(8.5,2.5)}{R_2}
\StateVar[L_1]{(3,1)}{L_1}
\EdgeLineDouble
\LArcL{R_1}{L_2}{}
\VCurveR[]{angleA=-120,angleB=-120,ncurv=1.2}{R_1}{L_1}{v_1}
\VCurveR []{angleA=60,angleB=60,ncurv=1.2}{R_2}{L_2}{v_2}
\LArcL{R_2}{L_1}{}
\ArcL{L_2}{R_2}{u_2}
\ArcL{L_1}{R_1}{u_1}
\end{VCPicture}
}}
\caption{The next evolutions of these graphs depend on the length of the pats $u_i$, $v_i$ and $p_i$.}
\label{figure: rauzy graph of type 5,6,7,8}
\end{figure}

First, the following result will be helpful to characterize valid paths that goes infinitely often through the vertex $1$ in the graph of graphs.

\begin{fact}
\label{fact: graph of type 1}
We can suppose without loss of generality that the evolution of a Rauzy graph of type 1 to a Rauzy graph of type 1 is coded by $[0,10]$ or by $[01,1]$.
\end{fact}

{
\begin{proof}
The morphisms coding an evolution from a graph of type 1 to a graph of type 1 are $[0,10] = D_{1,0}$, $[10,0] = D_{1,0} E_{0,1}$, $[01,1] = D_{0,1}$ and $[1,01] = D_{0,1} E_{0,1}$ and that the morphisms coding an evolution from a graph of type 1 to a graph of type 7 or 8 are $[0,1^k0,1^{k-1}0]$ and $[1,0^k1,0^{k-1}1] = E_{0,1} [0,1^k0,1^{k-1}0]$. 

By induction, it is easily seen that for all integers $n \geq 0$, we have 
\[
	E_{0,1} \left\{ D_{0,1}, D_{1,0} \right\}^n E_{0,1} = \left\{ D_{0,1}, D_{1,0} \right\}^n.
\]
To conclude the proof of the result, we have to consider several possibilities.
\begin{enumerate}

	\item 	If for all $n$, $\gamma_{i_n}$ codes an evolution from a graph of type 1 to a graph of type 1 and if $(\gamma_{i_n})_{n \in \N}$ contains infinitely many occurrences of $D_{1,0} E_{0,1}$ and/or of $D_{0,1} E_{0,1}$, then the result trivially holds.
	
	\item 	If for all $n$, $\gamma_{i_n}$ codes an evolution from a graph of type 1 to a graph of type 1 and if $(\gamma_{i_n})_{n \in \N}$ contains a finite and even number of occurrences of $D_{1,0} E_{0,1}$ and/or of $D_{0,1} E_{0,1}$, then the result trivially holds too.
	
	\item 	If for all $n$, $\gamma_{i_n}$ codes an evolution from a graph of type 1 to a graph of type 1 and if $(\gamma_{i_n})_{n \in \N}$ contains a finite and odd number of occurrences of $D_{1,0} E_{0,1}$ and/or of $D_{0,1} E_{0,1}$, then it suffices to insert in $(\gamma_{i_n})_{n \in \N}$ infinitely many occurrences of the morphism $id = E_{0,1}^2$ and the result holds.
	
	\item 	Finally, if $\gamma_{i_r} \cdots \gamma_{i_s} \in \{ D_{1,0}, D_{1,0} E_{0,1}, D_{0,1}, D_{0,1} E_{0,1} \}^*$ codes a finite sequence of evolutions from graphs of type 1 to graphs of type 1 and if $\gamma_{i_{s+1}} \in \{[0,1^k0,1^{k-1}0], [1,0^k1,0^{k-1}1] \}$ that codes an evolution to a graph of type 7 or 8, then $\gamma_{i_r} \cdots \gamma_{i_s} \gamma_{i_{s+1}}$ can be replaced by $\gamma_{i_r}' \cdots \gamma_{i_s}' \gamma_{i_{s+1}}'$ with $\gamma_{i_r}' \cdots \gamma_{i_s}' \in \{D_{0,1}, D_{1,0}\}^*$ and $\gamma_{i_{s+1}}'\in \{[0,1^k0,1^{k-1}0], [1,0^k1,0^{k-1}1] \}$, depending on the number of occurrences of $D_{1,0} E_{0,1}$ and of $D_{1,0} E_{0,1}$ in $\gamma_{i_r} \cdots \gamma_{i_s}$.
\end{enumerate}
\end{proof}
}

Next, Lemma~\ref{Lemma: graph of type 5 or 6} implies that we can merge the vertices 5 and 6 to one vertex denoted by $5/6$ in $\G$ and that the outgoing edges of that vertex are the same as the outgoing edges of the vertex 6 in $\G$. However, we have to take care of the lengths of $p_1$ and $p_2$ in Figure~\ref{figure: graph with no loop'} to know which morphism in the labels of the edges can be applied.

\begin{lemma}
\label{Lemma: graph of type 5 or 6}
Let $G_k$ be a Rauzy graph as in Figure~\ref{figure: graph with no loop'} and let $i_n$ be the smallest integer in $(i_n)_{n \in \N}$ such that $i_n \geq k$. We have
\begin{multline*}
	\{ \text{Type of } G_{i_{n+1}} \mid G_{i_n} \text{ is of type } 6 \} = \\
	\{ \text{Type of } G_{i_{n+2}} \mid G_{i_n} \text{ is of type } 5 \text{ and }  U_{i_n} \text{ is not strong bispecial}\}
\end{multline*}
and
\[	
	\{ \gamma_{i_n} \mid G_{i_n} \text{ is of type } 6 \} = 
	\{ \gamma_{i_n} \circ \gamma_{i_{n+1}} \mid G_{i_n} \text{ is of type } 5 \text{ and } U_{i_n} \text{ is not strong bispecial} \}.
\]
\end{lemma}

\begin{proof}
The first equality can be easily checked on the graph of graphs (Figure~\ref{figure: graph of graphs} on page~\pageref{figure: graph of graphs}) and the second one is deduced from the computation of morphisms coding the needed evolutions. 
{
Those are given in Table~\ref{morphismes a partir de 5 ou 6} (take care to match $(U_{i_n},U_{i_{n+1}})$ and $(U_{i_{n+1}},U_{i_{n+2}})$).
\begin{table}[h!tbp]
\centering
\begin{tabular}{|l|l|c|l|l|}
\hline
From		& 	To 	& 	$(U_{i_n},U_{i_{n+1}})$		&	Morphisms 	& 	Conditions		\\
\hline
6	&	1	&	$(\star,B)$		&	$[x,yx], [yx,x]$		&	 	\\
\hhline{~----}
	&	7 or 8	&	$(\star,\star)$	&	$[1,0^k2,(0^{k-1}2)]$	&	$k \geq 1$									\\
\hhline{~~~--}
	&		&					&	$[x,y^kx,(y^{k-1}x)]$	&	$k \geq 2$ 	\\
\hhline{~----}
	&	10	&	$(\star,B)$		&	$[1,0 1,2]$						&										\\
\hhline{~~---}
	&		&	$(\star,R)$		&	$[1 2^k 0,2^{\ell}0]$				&	$k,\ell \geq 0$, $k + \ell \geq 1$	\\
\hhline{~~~--}
	&		&					&	$[1 2^k 0,2^{\ell}0,1 2^{k-1}0]$	&	$k \geq \ell \geq 0$, $k \geq 1$	\\
\hhline{~~~--}
	&		&					&	$[1 2^k 0,2^{\ell}0, 2^{\ell-1}0]$	&	$\ell > k \geq 0$					\\
\hline	
\hline
5  	& 	1	&	$(R,B)$			&	$[x,y]$		&	 	\\
\hhline{~----}
	&	10	&	$(R,B)$		&	$[1,2,0]$							&										\\
\hhline{~~---}
	&		&	$(B,R)$		&	$[1,0 1,2]$							&										\\
\hhline{~~~--}
	&		&				&	$[0^k2,1,(0^{k-1}2)]$				&	$k \geq 1$							\\	
\hhline{~~~--}
	&		&				&	$[2^k 0,1 2^{\ell}0]$				&	$k,\ell \geq 0$, $k + \ell \geq 1$	\\
\hhline{~~~--}
	&		&				&	$[2^k 0,1 2^{\ell}0,2^{k-1}0]$		&	$k \geq \ell \geq 0$, $k \geq 1$	\\
\hhline{~~~--}
	&		&				&	$[2^k 0,1 2^{\ell}0,1 2^{\ell-1}0]$	&	$\ell > k \geq 0$					\\
\hline	
\hline
1 	&	1	&	$(B,B)$	&	$[x,yx]$, $[yx,x]$	&	\\
\hhline{~----}
	&	7 or 8	&	$(B,\star)$	&	$[x,y^kx,(y^{k-1}x)]$	&	$k \geq 2$ \\
\hline
\hline
10 	& 	1	&	$(R,B)$		&	$[x,y]$		&		\\
\hhline{~----}
	&	7 or 8	&	$(R,\star)$		&	$[1,0,(2)]$					&	\\
\hhline{~~---}
	&			&	$(B,\star)$		&	$[0,2^k 1,(2^{k-1}1)]$		&	$k \geq 1$	\\
\hhline{~----}
	&		10	&	$(R,R)$		&	$[1,0,(2)]$		&	\\
\hhline{~~---}
	&			&	$(B,B)$		&	$[0,2 0,1]$		&	\\
\hhline{~~---}
	&			&	$(R,B)$		&	$[0,1,2]$		&	\\
\hhline{~~---}
	&			&	$(B,R)$		&	$[0 1^k 2,1^{\ell}2]$				&	$k,\ell \geq 0$, $k+\ell \geq 1$	\\
\hhline{~~~--}
	&			&				&	$[0 1^k 2,1^{\ell}2,0 1^{k-1}2]$	&	$k \geq 1$, $k \geq \ell \geq 0$	\\
\hhline{~~~--}
	&			&				&	$[0 1^k 2,1^{\ell}2,1^{\ell-1}2]$	&	$\ell > k \geq 0$					\\	
\hline	
\end{tabular}
\caption{The morphisms coding an evolution from a Rauzy graph of type 6 are exactly the morphisms $\gamma_{i_n}\gamma_{i_{n+1}}$ where $\gamma_{i_n}$ codes an evolution from a graph of type 5.}
\label{morphismes a partir de 5 ou 6}
\end{table}
}
The only thing to observe is that when a graph $G_{i_n}$ is of type 5 and if $U_{i_n}$ corresponds to the vertex $B$ in Figure~\ref{figure: type 5} (page~\pageref{figure: type 5}), then $U_{i_n}$ cannot be a strong bispecial factors, otherwise there would be 3 right special vertices in $G_{i_n+1}$ and this does not correspond to any considered type of graphs.
\end{proof}

\begin{remark}
\label{remark: vertex 5/6}
In order to describe all valid paths in the component $C_4$, we sometimes have to know the precise type of a graph corresponding to the vertex $5/6$. Indeed, when going to that vertex in the modified component (suppose the label of the edge is $\gamma_{i_n}$ and that $U_{i_{n+1}}$ corresponds to the vertex $R_1$ in Figure~\ref{figure: graph with no loop'}), we may want to leave it using the morphism $\gamma_{i_{n+1}} = [x,y^kx,(y^{k-1}x)]$ 
{
(see Appendix~\ref{appendix 6 to}).
}
However, the evolution corresponding to that morphism is such that the smallest bispecial factor that admits $U_{i_n+1}$ as a suffix is strong (the other right special vertex is therefore suffix of a weak bispecial factor). Consequently, we can leave the vertex $5/6$ with that morphism only if $U_{i_{n+1}}$ is not bispecial, \textit{i.e.}, the other right special vertex becomes bispecial before $U_{i_n+1}$. In other words, we must have $|p_1| \geq |p_2|$ in Figure~\ref{figure: graph with no loop'}.
\end{remark}

Next lemma deals with the same kind of stuffs as in Lemma~\ref{Lemma: graph of type 5 or 6} but for Rauzy graphs of type 7 and 8. As for graphs of type 5 and 6, it allows us to merge the vertices 7 and 8 to one vertex denoted $7/8$ in $\G$. 


{
\begin{lemma}
\label{Lemma: graph of type 7 or 8}
Let $G_t$ be a Rauzy graph as in Figure~\ref{figure: graph with 2 loops'} and let $i_n$ be the smallest integer in $(i_m)_{m \in \N}$ such that $i_n \geq t$. Suppose that $U_t$ is the vertex $R_1$ and that $\theta_t(1)$ goes $k$ times through the loop $v_2 u_2$. Let $\ell \in \Z$ such that
\begin{equation}
\label{eq: inequ 78}
	|u_1| + (\ell-1) (|u_1| + |v_1|) < |u_2| + (k-1) (|u_2| + |v_2|) \leq |u_1| + \ell (|u_1| + |v_1|).
\end{equation}
Then, the graph can evolve to a graph of type
\begin{enumerate}[i.]
	\item	1 and the composition of morphisms coding this evolution is in 
	\begin{multline*}
		\left\{ [0,10]^h \left\{ [01,1], [1,01] \right\} \mid 0 \leq h < \max\{1, \ell\} \right\}	\\
		\cup \left\{ [0,10]^h [x,y] \mid \{x,y\} = \{0,1\}, h = \max\{0, \ell\} \right\}
	\end{multline*}	
		
	\item 5 or 6 as in Figure~\ref{figure: graph with no loop'} and the composition of morphisms coding this evolution is in
	\[
		\left\{ [0,10,20]^h \{ [0x,y,(0y)],[x,0y,(y)] \} \mid  \{x,y\} = \{1,2\}, 0 \leq h < \max\{1,\ell\} \right\};
	\]

	\item 9 with the starting vertex $U_m$, $m > i_n$, corresponding to the vertex $B$ in Figure~\ref{figure: type 9} and the composition of morphisms coding this evolution is in 
	\[
		\left\{ [0,10,20]^h [0,x,y] \mid \{x,y\} = \{1,2\} , h = \max\{0,\ell\} \right\}.
	\]
\end{enumerate}
\end{lemma}

\begin{proof}
First let us study which are the bispecial vertices we have to deal with. It is a direct consequence of the definition of Rauzy graphs that for $i$ and $j$ in $\N$, the words $B_1(i) = \lambda \left( u_1 (v_1 u_1)^i \right)$ and $B_2(j) = \lambda \left( u_2 (v_2 u_2)^j \right)$ respectively admit $L_1$ and $L_2$ as prefixes and $R_1$ and $R_2$ as suffixes. For all $i,j$, we write $e_1(i) = |B_1(i)| = t + |u_1| + i (|u_1|+|v_1|)$ and $e_2(j) = |B_2(j)| = t + |u_2| + j (|u_2|+|v_2|)$. Inequality~\eqref{eq: inequ 78} therefore provides some information on the order the bispecial vertices $B_1(\ell-1)$, $B_2(k-1)$ and $B_1(\ell)$ (if they exist) explode.

By hypothesis, the path $u_2(v_2 u_2)^k$ is allowed in $G_t$ (since it is a subpath of a $t$-circuit). This implies that $B_2(j)$ is a bispecial factor in $\fac{}{X}$ for all $j \in \{ 0,1, \dots, k-1 \}$ and this also gives us some information on the way they explode in their respective Rauzy graphs. Indeed, if there are 2 (resp. 3) $t$-circuits starting from $R_1$ in $G_t$, then in the Rauzy graph $G_{e_2(j)}$, the vertex $B_2(j)$ explodes as in Figure~\ref{figure: $j < k-1$} if $j < k-1$ and as in Figure~\ref{figure: $j = k-1$ and 2 circuits} (resp. in Figure~\ref{figure: $j = k-1$ and 3 circuits}) if $j = k-1$. 

\begin{figure}[h!tbp]
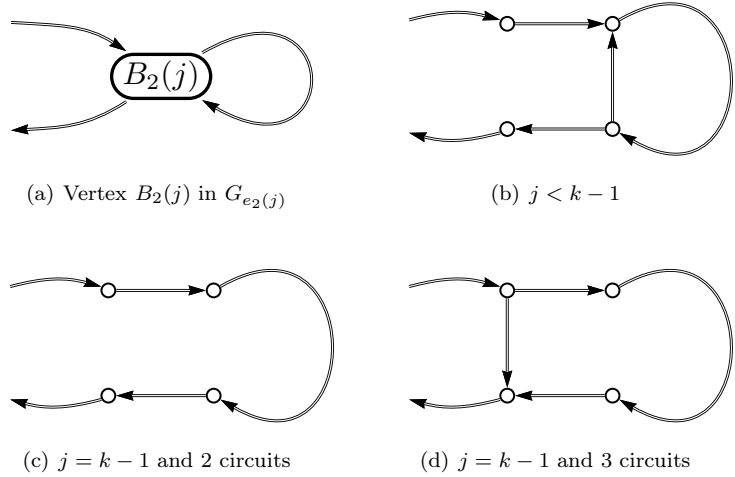

\centering
\subfigure[Vertex $B_2(j)$ in $G_{e_2(j)}$]{
\label{figure: Vertex $B_2(j)$ in $G_{e_2(j)}$}
\scalebox{0.7}{
\begin{VCPicture}{(0,0)(6,4)}
\ChgEdgeLabelScale{0.8}
\StateVar[B_2(j)]{(3,2)}{B2}
\ChgStateLineColor{white}
\EdgeLineDouble
\VSState{(0,3)}{HG}
\VSState{(0,1)}{BG}
\RstStateLineColor
\RstEdgeLineStyle
\ArcL{HG}{B2}{}
\ArcL{B2}{BG}{}
\LoopE{B2}{}
\end{VCPicture}
}}
\qquad
\subfigure[$j < k-1$]{
\label{figure: $j < k-1$}
\scalebox{0.7}{
\begin{VCPicture}{(0,0)(6,4)}
\ChgEdgeLabelScale{0.8}
\VSState{(2,3)}{aB}
\VSState{(2,1)}{Bb}
\VSState{(4,3)}{Ba}
\VSState{(4,1)}{bB}
\ChgStateLineColor{white}
\VSState{(0,3)}{HG}
\VSState{(0,1)}{BG}
\RstStateLineColor
\EdgeLineDouble
\ArcL{HG}{aB}{}
\ArcL{Bb}{BG}{}
\Edge{aB}{Ba}
\Edge{bB}{Ba}
\Edge{bB}{Bb}
\VCurveL[]{angleA=30,angleB=-30,ncurv=3}{Ba}{bB}{}
\RstEdgeLineStyle
\end{VCPicture}
}}
\\
\subfigure[$j = k-1$ and 2 circuits]{
\label{figure: $j = k-1$ and 2 circuits}
\scalebox{0.7}{
\begin{VCPicture}{(0,0)(6,4)}
\ChgEdgeLabelScale{0.8}
\VSState{(2,3)}{aB}
\VSState{(2,1)}{Bb}
\VSState{(4,3)}{Ba}
\VSState{(4,1)}{bB}
\ChgStateLineColor{white}
\VSState{(0,3)}{HG}
\VSState{(0,1)}{BG}
\RstStateLineColor
\EdgeLineDouble
\ArcL{HG}{aB}{}
\ArcL{Bb}{BG}{}
\Edge{aB}{Ba}
\Edge{bB}{Bb}
\VCurveL[]{angleA=30,angleB=-30,ncurv=3}{Ba}{bB}{}
\RstEdgeLineStyle
\end{VCPicture}
}}
\qquad
\subfigure[$j = k-1$ and 3 circuits]{
\label{figure: $j = k-1$ and 3 circuits}
\scalebox{0.7}{
\begin{VCPicture}{(0,0)(6,4)}
\ChgEdgeLabelScale{0.8}
\VSState{(2,3)}{aB}
\VSState{(2,1)}{Bb}
\VSState{(4,3)}{Ba}
\VSState{(4,1)}{bB}
\ChgStateLineColor{white}
\VSState{(0,3)}{HG}
\VSState{(0,1)}{BG}
\RstStateLineColor
\EdgeLineDouble
\ArcL{HG}{aB}{}
\ArcL{Bb}{BG}{}
\Edge{aB}{Ba}
\Edge{aB}{Bb}
\Edge{bB}{Bb}
\VCurveL[]{angleA=30,angleB=-30,ncurv=3}{Ba}{bB}{}
\RstEdgeLineStyle
\end{VCPicture}
}}
\caption{Explosion of the vertex $B_2(j)$ in $G_{e_2(j)}$.}
\end{figure}

As $U_t = R_1$, we know from Lemma~\ref{lemme: eta = identite ssi} and from Section~\ref{subsection: a procedure to assign letters to circuits} that the explosion of the vertices $B_2(j)$ are coded by the identity morphism for $j \in \{0,\dots, k-2\}$ and by a letter-to-letter morphism for $j = k-1$.

Now let us study the behaviour of the vertex $R_1$. As we do not have any information about the circuits starting from $R_2$, there are several possibilities for the explosion of the vertices $B_1(i)$. First, we can observe that, if for some integer $i < \ell$, the word $B_1(i)$ belongs to $\fac{}{X}$, then for all $h < i$, the word $B_1(h)$ is a bispecial factor in $\fac{}{X}$ and it explodes like $B_2(j)$ in Figure~\ref{figure: $j < k-1$}. Each of these evolutions is coded by $[0,10,20]$ (or by $[0,10]$ if there are only 2 circuits). On the other hand, if $B_1(i)$ is a bispecial factor of length $l < e_2(k-1)$ in $\fac{}{X}$ and if it explodes in $G_l$ similarly to $B_2(j)$ in Figure~\ref{figure: $j = k-1$ and 3 circuits}, then $G_l$ evolves to a graph of type 9 such that the starting vertex of the circuits corresponds to the vertex $R$ in Figure~\ref{figure: type 9}. Consequently, the right special vertex in $G_{l+1}$ that arises from $B_1(i)$ will not become bispecial until $B_2(k-1)$ has exploded. The evolution from $G_l$ to $G_{l+1}$ is coded by the morphism $[01,1]$ or $[1,01]$ if there are only 2 $l$-circuits and by one of the four following morphisms if there are three $l$-circuits: $[01,1,02]$, $[1,01,2]$, $[01,2,(02)]$ and $[1,02,(2)]$. Observe that $B_1(i)$ cannot explode similarly to $B_2(j)$ in Figure~\ref{figure: $j = k-1$ and 2 circuits} as that would imply that the sequence of right special vertices $(U_n)_{n \in \N}$ is finite.

To conclude the proof, it suffices to list all the possibilities for the explosions of the vertices $B_1(i)$. By hypothesis, $\ell$ is an integer such that
\[
	|u_1| + (\ell-1) (|u_1| + |v_1|) < |u_2| + (k-1) (|u_2| + |v_2|) \leq |u_1| + \ell (|u_1| + |v_1|)
\]
and we know that the vertices $B_1(i)$ and $B_2(j)$ respectively have length $e_1(i) = t+ |u_1| + i (|u_1| + |v_1|)$ and $e_2(j) = t + |u_2| + j (|u_2| + |v_2|)$ for all non-negative integers $i$ and $j$. Consequently, while $B_2(k-1)$ has not exploded yet, the vertex $B_1(i)$ (if it exists) has two possibilities: either it makes the graph evolve to a graph of type 7 or 8 with the morphism $[0,10,(20)]$, or it makes it evolve to a graph of type 9 with one of the morphisms $[01,1,(02)]$, $[1,01,(2)]$, $[01,2,(02)]$ and $[1,02,(2)]$.

First suppose that the graph is not of type 7 or 8 anymore when the vertex $B_2(k -1)$ explodes. The only possibility is that $\ell \geq 1$ and that a vertex $B_1(i)$, $0 \leq i \leq \ell-1$, has exploded as in Figure~\ref{figure: $j = k-1$ and 3 circuits}, making the graph evolve to a graph of type 9 with one of the morphisms $[01,1,(02)]$, $[1,01,(2)]$, $[01,2,(02)]$ and $[1,02,(2)]$. Observe that each of the explosions of $B_1(0), B_1(1), \dots, B_1(i-1)$ is coded by $[0,10,20]$. Then, the only bispecial vertices that occur in the next Rauzy graphs are vertices $B_2(j)$ for $j \in \{l',\dots, k-1\}$ and $l'$ the smallest integer such that $e_2(l') \geq e_1(i)$. They imply the following behaviours: for $j<k-1$, the explosions of $B_2(j)$ are coded by the identity morphism. For $j=k-1$, if there are three circuits starting from $B_1(i)$ and if its explosion is coded by the morphism $[01,1,02]$ or $[1,01,2]$ (resp. $[01,2,(02)]$ or $[1,02,(2)]$), then the explosion of $B_2(k)$ is coded by $[2,1,0]$ (resp. $[0,1,2]$). Consequently, the graph eventually evolves to a graph of type 5 or 6 and the composition of the morphisms is in
\begin{equation}
\label{eq: morphism lemma 78 1}
	\left\{	[0,10,20]^h \left\{ [0x,y,(0y)],[x,0y,(y)] \right\} \mid \{x,y\} = \{1,2\}, 0 \leq h < \max\{1,\ell\} \right\}.
\end{equation}
Still for $j = k-1$, if there are 2 circuits starting from $B_1(i)$, then the morphism coding its explosion is $[01,1]$ or $[1,01]$ and then the graph will evolve to a graph of type 1 with the morphism $[0,1]$ or $[1,0]$ (by exploding vertices $B_2(j)$). Consequently, the composition of morphisms coding this sequence of evolutions is in
\begin{eqnarray}
\label{eq: morphism lemma 78 2}
	\left\{	[0,10]^h \left\{ [01,1], [1,01] \right\} \mid 0 \leq i < \max\{1,\ell\} \right\}.
\end{eqnarray}

Now suppose that the graph is still of type 7 or 8 when the vertex $B_2(k -1)$ has exploded. If $\ell \geq 1$, this implies that the vertices $B_1(i)$ have exploded with the morphism $[0,10,(20)]$ for $i = 0,\dots,\ell-1$ (so we have $[0,10,(20)]^{\ell}$). Then, when the vertex $B_2(k-1)$ explodes, it makes the graph evolve to a graph $G_{i_m}$ of type 1 or 9 depending on the number of circuits (2 or 3 respectively). If the vertex $B_1(\ell)$ has the same length, we can suppose from Lemma~\ref{lemma: decomposition du mot directeur} that it does not explode at the same time so we can suppose that the graph does not evolve to a graph of type 7 or 8 (like it actually could with the morphism $[x,y^mx,(y^{m-1}x)]$). Consequently, we only have to consider the evolutions to graphs of type 1 or 9. They are respectively coded by $[0,1]$ or $[1,0]$ and by $[0,1,2]$ or $[0,2,1]$ and once this evolution is done, the next bispecial vertex is in $(U_n)_{n \in \N}$. 
\end{proof}

The next lemma will allow us to delete the vertex 9 in $\G$. Indeed, we can see in Figure~\ref{figure: graph of graphs} (page~\pageref{figure: graph of graphs}) that the only types of graphs that can evolve to a graph of type 9 are types 9 and 7 or 8. The next lemma states that we can modify the outgoing edges of the vertex $7/8$ such that the vertex 9 is isolated in $\G$.

\begin{lemma}
\label{lemma: C4 type 9}
In Lemma~\ref{Lemma: graph of type 7 or 8}, we can delete the third case of all possible evolutions (the one to graphs of type 9) by replacing the set of morphisms coding the evolutions to graphs of type 5 or 6 (the second case) by 
\[
	\left\{ [0,10,20]^h \{ [0x,y,(0y)],[x,0y,(y)] \} \mid \{x,y\} = \{1,2\}, h \in \N \right\}.
\]
We can also replace the morphisms coding the evolution to graphs of type 1 (the first case) by
\[
	\left\{ [0,10]^h \left\{ [01,1], [1,01] \right\} \mid h \in \N \right\}	
	\cup \left\{ [0,10]^h [x,y] \mid \{x,y\} = \{0,1\}, h \geq \max\{0, \ell\} \right\}
\]
\end{lemma}

\begin{proof}
Indeed, in Lemma~\ref{Lemma: graph of type 7 or 8} the morphisms coding the evolution to a graph of type 9 are in
\[
	\left\{ [0,10,20]^h [0,x,y] \mid \{x,y\} = \{1,2\} , h = \max\{0,\ell\} \right\}.
\]
But, once the graph is of type 9 with $U_{i_n} = B$, it can only evolve either to a graph of type 9 with $U_{i_{n+1}} = B$, or to a graph of type 5 or 6 with a morphism in $\{ [0x,y,(0y)], [x,0y,(y)] \mid \{x,y\} = \{1,2\} \}$. Consequently, the composition of evolution
\[
	7/8 (\to 9)^j \to 5/6 
\]
is coded by a morphism in
\[
	\left\{ [0,10,20]^h [0,x,y] [0,x0,y0]^j \{[0x,y,(0y)], [x,0y,(y)] \} \mid 
	\{x,y\} = \{1,2\} , h = \max\{0,\ell\} \right\}.
\]
Since $j$ can be arbitrarily large, this set is equal to
\[
	\left\{ [0,10,20]^h \{ [0x,y,(0y)],[x,0y,(y)] \} \mid \{x,y\} = \{1,2\}, h \in \N \right\}.
\]
For the second part (evolution to graphs of type 1), it suffices to observe that all considered morphisms also code evolutions from a graph of type 1 to a graph of type 1. Consequently, if $h$ is chosen greater than $\max\{0,\ell\}$, the morphism $[0,10]^{h - \max\{0,\ell\}}$ is simply coding $h - \max\{0,\ell\}$ evolutions from 1 to 1.
\end{proof}
}

{
}

The last type of graph that has not been treated yet is the type 10. The next lemma does it.

\begin{lemma}
\label{Lemma: graph of type 10}
Let $G_{i_n}$ be a Rauzy graph of type 10. Suppose that $U_{i_n}$ corresponds to the vertex $R$ in Figure~\ref{figure: type 10} and that the two $i_n$-circuits $\theta_{i_n}(0)$ and $\theta_{i_n}(1)$ respectively go through the loop $k$ and $\ell$ times with $k, \ell \geq 0$ and $k + \ell \geq 1$.
\newline
If the circuit $\theta_{i_n}(2)$ exists and starts like $\theta_{i_n}(0)$ does (recall that $\ell \leq k$ in this case), then
\begin{enumerate}[i.]
	\item	if $\ell = k$, $G_{i_n}$ will evolve to a Rauzy graph $G_{i_m}$, $m>n$, of type 10 such that $U_{i_m}$ corresponds to the vertex $B$ in Figure~\ref{figure: type 10}. This evolution is coded by the morphism $[1,0,2]$;
	\item	if $\ell < k$, $G_{i_n}$ will evolve to a Rauzy graph $G_{i_m}$, $m>n$, of type 7 or 8 such that the $i_m$-circuit $\theta_{i_n}(1)$ starting from $U_{i_m}$ goes through the loop $k' = k-\ell$ times. This evolution is also coded by the morphism $[1,0,2]$.
\end{enumerate}
If the circuit $\theta_{i_n}(2)$ exists and starts like $\theta_{i_n}(1)$ do (recall that $k \leq \ell -1$ in this case), then
\begin{enumerate}[i.]
	\item	if $k = \ell-1$, $G_{i_n}$ will evolve to a Rauzy graph $G_{i_m}$, $m>n$, of type 10 such that $U_{i_m}$ corresponds to the vertex $B$ in Figure~\ref{figure: type 10}. This evolution is coded by the morphism $[0,1,2]$;
	\item	if $k < \ell -1$, $G_{i_n}$ will evolve to a Rauzy graph $G_{i_m}$, $m>n$, of type 7 or 8 such that the $i_m$-circuit $\theta_{i_n}(1)$ starting from $U_{i_m}$ goes through the loop $k' = \ell-k-1$ times. This evolution is again coded by the morphism $[0,1,2]$.	
\end{enumerate}
If the circuit $\theta_{i_n}(2)$ does not exist, then
\begin{enumerate}[i.]
	\item	if $\ell \in \{k,k+1\}$ , $G_{i_n}$ will evolve to a Rauzy graph $G_{i_m}$, $m>n$, of type 1. This evolution is coded by a morphism in $\{[0,1],[1,0]\}$;
	\item	if $\ell<k$, $G_{i_n}$ will evolve to a Rauzy graph $G_{i_m}$, $m>n$, of type 7 or 8 such that the $i_m$-circuit $\theta_{i_n}(1)$ starting from $U_{i_m}$ goes through the loop $k' = k-\ell$ times. This evolution is coded by the morphism $[1,0]$.
	\item	if $\ell>k+1$, $G_{i_n}$ will evolve to a Rauzy graph $G_{i_m}$, $m>n$, of type 7 or 8 such that the $i_m$-circuit $\theta_{i_n}(1)$ starting from $U_{i_m}$ goes through the loop $k' = \ell-k-1$ times. This evolution is coded by the morphism $[0,1]$.
\end{enumerate}
\end{lemma}

\begin{proof}
Indeed, if the vertex $B$ in Figure~\ref{figure: type 10} explodes as in Figure~\ref{figure: evolution 1 of type 10}, the new graph is still of type 10. This evolution is coded by the morphism $[1,0,(2)]$. Moreover, if we denote by $k_{i_n}(0)$ (resp. $k_{i_n}(1)$, $k_{i_n}(2)$) the number of times that the $i_n$-circuit $\theta_{i_n}(0)$ (resp. $\theta_{i_n}(1)$, $\theta_{i_n}(2)$) goes through the loop, then we have $k_{i_n+1}(0) = k_{i_n}(1)-1$ and $k_{i_n+1}(1) = k_{i_n}(0)$. We also have $k_{i_n+1}(2) = k_{i_n}(2)$ if the $i_n$-circuit $\theta_{i_n}(2)$ starts like $\theta_{i_n}(0)$ does and $k_{i_n+1}(2) = k_{i_n}(2)-1$ if the $i_n$-circuit $\theta_{i_n}(2)$ starts like $\theta_{i_n}(1)$ does. Consequently, this evolution is repeated until either $k_{i_{n'}}(1) =0$ or $k_{i_{n'}}(0) = 0$ and $k_{i_{n'}}(1) = 1$ for some $n' \geq n$. Then the graph $G_{i_{n'}}$ evolves to a Rauzy graph of type 1, 7, 8 or 9 depending on $k_{i_{n'}}(0)$, $k_{i_{n'}}(1)$ and $k_{i_{n'}}(2)$ (if the circuit $\theta_{i_n}(2)$ exists). The computation of the morphism coding this last evolution is left to the reader.
\end{proof}

\begin{figure}[h!tbp]
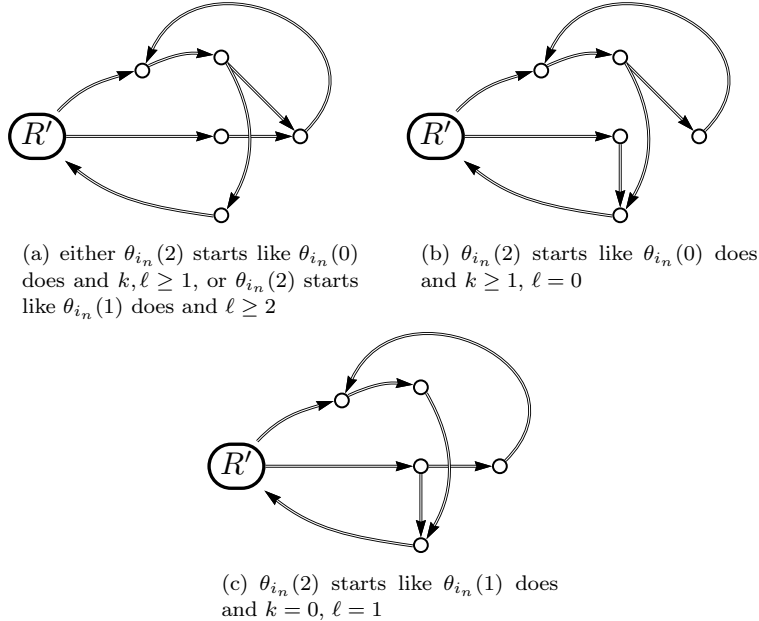

\begin{center}
\subfigure[either $\theta_{i_n}(2)$ starts like $\theta_{i_n}(0)$ does and $k, \ell \geq 1$, or $\theta_{i_n}(2)$ starts like $\theta_{i_n}(1)$ does and $\ell \geq 2$]{
\label{figure: evolution 1 of type 10}
\scalebox{0.7}{
\begin{VCPicture}{(0,0.5)(6,5)}
\ChgEdgeLabelScale{0.8}
\StateVar[R']{(0,2.5)}{R'}
\VSState{(2,3.75)}{HGG}
\VSState{(3.5,4)}{HG}
\VSState{(3.5,2.5)}{MG}
\VSState{(3.5,1)}{BG}
\VSState{(5,2.5)}{MD}
\EdgeLineDouble
\Edge{HG}{MD}
\Edge{MG}{MD}
\LArcL{HG}{BG}{}
\EdgeLineDouble
\ArcL{R'}{HGG}{}
\ArcL{HGG}{HG}{}
\Edge{R'}{MG}
\ArcL{BG}{R'}{}
\VCurveR[]{angleA=45,angleB=65,ncurv=1.5}{MD}{HGG}{}
\end{VCPicture}
}}
\qquad
\subfigure[$\theta_{i_n}(2)$ starts like $\theta_{i_n}(0)$ does and $k \geq 1$, $\ell =0$]{
\label{figure: evolution 2 of type 10}
\scalebox{0.7}{
\begin{VCPicture}{(0,0.5)(6,5)}
\ChgEdgeLabelScale{0.8}
\StateVar[R']{(0,2.5)}{R'}
\VSState{(2,3.75)}{HGG}
\VSState{(3.5,4)}{HG}
\VSState{(3.5,2.5)}{MG}
\VSState{(3.5,1)}{BG}
\VSState{(5,2.5)}{MD}
\EdgeLineDouble
\Edge{HG}{MD}
\Edge{MG}{BG}
\LArcL{HG}{BG}{}
\EdgeLineDouble
\ArcL{R'}{HGG}{}
\ArcL{HGG}{HG}{}
\Edge{R'}{MG}
\ArcL{BG}{R'}{}
\VCurveR[]{angleA=45,angleB=65,ncurv=1.5}{MD}{HGG}{}
\end{VCPicture}
}}
\qquad
\subfigure[$\theta_{i_n}(2)$ starts like $\theta_{i_n}(1)$ does and $k=0$, $\ell =1$]{
\label{figure: evolution 3 of type 10}
\scalebox{0.7}{
\begin{VCPicture}{(0,0.5)(6,5)}
\ChgEdgeLabelScale{0.8}
\StateVar[R']{(0,2.5)}{R'}
\VSState{(2,3.75)}{HGG}
\VSState{(3.5,4)}{HG}
\VSState{(3.5,2.5)}{MG}
\VSState{(3.5,1)}{BG}
\VSState{(5,2.5)}{MD}
\EdgeLineDouble
\Edge{MG}{MD}
\Edge{MG}{BG}
\LArcL{HG}{BG}{}
\EdgeLineDouble
\ArcL{R'}{HGG}{}
\ArcL{HGG}{HG}{}
\Edge{R'}{MG}
\ArcL{BG}{R'}{}
\VCurveR[]{angleA=45,angleB=65,ncurv=1.5}{MD}{HGG}{}
\end{VCPicture}
}}
\end{center}
\caption[Evolutions of a graph of type 10 with 3 circuits from $R$.]{Evolutions of a graph of type 10 with 3 circuits starting from $R$.}
\label{figure: evolutions of a type 10 with C}
\end{figure}

\begin{figure}[h!tbp]
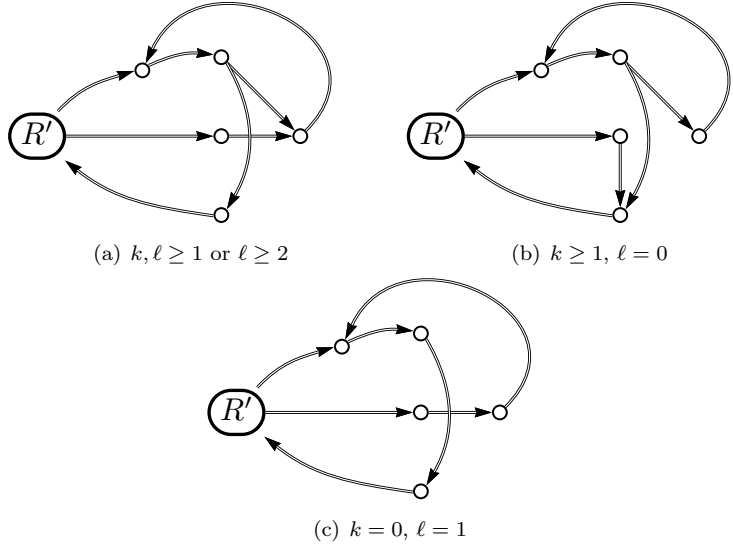

\begin{center}
\subfigure[$k,\ell \geq 1$ or $\ell \geq 2$]{
\label{figure: evolution 4 of type 10}
\scalebox{0.7}{
\begin{VCPicture}{(0,0.5)(6,5)}
\ChgEdgeLabelScale{0.8}
\StateVar[R']{(0,2.5)}{R'}
\VSState{(2,3.75)}{HGG}
\VSState{(3.5,4)}{HG}
\VSState{(3.5,2.5)}{MG}
\VSState{(3.5,1)}{BG}
\VSState{(5,2.5)}{MD}
\EdgeLineDouble
\Edge{HG}{MD}
\Edge{MG}{MD}
\LArcL{HG}{BG}{}
\EdgeLineDouble
\ArcL{R'}{HGG}{}
\ArcL{HGG}{HG}{}
\Edge{R'}{MG}
\ArcL{BG}{R'}{}
\VCurveR[]{angleA=45,angleB=65,ncurv=1.5}{MD}{HGG}{}
\end{VCPicture}
}}
\qquad
\subfigure[$k \geq 1$, $\ell =0$]{
\label{figure: evolution 5 of type 10}
\scalebox{0.7}{
\begin{VCPicture}{(0,0.5)(6,5)}
\ChgEdgeLabelScale{0.8}
\StateVar[R']{(0,2.5)}{R'}
\VSState{(2,3.75)}{HGG}
\VSState{(3.5,4)}{HG}
\VSState{(3.5,2.5)}{MG}
\VSState{(3.5,1)}{BG}
\VSState{(5,2.5)}{MD}
\EdgeLineDouble
\Edge{HG}{MD}
\Edge{MG}{BG}
\LArcL{HG}{BG}{}
\EdgeLineDouble
\ArcL{R'}{HGG}{}
\ArcL{HGG}{HG}{}
\Edge{R'}{MG}
\ArcL{BG}{R'}{}
\VCurveR[]{angleA=45,angleB=65,ncurv=1.5}{MD}{HGG}{}
\end{VCPicture}
}}
\qquad
\subfigure[$k=0$, $\ell =1$]{
\label{figure: evolution 6 of type 10}
\scalebox{0.7}{
\begin{VCPicture}{(0,0.5)(6,5)}
\ChgEdgeLabelScale{0.8}
\StateVar[R']{(0,2.5)}{R'}
\VSState{(2,3.75)}{HGG}
\VSState{(3.5,4)}{HG}
\VSState{(3.5,2.5)}{MG}
\VSState{(3.5,1)}{BG}
\VSState{(5,2.5)}{MD}
\EdgeLineDouble
\Edge{MG}{MD}
\LArcL{HG}{BG}{}
\ArcL{R'}{HGG}{}
\ArcL{HGG}{HG}{}
\Edge{R'}{MG}
\ArcL{BG}{R'}{}
\VCurveR[]{angleA=45,angleB=65,ncurv=1.5}{MD}{HGG}{}
\end{VCPicture}
}}
\end{center}
\caption[Evolutions of a graph of type 10 with 2 circuits from $R$.]{Evolutions of a graph of type 10 with 2 circuits starting from $R$.}
\label{figure: evolutions of a type 10 without C}
\end{figure}

\subsubsection{Modification of Component $C_4$}
\label{subsubsection: modification of component C_4}

Now we can modify the component $C_4$ of $\G$.

First let us modify the vertices. Lemmas~\ref{Lemma: graph of type 5 or 6} and~\ref{Lemma: graph of type 7 or 8} allow to merge the vertices 5 and 6 to one vertex $5/6$ and the vertices 7 and 8 to one vertex $7/8$. As already mentioned, the vertex 9 can also be deleted (thanks to Lemma~\ref{lemma: C4 type 9}). Finally, Lemma~\ref{Lemma: graph of type 10} describes the sequence of evolutions while $U_{i_n}$ corresponds to the vertex $R$ in a graph of type 10. Consequently, if a graph evolves to a graph of type 10 such that $U_{i_n} = R$, there is only one possible finite sequence of evolutions, the one given by Lemma~\ref{Lemma: graph of type 10}. Consequently, we can simply treat these evolutions by modifying the edges in $C_4$ as explained just below and we rename vertex $10$ by $10B$, meaning that the vertex $U_{i_n}$ always corresponds to the vertex $B$ in Figure~\ref{figure: type 10}.

Now let us modify the edges and/or their labels. All modifications are direct consequences of Fact~\ref{fact: graph of type 1}, Lemma~\ref{Lemma: graph of type 5 or 6}, Lemma~\ref{Lemma: graph of type 7 or 8}, Lemma~\ref{lemma: C4 type 9} and Lemma~\ref{Lemma: graph of type 10}: 
\begin{itemize}
	\item Fact~\ref{fact: graph of type 1} implies that we can consider only two morphisms to label the loop on vertex 1.

	\item Lemma~\ref{Lemma: graph of type 5 or 6} implies that the edges starting from $5/6$ are the same as those starting from $6$ in $\G$.

	\item By Lemma~\ref{Lemma: graph of type 10}, we can replace each morphism $\gamma_{i_n}$ labelling an edge coming to $10$ in $\G$ such that $U_{i_{n+1}}=R$ by the corresponding behaviour given in that lemma. For instance, in $\G$, the morphism $\gamma_{i_n} = [12^k0,2^{\ell}0,12^{k-1}0]$ labels an edge from 6 to 10. By Lemma~\ref{Lemma: graph of type 10}, this morphisms makes the graph of type 10 evolve to a graph of type 7 or 8 or 10 depending on $k$ and $\ell$. Consequently, we delete this morphism and add two morphisms: the morphism $\gamma_{i_n} \circ [1,0,2]$ from $5/6$ to $10B$ with $k = \ell$ (case \textit{i.}) and the morphism $\gamma_{i_n} \circ [1,0,2]$ from $5/6$ to $7/8$ with $\ell < k$.
	
	\item In Lemma~\ref{Lemma: graph of type 7 or 8} (so also in Lemma~\ref{lemma: C4 type 9}), as the behaviours depend on some lengths in Rauzy graphs, we simply consider the needed outgoing edges of the vertex $7/8$ to be able to follow all described behaviours and put some restrictions on the choices in Proposition~\ref{proposition: valid path in C4}. 
\end{itemize}

We then obtain the modified component $C_4$ represented in Figure~\ref{Figure: pre-graph C_4} with labels as given below; those are trivially compositions of morphisms of $\S$. We will also see that it is more convenient to modify a bit more that component.

{
\begin{figure}[h!tbp]
\centering
\begin{minipage}[c]{0.95\textwidth}
\centering
\scalebox{0.7}{
\begin{VCPicture}{(0,0.3)(8,7.5)}
\StateVar[5/6]{(4,1)}{56}
\StateVar[7/8]{(7,3.5)}{78}
\StateVar[1]{(4,6)}{1}
\StateVar[10B]{(1,3.5)}{10}
\LoopN{1}{}
\CLoopW{10}{}
\CLoopE{78}{}
\ForthBackOffset	
\Edge{78}{56}	\Edge{56}{78}
\Edge{1}{78}	\Edge{78}{1}
\RstEdgeOffset
\Edge{56}{10}
\Edge{10}{78}
\Edge{10}{1}
\Edge{56}{1}
\end{VCPicture}
}
\caption{First attempt to modify the component $C_4$ in $\G$.}
\label{Figure: pre-graph C_4}
\end{minipage}
\\
\vspace{1cm}
\begin{minipage}[c]{0.95\textwidth}
\centering
\begin{tabular}{|l|l|l|l|}
	\hline
	From 	&	 to 	& 	Labels				& 	Conditions	\\
	\hline
	1		&	1		&	$[0,10]$, $[01,1]$	&	\\
	\hhline{~---}
			&	7/8		&	$[x,y^kx,(y^{k-1}x)]$	&	$k \geq 2$	\\
	\hline
	5/6		&	1	&	$[x,yx], [yx,x]$	&	 	\\
	\hhline{~~--}
			&		&	$[1 2^k 0,2^k0]$, $[2^k 0,1 2^k0]$	& 	$k \geq 1$	\\
	\hhline{~~--}
			&		&	$[1 2^k 0,2^{k+1}0]$, $[2^{k+1} 0,1 2^k0]$	& 	$k \geq 0$	\\
	\hhline{~---}
			&	7/8	&	$[1,0^k2,(0^{k-1}2)]$	&	$k \geq 1$	\\
	\hhline{~~--}
			&		&	$[x,y^kx,(y^{k-1}x)]$	&	$k \geq 2$ 	\\
	\hhline{~~--}	
			&		&	$[2^{\ell}0,1 2^k 0,(1 2^{k-1}0)]$	&	$k > \ell \geq 0$	\\
	\hhline{~~--}	
			&		&	$[1 2^k 0,2^{\ell}0, (2^{\ell-1}0)]$	&	$\ell > k + 1 \geq 1$	\\
	\hhline{~---}
			&	10B	&	$[1,01,2]$	&	\\
	\hhline{~~--}
			&		&	$[2^k 0,12^k 0,1 2^{k-1}0]$		&	$k \geq 1$	\\
	\hhline{~~--}
			&		&	$[1 2^k 0,2^{k+1}0, 2^k 0]$		&	$k \geq 0$	\\
	\hline
	7/8 	&	1	&	$[01,1]$, $[1,01]$, $[x,y]$				&	 \\
	\hhline{~---}
			&	5/6	&	$[0 x,y,(0 y)], [x,0 y,(y)]$	&		\\
	\hhline{~---}
			&	7/8	&	$[0,1 0,(2 0)]$		&	\\
	\hline
	10B  	&	1	&	$[0 1^k 2,1^k2]$, $[1^k 2,0 1^k2]$	&	$k \geq 1$	\\
	\hhline{~~--}
			&		&	$[0 1^k 2,1^{k+1}2]$, $[1^{k+1} 2,0 1^k2]$	&	$k \geq 0$		\\
	\hhline{~---}
			&	7/8	&	$[0,2^k 1,2^{k-1}1]$		&	$k \geq 1$		\\
	\hhline{~~--}
			&		&	$[1^{\ell}2,0 1^k 2,(0 1^{k-1}2)]$	&	$k > \ell \geq 0$	\\	
	\hhline{~~--}
			&		&	$[0 1^k 2,1^{\ell}2,(1^{\ell-1}2)]$	&	$\ell > k + 1 \geq 1$	\\		
	\hhline{~---}
			&	10B	&	$[0,2 0,1]$		&	\\
	\hhline{~~--}
			&		&	$[1^k2,0 1^k 2,0 1^{k-1}2]$	&	$k \geq 1$	\\		
	\hhline{~~--}		
			&		&	$[0 1^k 2,1^{k+1}2,1^k2]$	&	$k \geq 0$	\\
	\hline	
\end{tabular}
\captionof{table}{Labels of edges in Figure~\ref{Figure: pre-graph C_4}}
\label{table: label C_4_1}
\end{minipage}
\end{figure}
}

The next lemma describes paths in Figure~\ref{Figure: pre-graph C_4} whose label is weakly primitive.

\begin{lemma}
\label{lemma: weakly primitive C_4}
An infinite path $p$ in Figure~\ref{Figure: pre-graph C_4} has a weakly primitive label if and only if one of the following conditions is satisfied:

\begin{enumerate}

	\item $p$ ultimately stays in vertex 1 and both morphisms $[0,10]$ and $[01,1]$ occur infinitely often in its label;

	\item $p$ ultimately stays in the subgraph $\{1,7/8\}$, goes through both vertices infinitely often and for all suffixes $p'$ of $p$ starting in vertex $7/8$, the label of $p'$ is not only composed of finite sub-sequences of morphisms in
	\[
		\left( [0,10]^* [0,1] [0,10]^* \{[0,1^k0] \mid k \geq 2\} \right)
		\cup \left( [0,10]^* [1,0] [01,1]^* \{ [1,0^k1] \mid k \geq 2\} \right);
	\]

	\item $p$ contains infinitely many occurrences of sub-paths $q$ that start in vertex $1$ and end in vertex $5/6$.

	\item $p$ ultimately stays in the subgraph $\{5/6,7/8,10B\}$ and does not ultimately correspond to one of the two following configurations:
	
	\begin{enumerate}
		\item 	the path ultimately stays in vertex $7/8$;

		\item 	\begin{itemize}
										
					\item the edge from $7/8$ to $5/6$ is labelled by $[1,02,2]$ or by $[01,2,02]$;
					
					\item the edge from $5/6$ to $7/8$ is labelled by $[1,02,2]$;
					
					\item the edge from $5/6$ to $10B$ is labelled by $[1,01,2]$;
					
					\item for all sub-paths $q$ uniquely composed of loops over $10B$, the label of $q$ contains only 							occurrences of morphisms in
					\[
						\left\{ {[0,20,1]}^{2n}, [02,12,2] \mid n \in \N \right\};
					\]
					
					\item for all finite sub-paths $q$ composed of loops over $10B$ and followed by the edge from $10B$ to $7/8$, the label of $q$ is in
					\[
						\left\{ {[0,20,1]}^{2n}, [02,12,2] \mid n \in \N \right\}^*  
							\{[2,012,02],[0,20,1][0,21,1] \} ;
					\]
				\end{itemize}	
		
		\item 	\begin{itemize}
					\item the paths does not go through the loop over vertex $7/8$;
					
					\item the loop over vertex $10B$ is labelled by $[12^k0,2^{k+1}0,2^k0]$ for some integer $k \geq 0$;
					
					\item the edge from $5/6$ to $7/8$ is labelled either by $[1,0^k2,0^{k-1}2]$ for some integer $k \geq 1$ or by $[12^k0,2^{\ell}0,2^{\ell-1}0]$ for some integers $k$ and $\ell$ such that $\ell > k+1 \geq 1$;
					
					\item the edge from $7/8$ to $5/6$ is labelled by $[1,02,2]$ or by $[2,01,1]$;
					
					\item the edge from $10B$ to $7/8$ is labelled by $[0,2^k1,2^{k-1}1]$ for some integer $k \geq 1$.
					
				\end{itemize}

	\end{enumerate}
\end{enumerate}
\end{lemma}

\begin{proof}
The proof of this lemma is not really hard, but rather long so it is given in 
Appendix~\ref{appendix: weakly primitive} page~\pageref{appendix: weakly primitive}.
\end{proof}

\label{valid path pas propre}
As in the previous cases, we would like to ensure that any valid path in Figure~\ref{Figure: pre-graph C_4} can be chosen in such a way that its label contains infinitely many right proper morphisms, which is currently not the case.
For instance, any path oscillating between $5/6$ and $7/8$ such that the edge from $5/6$ to $7/8$ is labelled by $[1,0^k2,0^{k-1}2]$ does not contain any right proper morphism but can be a suffix of a valid path (Lemma~\ref{Lemma: graph of type 5 or 6} and Lemma~\ref{Lemma: graph of type 7 or 8} ensure that the local condition of Proposition~\ref{prop: valid path} is satisfied). 
Thus, we have to modify Figure~\ref{Figure: pre-graph C_4} in such a way that a contraction of such a sequence of morphisms labels another path and contains infinitely many right proper morphisms. 
 
As proved in Proposition~\ref{proposition: valid path in C4}, this kind of problem can be solved by adding two edges in Figure~\ref{Figure: pre-graph C_4} labelled by the following additional morphisms. We then obtain the modified component as represented in Figure~\ref{Figure: graph C_4}.

{
\begin{figure}[h!tbp]
\centering
\begin{minipage}[c]{0.95\textwidth}
\centering
\scalebox{0.7}{
\begin{VCPicture}{(0,-0.7)(8,7)}
\StateVar[5/6]{(4,1)}{56}
\StateVar[7/8]{(7,3.5)}{78}
\StateVar[1]{(4,6)}{1}
\StateVar[10B]{(1,3.5)}{10}
\LoopN{1}{}
\CLoopW{10}{}
\CLoopE{78}{}
\LoopS{56}{}
\ForthBackOffset	
\Edge{78}{56}	\Edge{56}{78}
\Edge{10}{56}	\Edge{56}{10}
\Edge{1}{78}	\Edge{78}{1}
\RstEdgeOffset
\Edge{10}{78}
\Edge{10}{1}
\Edge{56}{1}
\end{VCPicture}
}
\caption{Graph corresponding to the component $C_4$ in $\G$.}
\label{Figure: graph C_4}
\end{minipage}
\\
\vspace{1cm}
\begin{minipage}[c]{0.95\textwidth}
\centering
\begin{tabular}{|l|l|l|l|}
	\hline
	From	&	To	&	Labels	& 	Conditions	\\
	\hline
	5/6		&	5/6	&	$[10^k2,0^{k-1}2,10^{k-1}2]$		&	 $k \geq 1$	\\
			&		&	$[10^{k-1}2,0^k2,10^k2]$			& 	\\
			&		&	$[0^k2,10^{k-1}2,0^{k-1}2]$			& 	\\
			&		&	$[0^{k-1}2,10^k2,0^k2]$				& 	\\	
	\hline
	10B		&	5/6	& 	$[02^k1,2^{k-1}1,02^{k-1}1]$		&  $k \geq 1$	\\
	 		&		&	$[02^{k-1}1,2^k1,02^k1]$			&  	\\
	 		&		&	$[2^k1,02^{k-1}1,2^{k-1}1]$			&  	\\
	 		&		&	$[2^{k-1}1,02^k1,2^k1]$				&  	\\
	\hline
\end{tabular}
\captionof{table}{Labels of the two additional edges in Figure~\ref{Figure: graph C_4}}
\label{table: label C_4_2}
\end{minipage}
\end{figure}
}

%
%
%
%
%
%
%
%

\begin{proposition}
\label{proposition: valid path in C4}
An infinite path $p$ in $\G$ labelled by $(\gamma_{i_n})_{n \geq N}$ is a valid suffix that always stays in component $C_4$ and that is such that $U_{i_N}$ is bispecial if and only if there is a contraction $(\alpha_n)_{n \geq N}$ of $(\gamma_{i_n})_{n \geq N}$ such that 

\begin{enumerate}
	\item 	\label{cond proper} there are infinitely many right proper morphisms in $(\alpha_n)_{n \geq N}$;

	\item 	$(\alpha_n)_{n \geq N}$ labels an infinite path $p$ in the graph represented in Figure~\ref{Figure: graph C_4} (whose labels are given in Table~\ref{table: label C_4_1} and Table~\ref{table: label C_4_2}) such that
	
	\begin{enumerate}[(A)]
		
	\item \label{cond 5/6} if for some integer $n \geq N$, $\alpha_n$ labels an edge to $5/6$, then $\alpha_{n+1}$ can belong to $\{[x,y^kx,(y^{k-1}x)] \mid \{x,y\} = \{0,1\}, k \geq 2 \}$ only if $|p_1| \geq |p_2|$;

	\item \label{cond 7/8} if for some integer $n \geq N$, $\alpha_n$ labels an edge to $7/8$ but not from $7/8$ (so it is equal to $[w_1,w_2 w_3^\mathfrak{k} w_4,w_2 w_3^{\mathfrak{k}-1} w_4]$ for some words $w_1$, $w_2$, $w_3$ and $w_4$ and for an integer $\mathfrak{k} \geq 1$ which corresponds to the greatest number of times that a circuit goes through the loop $v_2 u_2$ in Figure~\ref{figure: graph with 2 loops'}), if $h$ is the greatest integer such that $\alpha_{n+i} = [0,10,20]$ for all $i = 1, \dots, h$, then $h$ is finite and $\alpha_{n+h+1}$ can be in $\{[0,1],[1,0]\}$ if and only if $|u_1| + h (|u_1|+|v_1|) \geq |u_2| + (\mathfrak{k}-1) (|u_2|+|v_2|)$;

	\end{enumerate}
	
and such that one of the following conditions is satisfied

	\begin{enumerate}[(i)]

	\item \label{cond AP1} $p$ ultimately stays in vertex 1 and both morphisms $[0,10]$ and $[01,1]$ occur infinitely often in $(\alpha_n)_{n \geq N}$;

	\item $p$ ultimately stays in the subgraph $\{1,7/8\}$, goes through both vertices infinitely often and for all suffixes $p'$ of $p$ starting in vertex $7/8$, the label of $p'$ is not only composed of finite sub-sequences of morphisms in
	\[
		\left( [0,10]^* [0,1] [0,10]^* \{[0,1^k0] \mid k \geq 2\} \right)	
		\cup \left( [0,10]^* [1,0] [01,1]^* \{ [1,0^k1] \mid k \geq 2\} \right);
	\]

	\item $p$ contains infinitely many occurrences of sub-paths $q$ that start in vertex $1$ and end in vertex $5/6$.

	\item \label{cond AP4} $p$ ultimately stays in the subgraph $\{5/6,7/8,10B\}$ and does not ultimately correspond to one of the two following configurations:
	
	\begin{enumerate}[(a)]
		\item 	the path ultimately stays in vertex $7/8$;

		\item 	\begin{itemize}
					\item the the loop over $5/6$ is always labelled by in $[02,12,2]$ or $[102,2,12]$;					
														
					\item the edge from $5/6$ to $7/8$ is always labelled by $[1,02,2]$;
				
					\item the edge from $5/6$ to $10B$ is always labelled by $[1,01,2]$;

					\item the edge from $7/8$ to $5/6$ is always labelled by $[1,02,2]$ or by $[01,2,02]$;
									
					\item for all sub-paths $q$ uniquely composed of loops over $10B$, the label of $q$ contains 							only occurrences of morphisms in
					\[
						\left\{ {[0,20,1]}^{2n}, [02,12,2] \mid n \in \N \right\};
					\]
					
					\item for all finite sub-paths $q$ composed of loops over $10B$ and followed by the edge from $10B$ to $5/6$, the label of $q$ is in
					\[
						\left\{ {[0,20,1]}^{2n}, [02,12,2] \mid n \in \N \right\}^*
							[0,20,1] \left\{[21,01,1],[021,1,01] \right\} ;
					\]
				\end{itemize}	
		
		\item 	\begin{itemize}
					\item the paths does not go through the loop over the vertex $7/8$;
					
					\item the loop over the vertex $5/6$ is always labelled by $[0^k2,10^{k-1}2,0^{k-1}2]$ or by $[0^{k-1}2,10^k2,0^k2]$ for some integer $k \geq 1$;
					
					\item the loop over the vertex $10B$ is always labelled by $[12^k0,2^{k+1}0,2^k0]$ for some integer $k \geq 0$;
					
					\item the edge from $5/6$ to $7/8$ is always labelled either by $[1,0^k2,0^{k-1}2]$ for some integer $k \geq 1$ or by $[12^k0,2^{\ell}0,2^{\ell-1}0]$ for some integers $k$ and $\ell$ such that $\ell > k+1 \geq 1$;
					
					\item the edge from $7/8$ to $5/6$ is always labelled by $[1,02,2]$ or by $[2,01,1]$;

					\item the edge from $10B$ to $5/6$ is always labelled by $[2^k1,02^{k-1}1,2^{k-1}1]$ or by $[2^{k-1}1,02^k1,2^k1]$ for some integer $k \geq 1$;
					
					\item the edge from $10B$ to $7/8$ is always labelled by $[0,2^k1,2^{k-1}1]$ for some integer $k \geq 1$.
					
				\end{itemize}
	\end{enumerate}
\end{enumerate}
\end{enumerate}
\end{proposition}

\begin{proof}
Our aim is to describe valid suffix in $\G$ that stay in component $C_4$, accordingly to Proposition~\ref{prop: valid path}. The first step is to ensure that to any valid path $p$ in $\G$, there is a contraction $(\alpha_n)_{n \geq N}$ of its label that labels a path in Figure~\ref{Figure: graph C_4} and that contains infinitely many right proper morphisms. Up to know, the results in Section~\ref{subsection: component C_4} state that such a contraction labels a path in Figure~\ref{Figure: pre-graph C_4}, but some of them can contain only finitely many right proper morphisms. One can check that all of them label paths in Figure~\ref{Figure: non proper C_4} where


\begin{enumerate}
	\item 	the edge from $5/6$ to $10B$ is labelled by $[1,01,2]$;
	\item 	the edge from $5/6$ to $7/8$ is labelled by $[1,0^k2,0^{k-1}2]$;
	\item 	the edge from $7/8$ to $5/6$ is labelled by $[0x,y,0y]$ and $[x,0y,x]$;
	\item 	the edge from $10B$ to $7/8$ is labelled by $[0,2^k1,2^{k-1}1]$;
	\item 	the loop on $10B$ is labelled by $[0,20,1]$.
\end{enumerate}

\begin{figure}[h!tbp]
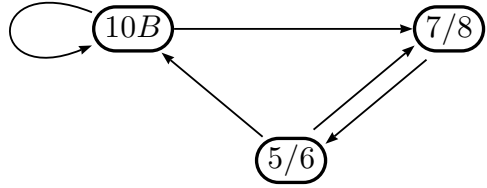

\centering
\scalebox{0.7}{
\begin{VCPicture}{(0,-0.7)(8,4)}
\StateVar[5/6]{(4,1)}{56}
\StateVar[7/8]{(7,3.5)}{78}
\StateVar[10B]{(1,3.5)}{10}
\LoopN{1}{}
\CLoopW{10}{}
\ForthBackOffset	
\Edge{78}{56}	\Edge{56}{78}
\RstEdgeOffset
\Edge{56}{10}
\Edge{10}{78}
\end{VCPicture}
}
\caption[Part of Figure~\ref{Figure: pre-graph C_4} with non-right proper morphisms.]{Part of Figure~\ref{Figure: pre-graph C_4} where there might be some valid labelled path with only non-right proper morphisms as labels.}
\label{Figure: non proper C_4}
\end{figure}

It is easily seen that labelled path in Figure~\ref{Figure: non proper C_4} that ultimately stay in vertex $10B$ are not valid. Moreover, the labels of the path of length 2 from $5/6$ to $5/6$ (passing through $7/8$) are right proper and equal to
\begin{align*}
	[1, 0^k 2, 0^{k-1} 2] \circ [01, 2, 02] 	&  	=  	[10^k 2, 0^{k-1} 2, 10^{k-1} 2]		\\
	[1, 0^k 2, 0^{k-1} 2] \circ [02, 1, 01] 	&  	=   [10^{k-1} 2, 0^k 2, 10^k 2]			\\
	[1, 0^k 2, 0^{k-1} 2] \circ [1, 02, 2]  	& 	= 	[0^k 2, 10^{k-1} 2, 0^{k-1} 2]		\\
	[1, 0^k 2, 0^{k-1} 2] \circ [2, 01, 1]  	& 	= 	[0^{k-1} 2, 10^k2, 0^k 2]				
\end{align*}
Similarly, the labels of the path of length 2 from $10B$ to $5/6$ (passing through $7/8$) are right proper and equal to
\begin{align*}
	[0, 2^k 1, 2^{k-1} 1] \circ [01, 2, 02] 	& 	= 	[02^k 1, 2^{k-1} 1, 02^{k-1} 1]		\\
	[0, 2^k 1, 2^{k-1} 1] \circ [02, 1, 01] 	& 	=   [02^{k-1} 1, 2^k 1, 02^k 1]			\\
	[0, 2^k 1, 2^{k-1} 1] \circ [1, 02, 2]  	&	= 	[2^k 1, 02^{k-1} 1, 2^{k-1} 1]		\\
	[0, 2^k 1, 2^{k-1} 1] \circ [2, 01, 1]  	&	= 	[2^{k-1} 1, 02^k1, 2^k 1]				
\end{align*}
To our aim, it suffices therefore to add two edges in Figure~\ref{Figure: pre-graph C_4}: one loop on $5/6$ labelled by the first four morphisms above and one edge from $10B$ to $5/6$ labelled by the last four morphisms above, which corresponds to Table~\ref{table: label C_4_2}. 

%

With that modification of Figure~\ref{Figure: pre-graph C_4}, the proper condition of Proposition~\ref{prop: valid path} is equivalent to the condition~\ref{cond proper} of the result. For the first condition of Proposition~\ref{prop: valid path} (the local one), it is a direct consequence of all previous lemmas and modifications of $C_4$:
\begin{enumerate}
	\item 	any finite path passing only through the vertex 1 is trivially valid;
	
	\item 	the condition~\ref{cond 5/6} of the result summarizes what is allowed according to Lemma~\ref{Lemma: graph of type 5 or 6} for vertex $5/6$;
	
	\item 	the condition~\ref{cond 7/8} summarizes what is allowed with vertex $7/8$ according to Lemma~\ref{Lemma: graph of type 7 or 8} and Lemma~\ref{lemma: C4 type 9};	

	\item 	the edges going to the vertex 10 in Figure~\ref{figure: graph of graphs} (page~\pageref{figure: graph of graphs}) have been modified according to Lemma~\ref{Lemma: graph of type 10}.
\end{enumerate}

It remains therefore to check the weakly primitive property. It is easily seen that conditions~\ref{cond AP1} to~\ref{cond AP4} are exactly those obtained in Lemma~\ref{lemma: weakly primitive C_4}, but modified according to the added edges.
\end{proof}

\subsection{Links between components}
\label{subsection: links between components}

Now that we know how the suffixes of valid paths in each component must behave, it remains to describe all links between them. To this aim, it suffices to look at the graph of graphs $\G$ (Figure~\ref{figure: graph of graphs} page~\pageref{figure: graph of graphs}) and, like we did in each component, to study the consequences of a given morphism $\gamma_{i_n}$ on the sequel in the directive word. For instance, in $\G$ there is an edge from 2 to 4 which is labelled by morphisms $\gamma_{i_n}$ depending on some exponents $k$ and $\ell$ and that are such that $U_{i_{n+1}}$ corresponds to the vertex $R$ in Figure~\ref{figure: type 4}. Then, Lemma~\ref{lemma: allowed path for type 4} (page~\pageref{lemma: allowed path for type 4}) states that, depending on $k$ and $\ell$, the graph will evolve to a graph of type 1, 4, 7 or 8 and 10 (with $U_{i_m} = B$) and it provides the morphism $\tau$ coding this evolution. Consequently, we add edges (if necessary) from 2 to $\{1,4B,7/8,10B\}$ labelled by $\gamma_{i_n} \circ \tau$. This yields to the \textit{modified graph of graphs} $\G'$ represented in Figure~\ref{figure: modified graph of graphs} (gray edges are simply those inner components). Labels of black edges are given below. In Table~\ref{table: morphismes a partir de 2}, Table~\ref{table: morphismes a partir de V_i} and Table~\ref{table: morphismes a partir de 4B}, we express in the column \enquote{Through} if the morphism is the result of a contraction like just explained. In the previous example, we would write $4R$ in the column \enquote{Through}, meaning that the morphisms is a composition of $\gamma_{i_n}$ and $\tau$ and that $\gamma_{i_n}$ codes an evolution to a Rauzy graph of type 4 such that $U_{i_n+1}$ corresponds to the vertex $R$ in Figure~\ref{figure: type 4}. 

Observe that, since black edges can only occur in a finite prefix of any valid path in $\G'$, we do not need to compute left conjugates of morphisms.

\begin{figure}[h!tbp]
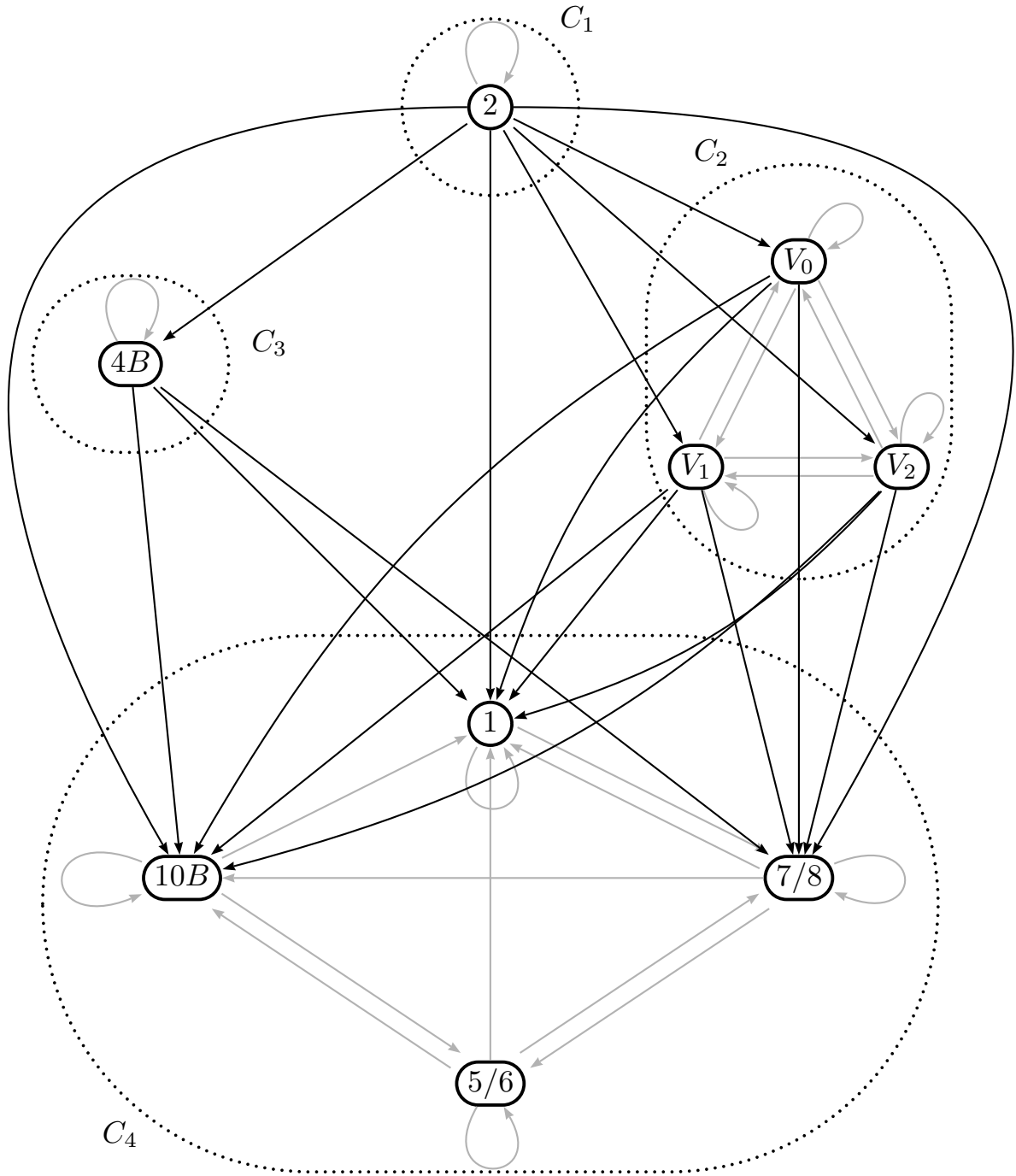

\centering
\scalebox{0.8}{
\begin{VCPicture}{(0,-2)(15,23)}
\ChgStateLineStyle{dotted}
\FixStateDiameter{3.5cm}
\StateVar[]{(7,20)}{C1}

\StateVar[\hspace{2cm}]{(0,15)}{C3}

\FixStateDiameter{6cm}
\VCPut[90]{(13,0)}{\StateVar[\hspace{4.5cm}]{(14.85,0)}{C2'}}

\FixStateDiameter{10.5cm}
\StateVar[\hspace{10cm}]{(7,4.5)}{C2}

\RstStateLineStyle
\MediumState

\ChgStateLineColor{white}
\StateVar[C_1]{(8.7,21.7)}{1}
\StateVar[C_2]{(11.3,19.1)}{2}
\StateVar[C_3]{(2.7,15.4)}{3}
\StateVar[C_4]{(-0.2,0)}{4}
\RstStateLineColor

\StateVar[2]{(7,20)}{2}
\StateVar[4B]{(0,15)}{4}
\StateVar[V_0]{(13,17)}{v0}
\StateVar[V_1]{(11,13)}{v1}
\StateVar[V_2]{(15,13)}{v2}
\StateVar[1]{(7,8)}{1}
\StateVar[10B]{(1,5)}{10}
\StateVar[7/8]{(13,5)}{78}
\StateVar[5/6]{(7,1)}{56}

\ChgEdgeLineColor{light-gray}
\LoopN{2}{}
\CLoopNE{v0}{}
\LoopN{4}{}
\LoopS{1}{}
\LoopS{56}{}
\CLoopR{-50}{v1}{}
\CLoopL{70}{v2}{}
\CLoopW{10}{}
\CLoopE{78}{}

\ForthBackOffset	
\Edge{v0}{v1}	\Edge{v1}{v0}
\Edge{v0}{v2}	\Edge{v2}{v0}
\Edge{v1}{v2}	\Edge{v2}{v1}
\Edge{78}{56}	\Edge{56}{78}
\Edge{10}{56}	\Edge{56}{10}
\Edge{1}{78}	\Edge{78}{1}
\RstEdgeOffset
\Edge{78}{10}
\Edge{10}{1}
\Edge{56}{1}

\RstEdgeLineColor
\Edge{2}{4}
\Edge{2}{v0}
\Edge{2}{v1}
\Edge{2}{v2}
\Edge{2}{1}
\VCurveL{angleA=0,angleB=60,ncurv=1.5}{2}{78}{}
\VCurveR{angleA=180,angleB=120,ncurv=1.3}{2}{10}{}

\ArcR{v0}{10}{}
\ArcR{v0}{1}{}
\Edge{v0}{78}

\Edge{v1}{10}
\Edge{v1}{1}
\Edge{v1}{78}

\ArcL{v2}{10}{}
\ArcL{v2}{1}{}
\Edge{v2}{78}

\Edge{4}{1}
\Edge{4}{10}
\Edge{4}{78}
\end{VCPicture}
}
\caption{Modified graph of graphs.}
\label{figure: modified graph of graphs}
\end{figure}

\begin{remark}
\label{remark: k et l differents}
It is important to notice that the exponents $k$ and $\ell$ in morphisms $\gamma_{i_n}$ do not always correspond to the integers $k$ and $\ell$ in Lemma~\ref{lemma: allowed path for type 4}, Lemma~\ref{Lemma: graph of type 7 or 8} and Lemma~\ref{Lemma: graph of type 10}. Indeed, if for instance we consider the evolution of a Rauzy graph of type 2 to a Rauzy graph of type 4 as represented in Figure~\ref{Figure: 2 to 4}. The morphism coding this evolution is either $[yz^kx,z^{\ell}x,yz^{k-1}x]$ or $[z^kx,yz^{\ell}x,z^{k-1}x]$ for some integers $k$ and $\ell$. But, the circuits $\theta_{i_n+1}(0)$ and $\theta_{i_n+1}(1)$ go respectively $k-1$ and $\ell-1$ times through the loop.
\end{remark}

\begin{figure}[h!tbp]
\centering
\subfigure[Before evolution]{
\scalebox{0.7}{
\begin{VCPicture}{(-1,-1.5)(6,6)}
\StateVar[U_{i_n}]{(2.5,3)}{B}
\VCurveR{angleA=30,angleB=150,ncurv=6}{B}{B}{\theta_{i_n}(x)}
\VCurveL{angleA=-30,angleB=-150,ncurv=12}{B}{B}{\theta_{i_n}(y)}
\VCurveL{angleA=-40,angleB=-140,ncurv=5}{B}{B}{\theta_{i_n}(z)}
\end{VCPicture}
}}
\subfigure[After evolution]{
\scalebox{0.7}{
\begin{VCPicture}{(-1,-1.5)(6,6)}
\StateVar[U_{i_n+1}]{(1.5,4)}{l1}
\VSState{(1.5,2.5)}{l2}
\VSState{(1.5,1)}{l3}
\VSState{(3.5,4)}{r1}
\VSState{(3.5,2.5)}{r2}
\VSState{(3.5,1)}{r3}
\VCurveR{angleA=30,angleB=150,ncurv=2}{r1}{l1}{}
\VCurveL{angleA=-40,angleB=-140,ncurv=5.5}{r2}{l2}{}
\VCurveL{angleA=-30,angleB=-150,ncurv=2}{r3}{l3}{}
\Edge{l1}{r2}
\Edge{l1}{r3}
\Edge{l2}{r3}
\Edge{l3}{r3}
\Edge{l3}{r1}
\end{VCPicture}
}}
\caption{Evolution of a graph of type 2 to a graph of type 4.}
\label{Figure: 2 to 4}
\end{figure}


\begin{table}[h!tbp]
\centering
\begin{tabular}{|l|c|l|l|}
\hline
	To		&	Through 	&	Labels 	& 	Conditions		\\
\hline
1			&		/	&	$[x,yzx]$, $[yzx,x]$, $[xy,zy]$	&	\\
			&			&	$[xy,zxy]$, $[zxy,xy]$			&	\\
\hhline{~---}
			&	$4R$		&	$[yz^kx,z^kx]$, $[z^kx,yz^kx]$	&	$k \geq 2$	\\
			&			&	$[yz^kx,z^{k-1}x]$, $[z^{k-1}x,yz^kx]$	&	\\
			&			&	$[yz^{k-1}x,z^kx]$, $[z^kx,yz^{k-1}x]$	&	\\
\hhline{~---}
			&	$10R$	&	$[(xy)^kz,y(xy)^kz]$, $[y(xy)^kz,(xy)^kz]$ 	&	$k \geq 1$	\\
\hhline{~~--}
			&			&	$[(xy)^kz,y(xy)^{k-1}z]$, $[y(xy)^{k-1}z,(xy)^kz]$ 		&	$k \geq 2$	\\
\hline
	$4B$		&		/	&	$[x,yx,yzx]$, $[y,yzx,yx]$		&	\\
\hhline{~---}
			&	$4R$		&	$[y^{k-1}z,xy^kz,xy^{k-1}z]$		&	$k \geq 2$	\\
			&			&	$[y^{k-1}z,xy^{k-1}z,xy^kz]$		&	\\
			&			&	$[xy^{k-1}z,y^kz,y^{k-1}z]$		&	\\
			&			&	$[xy^{k-1}z,y^{k-1}z,y^kz]$		&	\\
\hline
	$V_0$	&		/	&	$[0,120,20]$, $[0,10,210]$		&	\\
\hline
	$V_1$	&		/	&	$[01,1,201]$, $[021,1,21]$		&	\\
\hline
	$V_2$	&		/	&	$[02,102,2]$, $[012,12,2]$		&	\\
\hline
7/8			&		/	&	$[x,y^kzx,(y^{k-1}zx)]$				&	$k \geq 2$	\\
			&			&	$[x,zy^kx,(zy^{k-1}x)]$				&				\\
			&			&	$[x,{(yz)}^kx,({(yz)}^{k-1}x)]$		&				\\
			&			&	$[xy,z^kxy,(z^{k-1}xy)]$				&				\\
			&			&	$[xy,z^ky,(z^{k-1}y)]$				&				\\
\hhline{~~--}
			&			&	$[x,{(yz)}^kyx,({(yz)}^{k-1}yx)]$	&	$k \geq 1$	\\
\hhline{~---}
			&	$4R$		&	$[z^{\ell}x,yz^kx,yz^{k-1}x]$	&	$k-1 > \ell \geq 1$	\\
			&			&	$[yz^{\ell}x,z^kx,z^{k-1}x]$	&	\\
\hhline{~---}
			&	$10R$	&	$[y(xy)^{\ell}z,(xy)^kz,(xy)^{k-1}z]$		&	$k-1 > \ell \geq 0$	\\
\hhline{~~--}
			&			&	$[(xy)^kz,y(xy)^{\ell}z,y(xy)^{\ell-1}z]$	&	$\ell > k \geq 1$	\\
\hline
			&	/		&	$[xy,zxy,zy]$	&	\\
\hhline{~---}
			&	$4R$		&	$[z^kx,yz^kx,yz^{k-1}x]$	&	$k \geq 2$	\\
			&			&	$[yz^kx,z^kx,z^{k-1}x]$		&		\\
\hhline{~---}
			&	$10R$	&	$[y(xy)^{k-1}z,(xy)^kz,(xy)^{k-1}z]$	&	$k \geq 2$	\\
\hhline{~~--}
			&			&	$[(xy)^kz,y(xy)^kz,y(xy)^{k-1}z]$		&	\\
\hline
\end{tabular}
\caption{Morphisms labelling the black edges starting from 2 in $\G'$}
\label{table: morphismes a partir de 2}
\end{table}

\begin{table}[h!tbp]
\centering
\begin{tabular}{|l|c|l|l|}
\hline
	To		&	Through 	&	Labels 	& 	Conditions		\\
\hline
	1		&	/		&	$[x,iy]$, $[iy,x]$, $[xi,yi]$				&			\\
\hhline{~---}
			&	$10R$	&	$[xy^ki,y^ki]$, $[y^ki,xy^ki]$				&	$k \geq 1$	\\
\hhline{~~--}
			&			&	$[xy^ki,y^{k-1}i]$, $[y^{k-1}i,xy^ki]$		&	$k \geq 2$	\\
\hline
	7/8		&	/		&	$[i,xy^ki,xy^{k-1}i]$	&	$k \geq 1$	\\
\hhline{~~--}
			&			&	$[x,i^ky,i^{k-1}y]$		&	$k \geq 2$	\\
\hline{~---}
			&	$10R$	&	$[xy^{\ell}i,y^ki,y^{k-1}i]$	&	$k-1 > \ell \geq 0$		\\
\hhline{~~--}
			&			&	$[y^ki,xy^{\ell}i,xy^{\ell-1}i]$	&	$\ell > k \geq 1$		\\
\hline
	$10B$	&	/		&	$[x,ix,iy]$	&	\\
\hhline{~---}
			&	$10R$	&	$[xy^{k-1}i,y^ki,y^{k-1}i]$		&	$k \geq 2$	\\
\hhline{~~--}
			&			&	$[y^ki,xy^ki,xy^{k-1}i]$		&	$k \geq 1$	\\
\hline
\end{tabular}
\caption{Morphisms labelling the black edges starting from $V_i$ in $\G'$}
\label{table: morphismes a partir de V_i}
\end{table}

\begin{table}[h!tbp]
\centering
\begin{tabular}{|l|c|l|l|}
\hline
	To		&	Through 	&	Labels 	& 	Conditions		\\
\hline
1			&	$4R$		&	$[x^ky,0x^ky]$, $[0x^ky,x^ky]$			&	$k \geq 1$	\\
			&			&	$[x^{k-1}y,0x^ky]$, $[0x^ky,x^{k-1}y]$	&		\\
			&			&	$[x^ky,0x^{k-1}y]$, $[0x^{k-1}y,x^ky]$	&		\\
\hhline{~---}
			&	$10R$	&	$[0(x0)^ky,(x0)^ky]$, $[(x0)^ky,0(x0)^ky]$			&	$k \geq 1$	\\
			&			&	$[0(x0)^{k-1}y,(x0)^ky]$, $[(x0)^ky,0(x0)^{k-1}y]$	&	\\
\hline
7/8			&	/		&	$[0,x^ky0,x^{k-1}y0]$	&	$k \geq 1$	\\
\hhline{~---}
			&	$4R$		&	$[x^{\ell}y,0x^ky,0x^{k-1}y]$	&	$k-1 > \ell \geq 0$	\\
			&			&	$[0x^{\ell}y,x^ky,x^{k-1}y]$	&	\\
\hhline{~---}
			&	$10R$	&	$[(x0)^{\ell}y,0(x0)^ky,0(x0)^{k-1}y]$		&	$k > \ell \geq 0$	\\
\hhline{~~--}	
			&			&	$[0(x0)^ky,(x0)^{\ell}y,(x0)^{\ell-1}y]$	&	$\ell-1 > k \geq 0$	\\
\hline
$10B$		&	$4R$		&	$[x^ky,0x^ky,0x^{k-1}y]$	&	$k \geq 1$	\\
			&			&	$[0x^ky,x^ky,x^{k-1}y]$		&		\\
\hhline{~---}
			&	$10R$	&	$[(x0)^ky,0(x0)^ky,0(x0)^{k-1}y]$		&	$k \geq 1$	\\
			&			&	$[0(x0)^{k-1}y,(x0)^ky,(x0)^{k-1}y]$	&		\\
\hline
\end{tabular}
\caption{Morphisms labelling the black edges starting from $4B$ in $\G'$}
\label{table: morphismes a partir de 4B}
\end{table}

\subsection{Final Result}
\label{subsection: final result}

Now we can give an $\S$-adic characterization of minimal and aperiodic subshift with first difference of complexity bounded by 2. It suffices to put together all what we proved until now.

\begin{theorem}
\label{thm: 2n final}
Let $(X,T)$ be a subshift over an alphabet $A$ and let 
\[
	\S = \{G,D,M,E_{01},E_{12}\}
\]
be the set of 5 morphisms as defined on page~\pageref{section: $S$-adicity of subshifts with complexity $2n$}. Then, $(X,T)$ is minimal and satisfies $1 \leq p_X(n+1) - p_X(n) \leq 2$ for all $n \in \N$ if and only if $(X,T)$ is $\S$-adic such that there exists a contraction $(\Gamma_n)_{n \in \N}$ of its directive word 
and a sequence of morphisms $(\alpha)_{n \in \N}$ labelling an infinite path $p$ in the graph represented at Figure~\ref{figure: modified graph of graphs} and such that

\begin{enumerate}
	\item 	there are infinitely many right proper morphisms in $(\alpha_n)_{n \in \N}$ and for all integers $n \geq 0$, $\Gamma_n$ is either $\alpha_n$ or $\alpha_n^{(L)}$ and there are infinitely many right proper morphisms and infinitely many left proper morphisms in $(\Gamma_n)_{n \in \N}$;
	
	\item	if $p$ ultimately stays in component $C_1$ (resp. $C_2$, $C_3$), then the suffix of $p$ that stays in that component satisfies the conditions of Proposition~\ref{lemma: allowed path for type 2} (resp. Proposition~\ref{proposition: allowed path for type 3}, Proposition~\ref{proposition: valid path in C3});
	
	\item	\label{condition 78 thm final}	if $p$ ultimately stays in component $C_4$, then the suffix $p'$ of $p$ that stays in that component satisfies the conditions of Proposition~\ref{proposition: valid path in C4} with the following additional condition: if $p'$ starts in $7/8$, if the edge preceding $p'$ in $p$ is labelled by some morphism $\alpha_n = [w_1,w_2 w_3^\mathfrak{k} w_4,w_2 w_3^{\mathfrak{k}-1} w_4]$ such that $\mathfrak{k} \geq 1$ corresponds to the greatest number of times that a circuit goes through the loop $v_2 u_2$ in Figure~\ref{figure: graph with 2 loops'}), if $h$ is the greatest integer such that $\alpha_{n+i} = [0,10,20]$ for all $i = 1, \dots, h$, then $h$ is finite and $\alpha_{n+h+1}$ can be in $\{[0,1],[1,0]\}$ if and only if $|u_1| + h (|u_1|+|v_1|) \geq |u_2| + (\mathfrak{k}-1) (|u_2|+|v_2|)$;
\end{enumerate}
\end{theorem}


\begin{proof}
The last thing that remains to prove is that all morphisms $\Gamma_n$ belong to $\S^*$. To avoid long decompositions, we define the morphism $E_{0,2} = [2,1,0] = E_{0,1} E_{1,2} E_{0,1}$. We also define the following morphisms of $\S^*$. For $G_{x,y}$ (resp. $D_{x,y}$), read \enquote{add $y$ to the left (resp. right) of $x$}. For $M_{x,y}$, read \enquote{map $x$ to $y$}.
\begin{center}
\begin{tabular}{ll}
$G_{0,1} = [1 0,1,2] = G$								&	$D_{0,1} = [0 1,1,2] = D$ \\
$G_{0,2} = [2 0,1,2] = E_{1,2} G E_{1,2}$ 				&	$D_{0,2} = [0 2,1,2] = E_{1,2} D E_{1,2} $	\\
$G_{1,0} = [0,0 1,2] = E_{0,1} G E_{0,1}$				&	$D_{1,0} = [0,1 0,2] = E_{0,1} D E_{0,1}$ \\
$G_{1,2} = [0,2 1,2] = E_{0,1} G_{0,2} E_{0,1}$			&	$D_{1,2} = [0,1 2,2] = E_{0,1} D_{0,2} E_{0,1}$	\\
$G_{2,0} = [0,1,0 2] = E_{0,2} G_{0,2} E_{0,2}$			&	$D_{2,0} = [0,1,2 0] = E_{0,2} D_{0,2} E_{0,2}$	\\
$G_{2,1} = [0,1,1 2] = E_{1,2} G_{1,2} E_{1,2}$			&	$D_{2,1} = [0,1,2 1] = E_{1,2} D_{1,2} E_{1,2}$	\\
$M_{0,1} = [1,1,2]   = E_{0,2} M E_{0,2}$				&	$M_{1,0} = [0,0,2] 	 = E_{0,1} M_{0,1}$	\\
$M_{0,2} = [2,1,2]   = E_{0,1} E_{1,2} M E_{0,1}$		&	$M_{2,0} = [0,1,0]   = E_{0,2} M_{0,2}$	\\
$M_{1,2} = [0,2,2]   = E_{1,2} M$						&	$M_{2,1} = [0,1,1]   = M$	
\end{tabular}
\end{center}

Now we can compute all decompositions. The morphisms labelling inner edges in components $C_1$ and $C_2$ are easily seen to belong to $\S^*$. Hence, we can restrict ourselves to those labelling edges in components $C3$ and $C_4$ and labelling the black edges in Figure~\ref{figure: modified graph of graphs}. Furthermore, the only morphisms that really need some computation are those that depend on some exponents $k$ or $\ell$. Observe that the conditions on $k$ and $\ell$ given below are sometimes not restrictive enough; a given type morphism might label different edges, but the conditions on $k$ and $\ell$ can be different for these edges. The conditions we consider here are taken to be the most general. 

When having a look at the concerned morphisms in Proposition~\ref{proposition: valid path in C3}, Table~\ref{table: label C_4_1}, Table~\ref{table: label C_4_2}, Table~\ref{table: morphismes a partir de 2}, Table~\ref{table: morphismes a partir de V_i} and Table~\ref{table: morphismes a partir de 4B}, we see that all of them can be written as one of the following morphisms
\[
\begin{array}{ll}
	{[x,y^k x,y^{k-1}x]}, 				& 	k \geq 2				\\
	{[x,y^kz,y^{k-1}z]}, 				&	k \geq 1				\\	 
	{[x,zy^kx,zy^{k-1}x]},				&	k \geq 1				\\
	{[y^\ell x,zy^kx,zy^{k-1}x]},		&	k > \ell \geq 0		\\		
	{[y^k x,zy^kx,zy^{k-1}x]},			&	k \geq 1				\\
	{[x y^\ell z,y^k z,y^{k-1}z]},		&	k \geq \ell \geq 0, k+ \ell \geq 1	\\
\end{array}
\]

possibly up to exchange of some images, or up to applying a few last morphisms from the list given above. For instance, 
the morphism $[y(xy)^{\ell}z,(xy)^kz,(xy)^{k-1}z]$ in Table~\ref{table: morphismes a partir de 2} can be written
\[
	D_{x,y} E_{x,y} [x y^\ell z,y^k z,y^{k-1}z].
\]
Thus, all we have to do is to compute the decompositions of the previous morphisms, as well as the decomposition of their respective left conjugates when they exist (except for $[x,yz^kx,yz^{k-1}x]$ that only occurs as label of black edges, where it is useless to consider left conjugates). We obtain
\begin{eqnarray*}
	[x,y^k x,y^{k-1}x] 					& = &  	M_{z,x} G_{z,y}^{k-1} D_{y,z} [x,y,z]	\\
	{[x,x y^k ,xy^{k-1}]} 				& = &  	M_{z,x} D_{z,y}^{k-1} G_{y,z} [x,y,z]	\\
	{[x,y^k z,y^{k-1}z]} 				& = &	G_{z,y}^{k-1} D_{y,z} [x,y,z] [x,y,z]	\\
	{[x,zy^k x,zy^{k-1}x]}				& = &	D_{z,y}^{k-1} G_{y,z} D_{y,x} D_{z,x} [x,y,z]	\\
	{[y^\ell x,zy^kx,zy^{k-1}x]}			& = &	D_{z,y}^{k-\ell-1} G_{x,y}^\ell G_{y,z} D_{y,x} D_{z,x} [x,y,z]	\\
	{[xy^\ell ,xzy^k,xzy^{k-1}]}			& = &	G_{z,x} D_{z,y}^{k-1} D_{x,y}^\ell G_{y,z} [x,y,z]	\\
	{[y^k x,zy^kx,zy^{k-1}x]}			& = &	G_{x,y}^{k-1} D_{y,x} G_{x,z} D_{z,y} [y,z,x]		\\
	{[xy^k,xzy^k,xzy^{k-1}]}				& = &	G_{z,x} D_{z,y}^{k-1} D_{x,y}^k G_{y,z} [x,y,z]	\\
	{[x y^\ell z,y^k z,y^{k-1}z]}		& = &	D_{x,y}^\ell D_{x,z} G_{z,y}^{k-1} D_{y,z} [x,y,z]	\\
	{[zx y^\ell,zy^k,zy^{k-1}]}			& = &	G_{x,y}^\ell G_{x,z} D_{z,y}^{k-1} G_{y,z} [x,y,z]		\\
\end{eqnarray*}
which concludes the proof.
\end{proof}

\begin{remark}
Up to know, Theorem~\ref{thm: 2n final} is stated in such a way that we have to keep track of the Rauzy graphs to be able to compute the length of some paths $p_1$ and $p_2$ in Figure~\ref{figure: graph with no loop'} and $u_1$, $u_2$, $v_1$ and $v_2$ in Figure~\ref{figure: graph with 2 loops'}. This can actually be avoided by expressing these lengths only using the first morphisms of the directive word. Indeed, if $p$ is a valid path in $\G'$ (labelled by $(\alpha_n)_{n \in \N}$) and if $p'$ is a prefix of $p$ ending in $5/6$ (resp. in $7/8$), then the lengths $|p_1|$ and $|p_2|$ (resp. $|u_1|$, $|u_2|$, $|v_1|$ and $|v_2|$) can be expressed only with the label $(\alpha_n)_{0 \leq n \leq N}$ of $p'$. 
The interested reader can find the calculation in Section~\ref{appendix: computation of length}.
\end{remark}

To obtain the exact complexities $p(n) = 2n$ or $p(n) = 2n+1$, it suffices to impose respectively that $p(1) = 2$ or $p(1)=3$ and that for all $n \geq 1$, $p(n+1)-p(n) = 2$. This can be expressed by the fact the Rauzy graphs cannot be of type 1 (because these graphs are such that $p(n+1)-p(n)=1$). Consequently, one just has to impose that the path $p$ of the theorem does no go through vertex 1 except in some particular cases depending on the lengths $|u_1|$, $|u_2|$, $|v_1|$, $|v_2|$, $|p_1|$ and $|p_2|$.

\begin{corollary}
\label{2n exact}
A subshift $(X,T)$ is minimal and has complexity $p(n)=2n$ (resp. $p(n)=2n+1$) for all $n \geq 1$ if and only if it is an $\S$-adic subshift satisfying Theorem~\ref{thm: 2n final} and the following additional conditions:
\begin{enumerate}
	\item the path $p$ of Theorem~\ref{thm: 2n final} starts in vertex 1 or starts in vertex 2 and then $\alpha_0$ labels the edge to vertex $7/8$;
	
	\item in Condition~\ref{cond 7/8} of Proposition~\ref{proposition: valid path in C4} and in Condition~\ref{condition 78 thm final} of Theorem~\ref{thm: 2n final}, the inequality
	\[
		|u_1| + h (|u_1|+|v_1|) \geq |u_2| + (\mathfrak{k}-1) (|u_2|+|v_2|)
	\]
	is replaced by 
	\[
		|u_1| + h (|u_1|+|v_1|) = |u_2| + (\mathfrak{k}-1) (|u_2|+|v_2|)
	\]
	and in that case, $\alpha_{n+h+2}$ must label the edge from $1$ to $7/8$;

	\item in Condition~\ref{cond 7/8} of Proposition~\ref{proposition: valid path in C4}, the inequality $|p_1| \geq |p_2|$ is replaced by $|p_1| = |p_2|$.
\end{enumerate}
\end{corollary}


\section{Acknowledgement}
The author would like to thank Fabien Durand and Gwena\"el Richomme for their help and useful comments about this work. This paper was mainly written during the author's PhD thesis~\cite{Leroy_these} at the University of Picardie. It is also supported by the internal research project F1R-MTH-PUL-12RDO2 of
the University of Luxembourg.

\bibliographystyle{alpha}
\bibliography{bibliography}

\appendix

\newpage

\section{Evolution of Rauzy graphs such that $1 \leq p(n+1)-p(n) \leq 2$}	
\label{appendix: evolutions}

Here we present all possible evolutions of Rauzy graphs. They all correspond to edges in the graph of graphs (Figure~\ref{figure: graph of graphs}). We also give the corresponding morphisms $\gamma_{i_n}$. Note that some of them depend on the starting vertices $U_{i_n}$ and $U_{i_{n+1}}$ so we also give them accordingly to the graphs represented in Figure~\ref{figure: Rauzy graphs with at least 1 bispecial vertex}.

\subsection{Evolution of a Rauzy graph of type 1}

A graph of type 1 is represented in Figure~\ref{appendix type 1}. The possible evolutions are represented in Figure~\ref{appendix 1 to}.

{
\begin{figure}[h!tbp]
\centering
\begin{VCPicture}{(0,0.5)(4,3)}
\VSState{(2,2)}{B}
\EdgeLineDouble
\VCurveR[]{angleA=30,angleB=150,ncurv=20}{B}{B}{}
\VCurveL[]{angleA=-30,angleB=-150,ncurv=20}{B}{B}{}
\end{VCPicture}
\caption{Graph of type 1}
\label{appendix type 1}
\end{figure}

\begin{figure}[h!tbp]
\centering
\subfigure[To a graph of type 1]{
\begin{VCPicture}{(-0.5,-0.5)(4.5,4)}
\EdgeLineDouble
\VSState{(1,3)}{L1}
\VSState{(1,1)}{L2}
\VSState{(3,3)}{R1}
\VSState{(3,1)}{R2}
\VCurveR[]{angleA=30,angleB=150,ncurv=2}{R1}{L1}{}
\VCurveL[]{angleA=-30,angleB=-150,ncurv=2}{R2}{L2}{}
\Edge{L1}{R1}
\Edge{L1}{R2}
\Edge{L2}{R1}
\end{VCPicture}
}
\qquad
\subfigure[To a graph of type 7 or 8]{
\label{1 to 78}
\begin{VCPicture}{(-0.5,-0.5)(4.5,4)}
\EdgeLineDouble
\VSState{(1,3)}{L1}
\VSState{(1,1)}{L2}
\VSState{(3,3)}{R1}
\VSState{(3,1)}{R2}
\VCurveR[]{angleA=30,angleB=150,ncurv=2}{R1}{L1}{}
\VCurveL[]{angleA=-30,angleB=-150,ncurv=2}{R2}{L2}{}
\Edge{L1}{R1}
\Edge{L1}{R2}
\Edge{L2}{R1}
\Edge{L2}{R2}
\end{VCPicture}
}
\caption{Possible evolutions for a graph of type 1}
\label{appendix 1 to}
\end{figure}
}

\begin{table}[h!tbp]
\centering
\begin{tabular}{|l|c|l|l|}
\hline
From 1 to 	& 	$(U_{i_n},U_{i_{n+1}})$		&	Morphisms 	& 	Conditions		\\
\hline
1		&	$(B,B)$	&	$[x,yx]$, $[yx,x]$	&	\\
\hline
7 or 8	&	$(B,\star)$	&	$[x,y^kx,(y^{k-1}x)]$	&	$k \geq 2$ \\
\hline
\end{tabular}
\caption{List of morphisms coding the evolutions of a graph of type 1}
\label{List of morphisms coding the evolutions of a graph of type 1}
\end{table}

\newpage

\subsection{Evolution of a Rauzy graph of type 2}

A graph of type 2 is represented in Figure~\ref{appendix type 2}. The possible evolutions are represented in Figure~\ref{appendix 2 to 1}, Figure~\ref{appendix 2 to 234} and Figure~\ref{appendix 2 to 7810}.

\begin{figure}[h!tbp]
\centering
\begin{VCPicture}{(0,1)(4,4)}
\VSState{(2.5,3)}{B}
\EdgeLineDouble
\VCurveR{angleA=30,angleB=150,ncurv=20}{B}{B}{}
\VCurveL{angleA=-25,angleB=-155,ncurv=35}{B}{B}{}
\VCurveL{angleA=-40,angleB=-140,ncurv=20}{B}{B}{}
\end{VCPicture}
\caption{Graph of type 2}
\label{appendix type 2}
\end{figure}

\begin{figure}[h!tbp]
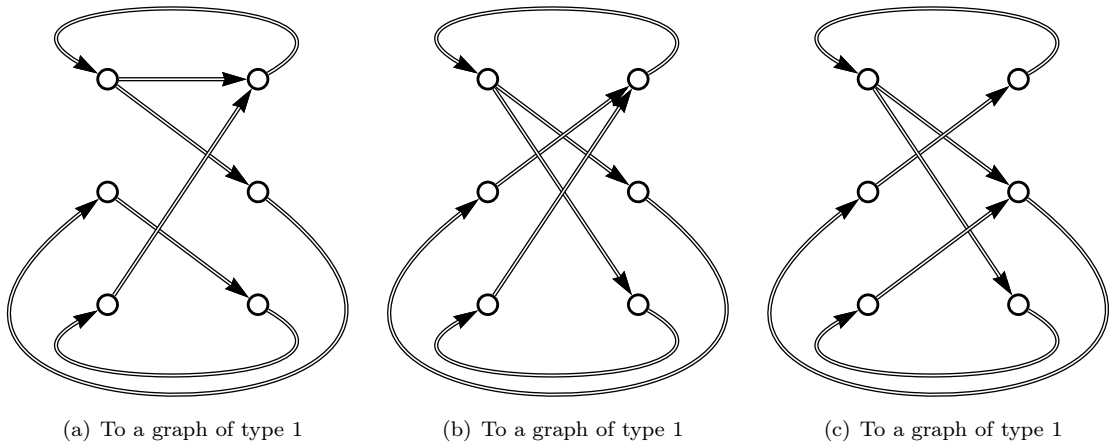

\centering
\subfigure[To a graph of type 1]{
\begin{VCPicture}{(0,-0.5)(4,5.5)}
\EdgeLineDouble
\VSState{(1,4)}{L1}
\VSState{(1,2.5)}{L2}
\VSState{(1,1)}{L3}
\VSState{(3,4)}{R1}
\VSState{(3,2.5)}{R2}
\VSState{(3,1)}{R3}
\VCurveR[]{angleA=30,angleB=150,ncurv=2}{R1}{L1}{}
\VCurveL[]{angleA=-38,angleB=-142,ncurv=5}{R2}{L2}{}
\VCurveL[]{angleA=-30,angleB=-150,ncurv=2}{R3}{L3}{}
\Edge{L1}{R1}
\Edge{L1}{R2}
\Edge{L2}{R3}
\Edge{L3}{R1}
\end{VCPicture}
}
\qquad
\subfigure[To a graph of type 1]{
\begin{VCPicture}{(0,-0.5)(4,5.5)}
\EdgeLineDouble
\VSState{(1,4)}{L1}
\VSState{(1,2.5)}{L2}
\VSState{(1,1)}{L3}
\VSState{(3,4)}{R1}
\VSState{(3,2.5)}{R2}
\VSState{(3,1)}{R3}
\VCurveR[]{angleA=30,angleB=150,ncurv=2}{R1}{L1}{}
\VCurveL[]{angleA=-38,angleB=-142,ncurv=5}{R2}{L2}{}
\VCurveL[]{angleA=-30,angleB=-150,ncurv=2}{R3}{L3}{}
\Edge{L1}{R2}
\Edge{L2}{R1}
\Edge{L1}{R3}
\Edge{L3}{R1}
\end{VCPicture}
}
\qquad
\subfigure[To a graph of type 1]{
\begin{VCPicture}{(0,-0.5)(4,5.5)}
\EdgeLineDouble
\VSState{(1,4)}{L1}
\VSState{(1,2.5)}{L2}
\VSState{(1,1)}{L3}
\VSState{(3,4)}{R1}
\VSState{(3,2.5)}{R2}
\VSState{(3,1)}{R3}
\VCurveR[]{angleA=30,angleB=150,ncurv=2}{R1}{L1}{}
\VCurveL[]{angleA=-38,angleB=-142,ncurv=5}{R2}{L2}{}
\VCurveL[]{angleA=-30,angleB=-150,ncurv=2}{R3}{L3}{}
\Edge{L1}{R2}
\Edge{L2}{R1}
\Edge{L1}{R3}
\Edge{L3}{R2}
\end{VCPicture}
}
\caption{Evolutions from 2 to 1}
\label{appendix 2 to 1}
\end{figure}

\begin{figure}[h!tbp]
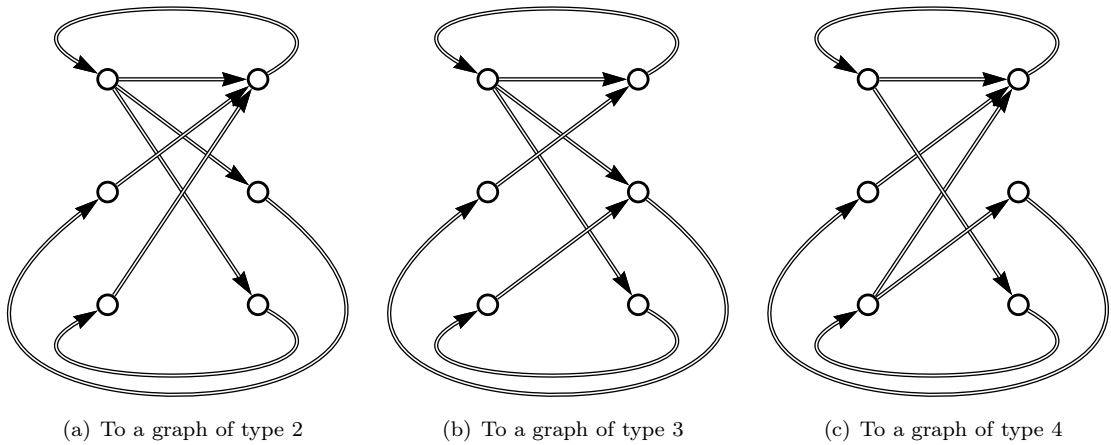

\centering
\subfigure[To a graph of type 2]{
\begin{VCPicture}{(0,-0.5)(4,5.5)}
\EdgeLineDouble
\VSState{(1,4)}{L1}
\VSState{(1,2.5)}{L2}
\VSState{(1,1)}{L3}
\VSState{(3,4)}{R1}
\VSState{(3,2.5)}{R2}
\VSState{(3,1)}{R3}
\VCurveR[]{angleA=30,angleB=150,ncurv=2}{R1}{L1}{}
\VCurveL[]{angleA=-38,angleB=-142,ncurv=5}{R2}{L2}{}
\VCurveL[]{angleA=-30,angleB=-150,ncurv=2}{R3}{L3}{}
\Edge{L1}{R1}
\Edge{L1}{R2}
\Edge{L1}{R3}
\Edge{L2}{R1}
\Edge{L3}{R1}
\end{VCPicture}
}
\qquad
\subfigure[To a graph of type 3]{
\begin{VCPicture}{(0,-0.5)(4,5.5)}
\EdgeLineDouble
\VSState{(1,4)}{L1}
\VSState{(1,2.5)}{L2}
\VSState{(1,1)}{L3}
\VSState{(3,4)}{R1}
\VSState{(3,2.5)}{R2}
\VSState{(3,1)}{R3}
\VCurveR[]{angleA=30,angleB=150,ncurv=2}{R1}{L1}{}
\VCurveL[]{angleA=-38,angleB=-142,ncurv=5}{R2}{L2}{}
\VCurveL[]{angleA=-30,angleB=-150,ncurv=2}{R3}{L3}{}
\Edge{L1}{R1}
\Edge{L1}{R2}
\Edge{L1}{R3}
\Edge{L2}{R1}
\Edge{L3}{R2}
\end{VCPicture}
}
\qquad
\subfigure[To a graph of type 4]{
\label{subfig: 2 to 4}
\begin{VCPicture}{(0,-0.5)(4,5.5)}
\EdgeLineDouble
\VSState{(1,4)}{L1}
\VSState{(1,2.5)}{L2}
\VSState{(1,1)}{L3}
\VSState{(3,4)}{R1}
\VSState{(3,2.5)}{R2}
\VSState{(3,1)}{R3}
\VCurveR[]{angleA=30,angleB=150,ncurv=2}{R1}{L1}{}
\VCurveL[]{angleA=-38,angleB=-142,ncurv=5}{R2}{L2}{}
\VCurveL[]{angleA=-30,angleB=-150,ncurv=2}{R3}{L3}{}
\Edge{L1}{R1}
\Edge{L2}{R1}
\Edge{L3}{R1}
\Edge{L1}{R3}
\Edge{L3}{R2}
\end{VCPicture}
}
\caption{Evolutions from 2 to $\{1,2,3,4\}$}
\label{appendix 2 to 234}
\end{figure}

\begin{figure}[h!tbp]
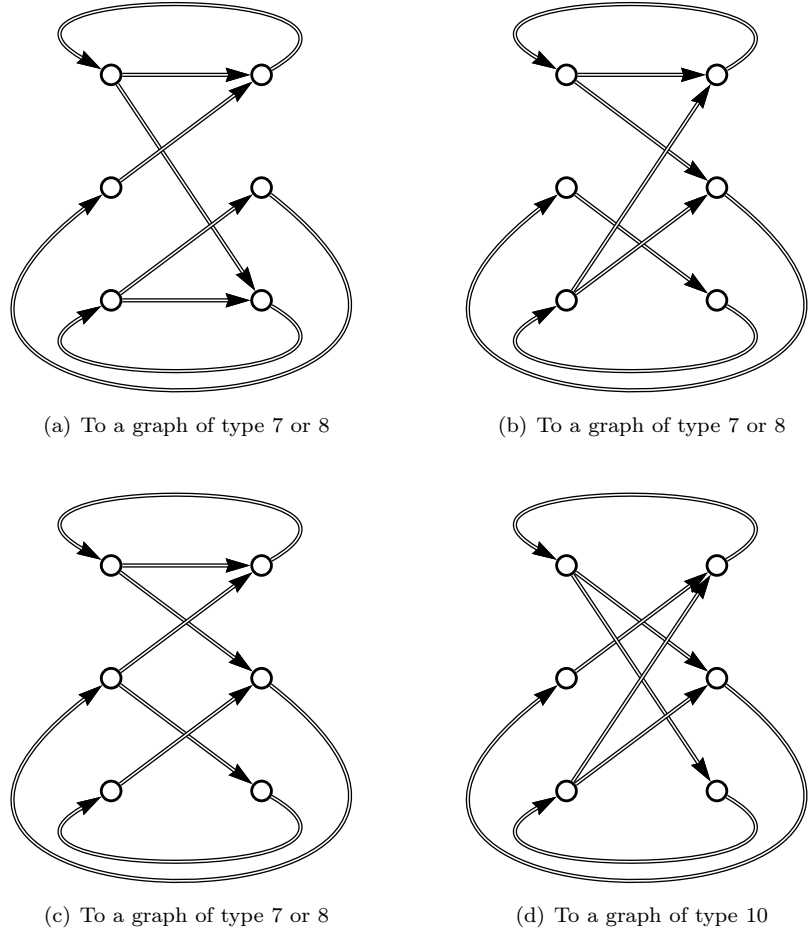

\centering
\subfigure[To a graph of type 7 or 8]{
\label{2 to 78 1}
\begin{VCPicture}{(-0.5,-0.5)(4.5,5.5)}
\EdgeLineDouble
\VSState{(1,4)}{L1}
\VSState{(1,2.5)}{L2}
\VSState{(1,1)}{L3}
\VSState{(3,4)}{R1}
\VSState{(3,2.5)}{R2}
\VSState{(3,1)}{R3}
\VCurveR[]{angleA=30,angleB=150,ncurv=2}{R1}{L1}{}
\VCurveL[]{angleA=-38,angleB=-142,ncurv=5}{R2}{L2}{}
\VCurveL[]{angleA=-30,angleB=-150,ncurv=2}{R3}{L3}{}
\Edge{L1}{R1}
\Edge{L1}{R3}
\Edge{L3}{R3}
\Edge{L3}{R2}
\Edge{L2}{R1}
\end{VCPicture}
}
\qquad
\subfigure[To a graph of type 7 or 8]{
\label{2 to 78 2}
\begin{VCPicture}{(-0.5,-0.5)(4.5,5.5)}
\EdgeLineDouble
\VSState{(1,4)}{L1}
\VSState{(1,2.5)}{L2}
\VSState{(1,1)}{L3}
\VSState{(3,4)}{R1}
\VSState{(3,2.5)}{R2}
\VSState{(3,1)}{R3}
\VCurveR[]{angleA=30,angleB=150,ncurv=2}{R1}{L1}{}
\VCurveL[]{angleA=-38,angleB=-142,ncurv=5}{R2}{L2}{}
\VCurveL[]{angleA=-30,angleB=-150,ncurv=2}{R3}{L3}{}
\Edge{L1}{R1}
\Edge{L1}{R2}
\Edge{L2}{R3}
\Edge{L3}{R2}
\Edge{L3}{R1}
\end{VCPicture}
}
\qquad
\subfigure[To a graph of type 7 or 8]{
\label{2 to 78 3}
\begin{VCPicture}{(-0.5,-0.5)(4.5,5.5)}
\EdgeLineDouble
\VSState{(1,4)}{L1}
\VSState{(1,2.5)}{L2}
\VSState{(1,1)}{L3}
\VSState{(3,4)}{R1}
\VSState{(3,2.5)}{R2}
\VSState{(3,1)}{R3}
\VCurveR[]{angleA=30,angleB=150,ncurv=2}{R1}{L1}{}
\VCurveL[]{angleA=-38,angleB=-142,ncurv=5}{R2}{L2}{}
\VCurveL[]{angleA=-30,angleB=-150,ncurv=2}{R3}{L3}{}
\Edge{L1}{R1}
\Edge{L1}{R2}
\Edge{L2}{R3}
\Edge{L3}{R2}
\Edge{L2}{R1}
\end{VCPicture}
}
\qquad
\subfigure[To a graph of type 10]{
\label{2 to 10}
\begin{VCPicture}{(-0.5,-0.5)(4.5,5.5)}
\EdgeLineDouble
\VSState{(1,4)}{L1}
\VSState{(1,2.5)}{L2}
\VSState{(1,1)}{L3}
\VSState{(3,4)}{R1}
\VSState{(3,2.5)}{R2}
\VSState{(3,1)}{R3}
\VCurveR[]{angleA=30,angleB=150,ncurv=2}{R1}{L1}{}
\VCurveL[]{angleA=-38,angleB=-142,ncurv=5}{R2}{L2}{}
\VCurveL[]{angleA=-30,angleB=-150,ncurv=2}{R3}{L3}{}
\Edge{L1}{R2}
\Edge{L2}{R1}
\Edge{L1}{R3}
\Edge{L3}{R2}
\Edge{L3}{R1}
\end{VCPicture}
}
\caption{Evolutions from 2 to $\{7,8,10\}$}
\label{appendix 2 to 7810}
\end{figure}

\begin{table}[h!tbp]
\centering
\begin{tabular}{|l|c|l|l|}
\hline
From 2 to 	& 	$(U_{i_n},U_{i_{n+1}})$		&	Morphisms 	& 	Conditions		\\
\hline
1	&	$(B,B)$		&	$[x,yzx]$, $[yzx,x]$, $[xy,zy]$	&	\\
	&				&	$[xy,zxy]$, $[zxy,xy]$			&	\\
\hline
2	&	$(B,B)$		&	$[0,1 0,2 0]$, $[0 1,1,2 1]$	&	\\
	&				&	$[0 2,1 2,2]$					&	\\
\hline
3	&	$(B,B)$		&	$[0,1 0,2 1 0]$, $[0,1 2 0,2 0]$	&	\\
	&				&	$[0 1,1,2 0 1]$, $[0 2 1,1,2 1]$	&	\\
	&				&	$[0 2,1 0 2,2]$, $[0 1 2,1 2,2]$	&	\\
\hline
4	&	$(B,R)$		&	$[xy^kz,y^{\ell}z,(xy^{k-1}z)]$	&	$k \geq \ell \geq 1$,	\\
	&				&	$[y^kz,xy^{\ell}z,(y^{k-1}z)]$	&	$k + \ell \geq 3$		\\
\hhline{~---}
	&	$(B,B)$		&	$[x,yx,yzx]$, $[x,yzx,yx]$		&	\\
\hline
7 or 8	&	$(B,\star)$	&	$[x,y^kzx,(y^{k-1}zx)]$				&	$k \geq 2$	\\
		&				&	$[x,zy^kx,(zy^{k-1}x)]$				&				\\
		&				&	$[x,{(yz)}^kx,({(yz)}^{k-1}x)]$		&				\\
		&				&	$[x,{(yz)}^kyx,({(yz)}^{k-1}yx)]$	&				\\
		&				&	$[xy,z^kxy,(z^{k-1}xy)]$			&				\\
		&				&	$[xy,z^ky,(z^{k-1}y)]$				&				\\
\hline
10	&	$(B,R)$		&	$[(xy)^kz,y(xy)^{\ell}z]$					&	$k \geq 1, \ell \geq 0, k+\ell \geq 2$	\\
\hhline{~~--}			
	&				&	$[(xy)^kz,y(xy)^{\ell}z,(xy)^{k-1}z]$		&	$k \geq 2, k > \ell \geq 0$			\\
\hhline{~~--}
	&				&	$[(xy)^kz,y(xy)^{\ell}z,y(xy)^{\ell-1}z]$	&	$\ell \geq k \geq 1$						\\
\hhline{~---}
	&	$(B,B)$		&	$[xy,zxy,zy]$	&	\\
\hline
\end{tabular}
\caption{List of morphisms coding the evolutions of a graph of type 2}
\label{List of morphisms coding the evolutions of a graph of type 2}
\end{table}

\newpage

\subsection{Evolution of a Rauzy graph of type 3}

A graph of type 3 is represented in Figure~\ref{appendix type 3}. The possible evolutions are represented in Figure~\ref{appendix 3 to}.

\begin{figure}[h!tbp]
\centering
\begin{VCPicture}{(0,0.5)(4,3.5)}
\VSState{(2,2)}{B}
\VSState{(2,0.5)}{L}
\EdgeLineDouble
\VCurveL[]{angleA=45,angleB=135,ncurv=20}{B}{B}{}
\VCurveR[]{angleA=-30,angleB=-10,ncurv=2}{B}{L}{}
\VCurveR[]{angleA=-150,angleB=190,ncurv=2}{B}{L}{}
\EdgeL{L}{B}{}
\end{VCPicture}
\caption{Graph of type 3}
\label{appendix type 3}
\end{figure}

\begin{figure}[h!tbp]
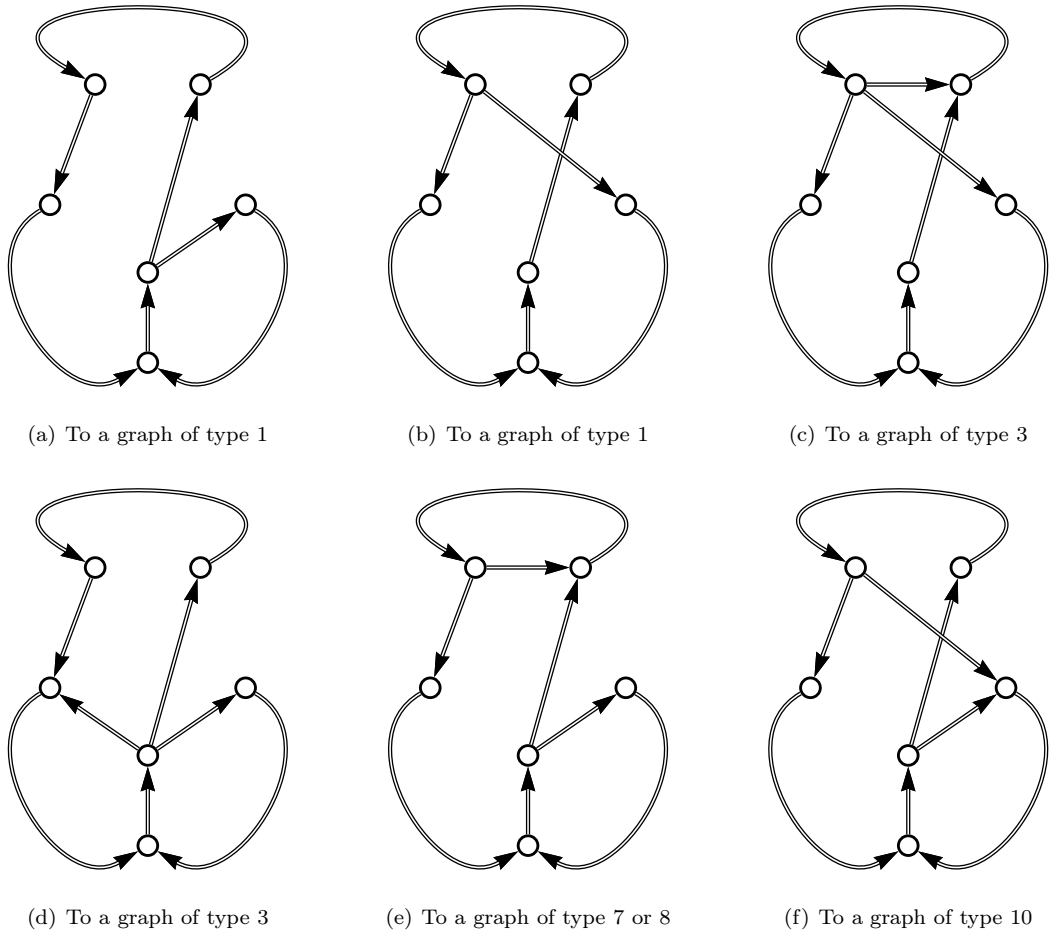

\centering
\subfigure[To a graph of type 1]{
\begin{VCPicture}{(0,-0.5)(4,5.4)}
\EdgeLineDouble
\VSState{(1.3,4)}{L1}
\VSState{(2.7,4)}{R1}
\VSState{(0.7,2.4)}{L2}
\VSState{(3.3,2.4)}{R2}
\VSState{(2,1.5)}{C1}
\VSState{(2,0.3)}{C2}
\VCurveR[]{angleA=30,angleB=150,ncurv=3}{R1}{L1}{}
\VCurveL[]{angleA=-30,angleB=-30,ncurv=1}{R2}{C2}{}
\VCurveL[]{angleA=-150,angleB=-150,ncurv=1}{L2}{C2}{}
\Edge{C2}{C1}
\Edge{C1}{R1}
\Edge{L1}{L2}
\Edge{C1}{R2}
\end{VCPicture}
}
\qquad
\subfigure[To a graph of type 1]{
\begin{VCPicture}{(0,-0.5)(4,5.4)}
\EdgeLineDouble
\VSState{(1.3,4)}{L1}
\VSState{(2.7,4)}{R1}
\VSState{(0.7,2.4)}{L2}
\VSState{(3.3,2.4)}{R2}
\VSState{(2,1.5)}{C1}
\VSState{(2,0.3)}{C2}
\VCurveR[]{angleA=30,angleB=150,ncurv=3}{R1}{L1}{}
\VCurveL[]{angleA=-30,angleB=-30,ncurv=1}{R2}{C2}{}
\VCurveL[]{angleA=-150,angleB=-150,ncurv=1}{L2}{C2}{}
\Edge{C2}{C1}
\Edge{C1}{R1}
\Edge{L1}{L2}
\Edge{L1}{R2}
\end{VCPicture}
}
\qquad
\subfigure[To a graph of type 3]{
\begin{VCPicture}{(0,-0.5)(4,5.4)}
\EdgeLineDouble
\VSState{(1.3,4)}{L1}
\VSState{(2.7,4)}{R1}
\VSState{(0.7,2.4)}{L2}
\VSState{(3.3,2.4)}{R2}
\VSState{(2,1.5)}{C1}
\VSState{(2,0.3)}{C2}
\VCurveR[]{angleA=30,angleB=150,ncurv=3}{R1}{L1}{}
\VCurveL[]{angleA=-30,angleB=-30,ncurv=1}{R2}{C2}{}
\VCurveL[]{angleA=-150,angleB=-150,ncurv=1}{L2}{C2}{}
\Edge{C2}{C1}
\Edge{C1}{R1}
\Edge{L1}{L2}
\Edge{L1}{R2}
\Edge{L1}{R1}
\end{VCPicture}
}
\qquad
\subfigure[To a graph of type 3]{
\begin{VCPicture}{(0,-0.5)(4,5.4)}
\EdgeLineDouble
\VSState{(1.3,4)}{L1}
\VSState{(2.7,4)}{R1}
\VSState{(0.7,2.4)}{L2}
\VSState{(3.3,2.4)}{R2}
\VSState{(2,1.5)}{C1}
\VSState{(2,0.3)}{C2}
\VCurveR[]{angleA=30,angleB=150,ncurv=3}{R1}{L1}{}
\VCurveL[]{angleA=-30,angleB=-30,ncurv=1}{R2}{C2}{}
\VCurveL[]{angleA=-150,angleB=-150,ncurv=1}{L2}{C2}{}
\Edge{C2}{C1}
\Edge{C1}{R1}
\Edge{L1}{L2}
\Edge{C1}{L2}
\Edge{C1}{R2}
\end{VCPicture}
}
\qquad
\subfigure[To a graph of type 7 or 8]{
\label{3 to 78}
\begin{VCPicture}{(0,-0.5)(4,5.4)}
\EdgeLineDouble
\VSState{(1.3,4)}{L1}
\VSState{(2.7,4)}{R1}
\VSState{(0.7,2.4)}{L2}
\VSState{(3.3,2.4)}{R2}
\VSState{(2,1.5)}{C1}
\VSState{(2,0.3)}{C2}
\VCurveR[]{angleA=30,angleB=150,ncurv=3}{R1}{L1}{}
\VCurveL[]{angleA=-30,angleB=-30,ncurv=1}{R2}{C2}{}
\VCurveL[]{angleA=-150,angleB=-150,ncurv=1}{L2}{C2}{}
\Edge{C2}{C1}
\Edge{C1}{R1}
\Edge{L1}{L2}
\Edge{L1}{R1}
\Edge{C1}{R2}
\end{VCPicture}
}
\qquad
\subfigure[To a graph of type 10]{
\label{3 to 10}
\begin{VCPicture}{(0,-0.5)(4,5.4)}
\EdgeLineDouble
\VSState{(1.3,4)}{L1}
\VSState{(2.7,4)}{R1}
\VSState{(0.7,2.4)}{L2}
\VSState{(3.3,2.4)}{R2}
\VSState{(2,1.5)}{C1}
\VSState{(2,0.3)}{C2}
\VCurveR[]{angleA=30,angleB=150,ncurv=3}{R1}{L1}{}
\VCurveL[]{angleA=-30,angleB=-30,ncurv=1}{R2}{C2}{}
\VCurveL[]{angleA=-150,angleB=-150,ncurv=1}{L2}{C2}{}
\Edge{C2}{C1}
\Edge{C1}{R1}
\Edge{L1}{L2}
\Edge{C1}{R2}
\Edge{L1}{R2}
\end{VCPicture}
}
\caption{Possible evolutions of a graph of type 3}
\label{appendix 3 to}
\end{figure}

\begin{table}[h!tbp]
\centering
\begin{tabular}{|l|c|l|l|}
\hline
From 3 to 	& 	$(U_{i_n},U_{i_{n+1}})$		&	Morphisms 	& 	Conditions		\\
\hline
1	&	$(B,B)$ 	&	$[xy,zy]$, $[xy,z]$, $[x,yz]$		&	\\
\hline
3	&	$(B,B)$		&	$[0,1 0,2 0]$, $[0,1 0,2]$, $[0,1,2 0]$	&	\\
	&				&	$[0 1,1,2 1]$, $[0 1,1,2]$, $[0,1,2 1]$	&	\\
	&				&	$[0 2,1 2,2]$, $[0 2,1,2]$, $[0,1 2,2]$	&	\\
\hline
7 or 8	&	$(B,\star)$	&	$[x,yz^kx,(yz^{k-1}x)]$	&	$ k \geq 1$	\\
\hhline{~~--}
		&				&	$[x,y^kz,(y^{k-1}z)]$	&	$k \geq 2$	\\
\hline
10	&	$(B,B)$		&	$[x,yx,yz]$					&		\\
\hhline{~---}
	&	$(B,R)$		&	$[x^ky,zx^{\ell}y]$					&	$k \geq 1$, $\ell \geq 0$, $k + \ell \geq 2$	\\
\hhline{~~--}
	&				&	$[x^ky,zx^{\ell}y,(x^{k-1}y)]$		&	$k \geq 2$, $k > \ell \geq 0$					\\
\hhline{~~--}
	&				&	$[x^ky,zx^{\ell}y,(zx^{\ell-1}y)]$	&	$\ell \geq k \geq 1$							\\	
\hline
\end{tabular}
\caption{List of morphisms coding the evolutions of a graph of type 3}
\label{List of morphisms coding the evolutions of a graph of type 3}
\end{table}

\newpage

\subsection{Evolution of a Rauzy graph of type 4}

A graph of type 3 is represented in Figure~\ref{appendix type 4}. The possible evolutions are represented in Figure~\ref{appendix 4 to}.

\begin{figure}[h!tbp]
\centering
\begin{VCPicture}{(0,0)(5,3)}
\VSState{(0,1.5)}{R}
\VSState{(4,1.5)}{B}
\EdgeLineDouble
\LArcL{R}{B}{}
\EdgeL{R}{B}{}
\VCurveR[]{angleA=40,angleB=-40,ncurv=20}{B}{B}{}
\LArcL{B}{R}{}
\end{VCPicture}
\caption{Graph of type 4}
\label{appendix type 4}
\end{figure}

\begin{figure}[h!tbp]
\centering
\subfigure[To a graph of type 1]{
\begin{VCPicture}{(0,0)(6,4)}
\EdgeLineDouble
\VSState{(0,2)}{LL}
\VSState{(3.5,3.5)}{L1}
\VSState{(3.5,2)}{L2}
\VSState{(3.5,0.5)}{L3}
\VSState{(5,2.5)}{R1}
\VSState{(5,1.5)}{R2}
\VCurveR[]{angleA=20,angleB=-20,ncurv=2}{R1}{R2}{}
\ArcL{LL}{L1}{}
\Edge{LL}{L2}{}
\ArcL{L3}{LL}{}
\Edge{R2}{L3}
\Edge{L1}{R1}
\Edge{L2}{L3}
\end{VCPicture}
}
\qquad
\subfigure[To a graph of type 1]{
\begin{VCPicture}{(0,0)(6,4)}
\EdgeLineDouble
\VSState{(0,2)}{LL}
\VSState{(3.5,3.5)}{L1}
\VSState{(3.5,2)}{L2}
\VSState{(3.5,0.5)}{L3}
\VSState{(5,2.5)}{R1}
\VSState{(5,1.5)}{R2}
\VCurveR[]{angleA=20,angleB=-20,ncurv=2}{R1}{R2}{}
\ArcL{LL}{L1}{}
\Edge{LL}{L2}{}
\ArcL{L3}{LL}{}
\Edge{R2}{L3}
\Edge{L1}{R1}
\Edge{L2}{R1}
\end{VCPicture}
}
\qquad
\subfigure[To a graph of type 4]{
\label{4 to 4}
\begin{VCPicture}{(0,0)(6,4)}
\EdgeLineDouble
\VSState{(0,2)}{LL}
\VSState{(3.5,3.5)}{L1}
\VSState{(3.5,2)}{L2}
\VSState{(3.5,0.5)}{L3}
\VSState{(5,2.5)}{R1}
\VSState{(5,1.5)}{R2}
\VCurveR[]{angleA=20,angleB=-20,ncurv=2}{R1}{R2}{}
\ArcL{LL}{L1}{}
\Edge{LL}{L2}{}
\ArcL{L3}{LL}{}
\Edge{R2}{L3}
\Edge{L1}{R1}
\Edge{R2}{R1}
\Edge{L2}{R1}
\end{VCPicture}
}
\qquad
\subfigure[To a graph of type 4]{
\label{4 to 4 bis}
\begin{VCPicture}{(0,0)(6,4)}
\EdgeLineDouble
\VSState{(0,2)}{LL}
\VSState{(3.5,3.5)}{L1}
\VSState{(3.5,2)}{L2}
\VSState{(3.5,0.5)}{L3}
\VSState{(5,2.5)}{R1}
\VSState{(5,1.5)}{R2}
\VCurveR[]{angleA=20,angleB=-20,ncurv=2}{R1}{R2}{}
\ArcL{LL}{L1}{}
\Edge{LL}{L2}{}
\ArcL{L3}{LL}{}
\Edge{R2}{L3}
\Edge{L1}{R1}
\Edge{L2}{L3}
\LArcL{L1}{L3}{}
\end{VCPicture}
}
\qquad
\subfigure[To a graph of type 7 or 8]{
\label{4 to 78}
\begin{VCPicture}{(0,0)(6,4)}
\EdgeLineDouble
\VSState{(0,2)}{LL}
\VSState{(3.5,3.5)}{L1}
\VSState{(3.5,2)}{L2}
\VSState{(3.5,0.5)}{L3}
\VSState{(5,2.5)}{R1}
\VSState{(5,1.5)}{R2}
\VCurveR[]{angleA=20,angleB=-20,ncurv=2}{R1}{R2}{}
\ArcL{LL}{L1}{}
\Edge{LL}{L2}{}
\ArcL{L3}{LL}{}
\Edge{R2}{L3}
\Edge{L1}{R1}
\Edge{R2}{R1}
\Edge{L2}{L3}
\end{VCPicture}
}
\qquad
\subfigure[To a graph of type 10]{
\label{4 to 10}
\begin{VCPicture}{(0,0)(6,4)}
\EdgeLineDouble
\VSState{(0,2)}{LL}
\VSState{(3.5,3.5)}{L1}
\VSState{(3.5,2)}{L2}
\VSState{(3.5,0.5)}{L3}
\VSState{(5,2.5)}{R1}
\VSState{(5,1.5)}{R2}
\VCurveR[]{angleA=20,angleB=-20,ncurv=2}{R1}{R2}{}
\ArcL{LL}{L1}{}
\Edge{LL}{L2}{}
\ArcL{L3}{LL}{}
\Edge{R2}{L3}
\Edge{L1}{R1}
\Edge{L2}{L3}
\Edge{L2}{R1}
\end{VCPicture}
}
\caption{Possible evolutions of a graph of type 4}
\label{appendix 4 to}
\end{figure}

\begin{table}[h!tbp]
\centering
\begin{tabular}{|l|c|l|l|}
\hline
From 4 to 	& 	$(U_{i_n},U_{i_{n+1}})$		&	Morphisms 	& 	Conditions		\\
\hline
1	&	$(R,B)$		&	$[x,y]$		&	$\card(C_n) = 2$ 	\\
\hline
4	&	$(R,R)$		&	$[0,1,(2)]$					&	\\
\hhline{~---}
	&	$(B,B)$		&	$[0,1 0,2 0]$, $[0,2 0,1 0]$	&	\\
\hhline{~---}
	&	$(R,B)$		&	$[1,0,2]$, $[1,2,0]$		&	\\
\hhline{~---}
	&	$(B,R)$		&	$[0 x^ky,x^{\ell}y,(0 x^{k-1}y)]$		&	$k \geq 1$, $k \geq \ell \geq 0$	\\
	&				&	$[x^ky,0 x^{\ell}y,(x^{k-1}y)]$			&										\\
\hline
7 or 8	&	$(R,\star)$		&	$[1,0,(2)]$		&	\\
\hhline{~---}
		&	$(B,\star)$		&	$[0,x^ky 0,(x^{k-1}y 0)]$		&	$k \geq 1$	\\
\hline
10	&	$(R,B)$		&	$[1,0,2]$	&	\\
\hhline{~---}
	&	$(B,R)$		&	$[0(x 0)^ky,(x 0)^{\ell}y]$						&	$k,\ell \geq 0$, $k+\ell \geq 1$	\\
\hhline{~~--}
	&				&	$[0(x 0)^ky,(x 0)^{\ell}y,0(x 0)^{k-1}y]$		&	$k \geq 1$, $k \geq \ell \geq 0$	\\
\hhline{~~--}
	&				&	$[0(x 0)^ky,(x 0)^{\ell}y,(x 0)^{\ell-1}y]$		&	$\ell > k \geq 0$					\\
\hline	
\end{tabular}
\caption{List of morphisms coding the evolutions of a graph of type 4}
\label{List of morphisms coding the evolutions of a graph of type 4}
\end{table}

\newpage

\subsection{Evolution of a Rauzy graph of type 5}

A graph of type 3 is represented in Figure~\ref{appendix type 5}. The possible evolutions are represented in Figure~\ref{appendix 5 to}.

\begin{figure}[h!tbp]
\centering
\begin{VCPicture}{(0,0.5)(6,3.5)}
\VSState{(0,2)}{B}
\VSState{(6,2)}{R}
\VSState{(4.5,2)}{L}
\EdgeLineDouble
\ArcL{B}{L}{}
\ArcR{B}{L}{}
\Edge{L}{R}
\VCurveL{angleA=-130,angleB=-50,ncurv=0.5}{R}{B}{}
\VCurveR{angleA=130,angleB=50,ncurv=0.5}{R}{B}{}		
\end{VCPicture}
\caption{Graph of type 5}
\label{appendix type 5}
\end{figure}

\begin{figure}[h!tbp]
\centering
\subfigure[To a graph of type 1]{
\begin{VCPicture}{(0,0)(6,4)}
\EdgeLineDouble
\VSState{(0,3.5)}{L1}
\VSState{(1.5,2.5)}{L2}
\VSState{(1.5,1.5)}{L3}
\VSState{(0,0.5)}{L4}
\VSState{(4,2)}{C}
\VSState{(5.5,2)}{R}
\Edge{C}{R}
\ArcR{R}{L1}{}
\ArcL{R}{L4}{}
\ArcL{L2}{C}{}
\ArcR{L3}{C}{}
\Edge{L1}{L2}
\Edge{L4}{L3}
\end{VCPicture}
}
\qquad
\subfigure[To a graph of type 10]{
\label{5 to 10}
\begin{VCPicture}{(0,0)(6,4)}
\EdgeLineDouble
\VSState{(0,3.5)}{L1}
\VSState{(1.5,2.5)}{L2}
\VSState{(1.5,1.5)}{L3}
\VSState{(0,0.5)}{L4}
\VSState{(4,2)}{C}
\VSState{(5.5,2)}{R}
\Edge{C}{R}
\ArcR{R}{L1}{}
\ArcL{R}{L4}{}
\ArcL{L2}{C}{}
\ArcR{L3}{C}{}
\Edge{L1}{L2}
\Edge{L4}{L3}
\Edge{L1}{L3}
\end{VCPicture}
}
\caption{Possible evolutions of a graph of type 5}
\label{appendix 5 to}
\end{figure}

\begin{table}[h!tbp]
\centering
\begin{tabular}{|l|c|l|l|}
\hline
From 5 to 	& 	$(U_{i_n},U_{i_{n+1}})$		&	Morphisms 	& 	Conditions		\\
\hline
1	&	$(R,B)$		&	$[x,y]$		&	$\card(C_n) = 2$ 	\\
\hline
10	&	$(R,B)$		&	$[1,2,0]$							&										\\
\hhline{~---}
	&	$(B,R)$		&	$[1,0 1,2]$							&										\\
\hhline{~~--}
	&				&	$[0^kc,1,(0^{k-1}2)]$				&	$k \geq 1$							\\	
\hhline{~~--}
	&				&	$[2^k 0,1 2^{\ell}0]$				&	$k,\ell \geq 0$, $k + \ell \geq 1$	\\
\hhline{~~--}
	&				&	$[2^k 0,1 2^{\ell}0,2^{k-1}0]$		&	$k \geq \ell \geq 0$, $k \geq 1$	\\
\hhline{~~--}
	&				&	$[2^k 0,1 2^{\ell}0,1 2^{\ell-1}0]$	&	$\ell > k \geq 0$					\\
\hline	
\end{tabular}
\caption{List of morphisms coding the evolutions of a graph of type 5}
\label{List of morphisms coding the evolutions of a graph of type 5}
\end{table}

\newpage

\subsection{Evolution of a Rauzy graph of type 6}

A graph of type 3 is represented in Figure~\ref{appendix type 6}. The possible evolutions are represented in Figure~\ref{appendix 6 to}.

\begin{figure}[h!tbp]
\centering
\begin{VCPicture}{(0,0.5)(6,3.5)}
\VSState{(0,2)}{L}
\VSState{(6,2)}{R}
\EdgeLineDouble
\ArcL{L}{R}{}
\ArcR{L}{R}{}
\VCurveL{angleA=-130,angleB=-50,ncurv=0.5}{R}{L}{}
\VCurveR{angleA=130,angleB=50,ncurv=0.5}{R}{L}{}		
\end{VCPicture}
\caption{Graph of type 6}
\label{appendix type 6}
\end{figure}

\begin{figure}[h!tbp]
\centering
\subfigure[To a graph of type 1]{
\begin{VCPicture}{(0,-0.5)(6,4)}
\EdgeLineDouble
\VSState{(0,3.5)}{L1}
\VSState{(1.5,2.5)}{L2}
\VSState{(1.5,1.5)}{L3}
\VSState{(0,0.5)}{L4}
\VSState{(6,3.5)}{R1}
\VSState{(4.5,2.5)}{R2}
\VSState{(4.5,1.5)}{R3}
\VSState{(6,0.5)}{R4}
\ArcR{R1}{L1}{}
\ArcL{R4}{L4}{}
\ArcL{L2}{R2}{}
\ArcR{L3}{R3}{}
\Edge{L1}{L2}
\Edge{L1}{L3}
\Edge{R2}{R1}
\Edge{R3}{R4}
\Edge{L4}{L2}
\end{VCPicture}
}
\qquad
\subfigure[To a graph of type 7 or 8]{
\label{6 to 78}
\begin{VCPicture}{(0,-0.5)(6,4)}
\EdgeLineDouble
\VSState{(0,3.5)}{L1}
\VSState{(1.5,2.5)}{L2}
\VSState{(1.5,1.5)}{L3}
\VSState{(0,0.5)}{L4}
\VSState{(6,3.5)}{R1}
\VSState{(4.5,2.5)}{R2}
\VSState{(4.5,1.5)}{R3}
\VSState{(6,0.5)}{R4}
\ArcR{R1}{L1}{}
\ArcL{R4}{L4}{}
\ArcL{L2}{R2}{}
\ArcR{L3}{R3}{}
\Edge{L1}{L2}
\Edge{L1}{L3}
\Edge{R2}{R1}
\Edge{R3}{R4}
\Edge{L4}{L2}
\Edge{L4}{L3}
\end{VCPicture}
}
\qquad
\subfigure[To a graph of type 7 or 8]{
\label{6 to 78 bis}
\begin{VCPicture}{(0,-0.5)(6,5)}
\EdgeLineDouble
\VSState{(0,3.5)}{L1}
\VSState{(1.5,2.5)}{L2}
\VSState{(1.5,1.5)}{L3}
\VSState{(0,0.5)}{L4}
\VSState{(6,3.5)}{R1}
\VSState{(4.5,2.5)}{R2}
\VSState{(4.5,1.5)}{R3}
\VSState{(6,0.5)}{R4}
\ArcR{R1}{L1}{}
\ArcL{R4}{L4}{}
\ArcL{L2}{R2}{}
\ArcR{L3}{R3}{}
\Edge{L1}{L2}
\Edge{L1}{L3}
\Edge{R2}{R1}
\Edge{R3}{R4}
\Edge{L4}{L3}
\Edge{R3}{R1}
\end{VCPicture}
}
\qquad
\subfigure[To a graph of type 10]{
\label{6 to 10}
\begin{VCPicture}{(0,-0.5)(6,5)}
\EdgeLineDouble
\VSState{(0,3.5)}{L1}
\VSState{(1.5,2.5)}{L2}
\VSState{(1.5,1.5)}{L3}
\VSState{(0,0.5)}{L4}
\VSState{(6,3.5)}{R1}
\VSState{(4.5,2.5)}{R2}
\VSState{(4.5,1.5)}{R3}
\VSState{(6,0.5)}{R4}
\ArcR{R1}{L1}{}
\ArcL{R4}{L4}{}
\ArcL{L2}{R2}{}
\ArcR{L3}{R3}{}
\Edge{L1}{L2}
\Edge{L1}{L3}
\Edge{R3}{R1}
\Edge{R3}{R4}
\Edge{L4}{L3}
\Edge{R2}{R4}
\end{VCPicture}
}
\caption{Possible evolutions of a graph of type 6}
\label{appendix 6 to}
\end{figure}

\begin{table}[h!tbp]
\centering
\begin{tabular}{|l|c|l|l|}
\hline
From 6 to 	& 	$(U_{i_n},U_{i_{n+1}})$		&	Morphisms 	& 	Conditions		\\
\hline
1	&	$(\star,B)$		&	$[x,yx], [yx,x]$		&	$\card(C_n) = 2$ 	\\
\hline
7 or 8	&	$(\star,\star)$	&	$[1,0^k2,(0^{k-1}2)]$	&	$k \geq 1$									\\
\hhline{~~--}
		&					&	$[x,y^kx,(y^{k-1}x)]$	&	$k \geq 2$ and $\card(C_n) = 2$		\\
\hline
10	&	$(\star,B)$		&	$[1,0 1,2]$						&										\\
\hhline{~---}
	&	$(\star,R)$		&	$[1 2^k 0,2^{\ell}0]$				&	$k,\ell \geq 0$, $k + \ell \geq 1$	\\
\hhline{~~--}
	&					&	$[1 2^k 0,2^{\ell}0,1 2^{k-1}0]$	&	$k \geq \ell \geq 0$, $k \geq 1$	\\
\hhline{~~--}
	&					&	$[1 2^k 0,2^{\ell}0, 2^{\ell-1}0]$	&	$\ell > k \geq 0$					\\
\hline	
\end{tabular}
\caption{List of morphisms coding the evolutions of a graph of type 6}
\label{List of morphisms coding the evolutions of a graph of type 6}
\end{table}

\newpage

\subsection{Evolution of a Rauzy graph of type 7}

A graph of type 3 is represented in Figure~\ref{appendix type 7}. The possible evolutions are represented in Figure~\ref{appendix 7 to}.

\begin{figure}[h!tbp]
\centering
\begin{VCPicture}{(0,1)(6,4)}
\VSState{(1,2)}{B}
\VSState{(5.5,2)}{R}
\VSState{(4,2.8)}{L}
\EdgeLineDouble
\LArcL{R}{B}{}
\ArcL{B}{L}{}
\ArcL{L}{R}{}
\VCurveL{angleA=-130,angleB=130,ncurv=20}{B}{B}{}
\VCurveL{angleA=60,angleB=60,ncurv=1.15}{R}{L}{}
\end{VCPicture}
\caption{Graph of type 7}
\label{appendix type 7}
\end{figure}

\begin{figure}[h!tbp]
\centering
\subfigure[To a graph of type 1]{
\begin{VCPicture}{(0,0.5)(6,4)}
\VSState{(1,2.5)}{L1}
\VSState{(1,1)}{L2}
\VSState{(2.5,2.5)}{R1}
\VSState{(2.5,1)}{R2}
\VSState{(5.5,1.75)}{R}
\VSState{(4,2.8)}{L}
\EdgeLineDouble
\ArcL{R}{R2}{}
\ArcL{R1}{L}{}
\ArcL{L}{R}{}
\Edge{L1}{R1}
\Edge{R2}{L2}
\VCurveL{angleA=180,angleB=180,ncurv=2}{L2}{L1}{}
\VCurveL{angleA=60,angleB=60,ncurv=1.15}{R}{L}{}
\end{VCPicture}
}
\qquad
\subfigure[To a graph of type 7 or 8]{
\label{7 to 78}
\begin{VCPicture}{(0,0.5)(6,4)}
\VSState{(1,2.5)}{L1}
\VSState{(1,1)}{L2}
\VSState{(2.5,2.5)}{R1}
\VSState{(2.5,1)}{R2}
\VSState{(5.5,1.75)}{R}
\VSState{(4,2.8)}{L}
\EdgeLineDouble
\ArcL{R}{R2}{}
\ArcL{R1}{L}{}
\ArcL{L}{R}{}
\Edge{L1}{R1}
\Edge{R2}{L2}
\Edge{L1}{L2}
\VCurveL{angleA=180,angleB=180,ncurv=2}{L2}{L1}{}
\VCurveL{angleA=60,angleB=60,ncurv=1.15}{R}{L}{}
\end{VCPicture}
}
\qquad
\subfigure[To a graph of type 9]{
\label{7 to 9}
\begin{VCPicture}{(0,0.5)(6,4)}
\VSState{(1,2.5)}{L1}
\VSState{(1,1)}{L2}
\VSState{(2.5,2.5)}{R1}
\VSState{(2.5,1)}{R2}
\VSState{(5.5,1.75)}{R}
\VSState{(4,2.8)}{L}
\EdgeLineDouble
\ArcL{R}{R2}{}
\ArcL{R1}{L}{}
\ArcL{L}{R}{}
\Edge{L1}{R1}
\Edge{R2}{L2}
\Edge{R2}{R1}
\VCurveL{angleA=180,angleB=180,ncurv=2}{L2}{L1}{}
\VCurveL{angleA=60,angleB=60,ncurv=1.15}{R}{L}{}
\end{VCPicture}
}
\caption{Possible evolutions of a graph of type 7}
\label{appendix 7 to}
\end{figure}

\begin{table}[h!tbp]
\centering
\begin{tabular}{|l|c|l|l|}
\hline
From 7 to 	& 	$(U_{i_n},U_{i_{n+1}})$		&	Morphisms 	& 	Conditions		\\
\hline
1	&	$(R,B)$		&	$[x,y]$		&	$\card(C_n) = 2$ 	\\
\hline
7 or 8	&	$(R,\star)$		&	$[0,1,(2)]$		&	\\
\hhline{~---}
		&	$(B,\star)$		&	$[0,1 0,(2 0)]$	&	\\
\hline
9	&	$(R,B)$		&	$[0,x,y]$	&	\\
\hhline{~---}
	&	$(B,R)$		&	$[0 1,1,(0 2)]$, $[1,0 1,(2)]$		&	\\
\hhline{~~--}
	&				&	$[0 1,2,(0 2)]$, $[1,0 2,(2)]$		&	$\card(C_n) = 3$	 	\\
\hline	
\end{tabular}
\caption{List of morphisms coding the evolutions of a graph of type 7}
\label{List of morphisms coding the evolutions of a graph of type 7}
\end{table}

\newpage

\subsection{Evolution of a Rauzy graph of type 8}

A graph of type 3 is represented in Figure~\ref{appendix type 8}. The possible evolutions are represented in Figure~\ref{appendix 8 to}.

\begin{figure}[h!tbp]
\centering
\begin{VCPicture}{(0,1)(6,4)}
\VSState{(1,2)}{L}
\VSState{(5,2)}{R}
\EdgeLineDouble
\LArcL{L}{R}{}
\LArcL{R}{L}{}
\VCurveL{angleA=-130,angleB=130,ncurv=20}{L}{L}{}
\VCurveL{angleA=50,angleB=-50,ncurv=20}{R}{R}{}
\end{VCPicture}
\caption{Graph of type 8}
\label{appendix type 8}
\end{figure}

\begin{figure}[h!tbp]
\centering
\subfigure[To a graph of type 1]{
\begin{VCPicture}{(0,0.5)(6,3)}
\VSState{(1,2)}{L1}
\VSState{(1,1)}{L2}
\VSState{(2,2)}{L3}
\VSState{(2,1)}{L4}
\VSState{(4,2)}{R1}
\VSState{(4,1)}{R2}
\VSState{(5,2)}{R3}
\VSState{(5,1)}{R4}
\EdgeLineDouble
\ArcL{L3}{R1}{}
\ArcL{R2}{L4}{}
\Edge{L1}{L3}
\Edge{L4}{L2}
\Edge{R1}{R3}
\Edge{R4}{R2}
\Edge{L1}{L2}
\VCurveL{angleA=180,angleB=180,ncurv=2}{L2}{L1}{}
\VCurveL{angleA=0,angleB=0,ncurv=2}{R3}{R4}{}
\end{VCPicture}
}
\qquad
\subfigure[To a graph of type 1]{
\begin{VCPicture}{(0,0.5)(6,3)}
\VSState{(1,2)}{L1}
\VSState{(1,1)}{L2}
\VSState{(2,2)}{L3}
\VSState{(2,1)}{L4}
\VSState{(4,2)}{R1}
\VSState{(4,1)}{R2}
\VSState{(5,2)}{R3}
\VSState{(5,1)}{R4}
\EdgeLineDouble
\ArcL{L3}{R1}{}
\ArcL{R2}{L4}{}
\Edge{L1}{L3}
\Edge{L4}{L2}
\Edge{R1}{R3}
\Edge{R4}{R2}
\Edge{L4}{L3}
\VCurveL{angleA=180,angleB=180,ncurv=2}{L2}{L1}{}
\VCurveL{angleA=0,angleB=0,ncurv=2}{R3}{R4}{}
\end{VCPicture}
}
\qquad
\subfigure[To a graph of type 5 or 6]{
\begin{VCPicture}{(0,0.5)(6,3)}
\VSState{(1,2)}{L1}
\VSState{(1,1)}{L2}
\VSState{(2,2)}{L3}
\VSState{(2,1)}{L4}
\VSState{(4,2)}{R1}
\VSState{(4,1)}{R2}
\VSState{(5,2)}{R3}
\VSState{(5,1)}{R4}
\EdgeLineDouble
\ArcL{L3}{R1}{}
\ArcL{R2}{L4}{}
\Edge{L1}{L3}
\Edge{L4}{L2}
\Edge{R1}{R3}
\Edge{R4}{R2}
\Edge{L4}{L3}
\Edge{R1}{R2}
\VCurveL{angleA=180,angleB=180,ncurv=2}{L2}{L1}{}
\VCurveL{angleA=0,angleB=0,ncurv=2}{R3}{R4}{}
\end{VCPicture}
}
\qquad
\subfigure[To a graph of type 7 or 8]{
\begin{VCPicture}{(0,0.5)(6,3)}
\VSState{(1,2)}{L1}
\VSState{(1,1)}{L2}
\VSState{(2,2)}{L3}
\VSState{(2,1)}{L4}
\VSState{(4,2)}{R1}
\VSState{(4,1)}{R2}
\VSState{(5,2)}{R3}
\VSState{(5,1)}{R4}
\EdgeLineDouble
\ArcL{L3}{R1}{}
\ArcL{R2}{L4}{}
\Edge{L1}{L3}
\Edge{L4}{L2}
\Edge{R1}{R3}
\Edge{R4}{R2}
\Edge{R4}{R3}
\Edge{L1}{L2}
\VCurveL{angleA=180,angleB=180,ncurv=2}{L2}{L1}{}
\VCurveL{angleA=0,angleB=0,ncurv=2}{R3}{R4}{}
\end{VCPicture}
}
\qquad
\subfigure[To a graph of type 7 or 8]{
\begin{VCPicture}{(0,0.5)(6,3)}
\VSState{(1,2)}{L1}
\VSState{(1,1)}{L2}
\VSState{(2,2)}{L3}
\VSState{(2,1)}{L4}
\VSState{(4,2)}{R1}
\VSState{(4,1)}{R2}
\VSState{(5,2)}{R3}
\VSState{(5,1)}{R4}
\EdgeLineDouble
\ArcL{L3}{R1}{}
\ArcL{R2}{L4}{}
\Edge{L1}{L3}
\Edge{L4}{L2}
\Edge{R1}{R3}
\Edge{R4}{R2}
\Edge{L1}{L2}
\Edge{L4}{L3}
\VCurveL{angleA=180,angleB=180,ncurv=2}{L2}{L1}{}
\VCurveL{angleA=0,angleB=0,ncurv=2}{R3}{R4}{}
\end{VCPicture}
}
\qquad
\subfigure[To a graph of type 9]{
\begin{VCPicture}{(0,0.5)(6,3)}
\VSState{(1,2)}{L1}
\VSState{(1,1)}{L2}
\VSState{(2,2)}{L3}
\VSState{(2,1)}{L4}
\VSState{(4,2)}{R1}
\VSState{(4,1)}{R2}
\VSState{(5,2)}{R3}
\VSState{(5,1)}{R4}
\EdgeLineDouble
\ArcL{L3}{R1}{}
\ArcL{R2}{L4}{}
\Edge{L1}{L3}
\Edge{L4}{L2}
\Edge{R1}{R3}
\Edge{R4}{R2}
\Edge{L4}{L3}
\Edge{R4}{R3}
\VCurveL{angleA=180,angleB=180,ncurv=2}{L2}{L1}{}
\VCurveL{angleA=0,angleB=0,ncurv=2}{R3}{R4}{}
\end{VCPicture}
}
\caption{Possible evolutions of a graph of type 7}
\label{appendix 8 to}
\end{figure}

\begin{table}[h!tbp]
\centering
\begin{tabular}{|l|c|l|l|}
\hline
From 8 to 	& 	$(U_{i_n},U_{i_{n+1}})$		&	Morphisms 	& 	Conditions		\\
\hline
1	&	$(\star,B)$		&	$[x,yx]$, $[yx,x]$	&	$\card(C_n) = 2$ 	\\
\hline
5 or 6	&	$(\star,\star)$		&	$[0 x,y,(0 y)], [x,0 y,(y)]$	&	$\card(C_n) = 3$	\\
\hline
7 or 8	&	$(\star,\star)$		&	$[0,1 0,(2 0)]$			&	\\
\hhline{~~--}
		&						&	$[x,y^kx,(y^{k-1}x)]$	&	$k \geq 2$ and $\card(C_n) = 2$				\\
\hline
9	&	$(\star,B)$		&	$[0,x 0,y 0]$		&	\\
\hhline{~---}
	&	$(\star,R)$		&	$[0 1,1,(0 2)]$, $[1,0 1,(2)]$	&	$\card(C_n) = 3$ 	\\
\hhline{~~--}
	&					&	$[0 1,2,(0 2)]$, $[1,0 2,(2)]$	&					\\
\hline	
\end{tabular}
\caption{List of morphisms coding the evolutions of a graph of type 8}
\label{List of morphisms coding the evolutions of a graph of type 8}
\end{table}

\newpage

\subsection{Evolution of a Rauzy graph of type 9}

A graph of type 3 is represented in Figure~\ref{appendix type 9}. The possible evolutions are represented in Figure~\ref{appendix 9 to}.

\begin{figure}[h!tbp]
\centering
\begin{VCPicture}{(0,1)(6,4)}
\VSState{(0.5,2)}{R}
\VSState{(3.5,2)}{L}
\VSState{(4.5,2)}{B}
\EdgeLineDouble
\ArcL{R}{L}{}
\ArcR{R}{L}{}
\EdgeL{L}{B}{}
\LArcL{B}{R}{}
\VCurveL{angleA=45,angleB=-45,ncurv=20}{B}{B}{}
\end{VCPicture}
\caption{Graph of type 9}
\label{appendix type 9}
\end{figure}

\begin{figure}[h!tbp]
\centering
\subfigure[To a graph of type 1]{
\label{9 to 1}
\begin{VCPicture}{(0,0.5)(6,4)}
\VSState{(0,2.5)}{R}
\VSState{(3,2.5)}{L}
\VSState{(4,2.5)}{L1}
\VSState{(4,1.5)}{L2}
\VSState{(5,2.5)}{R1}
\VSState{(5,1.5)}{R2}
\EdgeLineDouble
\ArcL{R}{L}{}
\ArcR{R}{L}{}
\Edge{L}{L1}
\LArcL{L2}{R}{}
\Edge{L1}{R1}
\Edge{R2}{L2}
\VCurveL{angleA=10,angleB=-10,ncurv=2}{R1}{R2}{}
\end{VCPicture}
}
\qquad
\subfigure[To a graph of type 5 or 6]{
\label{9 to 56}
\begin{VCPicture}{(0,0.5)(6,4)}
\VSState{(0,2.5)}{R}
\VSState{(3,2.5)}{L}
\VSState{(4,2.5)}{L1}
\VSState{(4,1.5)}{L2}
\VSState{(5,2.5)}{R1}
\VSState{(5,1.5)}{R2}
\EdgeLineDouble
\ArcL{R}{L}{}
\ArcR{R}{L}{}
\Edge{L}{L1}
\LArcL{L2}{R}{}
\Edge{L1}{R1}
\Edge{R2}{L2}
\Edge{L1}{L2}
\VCurveL{angleA=10,angleB=-10,ncurv=2}{R1}{R2}{}
\end{VCPicture}
}
\qquad
\subfigure[To a graph of type 9]{
\label{9 to 9}
\begin{VCPicture}{(0,0.5)(6,4)}
\VSState{(0,2.5)}{R}
\VSState{(3,2.5)}{L}
\VSState{(4,2.5)}{L1}
\VSState{(4,1.5)}{L2}
\VSState{(5,2.5)}{R1}
\VSState{(5,1.5)}{R2}
\EdgeLineDouble
\ArcL{R}{L}{}
\ArcR{R}{L}{}
\Edge{L}{L1}
\LArcL{L2}{R}{}
\Edge{L1}{R1}
\Edge{R2}{L2}
\Edge{R2}{R1}
\VCurveL{angleA=10,angleB=-10,ncurv=2}{R1}{R2}{}
\end{VCPicture}
}
\caption{Possible evolutions of a graph of type 9}
\label{appendix 9 to}
\end{figure}

\begin{table}[h!tbp]
\centering
\begin{tabular}{|l|c|l|l|}
\hline
From 9 to 	& 	$(U_{i_n},U_{i_{n+1}})$		&	Morphisms 	& 	Conditions		\\
\hline
1	&	$(R,B)$		&	$[x,y]$		&	$\card(C_n) = 2$ 	\\
\hline
5 or 6	&	$(R,\star)$		&	$[0,1,(2)]$, $[2,1,0]$		&	\\
\hhline{~---}
		&	$(B,\star)$		&	$[0 x,y,(0 y)]$, $[x,0 y,(y)]$		&	\\
\hline
9	&	$(R,R)$		&	$[0,1,(2)]$		&	\\
\hhline{~---}
	&	$(B,B)$		&	$[0,x 0,y 0]$		&	\\		
\hline	
\end{tabular}
\caption{List of morphisms coding the evolutions of a graph of type 9}
\label{List of morphisms coding the evolutions of a graph of type 9}
\end{table}

\newpage

\subsection{Evolution of a Rauzy graph of type 10}

A graph of type 3 is represented in Figure~\ref{appendix type 10}. The possible evolutions are represented in Figure~\ref{appendix 10 to}.

\begin{figure}[h!tbp]
\centering
\begin{VCPicture}{(0,0.5)(6,4)}
\VSState{(0,2)}{R}
\VSState{(4,3)}{L}
\VSState{(6,2)}{B}
\EdgeLineDouble
\ArcL{R}{L}{}
\ArcL{L}{B}{}
\Edge{R}{B}
\LArcL{B}{R}{}
\VCurveR{angleA=45,angleB=45,ncurv=1}{B}{L}{}
\end{VCPicture}
\caption{Graph of type 10}
\label{appendix type 10}
\end{figure}

\begin{figure}[h!tbp]
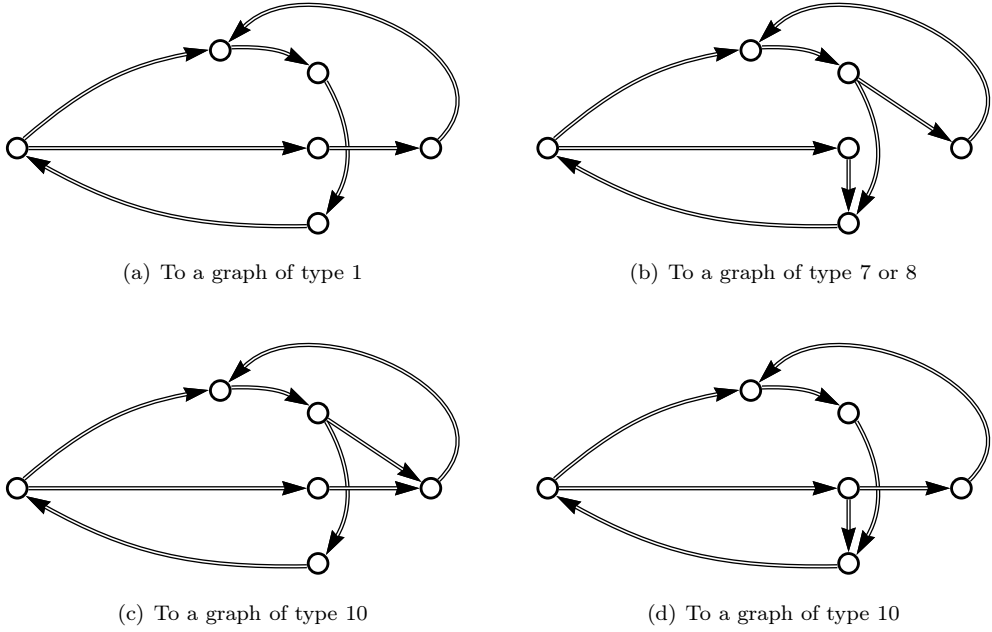

\centering
\subfigure[To a graph of type 1]{
\begin{VCPicture}{(0,0)(6,4)}
\VSState{(0,1.5)}{R}
\VSState{(2.7,2.8)}{L}
\VSState{(4,2.5)}{L1}
\VSState{(4,1.5)}{L2}
\VSState{(4,0.5)}{L3}
\VSState{(5.5,1.5)}{R1}
\EdgeLineDouble
\ArcL{R}{L}{}
\ArcL{L}{L1}{}
\Edge{R}{L2}
\ArcL{L3}{R}{}
\LArcL{L1}{L3}{}
\Edge{L2}{R1}
\VCurveR{angleA=45,angleB=45,ncurv=1}{R1}{L}{}
\end{VCPicture}
}
\qquad
\subfigure[To a graph of type 7 or 8]{
\label{10 to 78}
\begin{VCPicture}{(0,0)(6,4)}
\VSState{(0,1.5)}{R}
\VSState{(2.7,2.8)}{L}
\VSState{(4,2.5)}{L1}
\VSState{(4,1.5)}{L2}
\VSState{(4,0.5)}{L3}
\VSState{(5.5,1.5)}{R1}
\EdgeLineDouble
\ArcL{R}{L}{}
\ArcL{L}{L1}{}
\Edge{R}{L2}
\ArcL{L3}{R}{}
\LArcL{L1}{L3}{}
\Edge{L1}{R1}
\Edge{L2}{L3}
\VCurveR{angleA=45,angleB=45,ncurv=1}{R1}{L}{}
\end{VCPicture}
}
\qquad
\subfigure[To a graph of type 10]{
\label{10 to 10}
\begin{VCPicture}{(0,0)(6,4)}
\VSState{(0,1.5)}{R}
\VSState{(2.7,2.8)}{L}
\VSState{(4,2.5)}{L1}
\VSState{(4,1.5)}{L2}
\VSState{(4,0.5)}{L3}
\VSState{(5.5,1.5)}{R1}
\EdgeLineDouble
\ArcL{R}{L}{}
\ArcL{L}{L1}{}
\Edge{R}{L2}
\ArcL{L3}{R}{}
\LArcL{L1}{L3}{}
\Edge{L2}{R1}
\Edge{L1}{R1}
\VCurveR{angleA=45,angleB=45,ncurv=1}{R1}{L}{}
\end{VCPicture}
}
\qquad
\subfigure[To a graph of type 10]{
\label{10 to 10bis}
\begin{VCPicture}{(0,0)(6,4)}
\VSState{(0,1.5)}{R}
\VSState{(2.7,2.8)}{L}
\VSState{(4,2.5)}{L1}
\VSState{(4,1.5)}{L2}
\VSState{(4,0.5)}{L3}
\VSState{(5.5,1.5)}{R1}
\EdgeLineDouble
\ArcL{R}{L}{}
\ArcL{L}{L1}{}
\Edge{R}{L2}
\ArcL{L3}{R}{}
\LArcL{L1}{L3}{}
\Edge{L2}{R1}
\Edge{L2}{L3}
\VCurveR{angleA=45,angleB=45,ncurv=1}{R1}{L}{}
\end{VCPicture}
}
\caption{Possible evolutions of a graph of type 10}
\label{appendix 10 to}
\end{figure}

\begin{table}[h!tbp]
\centering
\begin{tabular}{|l|c|l|l|}
\hline
From 10 to 	& 	$(U_{i_n},U_{i_{n+1}})$		&	Morphisms 	& 	Conditions		\\
\hline
1	&	$(R,B)$		&	$[x,y]$		&	$\card(C_n) = 2$	\\
\hline
7 or 8	&	$(R,\star)$		&	$[1,0,(2)]$					&	\\
\hhline{~---}
		&	$(B,\star)$		&	$[0,2^k 1,(2^{k-1}1)]$		&	$k \geq 1$	\\
\hline
10	&	$(R,R)$		&	$[1,0,(2)]$		&	\\
\hhline{~---}
	&	$(B,B)$		&	$[0,2 0,1]$		&	\\
\hhline{~---}
	&	$(R,B)$		&	$[0,1,2]$		&	$\card(C_n) = 3$	\\
\hhline{~---}
	&	$(B,R)$		&	$[0 1^k 2,1^{\ell}2]$				&	$k,\ell \geq 0$, $k+\ell \geq 1$	\\
\hhline{~~--}
	&				&	$[0 1^k 2,1^{\ell}2,0 1^{k-1}2]$	&	$k \geq 1$, $k \geq \ell \geq 0$	\\
\hhline{~~--}
	&				&	$[0 1^k 2,1^{\ell}2,1^{\ell-1}2]$	&	$\ell > k \geq 0$					\\	
\hline	
\end{tabular}
\caption{List of morphisms coding the evolutions of a graph of type 10}
\label{List of morphisms coding the evolutions of a graph of type 10}
\end{table}

\section{Proof of Lemma~\ref{lemma: weakly primitive C_4}}
\label{appendix: weakly primitive}

Let us prove the following result which is equivalent to Lemma~\ref{lemma: weakly primitive C_4} but with more details.

\begin{lemma}
\label{lemma: almost primitive C_4'}
An infinite path $p$ in Figure~\ref{Figure: pre-graph C_4} has a weakly primitive label $(\alpha_n)_{n \geq N}$ if and only if one of the following conditions is satisfied:

\begin{enumerate}

	\item $p$ ultimately stays in vertex 1 and both morphisms $[0,10]$ and $[01,1]$ occur infinitely often in $(\alpha_n)_{n \geq N}$;
	
	\item $p$ ultimately stays in vertex $10B$ and for all integers $r \geq N$, $(\alpha)_{n \geq r}$ does not only contain occurrences of $[0,20,1]$, neither of $[01^k2,1^{k+1}2,1^k2]$ for $k \in \N$ and is not only composed of finite sub-sequences of morphisms in
	\[
		\left\{ {[0,20,1]}^{2n}, {[02,12,2]}^n \mid n \in \N \setminus \{0\} \right\};
	\]

	\item $p$ ultimately stays in the subgraph $\{1,7/8\}$, goes through both vertices infinitely often and for all suffixes $p'$ of $p$ starting in vertex $7/8$, the label of $p'$ is not only composed of finite sub-sequences of morphisms in
	\[
		\left( [0,10]^* [0,1] [0,10]^* \{[0,1^k0] \mid k \geq 2\} \right)	
		\cup \left( [0,10]^* [1,0] [01,1]^* \{ [1,0^k1] \mid k \geq 2\} \right);
	\]

	\item $p$ ultimately stays in the subgraph $\{5/6,7/8\}$, goes through both vertices infinitely often and for all suffixes $p'$ of $p$ starting in vertex $7/8$, the label of $p'$ is not only composed of finite sub-sequences of morphisms in	
	\[
		[0,10,20]^* \left\{ [1,02,2], [01,2,02] \right\} [1,02,2]
	\]
and not only composed of finite sub-sequences of morphisms in 
	\[
		\{ [2,01,1], [1,02,2] \} \left\{ [1,0^k2,0^{k-1}2] , [12^{k-1}0,2^{\ell}0,2^{\ell-1}0] \mid \ell > k \geq 1 \right\};
	\]	

	\item $p$ ultimately stays in the subgraph $\{5/6,7/8,10B\}$, goes through the three vertices infinitely often and if $(q_n)_{n \in \N}$ (resp. $(t_n)_{n \in \N}$) is the sequence of finite sub-paths or $p$ that start and end in $7/8$ and go through $10B$ (resp. that start and end in $7/8$ and do not go through $10B$), then for all integers $r \geq N$, the following conditions hold:
	\begin{itemize}
		\item[-] if for all $n \geq r$, the label of $q_n$ is in
		\begin{multline*}
			\{ [1,02,2],[01,2,02] \} [1,01,2] \{ {[0,20,1]}^{2n}, [02,12,2] \mid n \in \N \}^* \\
			\{ [2,012,02], [0,20,1][0,21,1] \},
		\end{multline*}
		then the sequence $(t_n)_{n \in \N}$ is infinite and contains infinitely many occurrences of finite paths whose label is not in
		\[
			\left\{ [1,02,2], [01,2,02] \right\} [1,02,2];
		\]
		
		\item[-] if for all $n \geq r$, the label of $q_n$ is in
		\begin{multline*}
			\{ [1,02,2],[2,01,1] \} \{ [12^k0,2^{k+1}0,2^k0] \mid k \geq 0 \} \\
			\{ [01^k2,1^{k+1}2,1^k2] \mid k \geq 0 \} \{ [0,2^k1,2^{k-1}1] \mid k \geq 2 \},
		\end{multline*}
		then the path $p$ goes infinitely often through the loop on $7/8$ or, the sequence $(t_n)_{n \in \N}$ is infinite and contains infinitely many occurrences of finite paths whose label is not in
		\[
			\{ [2,01,1], [1,02,2] \} 
				\left\{ [1,0^k2,0^{k-1}2] , [12^{k-1}0,2^{\ell}0,2^{\ell-1}0] \mid \ell > k \geq 1 	\right\};
		\]	
	\end{itemize}

	\item $p$ contains infinitely many occurrences of sub-paths $q$ that start in $1$ and end in $5/6$.

\end{enumerate}
\end{lemma}

\begin{proof}
The method to prove this result is to study the almost primitivity in each subgraph of Figure~\ref{Figure: pre-graph C_4}. Among all these subgraphs, those in which there exist some infinite paths are
\[
	\{1\}, \{7/8\}, \{10B\}, \{1,7/8\}, \{5/6,7/8\}, 
	 \{1, 5/6, 7/8\}, \{5/6, 7/8, 10B\}, \{1,5/6, 7/8, 10B\}.
\]

It is easily seen that all valid paths in the subgraph $\{7/8\}$ do not have almost primitive labels. Also, for the subgraphs $\{1\}$, $\{10B\}$, the given conditions of the result are trivially equivalent to the almost primitivity.

\bigskip

Let us study the subgraph $\{1,7/8\}$. If $q$ is a path starting in vertex $7/8$, going through vertex 1, possibly staying in it for a while and then coming back to vertex $7/8$, then its label belongs to the set
\begin{multline*}
	Q = \left\{[x,y][x,yx],[xy,y] \mid \{x,y\} = \{0,1\} \right\} \{[0,10],[01,1]\}^* \\
	\left\{ [0,1^k0,1^{k-1}0],[1,0^k1,0^{k-1}1] \mid k \geq 2 \right\}.
\end{multline*}
If $p$ ultimately stays in the subgraph $\{1,7/8\}$, it means that its label is ultimately composed of finite subsequences of morphisms in that set and of occurrences of the morphism $[0,10,20]$ labelling the loop on vertex $7/8$. However, morphisms labelling the edge from $7/8$ to $1$ do not contain the letter 2 in their images. Consequently, the third component of all morphisms can be ignored. Now it can be checked that for all finite sequences of morphisms $\alpha_1 \cdots \alpha_m$ in $Q$, $\alpha_1 \cdots \alpha_m(1)$ contains some occurrences of both $0$ and $1$. Since the morphism labelling the loop on $7/8$ is $[0,10]$, the label $(\alpha_n)_{n \geq N}$ of any infinite path $p$ in $\{1,7/8\}$ is not almost primitive if and only if there is an integer $r \geq N$ such that for all $n \geq r$, $\alpha_r \alpha_{r+1} \cdots \alpha_n(0)=0$. To conclude the proof for the subgraph $\{1,7/8\}$, it suffices to notice that the finite sequences of morphisms $\alpha_1' \cdots \alpha_m'$ in
\[
	[0,1] [0,10]^* [0,1^k0] \cup [1,0] [01,1]^* [1,0^k1]
\]
are the only ones in $Q$ such that $\alpha_1' \cdots \alpha_m'(0)=0$.

\bigskip

Let us study the subgraph $\{5/6,7/8\}$. For any word $u$ over $\{0,1,2\}$ we let $\alp(u)$ be the smallest lexicographic word over $\{0,1,2\}$ such that all letters occurring in $u$ occur in  $\alp(u)$ too. By abuse of notation, for any path $q$ with label $\sigma = \alpha_1 \cdots \alpha_m$ we write 
\[
	\alp(q) = (\alp(\sigma(0)),\alp(\sigma(1)),\alp(\sigma(2))).
\]

It can be algorithmically checked that, if $q$ is a path of length two that starts in $7/8$ and goes through $5/6$ before coming back to $7/8$, then  $\alp(q)$ is one of those given in Table~\ref{table 785678}.
\begin{table}[h!tbp]
\centering
\begin{tabular}{|l|l|l|l|l|}
\hline
(01,12,1) & (01,12,12) 	& (012,12,12) & (02,12,12) 	& (02,12,2) 	\\
\hline
(012,012,012) & (01,012,012) & (02,012,012) & (12,012,012) 	& (1,012,012) 	\\
\hline
(2,012,012) & (1,012,01) 	& (2,012,02) 	& \multicolumn{2}{l}{} 	\\
\hhline{---~~}
\end{tabular}
\caption{List of $\alp(q)$ for $q= 7/8 \to 5/6 \to 7/8$.}
\label{table 785678}
\end{table}

We let $Q_1$ denote the set of paths $q$ of length 2 that start in $7/8$, go through $5/6$ and come back to $7/8$ and such that $\alp(q)$ is one of the following:
\begin{center}
\begin{tabular}{|l|l|l|l|}
\hline
(012,012,012) 	& (01,012,012) 	& 	(02,012,012) &	(12,012,012)	\\
\hline
(1,012,012) 	&	(2,012,012)	&	(1,012,01) 	 & 	\multicolumn{1}{l}{}	\\
\hhline{---~}
\end{tabular}
\end{center}
Obviously, the label $(\alpha_n)_{n \geq N}$ of any infinite path $p$ in the subgraph $\{5/6,7/8\}$ that contains infinitely many occurrences of sub-paths $q$ in $Q_1$ is almost primitive. Indeed, if $p$ is a finite path in the subgraph $\{5/6,7/8\}$ that contains two occurrences of paths in $Q_1$, then the letter $1$ occurs in the three components of $\alp(p)$ which makes $(\alpha_n)_{n \geq N}$ almost primitive because for all paths $q$ in $Q_1$, the second component of $\alp(q)$ contains occurrences of the three letters.

Let us consider an infinite path $p$ such that all sub-paths $q$ of length 2 that start in $7/8$ and go through $5/6$ do not belong to $Q_1$, so are such that $\alp(q)$ is one of the following:
\begin{center}
\begin{tabular}{|l|l|l|}
\hline
(01,12,1)  & (01,12,12) 	& (012,12,12)  	\\
\hline
(02,12,12) & (02,12,2) 	& (2,012,02) 	\\
\hline
\end{tabular}
\end{center}
For such paths $q$, we can see two problems for the almost primitivity:
\begin{itemize}
	\item[-] 	except for paths $q$ such that $\alp(q) = (2,012,02)$, the letter $0$ never occurs in the two last components of $\alp(q)$;
	\item[-]	for paths $q$ such that $\alp(q) \in \{ (02,12,2), (2,012,02) \}$, the letter $1$ never occurs in the first and in the last component of $\alp(q)$.
\end{itemize}
Consequently, the following holds true: the label of any infinite path $p$ in $\{5/6,7/8\}$ such that all sub-paths $q: 7/8 \to 5/6 \to 7/8$ are such that
\begin{enumerate}
	\item 	$\alp(q) \in \{(02,12,2),(2,012,02)\}$ cannot be almost primitive;
	\item 	$\alp(q) \in \{(01,12,1), (01,12,12), (012,12,12), (02,12,12), (02,12,2)\}$ is almost primitive if and only if $\alp(q)$ is not ultimately $(02,12,2)$ and the path $p$ goes infinitely often through the loop on $7/8$ (because it is labelled by $[0,10,20]$).
\end{enumerate}
One can also check that if there are infinitely many occurrences of paths $q$ and $q'$ in $p$ such that $\alp(q) = (2,012,02)$ and 
\[
	\alp(q') \in \{(01,12,1), (01,12,12), (012,12,12), (02,12,12) \},
\]
then the label of $p$ is almost primitive. 

To conclude the proof for the subgraph $\{5/6,7/8\}$, it suffices now to study which labelled paths $q = 7/8 \to 5/6 \to 7/8$ correspond to the \enquote{forbidden cases} listed just above. If $q$ is such a path and if $\alpha_1$ (resp. $\alpha_2$) labels the edge $7/8 \to 5/6$ (resp. $5/6 \to 7/8$), then we have
\begin{eqnarray*}
	&\alp(q) = (02,12,2)  \Leftrightarrow  
					\begin{cases}
						\alpha_1  =  [1,02,2] \\
						\alpha_2  =  [1,02,2]
					\end{cases}& \\
	&\alp(q) = (2,012,02)  \Leftrightarrow  
					\begin{cases}
						\alpha_1  =  [01,2,02] \\
						\alpha_2  =  [1,02,2]
					\end{cases}& 			
\end{eqnarray*}
and
\begin{eqnarray*}
&	\alp(q) \in \{ (01,12,1), (01,12,12), (012,12,12), (02,12,12), (02,12,2) \} 	& 	\\ 
&	\Updownarrow		&	\\
&	\begin{cases}
		\alpha_1 \in \{[1,02,2], [2,01,1]\} \\
		\alpha_2 \in \{[1,0^k2,0^{k-1}2] \mid k \geq 1 \} \cup \{[12^k0,2^{\ell}0,2^{\ell-1}0] \mid \ell > k+1 \geq 1 \}
	\end{cases} 	&
\end{eqnarray*}

\bigskip

Let us study the subgraph $\{5/6,7/8,10B\}$. As for $\{5/6,7/8\}$, it can be algorithmically checked that, if $q$ is a finite path in $\{5/6,7/8,10B\}$ that starts and ends in $7/8$ and that goes through $10B$, then $\alp(q)$ is one of those given in Table~\ref{table 78561078}.
\begin{table}[h!tbp]
\centering
\begin{tabular}{|l|l|l|l|l|}
\hline
(01,012,01) & (01,012,012) 	& (012,012,012) & (012,12,12) 	& (02,012,012) 	\\
\hline
(02,012,02) & (1,012,01) 	& (1,012,012) 	& (2,012,012) 	& (2,012,02) 	\\
\hline
\end{tabular}
\caption{List of $\alp(q)$ for $q = 7/8 \to 5/6 \to 10B (\to 10B)^* \to 7/8$.}
\label{table 78561078}
\end{table}

Let us start by determining some non-almost primitive infinite labelled paths. First, it is easily seen that if $p_1$ is an infinite path in $\{5/6,7/8,10B\}$ whose sub-paths $q_{1,1} = 7/8 \to 5/6 \to 10B (\to 10B)^* \to 7/8$ are ultimately such that $\alp(q_{1,1}) \in \{(2,012,02),(02,012,02)\}$, then the label of $p_1$ is almost primitive if and only if $p_1$ contains infinitely many occurrences of sub-paths $q_{1,2} = 7/8 \to 5/6 \to 7/8$ such that\footnote{The problem is the same as the one met in the subgraph $\{5/6,7/8\}$: the letter $1$ never occurs in the image of $02$.}
\[
	\alp(q_{1,2})  \notin  \{(02,12,2), (2,012,02)\}.
\]

Next, one can also see that if $p_2$ is an infinite path in $\{5/6,7/8,10B\}$ whose sub-paths $q_{2,1} = 7/8 \to 5/6 \to 10B (\to 10B)^* \to 7/8$ are ultimately such that $\alp(q_{2,1}) = (012,12,12)$, then the label of $p_2$ is almost primitive if and only if
$p_2$ contains infinitely many occurrences of loops $7/8 \to 7/8$ or of sub-paths $q_{2,2} = 7/8 \to 5/6 \to 7/8$ such that\footnote{This is again a problem met in the subgraph $\{5/6,7/8\}$: the letter $0$ never occurs in the image of $12$.}
\[
	\alp(q_{2,2})  \notin  \{(01,12,1), (01,12,12), (012,12,12), (02,12,12), (02,12,2)\}.
\]

Now let us show that all other infinite paths $p_3$ in $\{5/6,7/8,10B\}$ that goes infinitely often through the three vertices have an almost primitive label. We can see that in all remaining values of $\alp(q)$, i.e., for all paths $q = 7/8 \to 5/6 \to 10B (\to 10B)^* \to 7/8$ with 
\[
	\alp(q) \notin \{(2,012,02),(02,012,02),(012,12,12)\},
\]
the second component of $\alp(q)$ is $012$. This makes the label of $p_3$ almost primitive because if $p'$ is a finite path in $\{5/6,7/8,10B\}$ that contains two occurrences of paths $q = 7/8 \to 5/6 \to 10B (\to 10B)^* \to 7/8$ with 
\[
	\alp(q) \notin \{(2,012,02),(02,012,02),(012,12,12)\},
\]
then each component of $\alp(p')$ contains an occurrence of the letter $1$.

To conclude the proof for the subgraph $\{\{5/6,7/8,10B\}$, it suffices (like for the subgraph $\{5/6,7/8\}$) to study which labelled paths $q = 7/8 \to 5/6 \to 10B (\to 10B)^* \to 7/8$ correspond to the \enquote{forbidden cases}, i.e., which ones are such that
\[
	\alp(q) \in \{(2,012,02),(02,012,02),(012,12,12)\}.
\]
If the label of $q = 7/8 \to 5/6 \to 10B (\to 10B)^* \to 7/8$ is $\alpha_1 \alpha_2 \cdots \alpha_m$ with $m \geq 3$ such that $\alpha_1$ (resp. $\alpha_2$, $\alpha_m$) labels the edge $7/8 \to 5/6$ (resp. $5/6 \to 10B$, $10B \to 7/8$) and $\alpha_3 \cdots \alpha_{m-1}$ labels the loop $10B \to 10B$, then it is not difficult (though a bit long) to check that the following holds true:
\begin{eqnarray*}
	&	\alp(q) \in \{(2,012,02),(02,012,02)\}	& 	\\
	&	\Updownarrow	&	\\
	&	\begin{cases}
			\alpha_1 \alpha_2 \in \{ [1,02,2],[01,2,02]\} [1,01,2]	\\
			\alpha_3 \cdots \alpha_{m-2} \in \left\{ {[0,20,1]}^{2n}, {[02,12,2]}^n \mid n \in \N \right\}^*	\\
			\alpha_m = [2,012,02] \text{ or } (m \geq 4 \text{ and } \alpha_{m-1} \alpha_m = [0,20,1] [0,21,1])
		\end{cases}
\end{eqnarray*}
and
\begin{eqnarray*}
	&	\alp(q)  = [012,12,12]	& 	\\
	&	\Updownarrow	&	\\
	&	\begin{cases}
			\alpha_1 \alpha_2 \in \{ [1,02,2],[2,01,2]\} \left\{ [12^k0,2^{k+1}0,2^k0] \mid k \geq 0 \right\}	\\
			\alpha_3 \cdots \alpha_{m-1} \in \left\{ [01^k2,1^{k+1}2,1^k2] \mid k \geq 0 \right\}	\\
			\alpha_m \in \left\{ [0,2^k1,2^{k-1}1] \mid k \geq 2 \right\}
		\end{cases}.
\end{eqnarray*}

\bigskip

To conclude the whole proof, it remains to show that the label of any path that goes infinitely often through the four vertices or that ultimately stays in the subgraph  $\{1,5/6,7/8\}$ is almost primitive. This can be easily seen: any such path must contain infinitely many occurrences of finite paths $1 \to 7/8 \to 5/6$ and all these paths have a strongly primitive label.
\end{proof}

\section{Computation of length of paths in Rauzy graphs}	
\label{appendix: computation of length}

To complete the proof of Theorem~\ref{thm: 2n final}, we need to be able to compute some lengths in Rauzy graphs. However, when computing the $\S$-adic representation of our subshifts, we do not keep track of the order $n$ of $G_n$. Consequently, we cannot simply compute the desired Rauzy graph and count the number of edges in the paths we are interested in. Moreover, that technique would not be efficient since the Rauzy graphs are getting bigger and bigger, making them harder to compute. To avoid this problem, we will compute lengths using the morphisms already computed. In other words, if for instance $\tau$ is a morphism labelling an edge to the vertex $7/8$ and coding a loop (i.e., containing an exponent $k$ or $\ell$), we will express the lengths $|u_1|$, $|u_2|$, $|v_1|$ and $|v_2|$ using $\tau$ and morphisms preceding $\tau$ in the directive word. 

Let us introduce some notations. We let $(\gamma_{i_n})_{n \in \N}$ be the sequence of morphisms of Definition~\ref{definition: definition des morphismes} (and Remark~\ref{remark: sigme = identity}) and for all $n \geq 0$, we let $\gamma_{[0,n]}$ denote the morphism $\gamma_{i_0} \cdots \gamma_{i_n}$; thus it is the morphism coding the evolution from $G_0$ to $G_{i_n+1}$. For any two words (or paths) $u$ and $v$, we also let $\CPref(u,v)$ and $\CSuff(u,v)$ respectively denote the longest common prefix and suffix of $u$ and $v$.


The computation of lengths in Rauzy graphs is based on the following fact which is a direct consequence of the constructions. 
\begin{fact}
Let $G_{i_{n+1}}$ be a Rauzy graph of a minimal subshift whose first difference of complexity satisfies $1 \leq p(n+1)-p(n) \leq 2$ for all $n$. If $\gamma_{[0,n]}$ is the morphism coding the evolution from $G_0$ to $G_{i_{n+1}}$, then for all $x \in \{0,1,2\}$, we have
\[
	\gamma_{[0,n]}(x) = \lambda_{R,i_{n+1}} \circ \theta_{i_{n+1}}(x).
\]
\end{fact}

Observe that, since the sequence $(\alpha_n)_{n \in \N}$ of Theorem~\ref{thm: 2n final} is a contraction of $(\gamma_{i_n})_{n \in \N}$, this result can easily be translated using $(\alpha_n)_{n \in \N}$, provided that we consider the good indices $k$ for $\theta_k(x)$. On the other hand, it does not hold anymore if we replace $\gamma_{[0,n]}$ by $\Gamma_{[0,n]} = \Gamma_0 \cdots \Gamma_n$ where $(\Gamma_n)_{n \in \N}$ is as defined in Theorem~\ref{thm: 2n final}. We will also need the following lemma.

\begin{lemma}
\label{lemma: longueur des circuits}
Let $(X,T)$ be a subshift over $A$. For all words $u \in \fac{}{X}$, there is at most one return word $r$ to $u$ such that $|w| \leq \frac{|u|}{2}$. As a corollary, for all $n$ at most one $n$-circuit has for length at most $\frac{n}{2}$.
\end{lemma}

\begin{proof}
The last part of the lemma is a direct consequence of Remark~\ref{remark: lien entre n-circuit et mot de retour} (page~\pageref{remark: lien entre n-circuit et mot de retour}). 

Let $u \in \fac{}{X}$ and let $r$ be a return word to $u$ with minimal length. By definition, $u$ is suffix of $ur$. Therefore, if $|r| \leq \frac{|u|}{2}$, $r$ is a suffix of $u$ and we can write $u = r_{[j,|r|]} r^k$ with $k \in \N$, $k \geq 1$ and $j \in \{2, \ldots, |r|+1\}$. Consequently, $u$ is $|r|$-periodic, i.e., $u_{i+|r|} = u_i$ for all $i \in \{1, \dots ,|u|-|r|\}$. 

If there is another return word $s$ to $u$ such that $|s| \leq \frac{|u|}{2}$, we deduce similarly that $u$ is $|s|$-periodic. Moreover, since $|s| \geq |r|$, we have $s = r_{[t,|r|]} r^q$ with $q \in \N$, $q \geq 1$ and $t \in \{2,\dots,|r|+1\}$. By Fine and Wilf's Theorem (see Theorem 8.1.4 in~\cite{Lothaire}) the word $u$ is therefore also $p$-periodic with $p = gcd(|r|,|s|)$. Consequently, there is a word $v$ of length $p$ such that $u = v_{[i,|v|]} v^l$ with $l \geq 1$ and $i \in \{2,\dots,p+1\}$. We also have $r = v^m$ for an integer $m \geq 1$. Therefore, the word $u$ is suffix of $uv$ and does not occur more than twice in $uv$. So, by definition $v$ is a return word to $u$ and, by hypothesis on the length of $r$, we have $v = r$ hence $p=|r|$. Thus $s = r^q$ and there are $q+1$ occurrences of $u$ in $us$ (because $u = r_{[j,|r|]} r^k$). Consequently, $s$ is a return word to $u$ if and only if $s=r$.
\end{proof}

\subsection{Computation of $|u_1|$, $|u_2|$, $|v_1|$ and $|v_2|$}
\label{subsubsection: computation of lengths u v}

First let us compute the length of paths $u_1$, $u_2$, $v_1$ and $v_2$ in Rauzy graphs as represented in Figure~\ref{figure: rauzy graph of type 7,8}. As in Lemma~\ref{Lemma: graph of type 7 or 8} and Theorem~\ref{thm: 2n final}, we let $\k$ denote the maximal number of times that a circuit goes through the loop $v_2 u_2$. In case the graph is $G_{i_n+1}$, this corresponds to the number of times that the circuit $\theta_{i_n+1}(1)$ goes through that loop.

\begin{figure}[h!tbp]
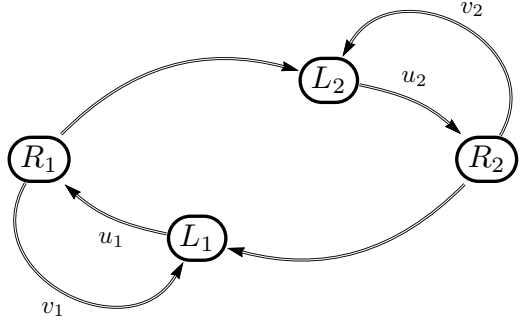

\centering
\scalebox{0.7}{
\begin{VCPicture}{(0,-0.5)(9,5)}
\ChgEdgeLabelScale{0.8}
\StateVar[R_1]{(0,2.5)}{R_1}
\StateVar[L_2]{(5.5,4)}{L_2}
\StateVar[R_2]{(8.5,2.5)}{R_2}
\StateVar[L_1]{(3,1)}{L_1}
\EdgeLineDouble
\LArcL{R_1}{L_2}{}
\VCurveR[]{angleA=-120,angleB=-120,ncurv=1.2}{R_1}{L_1}{v_1}
\VCurveR []{angleA=60,angleB=60,ncurv=1.2}{R_2}{L_2}{v_2}
\LArcL{R_2}{L_1}{}
\ArcL{L_2}{R_2}{u_2}
\ArcL{L_1}{R_1}{u_1}
\end{VCPicture}
}
\caption{Rauzy graphs of type 7 or 8}
\label{figure: rauzy graph of type 7,8}
\end{figure}

\subsubsection{Morphisms in Table~\ref{table: morphismes a partir de 2}}

Let us first consider the morphisms labelling the black edge from vertex $2$ to vertex $7/8$ in Figure~\ref{figure: modified graph of graphs}; they are listed in Table~\ref{table: morphismes a partir de 2}. 
The starting type of graph is represented in Figure~\ref{appendix type 2}.


\begin{enumerate}

	\item	$\gamma_{i_n} = [x,y^kzx,(y^{k-1}zx)]$ with $k \geq 2$ coming from the vertex $2$. The evolution corresponding to this morphism is represented in Figure~\ref{2 to 78 1} (page~\pageref{2 to 78 1}) with $U_{i_n+1}$ corresponding to the right special vertex on the top. We immediately obtain $|v_1| = |v_2| = 1$, $|u_1| = |\gamma_{[0,n-1]}(x)|-1$, $|u_2| = |\gamma_{[0,n-1]}(y)|-1$ and $\k = k-1$.

	\item	$\gamma_{i_n} = [x,zy^kx,(zy^{k-1}x)]$ with $k \geq 2$ coming from the vertex $2$. The evolution corresponding to this morphism is represented in Figure~\ref{2 to 78 1} (page~\pageref{2 to 78 1}) with $U_{i_n+1}$ corresponding to the right special vertex at the bottom. We immediately obtain $|v_1| = |v_2| = 1$, $|u_1| = |\gamma_{[0,n-1]}(x)|-1$, $|u_2| = |\gamma_{[0,n-1]}(y)|-1$ and $\k = k-1$.

	\item	$\gamma_{i_n} = [x,{(yz)}^kx,({(yz)}^{k-1}x)]$ with $k \geq 2$ coming from the vertex $2$. The evolution corresponding to this morphism is represented in Figure~\ref{2 to 78 2} (page~\pageref{2 to 78 2}) with $U_{i_n+1}$ corresponding to the right special vertex on the top. We immediately obtain $|v_1| = |v_2| = 1$, $|u_1| = |\gamma_{[0,n-1]}(x)|-1$, $|u_2| = |\gamma_{[0,n-1]}(yz)|-1$ and $\k = k-1$.

	\item	$\gamma_{i_n} = [xy,z^kxy,(z^{k-1}xy)]$ with $k \geq 2$ coming from the vertex $2$. The evolution corresponding to this morphism is represented in Figure~\ref{2 to 78 2} (page~\pageref{2 to 78 2}) with $U_{i_n+1}$ corresponding to the right special vertex at the bottom. We immediately obtain $|v_1| = |v_2| = 1$, $|u_1| = |\gamma_{[0,n-1]}(xy)|-1$, $|u_2| = |\gamma_{[0,n-1]}(z)|-1$ and $\k = k-1$.

	\item	$\gamma_{i_n} = [x,{(yz)}^kyx,({(yz)}^{k-1}yx)]$ with $k \geq 1$ coming from the vertex $2$. The evolution corresponding to this morphism is represented in Figure~\ref{2 to 78 3} (page~\pageref{2 to 78 3}) with $U_{i_n+1}$ corresponding to the right special vertex on the top. We immediately obtain $|v_1| = 1$, $|v_2| = |\gamma_{[0,n-1]}(z)|+1$, $|u_1| = |\gamma_{[0,n-1]}(x)|-1$, $|u_2| = |\gamma_{[0,n-1]}(y)|-1$ and $\k = k$.

	\item	$\gamma_{i_n} = [xy,z^ky,(z^{k-1}y)]$ with $k \geq 2$ coming from the vertex $2$. The evolution corresponding to this morphism is represented in Figure~\ref{2 to 78 3} (page~\pageref{2 to 78 3}) with $U_{i_n+1}$ corresponding to the right special vertex at the bottom. We immediately obtain $|v_1| = |\gamma_{[0,n-1]}(x)|+1$, $|v_2| = 1$, $|u_1| = |\gamma_{[0,n-1]}(y)|-1$, $|u_2| = |\gamma_{[0,n-1]}(z)|-1$ and $\k = k-1$.

	\item	$\gamma_{i_n} = [z^{\ell}x,yz^kx,yz^{k-1}x]$ with $k-1 > \ell \geq 1$ coming from the vertex $2$. The sequence of evolutions corresponding to that morphisms is the following. First, the graph evolves to a graph of type 4 as in Figure~\ref{subfig: 2 to 4} (page~\pageref{subfig: 2 to 4}) such that $\theta_{i_n+1}(0)$ and $\theta_{i_n+1}(1)$ go respectively $k-1$ and $\ell-1$ times through the loop. Then, the graph becomes a graph as in Figure~\ref{appendix type 4} and it evolves $\ell - 1$ times as represented in Figure~\ref{4 to 4}. Finally, it evolves to a a graph of type 7 or 8 as in Figure~\ref{4 to 78}. It is obviously seen that we have $|v_2| = 1$, $|u_2| = |\gamma_{[0,n-1]}(z)|-1$ and $|u_1| + |v_1| = |\gamma_{[0,n-1]}(z^{\ell}x)|$. Moreover, the path in Figure~\ref{appendix type 4} that will become $u_1$ corresponds to the segment which is not curved. After the first evolution (from 2 to 4), this path has for length $|\gamma_{[0,n-1]}(z)|$ (check in Figure~\ref{subfig: 2 to 4}) and at each evolution to a graph of type 4 (as in Figure~\ref{4 to 4}), its length increases by $|\gamma_{[0,n-1]}(z)|$. With the last evolution, we obtain $|u_1| = \ell |\gamma_{[0,n-1]}(z)| +1$. Finally we can check that $\k = k - \ell - 1$.

	\item	$\gamma_{i_n} = [yz^{\ell}x,z^kx,z^{k-1}x]$ with $k-1 > \ell \geq 1$ coming from the vertex $2$. The computation is the same as for the previous morphism. In this case we obtain $|v_2| = 1$, $|u_2| = |\gamma_{[0,n-1]}(z)|-1$, $|u_1|+|v_1| = |\gamma_{[0,n-1]}(yz^{\ell}x)|$, $|u_1| = |\gamma_{[0,n-1]}(y)| + \ell |\gamma_{[0,n-1]}(z)| +1$ and $\k = k - \ell - 1$.

	\item	$\gamma_{i_n} = [y(xy)^{\ell}z,(xy)^kz,(xy)^{k-1}z]$ with $k-1 > \ell \geq 1$ coming from the vertex $2$. The sequence of evolutions corresponding to that morphisms is the following. First, the graph evolves to a graph of type 10 as in Figure~\ref{2 to 10} (page~\pageref{2 to 10}) such that $\theta_{i_n+1}(0)$ and $\theta_{i_n+1}(1)$ go respectively $k-1$ and $\ell$ times through the loop. Then, the graph becomes a graph as in Figure~\ref{appendix type 10} and it evolves $2 \ell$ times as represented in Figure~\ref{10 to 10}. Finally, it evolves to a a graph of type 7 or 8 as in Figure~\ref{10 to 78}. It is obviously seen that we have $|u_1|+|v_1| = |\gamma_{[0,n]}(0)|$ and $|u_2| + |v_2| = |\gamma_{[0,n-1]}(xy)|$. In Figure~\ref{appendix type 10}, the path that will become $u_1$ is the segment from the bispecial vertex to the right special vertex. Once the graph has evolved as in Figure~\ref{2 to 10}, it has for length $|\gamma_{[0,n-1]}(z)|$ and we can see in Figure~\ref{10 to 10} that, during the $2 \ell$ evolutions to graphs of type 10, it keeps the same length. With the final evolution as in Figure~\ref{10 to 78}, we obtain $|u_1| = |\gamma_{[0,n-1]}(z)-1|$. For $|u_2|$ and $|v_2|$, we see in Figure~\ref{appendix type 10} that the path that will become $u_2$ is the path from the left special vertex to the bispecial vertex. Once the graph has evolved as in Figure~\ref{2 to 10}, we also see that it has for length $|\gamma_{[0,n-1]}(x)|$. Then, when the graph evolves as in Figure~\ref{10 to 10}, we see that the path that will become $u_2$ and $v_2$ always keep the same length but are exchanged at each time. However, since this evolution occurs $2 \ell$ times, we obtain (with the last evolution) $|u_2| = |\gamma_{[0,n-1]}(x)-1|$. We finally have $\k = k - \ell -1$.

	\item	$\gamma_{i_n} = [(xy)^kz,y(xy)^{\ell}z,y(xy)^{\ell-1}z]$ with $\ell > k \geq 1$ coming from the vertex $2$. The computation is the same as for the previous morphism. We still have $|u_1|+|v_1| = |\gamma_{[0,n]}(0)|$, $|u_2| + |v_2| = |\gamma_{[0,n-1]}(xy)|$ and $|u_1| = |\gamma_{[0,n-1]}(z)-1|$. However, once the graph has evolved as in Figure~\ref{2 to 10}, it evolves an odd number of times as in Figure~\ref{10 to 10} ($2(k-1)+1$ times). Consequently we have $|v_2| = |\gamma_{[0,n-1]}(x)-1|$ instead of $|u_2|$. We also have $\k = \ell - k$.

\end{enumerate}

\subsubsection{Morphisms in Table~\ref{table: morphismes a partir de V_i}}

Now let us consider the morphisms labelling the black edges from component $C_2$ to vertex $7/8$ in Figure~\ref{figure: modified graph of graphs}; they are listed in Table~\ref{table: morphismes a partir de V_i}. 
The starting type of graph is represented in Figure~\ref{appendix type 3}. 

For that kind of evolutions, we need to know the length of the path from the left special vertex to the right special vertex in Figure~\ref{appendix type 3}. Indeed, for instance in Figure~\ref{3 to 78}, we see that this path will become either $u_1$ or $u_2$, depending on the choice of the starting vertex $U_{i_n+1}$. This is achieved by the following lemma.


\begin{lemma}
\label{lemma: length type 3}
Let $G_{i_n}$ be a Rauzy graph of type 3 and let $\gamma_{[0,n-1]}$ be the morphism coding the evolution from $G_0$ to $G_{i_n}$. Suppose that $\{x,y,z\} = \{0,1,2\}$ and that $\theta_{i_n}(x)$ is the top loop in Figure~\ref{appendix type 3}. Let also $M$ be the length of the longest $i_{n+1}$-circuit. If $i$ and $j$ are such that $\min \{ |\gamma_{[0,n-1]}(x^i)|, |\gamma_{[0,n-1]}(y^j)|\} \geq 2M$, then the path from the left special vertex to the bispecial vertex has for length
\[
	|\CSuff(\gamma_{[0,n-1]}(y),\gamma_{[0,n-1]}(z))| - |\CSuff(\gamma_{[0,n-1]}(x^i),\gamma_{[0,n-1]}(y^j))|.
\]
\end{lemma}

\begin{proof}
Indeed, by Proposition~\ref{prop: path in Rauzy graphs} (page~\pageref{prop: path in Rauzy graphs}) we immediately deduce that the length of the path from the left special vertex to the bispecial vertex is
\[
	|\CSuff(\gamma_{[0,n-1]}(y),\gamma_{[0,n-1]}(z))| - i_n.
\]
Consequently, it suffices to prove that $i_n = |\CSuff(\gamma_{[0,n-1]}(x^i),\gamma_{[0,n-1]}(y^j))|$. By Lemma~\ref{lemma: longueur des circuits} we know that $2M$ is greater than $i_n$ and that so are $|\gamma_{[0,n-1]}(x^i))|$ and $|\gamma_{[0,n-1]}(y^j))|$. Consequently, Proposition~\ref{prop: path in Rauzy graphs} implies that both $\gamma_{[0,n-1]}(x^i))$ and $\gamma_{[0,n-1]}(y^j))$ admit the bispecial vertex $B$ as a suffix. Moreover, it is easily seen that if they have a longer common suffix, $B$ would not be bispecial so the result holds.
\end{proof}

In this section, we let $q$ denote the path from the left special vertex to the bispecial vertex in Figure~\ref{appendix type 3}.

\begin{enumerate}

	\item	$\gamma_{i_n} = [i,xy^ki,xy^{k-1}i]$ with $k \geq 1$ coming from the vertex $V_i$, $i \in \{0,1,2\}$. The evolution corresponding to that morphism is represented in Figure~\ref{3 to 78} with vertex $U_{i_n+1}$ corresponding to the right special vertex on the top. In that case we immediately have $|u_1| = |\gamma_{[0,n-1]}(i)| - 1$, $|v_1| = 1$, $|u_2| + |v_2| = |\gamma_{[0,n-1]}(y)|$ and $|u_2| = |q|-1$. We also have $\k = k$.

	\item	$\gamma_{i_n} = [x,i^ky,i^{k-1}y]$ with $k \geq 2$ coming from the vertex $V_i$, $i \in \{0,1,2\}$. The evolution corresponding to that morphism is represented in Figure~\ref{3 to 78} with vertex $U_{i_n+1}$ corresponding to the right special vertex at the bottom. In that case we immediately have $|u_2| = |\gamma_{[0,n-1](i)}| - 1$, $|v_2| = 1$, $|u_1| + |v_1| = |\gamma_{[0,n-1]}(x)|$ and $|u_1| = |q|-1$. We also have $\k = k-1$.

	\item	$\gamma_{i_n} = [xy^{\ell}i,y^ki,y^{k-1}i]$ with $k-1 > \ell \geq 0$ coming from the vertex $V_i$, $i \in \{0,1,2\}$. The sequence of evolutions corresponding to that morphism is the following. First the graph evolves to a graph of type 10 as in Figure~\ref{3 to 10} with starting vertex corresponding to the right special vertex on the top. Then, the graph becomes a graph as in Figure~\ref{appendix type 10} and evolves $2 \ell$ times to graphs of type 10 as in Figure~\ref{10 to 10}. Finally, the graph evolves as in Figure~\ref{10 to 78}. For this morphism, we directly see that $|u_1| + |v_1| = |\gamma_{[0,n]}(0)|$ and that $|u_2| + |v_2| = |\gamma_{[0,n-1]}(y)|$. We also see in Figure~\ref{appendix type 10} that the path that will become $u_2$ is the path from the left special vertex to the bispecial vertex. Once the graph has evolved as in Figure~\ref{3 to 10}, we see that this path has for length $|\gamma_{[0,n-1]}(y)| - |q| - 1$. Then, we see that its length is unchanged after $2$ evolutions as in Figure~\ref{10 to 10} (such an evolution exchanged the curved part of the loop in Figure~\ref{appendix type 10} with the other part). Consequently, we obtain $|u_2| = |\gamma_{[0,n-1]}(y)| - |q| - 1$. Next, in Figure~\ref{appendix type 10} we see that the path that will become $u_1$ is the segment from the bisepcial vertex to the right special vertex. Once the graph has evolved as in Figure~\ref{3 to 10}, we see that it has for length $|\gamma_{[0,n-1]}(i)|$. We also see in Figure~\ref{10 to 10} that it keeps the same length while these $2 \ell$ evolutions. While the last evolution as in Figure~\ref{10 to 78}, we have $|u_1| = |\gamma_{[0,n-1]}(i)|-1$. Finally, we have $\k = k - \ell - 1$.

	\item	$\gamma_{i_n} = [y^ki,xy^{\ell}i,xy^{\ell-1}i]$ with $\ell > k \geq 1$ coming from the vertex $V_i$, $i \in \{0,1,2\}$. The computation is the same as for the previous morphism. In this case we still have $|u_1| + |v_1| = |\gamma_{[0,n]}(0)|$, $|u_2| + |v_2| = |\gamma_{[0,n-1]}(y)|$ and $|u_1| = |\gamma_{[0,n-1]}(i)|-1$. However, in this case the graph evolves an odd number of times as in Figure~\ref{10 to 10} ($2(k-1)+1$ times) so we have $|v_2| = |\gamma_{[0,n-1]}(y)| - |q| - 1$ instead of $|u_2|$. We also have $\k = \ell - k$.

\end{enumerate}

\subsubsection{Morphisms in Table~\ref{table: morphismes a partir de 4B}}

Now let us consider the morphisms labelling the black edges from component $C_3$ to vertex $7/8$ in Figure~\ref{figure: modified graph of graphs}; they are listed in Table~\ref{table: morphismes a partir de 4B}. 
The starting type of graph is represented in Figure~\ref{appendix type 4}.

\begin{enumerate}

	\item	$\gamma_{i_n} = [0,x^ky0,x^{k-1}y0]$ with $k \geq 1$ coming from the vertex $4B$. The evolution corresponding to that morphism is represented in Figure~\ref{4 to 78}. In this case we immediately obtain the lengths $|u_1| = |\gamma_{[0,n-1]}(0)-1|$, $|v_1| = 1$, $|u_2| + |v_2| = |\gamma_{[0,n-1]}(x)|$, $|u_2| = |\CPref(\gamma_{[0,n-1]}(x),\gamma_{[0,n-1]}(y))| - 1$ and $\k = k$.

	\item	$\gamma_{i_n} = [x^{\ell}y,0x^ky,0x^{k-1}y]$ with $k-1 > \ell \geq 0$ coming from the vertex $4B$. The sequence of evolutions corresponding to that morphism is the following: first the graph evolves to graph of type 4 as in Figure~\ref{4 to 4 bis}. Then it becomes a graph as in Figure~\ref{appendix type 4} such that the starting vertex is not the bispecial vertex. It then evolves $\ell$ times as in Figure~\ref{4 to 4} and finally evolves as in Figure~\ref{4 to 78}. It is obviously seen that we have $|u_1| + |v_1| = |\gamma_{[0,n]}(0)|$, $|u_2| = |\gamma_{[0,n-1]}(x)|-1$ and that $|v_2|=1$. We also see that the path in Figure~\ref{appendix type 4} that will become $u_1$ is the segment from the bispecial vertex to the right special vertex. We see in Figure~\ref{4 to 4} that, during this evolution, it always keeps the same length. So, its has the same length than the path in Figure~\ref{4 to 4 bis} from the leftmost right special vertex to the right special vertex on the top. This path has for length $|\gamma_{[0,n-1]}(y)| - |\CPref(\gamma_{[0,n-1]}(x),\gamma_{[0,n-1]}(y))|$. With the last evolution (as in Figure~\ref{4 to 78}), we finally obtain $|u_1| = |\gamma_{[0,n-1]}(y)| - |\CPref(\gamma_{[0,n-1]}(x),\gamma_{[0,n-1]}(y))| - 1$. We also have $\k = k-1 - \ell$.

	\item	$\gamma_{i_n} = [0x^{\ell}y,x^ky,x^{k-1}y]$ with $k-1 > \ell \geq 0$ coming from the vertex $4B$. The computation and the lengths are exactly the same as for the previous morphism.

	\item	$\gamma_{i_n} = [(x0)^{\ell}y,0(x0)^ky,0(x0)^{k-1}y]$ with $k > \ell \geq 0$ coming from the vertex $4B$. The sequence of evolutions corresponding to that morphism is the following. First the graph evolves to a graph of type 10 as in Figure~\ref{4 to 10} and becomes a graph as in Figure~\ref{appendix type 10} such that the starting vertex is not the bispecial one. Then, the graph evolves $2\ell$ times as in Figure~\ref{10 to 10} and it finally evolves as in Figure~\ref{10 to 78}. We immediately have $|u_1|+|v_1| = |\gamma_{[0,n]}(0)|$ and $|u_2|+|v_2| = |\gamma_{[0,n-1]}(x0)|$. Moreover, we see that the path in Figure~\ref{appendix type 10} that will become $u_1$ is the segment from the bispecial vertex to the right special vertex. Once the graph has evolved as in Figure~\ref{4 to 10}, we see that this path has for length $|\CPref(\gamma_{[0,n-1]}(x),\gamma_{[0,n-1]}(y))|$. Then, we see in Figure~\ref{10 to 10} that after 2 such evolutions, this path still have the same length (the two segments starting from the right special vertex which is not bispecial get simply exchanged). Consequently, it still have the same length after the $2\ell$ evolutions to graphs of type 10. With the last evolution as in Figure~\ref{10 to 78} we obtain $|u_1| = |\CPref(\gamma_{[0,n-1]}(x),\gamma_{[0,n-1]}(y))| - 1$. We see that the paths in Figure~\ref{appendix type 10} that will become $u_2$ and $v_2$ are respectively the path $q$ from the left special vertex to the bispecial vertex and the path $q'$ from the bispecial vertex to the left special vertex. Once the graph has evolved as in Figure~\ref{4 to 10}, the path that will become $q$ has for length $|\gamma_{[0,n-1]}(0)|$. Then, at each evolution as in Figure~\ref{10 to 10}, $q$ and $q'$ are exchanged. As there is an even number of such evolutions, we finally get (after the last evolution as in Figure~\ref{10 to 78}) $|u_2| = |\gamma_{[0,n-1]}(0)|-1$. We also have $\k = k - \ell$.

	\item	$\gamma_{i_n} = [0(x0)^ky,(x0)^{\ell}y,(x0)^{\ell-1}y]$ with $\ell -1 > k \geq 0$ coming from the vertex $4B$. The computation is the same as for the previous morphism. In this case we still have $|u_1|+|v_1| = |\gamma_{[0,n]}(0)|$, $|u_2|+|v_2| = |\gamma_{[0,n-1]}(x0)|$ and $|u_1| = |\CPref(\gamma_{[0,n-1]}(x),\gamma_{[0,n-1]}(y))| - 1$. For $u_2$, in this case the graph evolves an odd number of times as in Figure~\ref{10 to 10} so we have $|v_2| = |\gamma_{[0,n-1]}(0)|-1$ instead of $|u_2|$. We also have $\k = \ell - k - 1$.
	
\end{enumerate}

\subsubsection{Morphisms in Table~\ref{table: label C_4_1}}

Now let us consider the morphisms labelling the black edges from component $C_4$ to vertex $7/8$ in Figure~\ref{figure: modified graph of graphs}; they are listed in Table~\ref{table: label C_4_1}. 
The starting types of graph are represented in Figure~\ref{appendix type 4}, Figure~\ref{appendix type 5}, Figure~\ref{appendix type 6} and Figure~\ref{appendix type 10}.

To compute lengths in this component, we have to be careful with the vertex $5/6$. Indeed, this vertex corresponds to the evolution from a graph of type 5 or 6 depending on the length of $p_1$ and $p_2$ in Figure~\ref{figure: graph with no loop'} (page~\pageref{figure: graph with no loop'}). To clearly explain how graphs evolve and how we compute lengths, we will always consider that the starting graph is of type 6. The reader is invited to check that all computations also hold when the graph is of type 5.

In the computations given below, we sometimes need to know the order of the starting Rauzy graph when it is of type 10. For this type of graph, we also need to know the length of the simple path from the left special vertex to the bispecial vertex. These information are given in the following lemma whose proof is similar to the proof of Lemma~\ref{lemma: length type 3} and left to the reader.

\begin{lemma}
\label{lemma: type 10 ordre et long}
Let $G_{i_n}$ be a Rauzy graph of type 10 as in Figure~\ref{appendix type 10}. Let $\gamma_{[0,n-1]}$ be the morphism coding the evolution from $G_0$ to $G_{i_n}$ and suppose that $U_{i_n}$ is the bispecial vertex. If $x \in \{0,1,2\}$ is such that $|\theta_{i_n}(x)| = \max \{ |\theta_{i_n}(i)| \mid i \in \{0,1,2 \} \}$ and if $l_0$, $l_1$ and $l_2$ are the smallest positive integers such that
\[
	\min \{ l_i |\gamma_{[0,n-1]}(i)| \mid i \in \{0,1,2 \} \} \geq 2 |\gamma_{[0,n-1]}(x)|,
\]
then we have
\[
	i_n = \left| \CSuff\left( \gamma_{[0,n-1]}(1^{l_1}), \gamma_{[0,n-1]}(2^{l_2}) \right) \right|.
\]
Moreover, the simple path from the left special vertex to the bispecial vertex in $G_{i_n}$ has for length
\[
	\left| \CSuff\left( \gamma_{[0,n-1]}(0^{l_0}), \gamma_{[0,n-1]}(1^{l_1}) \right) \right| - i_n.
\]
\end{lemma}

Now let us compute the lengths $|u_1|$, $|u_2|$, $|v_1|$ and $|v_2|$.

\begin{enumerate}

	\item $\gamma_{i_n} = [x,y^kx,y^{k-1}x]$ with $k \geq 2$ coming from the vertex 1 or from the vertex $5/6$. The evolutions corresponding to that morphism are represented in Figure~\ref{1 to 78} and in Figure~\ref{6 to 78}. We can easily see that $|u_1| = |\gamma_{[0,n-1]} (x)|-1$, $|v_1| = 1$, $|u_2| = |\gamma_{[0,n-1]} (y)|-1$ and $|v_2| = 1$. We also have $\k = k-1$.

	\item $\gamma_{i_n} = [1,0^k2,(0^{k-1}2)]$ with $k \geq 1$ coming from the vertex $5/6$. For this evolution, we directly have the lengths $|u_1|+|v_1| = |\gamma_{[0,n]}(0)|$, $|u_2|+|v_2| = |\gamma_{[0,n-1]}(0)|$, $|u_2| = |\CPref ( \gamma_{[0,n-1]} (0), \gamma_{[0,n-1]} (2) )| - 1$, $|u_1| = |\gamma_{[0,n-1]} (2) | - |\CPref ( \gamma_{[0,n-1]} (0), \gamma_{[0,n-1]} (2) )| - 1$ and $\k = k$.

	\item $\gamma_{i_n} = [2^{\ell}0,1 2^k 0,(1 2^{k-1}0)]$ with $k > \ell \geq 0$ coming from the vertex $5/6$. The sequence of evolutions corresponding to that morphism is the following. First the graph evolves to a graph of type 10 as in Figure~\ref{6 to 10} and becomes a graph as in Figure~\ref{appendix type 10} such that the starting vertex is not the bispecial one. Then, the graph evolves $2\ell$ times as in Figure~\ref{10 to 10} and it finally evolves as in Figure~\ref{10 to 78}. We immediately have $|u_1|+|v_1| = |\gamma_{[0,n]}(0)|$ and $|u_2|+|v_2| = |\gamma_{[0,n-1]}(2)|$. Moreover, we see that the path in Figure~\ref{appendix type 10} that will become $u_1$ is the segment from the bispecial vertex to the right special vertex. Once the graph has evolved as in Figure~\ref{6 to 10}, we see that this path has for length $|\gamma_{[0,n-1]}(0)|-|\CPref(\gamma_{[0,n-1]}(0),\gamma_{[0,n-1]}(1))|$. Then, we see in Figure~\ref{10 to 10} that after two such evolutions, this path still have the same length (because with such an evolution, the two segments starting from the right special vertex which is not bispecial simply get exchanged). Consequently, it still have the same length after the $2\ell$ evolutions to graphs of type 10. With the last evolution as in Figure~\ref{10 to 78} we obtain $|u_1| = |\gamma_{[0,n-1]}(0)|-|\CPref(\gamma_{[0,n-1]}(0),\gamma_{[0,n-1]}(2))| - 1$. For $u_2$ and $v_2$ we see that the paths in Figure~\ref{appendix type 10} that will become them are respectively the path $q$ from the left special vertex to the bispecial vertex and the path $q'$ from the bispecial vertex to the left special vertex. Once the graph has evolved as in Figure~\ref{6 to 10}, the path that will become $q$ has for length $|\gamma_{[0,n-1]}(2)| - |\CPref(\gamma_{[0,n-1]}(0),\gamma_{[0,n-1]}(2))|$. Then, at each evolution as in Figure~\ref{10 to 10}, $q$ and $q'$ are exchanged. Since there are an even number of such evolutions, we finally get (after the last evolution as in Figure~\ref{10 to 78}) $|u_2| = |\gamma_{[0,n-1]}(2)| - |\CPref(\gamma_{[0,n-1]}(0),\gamma_{[0,n-1]}(2))|-1$. We also have $\k = k - \ell$.

	\item $\gamma_{i_n} = [1 2^k 0,2^{\ell}0, (2^{\ell-1}0)]$ with $\ell > k+1 \geq 1$ coming from the vertex $5/6$. The computation is the same as for the previous morphism. In this case we still have $|u_1|+|v_1| = |\gamma_{[0,n]}(0)|$, $|u_2|+|v_2| = |\gamma_{[0,n-1]}(2)|$ and $|u_1| = |\gamma_{[0,n-1]}(0)|-|\CPref(\gamma_{[0,n-1]}(0),\gamma_{[0,n-1]}(2))| - 1$. For $u_2$, in this case the graph evolves an odd number of times as in Figure~\ref{10 to 10} so we have $|v_2| = |\gamma_{[0,n-1]}(2)| - |\CPref(\gamma_{[0,n-1]}(0),\gamma_{[0,n-1]}(2))|-1$ instead of $|u_2|$. We also have $\k = \ell - k - 1$.

	\item $\gamma_{i_n} = [0,2^k 1,2^{k-1}1]$ with $k \geq 1$ coming from the vertex $10B$. The evolution corresponding to that morphism is represented in Figure~\ref{10 to 78}. We immediately see that $|u_1| + |v_1| = |\gamma_{[0,n]}(0)|$, $|u_1| + |v_1| = |\gamma_{[0,n-1]}(2)|$, $|u_2| = |\CPref(\gamma_{[0,n-1]}(1),\gamma_{[0,n-1]}(2))|-1$. Moreover, by Lemma~\ref{lemma: type 10 ordre et long} we have (with the same notation) $|u_1| = \left| \CSuff\left( \gamma_{[0,n-1]}(0^{l_0}), \gamma_{[0,n-1]}(1^{l_1}) \right) \right| - \left| \CSuff\left( \gamma_{[0,n-1]}(1^{l_1}), \gamma_{[0,n-1]}(2^{l_2}) \right) \right| -1$. We also have $\k = k$.

	\item $\gamma_{i_n} = [1^{\ell}2,0 1^k 2,(0 1^{k-1}2)]$ with $k > \ell \geq 0$ coming from the vertex $10B$. The sequence of evolutions corresponding to that morphism is the following. First the graph evolves to a graph of type 10 as in Figure~\ref{10 to 10bis} and becomes a graph as in Figure~\ref{appendix type 10} such that the starting vertex is not the bispecial one. Then, the graph evolves $2\ell$ times as in Figure~\ref{10 to 10} and it finally evolves as in Figure~\ref{10 to 78}. We immediately have $|u_1|+|v_1| = |\gamma_{[0,n]}(0)|$ and $|u_2|+|v_2| = |\gamma_{[0,n-1]}(1)|$. Moreover, we see that the path in Figure~\ref{appendix type 10} that will become $u_1$ is the segment from the bispecial vertex to the right special vertex. Once the graph has evolved as in Figure~\ref{10 to 10bis}, we see that this path has for length $|\gamma_{[0,n-1]}(2)|-|\CPref(\gamma_{[0,n-1]}(1),\gamma_{[0,n-1]}(2))|$. Then, we see in Figure~\ref{10 to 10} that after two such evolutions, this path still have the same length (because with such an evolution, the two segments starting from the right special vertex which is not bispecial simply get exchanged). Consequently, it still has the same length after the $2\ell$ evolutions to graphs of type 10. With the last evolution as in Figure~\ref{10 to 78} we obtain $|u_1| = |\gamma_{[0,n-1]}(2)|-|\CPref(\gamma_{[0,n-1]}(1),\gamma_{[0,n-1]}(2))| - 1$. We see in Figure~\ref{appendix type 10} that the paths that will become $u_2$ and $v_2$ are respectively the path $q$ from the left special vertex to the bispecial vertex and the path $q'$ from the bispecial vertex to the left special vertex. Once the graph has evolved as in Figure~\ref{10 to 10bis}, we know from Lemma~\ref{lemma: type 10 ordre et long} that $q$ has for length $\left| \CSuff\left( \gamma_{[0,n-1]}(0^{l_0}), \gamma_{[0,n-1]}(1^{l_1}) \right) \right| - \left| \CSuff\left( \gamma_{[0,n-1]}(1^{l_1}), \gamma_{[0,n-1]}(2^{l_2}) \right) \right|$. Then, at each evolution as in Figure~\ref{10 to 10}, $q$ and $q'$ are exchanged. As there are an even number of such evolutions, we finally get (after the last evolution as in Figure~\ref{10 to 78}) $|u_2| = \left| \CSuff\left( \gamma_{[0,n-1]}(0^{l_0}), \gamma_{[0,n-1]}(1^{l_1}) \right) \right| - \left| \CSuff\left( \gamma_{[0,n-1]}(1^{l_1}), \gamma_{[0,n-1]}(2^{l_2}) \right) \right|-1$. We also have $\k = k - \ell$.

	\item $\gamma_{i_n} = [0 1^k 2,1^{\ell}2,(1^{\ell-1}2)]$ with $\ell > k+1 \geq 1$ coming from the vertex $10B$. The computation is the same as for the previous morphism. In this case we still have $|u_1|+|v_1| = |\gamma_{[0,n]}(0)|$, $|u_2|+|v_2| = |\gamma_{[0,n-1]}(1)|$ and $|u_1| = |\gamma_{[0,n-1]}(2)|-|\CPref(\gamma_{[0,n-1]}(1),\gamma_{[0,n-1]}(2))| - 1$. For $u_2$, in this case the graph evolves an odd number of times as in Figure~\ref{10 to 10} so we have $|v_2| = \left| \CSuff\left( \gamma_{[0,n-1]}(0^{l_0}), \gamma_{[0,n-1]}(1^{l_1}) \right) \right| - \left| \CSuff\left( \gamma_{[0,n-1]}(1^{l_1}), \gamma_{[0,n-1]}(2^{l_2}) \right) \right|-1$ instead of $|u_2|$. We also have $\k = \ell - k - 1$.

\end{enumerate}

\subsection{Computation of $|p_1|$ and $|p_2|$}
\label{subsubsection: computation of lengths p}

The aim of this section is to compute the length of the paths $p_1$ and $p_2$ of Figure~\ref{figure: graph with no loop'} when evolving to such a graph, i.e., when considering an edge to the vertex $5/6$ in Figure~\ref{Figure: graph C_4}. These lengths do not only depend on the last morphism applied but on a finite number of morphisms. First, the next lemma shows how to compute these lengths when passing through the vertex $7/8$ in Figure~\ref{Figure: graph C_4}. The other cases will be particular cases of this one. Indeed, morphisms labelling the loop on vertex $5/6$ in Figure~\ref{Figure: graph C_4} are simply compositions
of the morphism $[1,0^k2,0^{k-1}2]$ (labelling the edge from $5/6$ to $7/8$) with a morphism in $\{ [0x,y,0y], [x,0y,y] \}$ (labelling the edge from $7/8$ to $5/6$). In other words, it simply corresponds to the case $h = 0$ in Lemma~\ref{lemma: length of p1 and p2} below. For morphisms labelling the edge from $10B$ to $5/6$ in Figure~\ref{Figure: graph C_4}, the reasoning is the same but this time, the morphisms labelling the edge from $10B$ to $5/6$ are compositions of the morphism $[0,2^k1,2^{k-1}1]$ (labelling the edge from $10B$ to $7/8$) with a morphism in $\{ [0x,y,0y], [x,0y,y] \}$ (labelling the edge from $7/8$ to $5/6$).

\begin{lemma}
\label{lemma: length of p1 and p2}
Let $G_{i_{n-1}+1}$ be a Rauzy graph as represented in Figure~\ref{figure: graph with 2 loops'} (page~\pageref{figure: graph with 2 loops'}) and let $\gamma_{[0,n-1]}$ be the morphism coding the evolution from $G_0$ to $G_{i_{n-1}+1}$ (so to $G_{i_n}$). Suppose that $U_{i_{n-1}+1}$ corresponds to the vertex $R_1$ in Figure~\ref{figure: graph with 2 loops'} and that the circuit $\theta_{i_{n-1}+1}(1)$ goes exactly $k$ times through the loop $v_2 u_2$.

Let $\ell$ be the unique integer such that 
\[
	|u_1| + (\ell-1) (|u_1| + |v_1|) < |u_2| + (k-1) (|u_2| + |v_2|) \leq |u_1| + \ell (|u_1| + |v_1|)
\]
and let $h$ be the greatest integer such that for all $r \in \{0,\dots,h-1\}$, $\gamma_{i_{n+r}} =[0,10,20]$.
Suppose that $\gamma_{i_{n+h}}$ labels the edge from $7/8$ to $5/6$ (so belongs to $\{ [0 x,y,(0 y)], [x,0 y,(y)] \mid \{x,y\} = \{1,2\} \}$), then $G_{i_{n+h}+1}$ is a graph as represented in Figure~\ref{figure: graph with no loop'} (page~\pageref{figure: graph with no loop'}) and the lengths of $p_1$ and $p_2$ are as follows.

If $h < \ell$, then for $k' = \min \{ i \in \N \mid |u_2| + i (|u_2| + |v_2|) \geq |u_1| + h (|u_1| + |v_1|) \}$, we have
\begin{eqnarray*}
|p_1| 	&=& \left| \CPref\left( \gamma_{[0,n-1]}(1),\gamma_{[0,n-1]}(2) \right) \right| - (k-1 - k') (|u_2| + |v_2|) \\
		& & - (|u_2| + k'(|u_2|+|v_2|) - (|u_1| + h(|u_1|+|v_1|))) - 1	\\
|p_2| 	&=& |\gamma_{[0,n-1]}(2)| - \left| \CPref \left(\gamma_{[0,n-1]}(1),\gamma_{[0,n-1]}(2) \right) \right| - 1
\end{eqnarray*}
and if $h \geq \ell$, we have
\begin{eqnarray*}
|p_1| 	&=& \left| \CPref( \gamma_{[0,n-1]}(1),\gamma_{[0,n-1]}(2) ) \right| -1	\\
|p_2| 	&=& \left| \gamma_{[0,n-1]}(2) \right| - \left| \CPref( \gamma_{[0,n-1]}(1),\gamma_{[0,n-1]}(2) ) \right|  \\
		& & - (|u_1| + \ell(|u_1|+|v_1|) - (|u_2| + (k-1)(|u_2|+|v_2|))) - 1.
\end{eqnarray*}
\end{lemma}

\begin{proof}
Let us recall notation introduced in the proof of Lemma~\ref{Lemma: graph of type 7 or 8}. For all non-negative integers $r$ and $s$, $B_1(r)$ and $B_2(s)$ are respectively the words $\lambda(u_1 (v_1u_1)^r)$ and $\lambda(u_2 (v_2u_2)^s)$. For $s \in \{0,\dots,k-1\}$, $B_2(s)$ is a bispecial vertex in $G_{|B_2(s)|}$ and $B_2(k)$ does not belong to the language of the considered subshift. Also, for all non-negative integers $r$, if $B_1(r)$ is in the language of the considered subshift, then it is a bispecial vertex in $G_{|B_1(r)|}$.

Now let us determine the sequence of evolutions corresponding to the sequence of morphisms $(\gamma_{i_m})_{n < m \leq n+h+1}$. The graph $G_{i_n+1}$ will evolve to a graph of type 7 or 8 depending on $|u_1|$ and $|v_1|$. Thanks to Lemma~\ref{lemma: decomposition du mot directeur} we can suppose without loss of generality that it evolves to a graph of type 7.

Let us start studying the behaviours of vertices $B_2(s)$. The hypothesis on $\theta_{i_{n-1}+1}(1)$ implies that for all $s \in \{0,\dots,k-2\}$, $B_2(s)$ will explode as represented in Figure~\ref{figure: $j < k-1$} (page~\pageref{figure: $j < k-1$}). Then, the hypothesis on $\gamma_{i_{n+h}}$ implies that $B_2(k-1)$ will explode as in Figure~\ref{figure: $j = k-1$ and 3 circuits} (because there are three distinct letters in its images). 

Now let us study the behaviours of vertices $B_1(r)$. By constructions of the morphisms $\gamma_{i_m}$, for $r \in \{0,\dots,h\}$, the hypothesis on $\gamma_{i_{n+r}}$ implies that $B_1(r)$ is a bispecial vertex of the subshift and that for $r \in \{0,\dots,h-1\}$, $B_1(r)$ explodes like $B_2(j)$ does in Figure~\ref{figure: $j < k-1$}. However, the hypothesis on $\ell$ implies that at most the first $\ell$ vertices among $B_1(0), B_1(1), \dots$ can explode strictly before that $B_2(k-1)$ explodes. Also, the hypothesis on $\gamma_{i_{n+h}}$ implies that $B_1(h)$ explodes like $B_2(j)$ does in Figure~\ref{figure: $j = k-1$ and 3 circuits}.

Now let us exactly describe the sequence of evolution depending on $h$ and $\ell$.

When $h < \ell$, the vertex $B_1(h)$ explodes before $B_2(k-1)$. Let $k'$ be the smallest integer such that $|B_2(k')| \geq |B_1(h)|$. We obviously have $k' \leq k-1$. Then, all bispecial vertices $B_1(0), \dots, B_1(h-1), B_2(0), \dots, B_2(k'-1)$ explode and make the graph keeping type 7 or 8. Then, the explosion of $B_1(h)$ makes the graph $G_{|B_1(h)|}$ evolve as represented in\footnote{Thanks to Lemma~\ref{lemma: decomposition du mot directeur}, we can still suppose that the graph os of type 7.} Figure~\ref{7 to 9} (page~\pageref{7 to 9}) so the graph evolves to a graph of type 9 as in Figure~\ref{appendix type 9}. Then, the explosions of $B_2(k'), \dots, B_2(k-2)$ make the graph evolve as in Figure~\ref{9 to 9}. Finally, the explosion of $B_2(k-1)$ makes the graph evolve as in Figure~\ref{9 to 56}.

When $h \geq \ell$, it means that vertex $B_1(h)$ will not explode strictly before that $B_2(k-1)$ explodes. In that case, Lemma~\ref{lemma: decomposition du mot directeur} allows us to suppose that $B_1(\ell)$ explodes strictly after that $B_2(k-1)$ has exploded and, as a consequence, that so does $B_1(h)$. Consequently, vertices $B_1(0), \dots, B_1(\ell-1), B_2(0), \dots, B_2(k-2)$ explode and make graphs keeping type 7 or 8. Then, the explosion of $B_2(k-1)$ makes the graph $G_{|B_2(k-1)|}$ evolve as in Figure~\ref{7 to 9} so it evolves to a graph of type 9 as in Figure~\ref{appendix type 9}. Then, vertices $B_1(\ell), \dots, B_1(h-1)$ make graphs keeping type 9 as in Figure~\ref{9 to 9}. Finally, the explosion of $B_1(h)$ makes the graph $G_{|B_1(h)|}$ evolve as in Figure~\ref{9 to 56}. 

Now let us compute $|p_1|$ and $|p_2|$. In Figure~\ref{9 to 56}, we see that the two paths in Figure~\ref{appendix type 9} that will become $p_1$ and $p_2$ are the path from the left special vertex to the bispecial vertex and the path from the bispecial vertex to the right special vertex\footnote{Which one is $p_1$ depends on the starting vertex for the circuits.}. In Figure~\ref{9 to 9}, we also see that, while graphs keep being graphs of type 9, these paths always have the same length (because, in Figure~\ref{9 to 9}, they are paths from a left special vertex to a left special vertex and from a right special vertex to a right special vertex). Consequently, the lengths of the paths in Figure~\ref{appendix type 9} that will become $p_1$ and $p_2$ can be computed in the evolution from the last graph of type 7 to the first graph of type 9, i.e., in the evolution of $G_{|B_1(h)|}$ when $h < \ell$ and of $G_{|B_2(k-1)|}$ otherwise.

Suppose that $h$ is smaller than $\ell$. It means that $G_{|B_1(h)|}$ is a  graph of type 7 as represented in Figure~\ref{appendix type 7} where $U_{|B_1(h)|} = B_1(h)$ is the bispecial vertex. It is easily seen that in Figure~\ref{appendix type 7}, the path from the left special vertex to the right special vertex has for length 
\[
	|B_2(k')| - |B_1(h)| = |u_2| + k'(|u_2|+|v_2|) - (|u_1| + h(|u_1|+|v_1|)).
\]
We also see in Figure~\ref{7 to 9} that the path in $G_{|B_1(h)|}$ that will become $p_1$ (resp. that will become $p_2$) is the path from $B_1(h)$ to the left special vertex (resp. from the right special vertex to $B_1(h)$). Consequently, we directly have
\[
	|p_2| = |\gamma_{[0,n-1]}(2)| - \left| \CPref \left(\gamma_{[0,n-1]}(1),\gamma_{[0,n-1]}(2) \right) \right| -1.
\]
To compute, $|p_1|$, we can notice that the longest common prefix of $\theta_{i_{n-1}+1}(1)$ and $\theta_{i_{n-1}+1}(2)$ has the same length as the path starting from $B_1(h)$, going $k-1 - k'$ times through the loop with label $\lambda_R(v_2u_2)$ and ending in the right special vertex which is not $B_1(h)$. Consequently, the path from $B_1(h)$ to the left special vertex has for length
\[
	\left| \CPref\left( \gamma_{[0,n-1]}(1),\gamma_{[0,n-1]}(2) \right) \right| - (k-1 - k') (|u_2| + |v_2|) - (|B_2(k')| - |B_1(h)|)
\]
so
\begin{eqnarray*}
	|p_1| &=& \left| \CPref\left( \gamma_{[0,n-1]}(1),\gamma_{[0,n-1]}(2) \right) \right| - (k-1 - k') (|u_2| + |v_2|) \\
		  & & - (|u_2| + k'(|u_2|+|v_2|) - (|u_1| + h(|u_1|+|v_1|))) -1
\end{eqnarray*}

Now suppose that $h$ is not smaller than $\ell$. It means that $G_{|B_2(k-1)|}$ is a  graph of type 7 as represented in Figure~\ref{appendix type 7} where $U_{|B_2(k-1)|}$ is not the bispecial vertex. It is easily seen that in Figure~\ref{appendix type 7}, the path from the left special vertex to the right special vertex has for length 
\[
	|B_1(\ell)| - |B_2(k-1)| = |u_1| + \ell(|u_1|+|v_1|) - (|u_2| + (k-1)(|u_2|+|v_2|)).
\]
From what precedes, we know that the paths in $G_{|B_2(k-1)|}$ that will become $p_1$ and $p_2$ are respectively the simple path from $U_{|B_2(k-1)|}$ to $B_2(k-1)$ and the path from $B_2(k-1)$ to the left special vertex. Consequently, we directly have
\[
	|p_1| = \left| \CPref( \gamma_{[0,n-1]}(1),\gamma_{[0,n-1]}(2) ) \right| -1
\] 
and
\begin{eqnarray*}
	|p_2| 	&=& \left| \gamma_{[0,n-1]}(2) \right| - \left| \CPref( \gamma_{[0,n-1]}(1),\gamma_{[0,n-1]}(2) ) \right|  \\
			& & - (|u_1| + \ell(|u_1|+|v_1|) - (|u_2| + (k-1)(|u_2|+|v_2|))) - 1.
\end{eqnarray*}
\end{proof}

\end{document}